\documentclass[italian,english]{article}
\usepackage[T1]{fontenc}
\usepackage[latin9]{inputenc}
\usepackage{geometry}
\usepackage{color}
\usepackage{array}
\usepackage{pifont}
\usepackage{float}
\usepackage{fancybox}
\usepackage{calc}
\usepackage{mathrsfs}
\usepackage{amsmath}
\usepackage{amsthm}
\usepackage{amssymb}
\usepackage{graphicx}
\usepackage{wasysym}

\makeatletter

\providecommand{\tabularnewline}{\\}

\theoremstyle{plain}
\newtheorem{thm}{\protect\theoremname}
  \theoremstyle{plain}
  \newtheorem{prop}[thm]{\protect\propositionname}
  \theoremstyle{remark}
  \newtheorem{rem}[thm]{\protect\remarkname}
  \theoremstyle{definition}
  \newtheorem{example}[thm]{\protect\examplename}
  \theoremstyle{plain}
  \newtheorem{cor}[thm]{\protect\corollaryname}

\usepackage{ dsfont }
\usepackage[all]{xy}
\usepackage{color}
\usepackage[svgnames]{xcolor}

\usepackage{geometry}
\usepackage{graphicx}

\makeatletter
\let\@fnsymbol\@arabic
\makeatother

\AtBeginDocument{
  
}

\makeatother

\usepackage{babel}
  \addto\captionsenglish{\renewcommand{\corollaryname}{Corollary}}
  \addto\captionsenglish{\renewcommand{\examplename}{Example}}
  \addto\captionsenglish{\renewcommand{\propositionname}{Proposition}}
  \addto\captionsenglish{\renewcommand{\remarkname}{Remark}}
  \addto\captionsenglish{\renewcommand{\theoremname}{Theorem}}
  \addto\captionsitalian{\renewcommand{\corollaryname}{Corollario}}
  \addto\captionsitalian{\renewcommand{\examplename}{Esempio}}
  \addto\captionsitalian{\renewcommand{\propositionname}{Proposizione}}
  \addto\captionsitalian{\renewcommand{\remarkname}{Osservazione}}
  \addto\captionsitalian{\renewcommand{\theoremname}{Teorema}}
  \providecommand{\corollaryname}{Corollary}
  \providecommand{\examplename}{Example}
  \providecommand{\propositionname}{Proposition}
  \providecommand{\remarkname}{Remark}
\providecommand{\theoremname}{Theorem}

\begin{document}
\global\long\def\u{u}
\global\long\def\p{P}
\global\long\def\X{X}
\global\long\def\me{\mu_{e}}
\global\long\def\sym{\textrm{sym}}
\global\long\def\grad{\nabla}
\global\long\def\le{\lambda_{e}}
\global\long\def\tr{\textrm{tr}}
\global\long\def\mc{\mu_{c}}
\global\long\def\skew{\textrm{skew}}
\global\long\def\curl{\textrm{Curl}}
\global\long\def\ac{\alpha_{c}}
\global\long\def\mh{\mu_{\mathrm{mic}}}
\global\long\def\lh{\lambda_{\textrm{mic}}}
\global\long\def\B{\mathscr{B}}
\global\long\def\R{\mathbb{R}}
\global\long\def\fr{\rightarrow}
\global\long\def\Q{\mathcal{Q}}
\global\long\def\A{\mathscr{A}}
\global\long\def\L{\mathscr{L}}
\global\long\def\D{\mathscr{D}}
\foreignlanguage{italian}{}\global\long\def\id{\mathds{1}}
\global\long\def\ds{\textrm{dev sym}}
\global\long\def\sph{\textrm{sph}}
\global\long\def\eg{\boldsymbol{\varepsilon}}
\global\long\def\aa{\boldsymbol{\alpha}}
\global\long\def\o{\boldsymbol{\omega}}
\global\long\def\ege{\boldsymbol{\varepsilon}_{e}}
\global\long\def\egp{\boldsymbol{\varepsilon}_{p}}
\global\long\def\punto{\,.\,}
\global\long\def\so{\mathfrak{so}\left(3\right)}
\global\long\def\Sym{\textrm{Sym}\left(3\right)}
\global\long\def\MM{\boldsymbol{\mathfrak{M}}}
\global\long\def\C{\mathbb{C}}
\global\long\def\gl{\mathfrak{gl}\left(3\right)}
\global\long\def\P{P}
\global\long\def\Lin{\textrm{Lin}}
\global\long\def\D{\boldsymbol{D}}
\global\long\def\a{\alpha}
\global\long\def\b{\beta}
\global\long\def\lle{\lambda_{e}}
\global\long\def\ce{\mathbb{C}_{e}}
\global\long\def\cm{\mathbb{C}_{\mathrm{micro}}}
\global\long\def\cc{\mathbb{C}_{c}}
\global\long\def\axl{\textrm{axl}}
\global\long\def\dev{\textrm{dev}}
\global\long\def\Ls{\widehat{\mathbb{L}}}
\global\long\def\mum{\mu_{\mathrm{macro}}}
\global\long\def\lam{\lambda_{\mathrm{macro}}}
\global\long\def\mh{\mu_{\mathrm{micro}}}
\global\long\def\lh{\lambda_{\textrm{micro}}}
\global\long\def\vau{\omega_{1}^{int}}
\global\long\def\vad{\omega_{2}^{int}}
\global\long\def\axl{\textrm{axl}}

\title{A panorama of dispersion curves for the weighted isotropic relaxed
micromorphic model}

\author{Marco Valerio d\textquoteright Agostino\thanks{Marco Valerio d'Agostino, corresponding author, marco-valerio.dagostino@insa-lyon.fr,
LGCIE, INSA-Lyon, Université de Lyon, 20 avenue Albert Einstein, 69621,
Villeurbanne cedex, France} \,and Gabriele Barbagallo\thanks{Gabriele Barbagallo, gabriele.barbagallo@insa-lyon.fr, LaMCoS-CNRS
\& LGCIE, INSA-Lyon, Universitité de Lyon, 20 avenue Albert Einstein,
69621, Villeurbanne cedex, France} \,and Ionel-Dumitrel Ghiba\thanks{Ionel-Dumitrel Ghiba, dumitrel.ghiba@uaic.ro, Lehrstuhl für Nichtlineare
Analysis und Modellierung, Fakultät für Mathematik, Universität Duisburg-Essen,
Thea-Leymann Str. 9, 45127 Essen, Germany; Alexandru Ioan Cuza University
of Ia\c{s}i, Department of Mathematics, Blvd. Carol I, no. 11, 700506
Ia\c{s}i, Romania; and Octav Mayer Institute of Mathematics of the
Romanian Academy, Ia\c{s}i Branch, 700505 Ia\c{s}i, Romania} \\
and Angela Madeo\thanks{Angela Madeo, angela.madeo@insa-lyon.fr, LGCIE, INSA-Lyon, Université
de Lyon, 20 avenue Albert Einstein, 69621, Villeurbanne cedex, France} $\;$and Patrizio Neff\,\thanks{Patrizio Neff, corresponding author, patrizio.neff@uni-due.de, Head
of Chair for Nonlinear Analysis and Modelling, Fakultät für Mathematik,
Universität Duisburg-Essen, Mathematik-Carrée, Thea-Leymann-Straße
9, 45127 Essen, Germany}\textit{}\linebreak{}
\textit{}\\
\textit{}\\
\textit{Dedicated to David J. Steigmann on the occasion of his 60th
birthday}}
\maketitle
\begin{abstract}
We consider the weighted isotropic relaxed micromorphic model and
provide an in depth investigation of the characteristic dispersion
curves when the constitutive parameters of the model are varied. The
weighted relaxed micromorphic model generalizes the classical relaxed
micromorphic model previously introduced by the authors, since it
features the Cartan-Lie decomposition of the tensors $\P_{,t}$ and
$\curl\,\P$ in their $\dev\,\sym$, $\skew$ and spheric part. It
is shown that the split of the tensor $\P_{,t}$ in the micro-inertia
provide an independent control of the cut-offs of the optic banches.
This is crucial for the future calibration of the relaxed micromorphic
model on real band-gap metamaterials.

Even if the physical interest of the introduction of the split of
the tensor $\curl\,\P$ is less evident than in the previous case,
we discuss in detail which is its effect on the dispersion curves.
Finally, we also provide a complete parametric study involving all
the constitutive parameters of the introduced model, so giving rise
to an exhaustive panorama of dispersion curves for the relaxed micromorphic
model.
\end{abstract}
\addtocounter{footnote}{5} \vspace{6mm}

\vspace{2mm}
\textbf{}

\newpage{}

\tableofcontents{}

\section*{Introduction}

The micromorphic framework is increasingly used as an algorithmic
device to regularize gradient-elasticity or gradient plasticity models
(see e.g. \cite{forest2009micromorphic,forest2016nonlinear}). In
these cases, the problem of understanding the genuine physical meaning
which can be associated to micromorphic models does not arise, since
the micromorphic framework is simply used as a tool for the regularization
of higher order models. With a completely different perspective, in
a series of works \cite{madeo2015band,madeo2015wave,madeo2016first,madeo2016reflection},
we started looking for real situations in which micromorphic models
can be used to properly convey important physical informations to
the modeling of the actual mechanical behavior of some microstructured
materials. More particularly, we focused our attention on the newly
introduced relaxed micromorphic model\footnote{We use the term \textbf{relaxed} in its proper english meaning and
not in the sense of finding the lower semi-continuous hull. Indeed,
the relaxed micromorphic continuum is always lower semi-continuous,
but, contrary to the classical micromorphic model, the assumption
on the constitutive coefficients are much weakened (relaxed). Notably,
constraining the micro-distortion $P=\nabla u$ does not lead to a
second-gradient model but leads back to classical linear elasticity
without characteristic length scale.} (see \cite{barbagallo2016transparent,ghiba2015relaxed,neff2004material,neff2006existence,neff2014existence,neff2014unifying,neff2015relaxed})
to investigate the unorthodox dynamical properties of band-gap metamaterials,
i.e. microstructured materials which are able to inhibit wave propagation
in precise frequency ranges. Similarly to the classical micromorphic
models originally introduced by Mindlin and Eringen \cite{mindlin1964micro,eringen2012microcontinuum},
the relaxed micromorphic model features an enriched kinematics to
account e.g. for microscopic motions in the interior of the considered
macroscopic continuum. Additionally to the classical macroscopic \textit{displacement}
vector field $u(x,t)$, the micromorphic models typically introduce
supplementary, microstructure-related, degrees of freedom by means
of a second order tensor field $P(x,t)$ which is known as \textit{micro-distortion}
tensor. The relaxed micromorphic model differs from more classical
micromorphic ones in the sense that the higher order space derivatives
of the field $P$ are constitutively introduced in the strain energy
density not through the whole gradient of $P$, but only through its
$\mathrm{Curl}$. The fact of using the $\mathrm{Curl}$ of the micro-distortion
tensor is rather common when dealing with dislocation based gradient
plasticity (see e.g. \cite{cordero2010size,berdichevskii1967dynamic,birsan2016dislocation,claus1969three,claus1971dislocation,ebobisse2010existence,ebobisse2015existence,ebobisse2016canonical,eringen1969micromorphic,kroner1discussion,neff2009notes,nesenenko2012well,popov1994dynamics,svendsen2009constitutive}),
but this is indeed not the case when considering pure elasticity in
which the standard formulations commonly introduce the whole gradient
$\nabla\,P$ of the micro-distortion tensor $P$. As a matter of fact,
the use of micromorphic models which only consider the $\mathrm{Curl}$
of the micro-distortion in a purely linear-elastic framework can shed
light on the modeling of non-local metamaterials which exhibit band-gap
behaviors \cite{madeo2015band,madeo2015wave,madeo2016first,madeo2016reflection,madeo2016complete,madeo2016review,madeo2016role}.

In the present work, we provide a generalization of the isotropic
relaxed micromorphic model used in \cite{madeo2015band,madeo2015wave,madeo2016first,madeo2016reflection}
based on the Cartan-Lie decomposition of the micro-distortion tensor
$P$ and of its $\mathrm{Curl}$. Such decomposition allows us to
introduce, in the isotropic setting, three parameters for the micro-inertia
and three internal length associated to the space derivatives of $P$
appearing through $\curl\,\P$. If the physical meaning of the three
micro-inertia parameters may be rather intuitively related to a distinction
of the weights attributed to the distortional, rotational and volumetric
expansion vibration modes at the level of the unit cell, a clear interpretation
of the introduction of three different characteristic lengths is less
immediate. In the view of applications, we will be able to show in
the short term whether it is worth introducing three different micro-inertia
parameters for real band-gap metamaterials. The phenomenological interest
of the actual distinction of the non-localities associated to three
different internal lengths will be also investigated in further works.

In this paper, we discuss the effect that the introduced split of
the micro-inertia and of the internal lengths has on the dispersion
curves of the considered relaxed micromorphic model. We present and
discuss in detail the specific effects that the micro-inertia parameters
and the characteristic lengths have on the characteristic of the dispersion
curves, in general, and of the band-gaps, in particular. The split
on the micro-inertia is found to be fundamental for the description
of real metamaterials, since it gives the possibility of controlling
separately the cut-offs of the optic curves in the dispersion diagram.

We obtain the previously introduced results \cite{madeo2015wave}
with a unique micro-inertia parameter and internal length as a suitable
limiting case of the more general model presented here. We then focus
our attention on another particular limiting case that is the one
with vanishing internal lengths. Such particular case of the relaxed
micromorphic model in which no derivatives of the micro-distortion
tensor appear can be called as an ``internal variable model'' (in
the Cosserat framework this approach has been named ``reduced Cosserat
model'', see \cite{kulesh2009problem,grekova2012linear}) and may
be of interest for the description of some band gap metamaterials
for which the so-called hypothesis of separation of scales is verified
(see e.g. \cite{sridhar2016homogenization,pham2013transient}).

For all the proposed cases, we show the direct effect of the variation
of any single parameter on the dispersion curves and on the band gap
characteristics. This paper is now organized as follows: 
\begin{itemize}
\item in chapter 1 we introduce the notations used in the paper, 
\item in chapter 2 we present the weighted relaxed micromorphic model in
the unbounded domain $\R^{3}$ in a variational form and we derive
the PDEs governing the system,
\item in chapter 3 we show how it is possible to recover the classical linear
elasticity model from the relaxed micromorphic model,
\item in chapter 4 we introduce the plane wave ansatz on the unknown kinematical
fields in order to show how it is possible to reduce the system of
governing PDEs to an algebraic problem, finding also the dispersion
curves.
\item in chapter 5 we perform a parametric study on the influence of the
material parameters on the behavior of the dispersion curves.
\end{itemize}

\section{Notation}

Throughout this paper the Einstein convention of sum over repeated
indexes is used if not differently specified. We denote by $\R^{3\times3}$
the set of real $3\times3$ second order tensors and by $\R^{3\times3\times3}$
the set of real $3\times3\times3$ third order tensors. The standard
Euclidean scalar product on $\R^{3\times3}$ is given by $\left\langle X,Y\right\rangle {}_{\R^{3\times3}}=\tr(X\cdot Y^{T})$
and, thus, the Frobenius tensor norm is $\|X\|^{2}=\left\langle X,X\right\rangle {}_{\R^{3\times3}}$.
Moreover, the identity tensor on $\R^{3\times3}$ will be denoted
by $\mathds{1}$, so that $\tr(X)=\left\langle X,\mathds{1}\right\rangle $.
We adopt the usual abbreviations of Lie-algebra theory, i.e.: 
\begin{itemize}
\item $\Sym:=\{X\in\R^{3\times3}\;|X^{T}=X\}$ denotes the vector-space
of all symmetric $3\times3$ matrices 
\item $\so:=\{X\in\R^{3\times3}\;|X^{T}=-X\}$ is the Lie-algebra of skew
symmetric tensors 
\item $\mathfrak{sl}(3):=\{X\in\R^{3\times3}\;|\tr(X)=0\}$ is the Lie-algebra
of traceless tensors 
\item $\R^{3\times3}\simeq\mathfrak{gl}(3)=\{\mathfrak{sl}(3)\cap\Sym\}\oplus\so\oplus\R\!\cdot\!\mathds{1}$
is the \emph{orthogonal Cartan-decomposition of the Lie-algebra} 
\end{itemize}
In other words, for all $X\in\R^{3\times3}$, we consider the decomposition
\begin{align}
X=\ds X+\skew X+\frac{1}{3}\mathrm{tr}(X)\,\mathds{1}\label{eq:cartan lie}
\end{align}
where: 
\begin{itemize}
\item $\sym\,X=\frac{1}{2}(X^{T}+X)\in\Sym$ is the symmetric part of $X$, 
\item $\skew\,X=\frac{1}{2}(X-X^{T})\in\so$ is the skew-symmetric part
of $X$, 
\item $\dev\,X=X-\frac{1}{3}\tr(X)\,\mathds{1}\in\mathfrak{sl}(3)$ is the
deviatoric part of $X$. 
\end{itemize}
Throughout this paper we denote: 
\begin{itemize}
\item the sixth order tensors $\Ls:\R^{3\times3\times3}\rightarrow\R^{3\times3\times3}$,
by a hat,
\item the fourth order tensors $\overline{\C}:\R^{3\times3}\rightarrow\R^{3\times3}$,
by an overline,
\item without superscripts, the classical fourth order tensors acting only
on symmetric matrices \\
 $\C:\Sym\rightarrow\Sym$ or skew-symmetric ones $\cc:\so\rightarrow\so$
,
\item the second order tensors $\widetilde{\cc}:\R^{3}\rightarrow\R^{3}$
appearing as elastic stiffness, by a tilde. 
\end{itemize}
We denote by $\overline{\C}\:X$ the linear application of a $4^{th}$
order tensor to a $2^{nd}$ order tensor and also for the linear application
of a $6^{th}$ order tensor $\Ls$ to a $3^{rd}$ order tensor. In
symbols, we have: 
\begin{align}
\left(\overline{\C}\:X\right)_{ij}=\overline{\C}_{ijhk}X_{hk}\,,\qquad\left(\Ls\:A\right)_{ijh}=\Ls_{ijhpqr}A_{pqr}\,.
\end{align}
The operation of simple contraction between tensors of suitable order
is denoted by a central dot as, for example: 
\begin{align}
\left(\widetilde{\C}\cdot v\right)_{i}=\widetilde{\C}_{ij}v_{j}\,,\qquad\left(\widetilde{\C}\cdot X\right)_{ij}=\widetilde{\C}_{ih}X_{hj}\,.
\end{align}

\noindent Typical conventions for differential operations are implied
such as a comma followed by a subscript to denote the partial derivative
with respect to the corresponding Cartesian coordinate, i.e. $\left(\cdot\right)_{,j}=\frac{\partial(\cdot)}{\partial x_{j}}$.

\noindent The $\textrm{curl}$ of a vector field $v$ is defined as\footnote{Given a third order tensors $A$ and a second order tensor $B$, the
double contraction $A:B$ is defined as $\left(A:B\right)_{i}=A_{ijk}B_{kj}$. }
\[
\left(\textrm{curl}\,v\right)_{i}=\varepsilon_{ijk}v_{k,j},
\]
where $\varepsilon_{ijk}$ is the Levi-Civita third order permutation
tensor. Let $X$ be a two order tensor field and $X_{1},X_{2},X_{3}$
three vector fields such that
\[
X=\begin{pmatrix}X_{1}^{T}\\
X_{2}^{T}\\
X_{3}^{T}
\end{pmatrix}.
\]

\noindent The $\curl$ of $X$ is defined as follows:
\[
\curl\,X=\begin{pmatrix}\left(\textrm{curl}\,X_{1}\right)^{T}\\
\left(\textrm{curl}\,X_{2}\right)^{T}\\
\left(\textrm{curl}\,X_{3}\right)^{T}
\end{pmatrix},
\]
that in indices is 
\[
\left(\curl\,X\right)_{ij}=\varepsilon_{jmn}X_{in,m}.
\]
 For the iterated $\curl$ we find
\begin{align*}
\left(\curl\:\curl\:P\right)_{ij} & =\varepsilon_{jmn}\left(\curl\,P\right)_{in,m}=\varepsilon_{jmn}\left(\varepsilon_{nab}P_{ib,a}\right)_{,m}=\varepsilon_{jmn}\varepsilon_{nab}P_{ib,am}\\
 & =-\,\varepsilon_{nmj}\varepsilon_{nab}P_{ib,am}=\left(\delta_{ma}\delta_{jb}-\delta_{mb}\delta_{ja}\right)P_{ib,am}=P_{im,jm}-P_{ij,mm}.
\end{align*}
The divergence $\textrm{div}\,v$ of a vector field $v$ is defined
as $\textrm{div}\,v=v_{i,i}$ and the divergence $\textrm{Div}\,X$
of a matrix $X$ as
\[
\textrm{Div}\,X=\begin{pmatrix}\textrm{div}\,X_{1}\\
\textrm{div}\,X_{2}\\
\textrm{div}\,X_{3}
\end{pmatrix}=\begin{pmatrix}\left(X_{1}\right)_{i,i}\\
\left(X_{2}\right)_{i,i}\\
\left(X_{3}\right)_{i,i}
\end{pmatrix}.
\]
Given two differentiable vector fields $u,v:\Omega\subseteq\R^{3}\fr\R^{3},$
we have that
\begin{equation}
\textrm{div}\left(u\times v\right)=\left\langle \textrm{curl}\,u,v\right\rangle -\left\langle u,\textrm{curl}\,v\right\rangle ,\label{eq:id div}
\end{equation}
since
\begin{align*}
\left(\varepsilon_{ijk}u_{j}v_{k}\right)_{,i} & =\varepsilon_{ijk}u_{j,i}v_{k}+\varepsilon_{ijk}u_{j}v_{k,i}=\varepsilon_{kij}u_{j,i}v_{k}-u_{j}\varepsilon_{jik}v_{k,i}\\
 & =\left\langle \textrm{curl}\,u,v\right\rangle -\left\langle u,\textrm{curl}\,v\right\rangle .
\end{align*}

\noindent 

\section{Variational formulation of the relaxed model}

The kinematical fields of the problem are the displacement $u$ and
the micro-distortion tensor field $P$:
\begin{gather*}
u:\overline{\Omega}\times I\fr\R^{3},\quad\left(x,t\right)\mapsto\u\left(x,t\right),\qquad P:\overline{\Omega}\times I\fr\R^{3\times3},\quad\left(x,t\right)\mapsto P\left(x,t\right),
\end{gather*}
where $\Omega$ is an open bounded domain in $\R^{3}$ with a piecewise
smooth boundary $\partial\Omega$ and closure $\overline{\Omega}$,
and $I=\left[0,T\right]\subseteq\R$ is the time interval. The mechanical
model is formulated in the variational context. This means that we
consider an action functional on an appropriate function-space. Setting
$\Omega_{0}=\Omega\times\left\{ 0\right\} $, the space of configurations
of the problem is 
\[
\mathcal{Q}:=\left\{ \left(\u,\P\right)\in\mathscr{C}^{1}\left(\overline{\Omega}\times I,\R^{3}\right)\times\mathscr{C}^{1}\left(\overline{\Omega}\times I,\R^{3\times3}\right):\left(\u,\P\right)\textrm{ verifies conditions }\left(\mathsf{B}_{1}\right)\textrm{ and }\left(\mathsf{B}_{2}\right)\right\} 
\]
where
\begin{itemize}
\item $\left(\mathsf{B}_{1}\right)$ are the boundary conditions $u\left(x,t\right)=\varphi\left(x,t\right)$
and $P_{i}\left(x,t\right)\times n=\psi_{i}\left(x,t\right)$, $i=1,2,3$,
$\quad\left(x,t\right)\in\partial\Omega\times\left[0,T\right]$, where
$n$ is the unit outward normal vector on $\partial\Omega\times\left[0,T\right]$,
$P_{i},\,i=1,2,3$ are the rows of $P$ and $\varphi,\psi_{i}$ are
prescribed functions,
\item $\left(\mathsf{B}_{2}\right)$ are the initial conditions $\left.\u\right|_{\Omega_{0}}=\u_{0},\left.\u_{,t}\right|_{\Omega_{0}}=\underline{\u}_{0},\left.P\right|_{\Omega_{0}}=P_{0},\left.P_{,t}\right|_{\Omega_{0}}=\underline{P}_{0}\;\textrm{in }\Omega_{0}$,
where $u_{0}\left(x\right),\underline{\u}_{0}\left(x\right),$ $P_{0}\left(x\right),\underline{P}_{0}\left(x\right)$
are prescribed functions.
\end{itemize}
The action functional $\mathscr{A}:\mathcal{Q}\fr\R,$ is the sum
of the internal and external action functionals $\mathscr{A}_{\mathscr{L}}^{int},\mathscr{A}^{ext}:\mathcal{Q}\fr\R$
defined as follows
\begin{align}
\mathscr{A}_{\mathscr{L}}^{int}\left[\left(\u,\P\right)\right] & :=\int_{I}\int_{\Omega}\L\left(\u_{,t},\P_{,t},\nabla\u,\P,\curl\,P\right)dv\,dt,\\
\mathscr{A}^{ext}\left[\left(\u,\P\right)\right] & :=\int_{I}\int_{\Omega}\left(\left\langle f^{ext},u\right\rangle +\left\langle M^{ext},P\right\rangle \right)dv\,dt,\nonumber 
\end{align}
where $\L$ is the Lagrangian density of the system and $f^{ext},M^{ext}$
are the body force and double body force. In this work we will consider
$f^{ext}=0,M^{ext}=0$. In order to find the stationary points of
the action functional, we have to calculate its first variation:
\begin{gather*}
\delta\mathscr{A}=\delta\mathscr{A}_{\L}^{int}=\delta\int_{I}\int_{\Omega}\L\left(\u_{,t},\P_{,t},\nabla\u,\P,\curl\,P\right)dv\,dt.
\end{gather*}
Well-posedness of this variational problem (existence, uniqueness
and stability of solution) has been proved in \cite{ghiba2015relaxed,neff2014unifying,neff2015relaxed}.

\subsection{Constitutive assumptions on the energy density and equations of motion
in strong form}

For the Lagrangian energy density we assume the standard split in
kinetic minus potential energy:

\[
\mathscr{L}\left(\u_{,t},\P_{,t},\nabla\u,\P,\curl\,\p\right)=J\left(\u_{,t},\P_{,t}\right)-W\left(\nabla\u,\P,\curl\,\p\right),
\]
 In general anisotropic linear elastic micromorphic media, as clearly
stated in \cite{barbagallo2016transparent,neff2014unifying}, we have
that the kinetic energy density and the potential have the following
expression
\begin{align*}
J\left(\u_{,t},\P_{,t}\right) & =\frac{1}{2}\left\langle \rho\,u_{,t},u_{,t}\right\rangle +\frac{1}{2}\left\langle \overline{\mathbb{J}}\:\P_{,t},\P_{,t}\right\rangle \\
W\left(\nabla\u,\P,\curl\,\p\right) & =\underbrace{\frac{1}{2}\left\langle \mathbb{C}_{e}\,\sym\left(\grad\u-P\right),\sym\left(\grad\u-P\right)\right\rangle _{\R^{3\times3}}}_{\textrm{anisotropic elastic - energy}}+\underbrace{\frac{1}{2}\left\langle \cm\,\sym\,P,\sym\,P\right\rangle _{\R^{3\times3}}}_{\textrm{micro - self - energy}}\\
 & \qquad+\underbrace{\frac{1}{2}\left\langle \cc\:\skew\left(\grad\u-P\right),\skew\left(\grad\u-P\right)\right\rangle _{\R^{3\times3}}}_{\textrm{invariant local anisotropic rotational elastic coupling}}+\underbrace{\mu\:\frac{L_{c}^{2}}{2}\left\langle \overline{\mathbb{L}}_{\textrm{aniso}}\,\curl\,P,\curl\,P\right\rangle _{\R^{3\times3}}}_{\textrm{ curvature}},
\end{align*}
where
\[
\begin{cases}
\rho:\Omega\fr\R^{+} & \textrm{is the macro-inertia density},\\
\overline{\mathbb{J}}:\R^{3\times3}\fr\R^{3\times3} & \textrm{is the \ensuremath{4^{th}}order micro-inertia density tensor},\\
\C_{e},\cm:\Sym\fr\Sym & \textrm{are the \ensuremath{4^{th}}order elasticity tensors with 21 independent components, }\\
\cc:\so\fr\so & \textrm{is a dimensionless \ensuremath{4^{th}} order tensor with 6 independent components},\\
\overline{\mathbb{L}}_{\textrm{aniso}}:\R^{3\times3}\fr\R^{3\times3} & \textrm{is a dimensionless \ensuremath{4^{th}} order tensor with almost 45 independent components},
\end{cases}
\]
and $L_{c}$ is the characteristic length of the relaxed micromorphic
model. We demand that the bilinear forms induced by $\overline{\mathbb{J}},\C_{e},\cm,\overline{\mathbb{L}}_{\textrm{aniso}}$
are positive definite, 
\begin{gather}
\exists\,c^{+},c_{e}^{+},c_{\textrm{micro}}^{+},c_{l}^{+}>0:\,\forall S\in\text{Sym}(3)\quad\begin{cases}
\left\langle \overline{\mathbb{J}}\:S,S\right\rangle {}_{\R^{3\times3}}\geq c^{+}\|S\|_{\R^{3\times3}}^{2},\\
\left\langle \ce\:S,S\right\rangle {}_{\R^{3\times3}}\geq c_{e}^{+}\|S\|_{\R^{3\times3}}^{2},\\
\left\langle \cm\:S,S\right\rangle {}_{\R^{3\times3}}\geq c_{\textrm{micro}}^{+}\|S\|_{\R^{3\times3}}^{2},\\
\left\langle \overline{\mathbb{L}}_{\textrm{aniso}}\:S,S\right\rangle {}_{\R^{3\times3}}\geq c_{\textrm{l}}^{+}\|S\|_{\R^{3\times3}}^{2},
\end{cases}
\end{gather}
and, in sharp contrast to the Mindlin-Eringen format, that the bilinear
form induced by $\cc$ is only positive semi-definite, i.e.\footnote{It is in virtue of such weakening of the theoretical framework needed
to prove its well posedness that the word ``relaxed'' was chosen
to distinguish the relaxed micromorphic model from Mindlin's one (see
\cite{barbagallo2016transparent,ghiba2015relaxed,neff2004material,neff2006existence,neff2014existence,neff2014unifying,neff2015relaxed}).}
\begin{equation}
\forall\overline{A}\in\so:\quad\left\langle \cc\:\overline{A},\overline{A}\right\rangle {}_{\R^{3\times3}}\geq0.
\end{equation}

In this work we introduce the hypothesis according to which the micromorphic
medium is \textbf{homogeneous and isotropic. }This leads to the following
particular expression for the kinetic and strain energy densities:
\begin{align}
J\left(\u_{,t},\P_{,t}\right) & \:{\displaystyle =\frac{1}{2}\,\rho\left\Vert \u_{,t}\right\Vert ^{2}+\frac{1}{2}\left(\eta_{\,1}\left\Vert \textrm{dev sym}\,\P_{,t}\right\Vert ^{2}+\eta_{\,2}\left\Vert \textrm{skew}\,\P_{,t}\right\Vert ^{2}+\frac{1}{3}\,\eta_{\,3}\left(\textrm{tr}\,\P_{,t}\right)^{2}\right)},\nonumber \\
W\left(\nabla\u,\P,\curl\,\p\right) & =\underbrace{{\displaystyle \me\left\Vert \sym\left(\grad\u-P\right)\right\Vert ^{2}+\frac{\le}{2}\left(\textrm{tr}\left(\grad\u-P\right)\right)^{2}+\mh\left\Vert \sym\,P\right\Vert ^{2}+\frac{\lh}{2}\left(\tr\,P\right)^{2}}{\displaystyle +\,\mc\left\Vert \skew\left(\grad\u-\P\right)\right\Vert ^{2}}}_{\mathsf{A}}\nonumber \\
 & \qquad+\underbrace{\mu_{e}\,\frac{L_{c}^{2}}{2}\left(\alpha_{1}\left\Vert \ds\,\curl\,\P\right\Vert ^{2}+\alpha_{2}\left\Vert \skew\,\curl\,\P\right\Vert ^{2}+\frac{1}{3}\,\alpha_{3}\left(\tr\,\curl\,\P\right)^{2}\right)}_{\mathsf{B}},\label{eq:lagrangian}
\end{align}
where $\rho$ is the macroscopic mass density, $L_{c}$ is the internal
length accounting for non-local effects, $\mc$ is the Cosserat couple
modulus, $\me,\le,\mh,\lh$ are the other elastic parameters featured
by the isotropic relaxed micromorphic model (see \cite{neff2014unifying}),
$\eta_{1},\eta_{2},\eta_{3}$ are the inertia weights and $\alpha_{1},\alpha_{2},\alpha_{3}$
are dimensionless parameters. It can be seen that the two tensor fields
$\P_{,t}$ and $\curl\,\P$ have been decomposed according to the
Cartan-Lie decomposition. Since the part $\mathsf{A}$ of the potential
energy is the same as in \cite{madeo2015wave}, in order to compute
the first variation of the action functional it is sufficient to evaluate
only the first variation of the kinetic energy and the second part
$\mathsf{B}$ of the potential energy.

We explicitly remark that the chosen expression for the micro-inertia
in terms of $\eta_{1},\eta_{2}$ and $\eta_{3}$ is more general than
the one introduced in \cite{madeo2015wave}. The same holds for the
non-local term in which the three constants $\alpha_{1},\alpha_{2}$
and $\alpha_{3}$ appear. A crucial point for further experimentally
oriented works will be the split of the kinetic energy that we introduce
here. Indeed, the fact of introducing three micro-inertia parameters
instead of one allows extra freedom for the fitting of the dispersion
curves on real band-gap metamaterials.

The particular case of the relaxed micromorphic model presented in
\cite{madeo2015wave} can be obtained by simply setting $\eta_{1}=\eta_{2}=\eta_{3}=10^{-2}\,\textrm{Kg}/\textrm{m}$,
and $\alpha_{1}=\alpha_{2}=\alpha_{3}=1$. The weights $\alpha_{1},\alpha_{2}$
and $\alpha_{3}$ allow to account for a refined splitting of the
non-localities present in the considered relaxed micromorphic model.
This possibility provides a certain freedom for future developments,
but it is too general to provide new physical understanding of band-gap
metamaterials currently studied. In fact, the most common band-gap
metamaterials are conceived letting non-local effects being very small
based on some sort of ``separation of scales'' hypothesis (see e.g.
\cite{pham2013transient,sridhar2016homogenization}). This means it
is sensible that, for such metamaterials, non-local effects may be
described by means of a unique characteristic length (case $\alpha_{1}=\alpha_{2}=\alpha_{3}=1$).
Nevertheless, the weighted higher-order terms presented here may allow
for more detailed descriptions of non-localities in new metamaterials
in which strong contrasts of the mechanical properties at the micro-level
occur.

The question is quite different for the isotropic weighted expression
of the micro-inertia which introduces the 3 parameters $\eta_{1},\eta_{2}$
and $\eta_{3}$. It is indeed sensible that, for some metamaterials,
the vibrations associated to distortion, rotation and volumetric expansion
of the unit cells at the micro-level do not occur with the same facility.
In other words, the three different modes might be more or less privileged
depending on the considered metamaterial.

The real interest of the presented micro-inertia splitting must be
tested by fitting the proposed relaxed micromorphic model on real
experiments on existing band-gap metamaterials. We leave this task
to a forthcoming paper, limiting ourselves here to discuss numerical
results which may be of interest for conceiving pertinent experimental
campaigns. 

We have shown elsewhere \cite{ghiba2015relaxed,neff2006existence,neff2014existence,neff2015relaxed},
that the static and dynamic problem in a bounded domain is well-posed
(existence and uniqueness) under the general assumptions on the constitutive
coefficients:
\begin{equation}
\begin{array}{ccccccccc}
3\,\le+2\,\me>0, &  & \me>0, &  & \mh>0, &  & 3\,\lh+2\,\mh>0, &  & \overline{\mathbb{J}}\textrm{ is positive definite},\\
\\
\rho>0, &  & \mc\geq0, &  & L_{c}>0 &  & \textrm{and} &  & \alpha_{1},\alpha_{2}>0,\,\alpha_{3}\geq0.
\end{array}\label{eq:J}
\end{equation}
Currently, it is not known whether assuming only
\begin{equation}
\left(\alpha_{1},\,\alpha_{2}>0,\,\alpha_{3}\geq0\right)\qquad\textrm{or}\qquad\left(\alpha_{1}>0,\,\alpha_{2},\,\alpha_{3}\geq0\right)\label{eq:conditions}
\end{equation}
is sufficient for well-posedness of the initial boundary value problem.
In our parametric study of the whole-space harmonic wave propagation
problem (\ref{eq:PDE system}), we will re-encounter the limit case
(\ref{eq:conditions}) showing no deficiency.

\bigskip{}

It is straightforward to derive (with the stronger regularity for
the kinematical fields $\left(u,P\right)\in\mathscr{C}^{2}\left(\overline{\Omega}\times I,\R^{3}\right)\times\mathscr{C}^{2}\left(\overline{\Omega}\times I,\R^{3\times3}\right)$)
the Euler-Lagrange equations corresponding to the Lagrangian associated
with the strain energy and kinetic energies (\ref{eq:lagrangian})
which, after projection on the orthogonal subspaces in (\ref{eq:cartan lie}),
read\footnote{The new calculations concerning the variation of the term $\mathsf{B}$
in (\ref{eq:lagrangian}) are presented in Appendix 1}:

\begin{figure}[H]
\centering{}%
\noindent\doublebox{\begin{minipage}[t]{1\columnwidth - 2\fboxsep - 7.5\fboxrule - 1pt}%
\smallskip{}

\begin{align}
\rho\,\u_{,tt} & =\textrm{Div}\left[2\,\me\,\sym\left(\grad\u-\P\right)+\le\,\tr\left(\grad\u-\P\right)\id+2\,\mc\,\skew\left(\grad\u-\P\right)\right],\nonumber \\
\nonumber \\
\eta_{\,1}\,\textrm{dev sym}\,\P_{,tt} & =2\,\me\,\ds\left(\grad\u-\P\right)-2\,\mh\,\ds\,\P\nonumber \\
 & \quad-\me\,L_{c}^{2}\,\ds\left(\alpha_{1}\,\curl\,\ds\,\curl\,\P+\alpha_{2}\,\curl\,\skew\,\curl\,\P+\frac{\alpha_{3}}{3}\,\curl\left(\tr\left(\curl\,P\right)\id\right)\right),\nonumber \\
\eta_{\,2}\,\textrm{skew}\,\P_{,tt} & =2\,\mc\,\skew\left(\grad\u-\P\right)\label{eq:PDE system}\\
 & \quad-\me\,L_{c}^{2}\,\skew\left(\alpha_{1}\,\curl\,\ds\,\curl\,\P+\alpha_{2}\,\curl\,\skew\,\curl\,\P+\frac{\alpha_{3}}{3}\,\curl\left(\tr\left(\curl\,P\right)\id\right)\right),\nonumber \\
\frac{1}{3}\,\eta_{\,3}\,\textrm{tr}\left(\P_{,tt}\right) & {\displaystyle \,\,=\left(\frac{2}{3}\,\me+\le\right)\,\tr\left(\grad\u-\P\right)-\left(\frac{2}{3}\,\mh+\lh\right)\,\tr\left(\P\right)}\nonumber \\
 & \quad-\me\,L_{c}^{2}\,\frac{1}{3}\,\tr\left(\alpha_{1}\,\curl\,\ds\,\curl\,\P+\alpha_{2}\,\curl\,\skew\,\curl\,\P+\frac{\alpha_{3}}{3}\,\curl\left(\tr\left(\curl\,P\right)\id\right)\right).\nonumber 
\end{align}

\smallskip{}
\end{minipage}}
\end{figure}

\subsection{Internal variable model}

The internal variable model can be easily obtained as a particular
case of the relaxed model simply setting the three parameters $\alpha_{1},\alpha_{2},\alpha_{3}$
to be simultaneously equal to zero and so setting to zero the energetic
part linked to the derivatives of the micro-distortion tensor $P$.
In this way we cannot directly control the space variation of $P$.
This hypothesis is reasonable if we are modeling the mechanical behavior
of a medium in which the variation of $P$ is very small, i.e. the
norm $\left\Vert \grad P\right\Vert $ is dominated by a small real
value $\varepsilon$. As we will see, this model represents, in a
suitable meaning, a pathological limit: the behavior of the dispersion
curves changes drastically with respect to the full relaxed micromorphic
case.

\section{Limit passage to classical linear elasticity for vanishing micro-inertia}

In this section we would like to show how to obtain classical linear
elasticity as a limit case of our relaxed micromorphic model. Indeed,
there are several ways to obtain classical linear elasticity. For
all shown cases we will also perform a limit dispersion analysis and
identify the limiting elastic moduli. 

Consider (for simplicity the relaxed micromorphic modulus $\le=0$),
$\mc=0$, $\alpha_{1}=\alpha_{2}=\alpha_{3}=1$ and $\eta_{1}=\eta_{2}=\eta_{3}=0$.
\begin{align}
\rho\,u_{,tt}=\textrm{Div}\left[2\,\me\,\sym\left(\nabla u-P\right)\right], & \qquad0=-\me\,L_{c}^{2}\,\curl\,\curl\,P+\sigma-s,\label{eq:aq ridotta 1}
\end{align}
where
\begin{gather*}
\sigma=2\,\me\:\sym\left(\nabla u-P\right),\qquad s=2\,\mh\,\sym\,P.
\end{gather*}
We look for solutions of (\ref{eq:aq ridotta 1}) in the form of
\begin{gather}
P=\beta^{+}\grad u\quad\textrm{with}\quad\beta^{+}>0.\label{eq:beta-1}
\end{gather}
Inserting (\ref{eq:beta-1}) into (\ref{eq:aq ridotta 1}) we obtain\footnote{We recall that $\curl\,\nabla u=0$.}
\begin{gather}
\rho\,u_{,tt}=\textrm{Div}\left[2\,\me\,\sym\left(\nabla u-\beta^{+}\grad u\right)\right],\quad\quad0=0+2\,\me\,\sym\left(\nabla u-\beta^{+}\grad u\right)-2\,\mh\,\sym\left(\beta^{+}\grad u\right)
\end{gather}
$\Longleftrightarrow$
\begin{gather*}
\rho\,u_{,tt}=\textrm{Div}\left[2\,\me\left(1-\beta^{+}\right)\sym\,\grad u\right],\quad\quad0=0+2\,\me\left(1-\beta^{+}\right)\sym\,\grad u-2\,\mh\,\beta^{+}\,\sym\,\grad u
\end{gather*}
$\Longleftrightarrow$
\begin{gather}
0=\left[2\,\me\left(1-\beta^{+}\right)-2\,\mh\,\beta^{+}\right]\sym\,\grad u.\label{eq:final equation}
\end{gather}
Since $\sym\,\grad u\neq0$ by assumption, equation (\ref{eq:final equation})
is verified if and only if
\begin{gather*}
2\,\me\left(1-\beta^{+}\right)-2\,\mh\,\beta^{+}=0,
\end{gather*}
this means
\begin{gather}
\me\left(1-\beta^{+}\right)=\mh\,\beta^{+}\quad\Longleftrightarrow\quad\frac{\me}{\mh}=\frac{\beta^{+}}{1-\beta^{+}}\quad\Longleftrightarrow\quad\beta^{+}=\frac{\me}{\me+\mh}.\label{eq:beta}
\end{gather}
Assuming generically that $\me<\mh$, we find the following inequalities
\begin{gather*}
\frac{\beta^{+}}{1-\beta^{+}}=\frac{\me}{\mh}<1\quad\Leftrightarrow\quad\beta^{+}<1-\beta^{+}\quad\Leftrightarrow\quad2\beta^{+}<1\quad\Leftrightarrow\quad\beta^{+}<\frac{1}{2}.
\end{gather*}
Inserting the last expression of $\beta^{+}$ in (\ref{eq:beta})
we find
\begin{gather*}
\rho\,u_{,tt}=\textrm{Div}[2\,\me\underbrace{\left(1-\beta^{+}\right)}_{=\,\,\frac{\underset{}{\mu_{\textrm{micro}}}}{\overset{}{\me\,+\,\mh}}}\sym\,\nabla u]
\end{gather*}
and therefore
\begin{gather}
\rho\,u_{,tt}=\textrm{Div}\left[2\,\frac{\me\,\mh}{\me+\mh}\sym\,\nabla u\right]=\textrm{Div}\left[2\,\mum\,\sym\,\nabla u\right],
\end{gather}
where we have set 
\[
\mum:=\frac{\me\,\mh}{\me+\mh}\qquad\qquad\textrm{(harmonic mean)},
\]
according to formula (50) in \cite{barbagallo2016transparent}. This
analysis can be repeated with $\le\neq0$ such that $2\,\me+3\,\le>0.$
In this case we obtain as limit model
\begin{gather}
\rho\,u_{,tt}=\textrm{Div}\left[2\,\mum\,\sym\,\nabla u+\lam\,\textrm{tr}\left(\nabla u\right)\id\right]
\end{gather}
with
\[
\mum:=\frac{\me\,\mh}{\me+\mh},\qquad\lam=\frac{1}{3}\frac{\left(2\mu_{e}+3\lambda_{e}\right)\left(2\mh+3\lh\right)}{2\left(\mu_{e}+\mh\right)+3\left(\lambda_{e}+\lh\right)}-\frac{2}{3}\frac{\mu_{e}\,\mh}{\mu_{e}+\mh}
\]
being consistent with
\[
\kappa_{\textrm{macro}}=\frac{2\,\mum+3\,\lam}{3}
\]
from \cite{barbagallo2016transparent}. Thus the relaxed micromorphic
model with $\mc=0$ and $\eta\equiv0$ provides a classical macroscopic,
first gradient solution with $\mum,\lam$ as elastic moduli, provided
that the micro-inertia is identically zero (or $\eta\fr0$).

\section{Plane wave propagation in isotropic relaxed micromorphic media}

In this section we introduce the plane wave ansatz on the unknown
kinematical fields. This hypothesis allows to study the main characteristics
of wave propagation of relaxed micromorphic media in the simplest
possible way. The problem of wave propagation still remains 3D (all
the components of the introduced unknown fields are non vanishing),
while the space dependence is only on one scalar direction $x_{1}$
which is also the direction of propagation of the plane wave. Under
this assumption, the bulk equations (\ref{eq:PDE system}) take a
simplified form because all the partial derivatives in $x_{2},x_{3}$-direction
are zero. 

Moreover, thanks to an opportune change of variables, we can completely
uncouple the system of PDE in (\ref{eq:PDE system}) as done in \cite{madeo2015wave}.
In order to do this, we project also the micro-distortion tensor $P$
on the component spaces of the Cartan-Lie decomposition of $\R^{3\times3}$.
We set for the deviatoric - symmetric part

\begin{equation}
\dev\,\sym\,\P=\frac{1}{2}\left(\P+\P^{T}\right)-\frac{1}{3}\textrm{tr}\left(\P\right)\id=\begin{pmatrix}P_{1}^{D} & P_{\left(12\right)} & P_{\left(13\right)}\\
P_{\left(12\right)} & P_{2}^{D} & P_{\left(23\right)}\\
P_{\left(13\right)} & P_{\left(23\right)} & P_{3}^{D}
\end{pmatrix},
\end{equation}
where we have defined 
\begin{align}
P_{\alpha}^{D}=P_{\alpha\alpha}-\frac{1}{3}\textrm{tr}\,\P,\qquad\textrm{and}\qquad P_{\left(\alpha\beta\right)}=P_{\left(\beta\alpha\right)}=\frac{1}{2}\left(P_{\alpha\beta}+P_{\beta\alpha}\right)\quad\textrm{if}\quad\alpha\neq\beta.\label{eq:symmetric components}
\end{align}
Moreover, for the skew-symmetric part of $P$, we set
\begin{equation}
\skew\,\P=\frac{1}{2}\left(\P-\P^{T}\right)=\begin{pmatrix}0 & P_{\left[12\right]} & P_{\left[13\right]}\\
-P_{\left[12\right]} & 0 & P_{\left[23\right]}\\
-P_{\left[13\right]} & -P_{\left[23\right]} & 0
\end{pmatrix},\label{eq:skew symmetric components}
\end{equation}
where $P_{\left[\alpha\beta\right]}=\frac{1}{2}\left(P_{\alpha\beta}-P_{\beta\alpha}\right)$
and $P_{\left[\beta\alpha\right]}=-P_{\left[\alpha\beta\right]}$
and finally for the spherical part, we introduce the variable 
\begin{equation}
P^{S}=\frac{1}{3}\,\textrm{tr}\,\P=\frac{1}{3}\sum_{\alpha=1}^{3}P_{\alpha\alpha}.\label{eq:ps}
\end{equation}
Further we introduce the last new variable 
\begin{equation}
P^{V}=P_{22}-P_{33}=P_{2}^{D}-P_{3}^{D},
\end{equation}
and remark the validity of the identity
\begin{equation}
\P_{22}+P_{33}=2\,\P^{S}-P_{1}^{D}.
\end{equation}
Also, in what follows, we set $\P^{D}=\P_{1}^{D}$. It can be checked
that the micro-distortion tensor $\P$ can be written in terms of
the new variables as:
\begin{equation}
P=\dev\,\sym\,\P+\skew\,\P+\frac{1}{3}\textrm{tr}\,\P=\begin{pmatrix}P^{D}+P^{S} & P_{\left(12\right)}+P_{\left[12\right]} & P_{\left(13\right)}+P_{\left[13\right]}\\
\\
P_{\left(12\right)}-P_{\left[12\right]} & P_{2}^{D}+P^{S} & P_{\left(23\right)}+P_{\left[23\right]}\\
\\
P_{\left(13\right)}-P_{\left[13\right]} & P_{\left(23\right)}-P_{\left[23\right]} & P_{3}^{D}+P^{S}
\end{pmatrix},
\end{equation}
and we find, starting from (\ref{eq:PDE system}), the following four
groups of completely uncoupled equations in the new unknown fields
(remembering the dependence of the kinematical fields only on the
$x_{1}-$ direction)
\begin{equation}
\left(u_{1},u_{2},u_{3},P^{D},P_{\left(12\right)},P_{\left(13\right)},\text{P}_{\left(23\right)},P_{\left[12\right]},P_{\left[13\right]},\text{P}_{\left[23\right]},P^{S},P^{V}\right):
\end{equation}

\begin{itemize}
\item a first group of PDEs in the unknowns $u_{1},P^{D},P^{S}$ (longitudinal
quantities)
\begin{align}
\u_{1,tt} & =\frac{2\,\me+\lle}{\rho}\,u_{1,11}-\frac{2\,\me}{\rho}\,P_{,1}^{D}-\frac{2\,\me+3\,\le}{\rho}\,P_{,1}^{S},\nonumber \\
\P_{,tt}^{D} & =\frac{4}{3}\,\frac{\me}{\eta_{\,1}}\,u_{1,1}+\frac{\alpha_{2}}{\eta_{\,1}}\,\frac{\me\,L_{c}^{2}}{3}\left(P_{,11}^{D}-2\,\P_{,11}^{S}\right)-\frac{2\left(\me+\mh\right)}{\eta_{1}}\,\P^{D},\label{eq:long_curves}\\
\P_{,tt}^{S} & =\frac{2\,\me+3\,\le}{3\,\eta_{3}}\,u_{1,1}-\frac{\alpha_{2}}{\eta_{3}}\,\frac{\me L_{c}^{2}}{3}\left(P_{,11}^{D}-2\,\P_{,11}^{S}\right)-\frac{3\left(\lle+\lh\right)+2\left(\me+\mh\right)}{\eta_{3}}\,P^{S},\nonumber 
\end{align}
\item a second and third group of PDEs involving only the transversal quantities
in the direction $x_{\xi}$ with $\xi\in\left\{ 2,3\right\} $
\begin{align}
\u_{\xi,tt} & =\frac{\me+\mc}{\rho}\,u_{\xi,11}-\frac{2\,\me}{\rho}\,P_{\left(1\xi\right),1}+\frac{2\,\mc}{\rho}\,P_{\left[1\xi\right],1},\nonumber \\
\P_{\left(1\xi\right),tt} & =\frac{\me}{\eta_{\,1}}\,u_{\left(\xi,1\right)}+\frac{\alpha_{1}+\alpha_{2}}{\eta_{\,1}}\,\frac{\mu_{e}\,L_{c}^{2}}{4}\left(\P_{\left(1\xi\right),11}+\P_{\left[1\xi\right],11}\right)-\frac{2\left(\me+\mh\right)}{\eta_{1}}\,\P_{\left(1\xi\right)},\label{eq:trans_curves12}\\
\P_{\left[1\xi\right],tt} & =-\frac{\mc}{\eta_{2}}\,\u_{\xi,1}+\frac{\alpha_{1}+\alpha_{2}}{\eta_{2}}\,\frac{\me\,L_{c}^{2}}{4}\left(\P_{\left(1\xi\right),11}+\P_{\left[1\xi\right],11}\right)-\frac{2\,\mc}{\eta_{2}}\,\P_{\left[1\xi\right]},\nonumber 
\end{align}
\item and three completely uncoupled equations
\begin{align}
\P_{\left(23\right),tt} & =-\frac{2\left(\me+\mh\right)}{\eta_{1}}\,\P_{\left(23\right)}+\frac{\alpha_{1}}{\eta_{1}}\,\me\,L_{c}^{2}\,\P_{\left(23\right),11},\nonumber \\
\P_{\left[23\right],tt} & =-\frac{2\,\mc}{\eta_{2}}\,\P_{\left[23\right]}+\frac{\alpha_{1}+2\,\alpha_{3}}{\eta_{2}}\,\frac{\me\,L_{c}^{2}}{3}\,\P_{\left[23\right],11},\label{eq:uncoupled curves}\\
\P_{,tt}^{V} & =-\frac{2\left(\me+\mh\right)}{\eta_{1}}\,\P^{V}+\frac{\alpha_{1}}{\eta_{1}}\,\me\,L_{c}^{2}\,\P_{,11}^{V}.\nonumber 
\end{align}
\end{itemize}
The systems (\ref{eq:long_curves}),(\ref{eq:trans_curves12}),(\ref{eq:uncoupled curves})
of PDEs are explicitly derived in the appendix. Now we consider the
plane wave form for the newly introduced fields, i.e. 
\begin{align}
\u\left(x,t\right) & =\widehat{u}\,e^{i\left(kx_{1}-\,\omega t\right)}\label{eq:wave strucutre-1}
\end{align}
where $\widehat{u}=\left(\widehat{u}_{1},\widehat{u}_{2},\widehat{u}_{3}\right)$
is the so called polarization vector in $\mathbb{C}^{3}$ and\footnote{The quantity $\omega$ is the (circular) frequency and $k$ is the
(possibly complex) wave number.} 
\[
\begin{array}{ccccc}
P_{\alpha}^{D}=\widehat{P}_{\alpha}^{D}\,e^{i\left(kx_{1}-\,\omega t\right)}, &  & P^{V}=\widehat{P}^{V}\,e^{i\left(kx_{1}-\,\omega t\right)}, &  & P^{S}=\widehat{P}^{S}\,e^{i\left(kx_{1}-\,\omega t\right)},\\
\\
P_{\left(\alpha\beta\right)}=\widehat{P}_{\left(\alpha\beta\right)}\,e^{i\left(kx_{1}-\,\omega t\right)}, &  & P_{\left[\alpha\beta\right]}=\widehat{P}_{\left[\alpha\beta\right]}\,e^{i\left(kx_{1}-\,\omega t\right)}, &  & \alpha,\beta\in\left\{ 1,2,3\right\} .
\end{array}
\]
Introducing the vector

\[
\boldsymbol{v}=\left(\widehat{u}_{1},\widehat{P}^{D},\widehat{P}^{S},\widehat{u}_{2},\widehat{P}_{\left(12\right)},\widehat{P}_{\left[12\right]},\widehat{u}_{3},\widehat{P}_{\left(13\right)},\widehat{P}_{\left[13\right]},\widehat{P}_{\left(23\right)},\widehat{P}_{\left[23\right]},\widehat{P}^{V}\right)\in\R^{12},
\]
if we divide the PDE system (\ref{eq:PDE system}) by $e^{i\left(kx_{1}-\,\omega t\right)}$
we obtain the associated algebraic system in the form
\[
\boldsymbol{D}\,\boldsymbol{v}=0,
\]
where the matrix $\boldsymbol{D}$ is a $12\times12$ matrix with
the following block-structure
\begin{equation}
\boldsymbol{D}=\begin{pmatrix}\mathsf{E}_{1} & 0 & 0 & 0\\
0 & \mathsf{E}_{2} & 0 & 0\\
0 & 0 & \mathsf{E}_{3} & 0\\
0 & 0 & 0 & \mathsf{E}_{4}
\end{pmatrix}\in\mathbb{C}^{12\times12},
\end{equation}
in which $\mathsf{E}_{1},\mathsf{E}_{2},\mathsf{E}_{3},\mathsf{E}_{4}$
are the following matrices in $\mathbb{C}^{3\times3}$:
\begin{align}
\mathsf{E}_{1} & =\left(\begin{array}{ccc}
k^{2}c_{\textrm{p}}^{2}-\omega^{2} & {\displaystyle 2ik\frac{\mu_{e}}{\rho}} & {\displaystyle ik\frac{\left(3\lambda_{e}+2\mu_{e}\right)}{\rho}}\\
\\
{\displaystyle -\frac{4}{3}\frac{\mu_{e}}{\eta_{\,1}}ik} & {\displaystyle \frac{1}{3}\frac{\alpha_{2}}{\eta_{\,1}}\,k^{2}\mu_{e}\,L_{c}^{2}+\omega_{s}^{2}-\omega^{2}} & -{\displaystyle \frac{2}{3}\,\frac{\alpha_{2}}{\eta_{\,1}}\,k^{2}\mu_{e}\,L_{c}^{2}}\\
\\
{\displaystyle -ik\left(\frac{3\lambda_{e}+2\mu_{e}}{3\,\eta_{\,3}}\right)} & {\displaystyle -\frac{1}{3}\frac{\alpha_{2}}{\eta_{3}}\,k^{2}\mu_{e}\,L_{c}^{2}} & {\displaystyle \frac{2}{3}\frac{\alpha_{2}}{\eta_{3}}\,k^{2}\mu_{e}\,L_{c}^{2}+\omega_{p}^{2}-\omega^{2}}
\end{array}\right),\nonumber \\
\nonumber \\
\mathsf{E}_{2}=\mathsf{E}_{3} & =\left(\begin{array}{ccc}
k^{2}c_{\mathrm{s}}^{2}-\omega^{2} & {\displaystyle 2ik\,\frac{\mu_{e}}{\rho}} & {\displaystyle -ik\,\omega_{r}^{2}\frac{\eta_{2}}{\rho}}\\
\\
{\displaystyle -ik\,\frac{\mu_{e}}{\eta_{\,1}}} & {\displaystyle k^{2}\mu_{e}\,L_{c}^{2}\,\frac{1}{4}\,\frac{\alpha_{1}+\alpha_{2}}{\eta_{\,1}}+\omega_{s}^{2}-\omega^{2}} & {\displaystyle k^{2}\mu_{e}\,L_{c}^{2}\frac{1}{4}\,\frac{\alpha_{1}+\alpha_{2}}{\eta_{\,1}}}\\
\\
{\displaystyle \frac{1}{2}ik\,\omega_{r}^{2}} & {\displaystyle k^{2}\mu_{e}\,L_{c}^{2}\,\frac{1}{4}\,\frac{\alpha_{1}+\alpha_{2}}{\eta_{\,2}}} & {\displaystyle k^{2}\mu_{e}\,L_{c}^{2}\,\frac{1}{4}\,\frac{\alpha_{1}+\alpha_{2}}{\eta_{\,2}}+\omega_{r}^{2}-\omega^{2}}
\end{array}\right),\nonumber \\
\nonumber \\
\mathsf{E}_{4} & =\left(\begin{array}{ccc}
k^{2}\left(c_{\textrm{m}}^{\textrm{d}}\right)^{2}+\omega_{s}^{2}-\omega^{2} & 0 & 0\\
\\
0 & {\displaystyle k^{2}\left(c_{\textrm{m}}^{\textrm{vd}}\right)^{2}+\omega_{r}^{2}-\omega^{2}} & 0\\
\\
0 & 0 & k^{2}\left(c_{\textrm{m}}^{\textrm{d}}\right)^{2}+\omega_{s}^{2}-\omega^{2}
\end{array}\right),\label{eq:E4}
\end{align}
where $c_{\textrm{p}},c_{\mathrm{s}},c_{\textrm{m}}^{\textrm{d}},c_{\textrm{m}}^{\textrm{vd}},\omega_{r},\omega_{s},\omega_{p}$,
are defined in (\ref{eq:omega}) and (\ref{eq:asintoti obliqui}).
Introducing the auxiliary matrices $\widehat{\mathsf{E}}_{1},\widehat{\mathsf{E}}_{2},\widehat{\mathsf{E}}_{3},\in\mathbb{C}^{12}\times\mathbb{C}^{12}$,$\quad\widehat{\mathsf{E}}_{4}\in\R^{12}\times\R^{12}$
\begin{align*}
\widehat{\mathsf{E}}_{1} & =\begin{pmatrix}\mathsf{E}_{1} & 0 & 0 & 0\\
0 & \id_{3} & 0 & 0\\
0 & 0 & \id_{3} & 0\\
0 & 0 & 0 & \id_{3}
\end{pmatrix}, & \widehat{\mathsf{E}}_{2} & =\begin{pmatrix}\id_{3} & 0 & 0 & 0\\
0 & \mathsf{E}_{2} & 0 & 0\\
0 & 0 & \id_{3} & 0\\
0 & 0 & 0 & \id_{3}
\end{pmatrix}, & \widehat{\mathsf{E}}_{3} & =\begin{pmatrix}\id_{3} & 0 & 0 & 0\\
0 & \id_{3} & 0 & 0\\
0 & 0 & \mathsf{E}_{3} & 0\\
0 & 0 & 0 & \id_{3}
\end{pmatrix}, & \widehat{\mathsf{E}}_{4} & =\begin{pmatrix}\id_{3} & 0 & 0 & 0\\
0 & \id_{3} & 0 & 0\\
0 & 0 & \id_{3} & 0\\
0 & 0 & 0 & \mathsf{E}_{4}
\end{pmatrix},
\end{align*}
where $\id_{3}$ is the identity of $\mathbb{C}^{3\times3}$ or $\mathbb{R}^{3\times3}$
, we remark that 
\[
\boldsymbol{D}=\widehat{\mathsf{E}}_{1}\,\widehat{\mathsf{E}}_{2}\,\widehat{\mathsf{E}}_{3}\,\widehat{\mathsf{E}}_{4}
\]
and therefore
\begin{align}
\det\,\boldsymbol{D} & =\det\,\widehat{\mathsf{E}}_{1}\cdot\det\,\widehat{\mathsf{E}}_{2}\cdot\det\,\widehat{\mathsf{E}}_{3}\cdot\det\,\widehat{\mathsf{E}}_{4}=\det\,\widehat{\mathsf{E}}_{1}\cdot\left(\det\,\widehat{\mathsf{E}}_{2}\right)^{2}\cdot\det\,\widehat{\mathsf{E}}_{4}\nonumber \\
 & =\det\,\mathsf{E}_{1}\cdot\left(\det\,\mathsf{E}_{2}\right)^{2}\cdot\det\,\mathsf{E}_{4}.
\end{align}
 In this way the study of the solutions $\widehat{\omega}=\widehat{\omega}\left(k\right)$
of $\det\,\boldsymbol{D}=\det\,\boldsymbol{D}\left(k,\omega\right)=0$
is equivalent of the study of the solutions of the three equations
\begin{align}
{\textstyle \det\,\mathsf{E}_{1}\left(k,\omega\right)} & =0, & \det\,\mathsf{E}_{2}\left(k,\omega\right) & =0, & \det\,\mathsf{E}_{4} & \left(k,\omega\right)=0.\label{eq:dets}
\end{align}
The solutions $\widehat{\omega}=\widehat{\omega}\left(k\right)$ of
these characteristic equations are known as the \textbf{dispersion
curves} of the considered continuum. Introducing the matrices
\begin{gather}
\mathsf{B}_{i}=\mathsf{E}_{i}+\omega^{2}\id,\qquad i\in\left\{ 1,\ldots,4\right\} ,\label{eq: matrici Bi}
\end{gather}
 we can regard the problems in (\ref{eq:dets}) equivalently as eigenvalue
problems
\begin{equation}
\det\left(\mathsf{B}_{i}-\omega^{2}\id\right)=0,\label{eq:eigenvalues problems}
\end{equation}
where $\mathsf{B}_{i}$ are the blocks of the \textbf{symmetric acoustic
tensor}.

\subsection{Analysis of dispersion curves}

The dispersion curves for the relaxed micromorphic model are the functions
$\widehat{\omega}_{i}=\widehat{\omega}_{i}\left(k\right),\;i\in\left\{ 1,\ldots,12\right\} $
that are solutions of the polynomial equations (\ref{eq:dets}) or
equivalently the eigenvalues of the matrices in (\ref{eq: matrici Bi}).
Thanks to the invariant property of the eigenvalues with respect to
similarity transformations, showing that our matrices $\mathsf{E}_{i}$
are similar to real symmetric matrices, we obtain that the dispersion
curves are real valued functions\footnote{The eigenvalues of a symmetric real matrix are always reals.}
\cite{mad2016}. To this aim, we introduce the scaling matrices
\begin{equation}
\mathsf{\mathsf{P}}_{1}:=\begin{pmatrix}\sqrt{\rho} & 0 & 0\\
0 & i\sqrt{\frac{3\,\eta_{1}}{2}} & 0\\
0 & 0 & i\sqrt{3\,\eta_{3}}
\end{pmatrix},\qquad\mathsf{\mathsf{P}}_{2}:=\begin{pmatrix}\sqrt{\rho} & 0 & 0\\
0 & i\sqrt{2\,\eta_{1}} & 0\\
0 & 0 & i\sqrt{2\,\eta_{2}}
\end{pmatrix}.
\end{equation}
It is immediately seen that 
\begin{align}
\mathsf{P}_{1}\cdot\mathsf{E}_{1}\cdot\mathsf{P}_{1}^{-1} & =\left(\begin{array}{ccc}
{\displaystyle c_{\textrm{p}}^{2}\,k^{2}-\omega^{2}} & {\displaystyle \frac{2\sqrt{2}\mu_{e}}{\sqrt{3\,\varrho\,\eta_{1}}}}\,k & {\displaystyle \frac{\left(3\lambda_{e}+2\mu_{e}\right)}{\sqrt{3\,\varrho\,\eta_{3}}}}\,k\\
\\
{\displaystyle \frac{2\sqrt{2}\mu_{e}}{\sqrt{3\,\varrho\,\eta_{1}}}}\,k & {\displaystyle \frac{\alpha_{2}}{3\eta_{1}}\,L_{c}^{2}\mu_{e}\,k^{2}+\omega_{s}^{2}-\omega^{2}} & {\displaystyle -\frac{\alpha_{2}}{\sqrt{\eta_{1}\eta_{3}}}\frac{\sqrt{2}\,\mu_{e}L_{c}^{2}}{3}\,k^{2}}\\
\\
{\displaystyle \frac{\left(3\lambda_{e}+2\mu_{e}\right)}{\sqrt{3\,\varrho\,\eta_{3}}}}\,k & {\displaystyle -\frac{\alpha_{2}}{\sqrt{\eta_{1}\eta_{3}}}\frac{\sqrt{2}\mu_{e}L_{c}^{2}}{3}\,k^{2}} & {\displaystyle \frac{\alpha_{2}}{\eta_{3}}\,\frac{2\,\mu_{e}L_{c}^{2}}{3}\,k^{2}+\omega_{p}^{2}-\omega^{2}}
\end{array}\right),\\
\nonumber \\
\mathsf{\mathsf{P}}_{2}\cdot\mathsf{E}_{2}\cdot\mathsf{\mathsf{P}}_{2}^{-1} & =\left(\begin{array}{ccc}
k^{2}c_{s}^{2}-\omega^{2} & {\displaystyle \frac{\sqrt{2}\mu_{e}}{\sqrt{\varrho\,\eta_{1}}}\,k} & {\displaystyle -\frac{\sqrt{2}\mu_{c}}{\sqrt{\varrho\,\eta_{2}}}\,k}\\
\\
{\displaystyle \frac{\sqrt{2}\mu_{e}}{\sqrt{\varrho\,\eta_{1}}}\,k} & {\displaystyle k^{2}\mu_{e}\,L_{c}^{2}\,\frac{1}{4}\,\frac{\alpha_{1}+\alpha_{2}}{\eta_{\,1}}+\omega_{s}^{2}-\omega^{2}} & {\displaystyle \frac{\alpha_{1}+\alpha_{2}}{\sqrt{\eta_{1}\eta_{2}}}\frac{\mu_{e}L_{c}^{2}}{4}\,k^{2}}\\
\\
{\displaystyle -\frac{\sqrt{2}\mu_{c}}{\sqrt{\varrho\,\eta_{2}}}\,k} & {\displaystyle \frac{\alpha_{1}+\alpha_{2}}{\sqrt{\eta_{1}\eta_{2}}}\frac{\mu_{e}L_{c}^{2}}{4}\,k^{2}} & {\displaystyle k^{2}\mu_{e}\,L_{c}^{2}\,\frac{1}{4}\,\frac{\alpha_{1}+\alpha_{2}}{\eta_{\,2}}+\omega_{r}^{2}-\omega^{2}}
\end{array}\right).
\end{align}
Since $\mathsf{E}_{4}$ has only two distinct eigenvalues, we have
only two distinct dispersion curves as solutions of the system $\det\,\mathsf{E}_{4}=0$. 

\subsubsection{Cut-off frequencies}

The cut-off frequencies are the solutions of the equation ${\textstyle \det}\,\D\left(k,\omega\right)=0$
when $k=0$ and give us the values of the dispersion curves $\widehat{\omega}_{i}\left(k\right)$
at $k=0$. We find only three different non trivial solutions for
the equation ${\textstyle \det}\,\D\left(0,\omega\right)=0:$ 
\begin{align}
\omega_{s}\left(\me,\mh,\eta_{1}\right)=\sqrt{\frac{2\left(\me+\mh\right)}{\eta_{1}}},\qquad\omega_{r}\left(\mc,\eta_{2}\right)=\sqrt{\frac{2\,\mc}{\eta_{2}}},\nonumber \\
\label{eq:omega}\\
\omega_{p}\left(\lle,\lh,\me,\mh,\eta_{3}\right)=\sqrt{\frac{3\left(\lle+\lh\right)+2\left(\me+\mh\right)}{\eta_{3}}},\nonumber 
\end{align}
 with multiplicity of 5,3,1, respectively. The null solution has multiplicity
3. This means that if 
\begin{itemize}
\item $\mc>0$ we have 3 acoustic curves, and 9 optic curves,
\item $\mc=0$ we have 6 acoustic curves, and 6 optic curves.
\end{itemize}
The first novel result with respect to \cite{madeo2015wave} is that
the presence of three micro-inertia terms $\eta_{1},\eta_{2},\eta_{3}$
makes the three cut-off frequencies completely independent. This means
that having fixed the parameters $\left(\lle,\lh,\me,\mh,\mc\right)$
we can obtain all positive values for the cut-offs by simply changing
the values of the three inertia parameters $\eta_{1},\eta_{2},\eta_{3}$.
Whether the fact of having $\eta_{1}\neq\eta_{2}\neq\eta_{3}$ may
be interesting for applications on real band-gap metamaterials must
be checked on real experiments. It will be the objective of a forthcoming
paper to show that this is indeed the case. 

\subsubsection{Oblique asymptotes }

In this sub-section we want to give a tool to determine the oblique
asymptotes to the unbounded dispersion curves $\widehat{\omega}\left(k\right)$,
solutions of the equation $\det\,\boldsymbol{D}\left(k,\omega\right)=0$.
First of all, it is useful to notice that the matrix $\D$ can be
written as:
\[
\D\left(k,\omega\right)=\boldsymbol{A}_{2}\,k^{2}+\boldsymbol{B}_{2}\,\omega^{2}+\boldsymbol{A}_{1}\,k+\boldsymbol{C}_{0},
\]
where $\boldsymbol{A}_{2},\boldsymbol{B}_{2},\boldsymbol{A}_{1}$
and $\boldsymbol{C}_{0}$ are suitable $12\times12$ constant real
matrices with $\boldsymbol{B}_{2}$ invertible. Thus we have that
\begin{align*}
\det\,\D\left(k,\omega\right) & =\det\left(\boldsymbol{A}_{2}\,k^{2}+\boldsymbol{B}_{2}\,\omega^{2}+\boldsymbol{A}_{1}\,k+\boldsymbol{C}_{0}\right)\\
 & =\det\,\boldsymbol{B}_{2}\cdot\det\left(\boldsymbol{B}_{2}^{-1}\boldsymbol{A}_{2}\,k^{2}+\omega^{2}\id+\boldsymbol{B}_{2}^{-1}\boldsymbol{A}_{1}\,k+\boldsymbol{B}_{2}^{-1}\boldsymbol{C}_{0}\right)\\
 & =k^{24}\,\det\boldsymbol{B}_{2}\cdot\det\left(\boldsymbol{B}_{2}^{-1}\boldsymbol{A}_{2}+\frac{\omega^{2}}{k^{2}}\,\id+\frac{1}{k}\,\boldsymbol{B}_{2}^{-1}\boldsymbol{A}_{1}+\frac{1}{k^{2}}\,\boldsymbol{B}_{2}^{-1}\boldsymbol{C}_{0}\right).
\end{align*}
Thus the equation $\det\,\boldsymbol{D}\left(k,\omega\right)=0$ is
equivalent to 
\begin{equation}
p\left(k,\omega\right)=\det\left(\boldsymbol{B}_{2}^{-1}\boldsymbol{A}_{2}+\frac{\omega^{2}}{k^{2}}\,\id+\frac{1}{k}\,\boldsymbol{B}_{2}^{-1}\boldsymbol{A}_{1}+\frac{1}{k^{2}}\,\boldsymbol{B}_{2}^{-1}\boldsymbol{C}_{0}\right)=0.\label{eq: caratt}
\end{equation}

\begin{prop}
\label{prop:oblique asy}Let us assume that the equation $\det\,\boldsymbol{D}\left(k,\omega\right)=0$
admits a non-empty set of solutions $\Delta=\left\{ \widehat{\omega}_{i}\left(k\right)\right\} _{i=1}^{n\in\mathbb{N}}.$
Let us consider the subset $\Delta_{\infty}=\left\{ \widehat{\omega}_{j}\right\} _{j=1}^{s\leq n}$
constituted by the solutions verifying the following conditions: 
\end{prop}
\begin{enumerate}
\item \textit{$\widehat{\omega}_{j}$ is a monotonically increasing function
of $k$ for $j=1,\ldots,s$,}
\item \textit{$\lim_{k\fr\infty}\frac{\widehat{\omega}_{j}\left(k\right)}{k}\neq0$
for $j=1,\ldots,s$, (which implies that $\widehat{\omega}_{j}$ is
unbounded and so without horizontal asymptote),}
\end{enumerate}
\textit{and we assume that $\Delta_{\infty}\neq\textrm{Ø}$. If we
consider a reduced problem
\begin{equation}
q\left(k,\omega\right)=\det\left(\boldsymbol{B}_{2}^{-1}\boldsymbol{A}_{2}+\left(\frac{\omega}{k}\right)^{2}\,\id\right)=0,\label{eq:prob obli 2}
\end{equation}
then this problem (\ref{eq:prob obli 2}) admits solutions $\left\{ \widetilde{\omega}_{j}\left(k\right)\right\} _{j=1}^{s}$
such that $\lim_{k\fr\infty}\left(\widehat{\omega}_{j}-\widetilde{\omega}_{j}\right)=0$
for every $j=1,\ldots,s$.}
\begin{proof}
This is a simple application of the property of the continuous dependence
of the roots of a polynomial on its coefficients. We can remark that
under condition 2 of the proposition, if we think the coefficients
of $p\left(k,\omega\right)$ as functions of $k$ (because we are
looking for solutions of the type $k\mapsto\left(k,\omega\left(k\right)\right)$),
then due to the continuity of the determinant, they converge to the
coefficients of $q\left(k,\omega\right)$ and so do its roots. 
\end{proof}

\begin{rem}
Proposition \ref{prop:oblique asy} does not work for the bounded
dispersion curves because in this case also the term $\frac{\widetilde{\omega}\left(k\right)}{k}$
converges to zero when $k\fr\infty$ because bounded curves violate
the conditions 2 of proposition \ref{prop:oblique asy}. It is for
this reason that we will give another argument to look for the horizontal
asymptotes.
\end{rem}
In our case the roots $\widetilde{\omega}_{j}\left(k\right)$ of the
reduced polynomial (\ref{eq:prob obli 2}) can be computed more easily
and are found to be straight lines with slopes:
\begin{align}
c_{\mathrm{m}}^{\mathrm{d}}=\sqrt{\frac{\alpha_{1}\,\me\,L_{c}^{2}}{\eta_{1}}}, &  & c_{\mathrm{m}}^{\mathrm{vd}}=\sqrt{\frac{\left(\alpha_{1}+2\,\alpha_{3}\right)\,\me\,L_{c}^{2}}{3\,\eta_{2}}}, &  & c_{\mathrm{m}}^{\mathrm{dr}}=\frac{1}{2}\sqrt{\frac{\left(\eta_{1}+\eta_{2}\right)}{\eta_{1}\,\eta_{2}}\left(\alpha_{1}+\alpha_{2}\right)\,\me\,L_{c}^{2}},\nonumber \\
\label{eq:asintoti obliqui}\\
c_{\mathrm{s}}=\sqrt{\frac{\me+\mc}{\rho}}, &  & c_{\mathrm{p}}=\sqrt{\frac{2\,\me+\lle}{\rho}}, &  & c_{\mathrm{m}}^{\mathrm{r}}=\sqrt{\frac{\left(2\,\eta_{1}+\eta_{3}\right)}{3\,\eta_{1}\,\eta_{3}}\,\alpha_{2}\,\me\,L_{c}^{2}}.\nonumber 
\end{align}

\subsubsection{Horizontal asymptotes}

In this subsection we want to investigate the behavior at infinity
of the dispersion curves that are bounded i.e. that have horizontal
asymptote. Thus, let $\widehat{\omega}\left(k\right)$ be a bounded
solution of the equation ${\textstyle \det}\,\D\left(k,\omega\right)=0$.
Under the assumption that this function is monotonically increasing
in $k$, setting 
\[
\widehat{\omega}_{*}:=\sup_{\R^{+}}\left\{ \widehat{\omega}\left(k\right)\right\} <\infty,
\]
it is straightforward to show that $\widehat{\omega}\left(k\right)$
admit a horizontal asymptote whose value is $\widehat{\omega}_{*}$.
Thanks to the particular expression of the function ${\textstyle \det}\,\D\left(k,\omega\right)$
we can find a necessary (and computable) condition on $\widehat{\omega}_{*}$
both in the general relaxed micromorphic case and in the internal
variable model. Indeed, in the general $12\times12$ case it can be
checked that the function ${\textstyle \det}\,\D\left(k,\omega\right)$
is a polynomial of even order in the two variables $k,\omega$, that
can be written as
\begin{equation}
{\textstyle \det}\,\D\left(k,\omega\right)=\sum_{h=0}^{12}c_{2h}\left(\omega^{2}\right)k^{2h},\qquad\textrm{with}\qquad c_{2h}:\left[0,+\infty\right]\fr\left[0,+\infty\right]\label{eq:poly}
\end{equation}
polynomial functions in $\omega^{2}$. Our calculation gives that
\[
c_{24}\left(\omega^{2}\right)=c_{22}\left(\omega^{2}\right)=c_{20}\left(\omega^{2}\right)\equiv0\qquad\textrm{and}\qquad c_{2h}\left(\omega^{2}\right)\neq0\quad\textrm{if}\quad h<10.
\]
In order to compare our relaxed model to the internal variable one
(which is obtained setting $\alpha_{1}=\alpha_{2}=\alpha_{3}=0$),
we can regard the polynomials $c_{2h}\left(\omega^{2}\right)$ as
functions of the three parameters $\alpha_{1},\alpha_{2}$ and $\alpha_{3}$.
Our calculation shows that the polynomials $c_{2h}\left(\omega^{2}\right)$
are zero for the following combinations of these three scalars:

\begin{table}[H]
\begin{centering}
\begin{tabular}{>{\centering}p{3cm}>{\centering}p{6cm}}
\noalign{\vskip4mm}
$c_{18}\left(\omega\right)$ & $\alpha_{1}=0\;\textrm{or}\;\alpha_{2}=0$\tabularnewline[2mm]
\hline 
\noalign{\vskip2mm}
$c_{16}\left(\omega\right)$ & $\alpha_{1}=0$\tabularnewline[2mm]
\hline 
\noalign{\vskip2mm}
$c_{14}\left(\omega\right)$ & $\alpha_{1}=0\;\textrm{and}\;\alpha_{2}=0$\tabularnewline[2mm]
\hline 
\noalign{\vskip2mm}
$c_{12}\left(\omega\right)$ & $\alpha_{1}=0\;\textrm{and}\;\alpha_{2}=0$\tabularnewline[2mm]
\hline 
\noalign{\vskip2mm}
$c_{10}\left(\omega\right)$ & $\alpha_{1}=0\;\textrm{and}\;\alpha_{2}=0$\tabularnewline[2mm]
\hline 
\noalign{\vskip2mm}
$c_{8}\left(\omega\right)$ & $\alpha_{1}=0\;\textrm{and}\;\alpha_{2}=0\;\textrm{and}\;\alpha_{3}=0$\tabularnewline[2mm]
\hline 
\noalign{\vskip2mm}
$c_{6}\left(\omega\right)$ & -\tabularnewline[2mm]
\hline 
\noalign{\vskip2mm}
$c_{4}\left(\omega\right)$ & -\tabularnewline[2mm]
\hline 
\noalign{\vskip2mm}
$c_{2}\left(\omega\right)$ & -\tabularnewline[2mm]
\hline 
\noalign{\vskip2mm}
$c_{0}\left(\omega\right)$ & -\tabularnewline[2mm]
\end{tabular}
\par\end{centering}
\caption{Effect of the parameters $\alpha_{1},\alpha_{2},\alpha_{3}$ on the
order in $k$ of the polynomial $\det_{\boldsymbol{D}}$.}
\end{table}
We can hence see that in the case of the internal variable model,
the order of ${\textstyle \det}\,\D\left(k,\omega\right)$ is smaller
and 
\[
{\textstyle \det^{\mathcal{I}}}\,\D\left(k,\omega\right)=\sum_{h=0}^{3}\hat{c}_{2h}\left(\omega^{2}\right)k^{2h},
\]
where the functions $\hat{c}_{2h}\left(\omega^{2}\right)$ and ${\textstyle \det^{\mathcal{I}}}\,\D\left(k,\omega\right)$
are obtained from the $c_{2h}\left(\omega^{2}\right)$ setting $\alpha_{1}=\alpha_{2}=\alpha_{3}=0$. 

Whit the purpose of clarify the general tool that we will find to
calculate the horizontal asymptote of ${\textstyle \det}\,\D\left(k,\omega\right)$,
we propose the following example.
\begin{example}
Let us consider the polynomial
\[
{\textstyle \det}\,\D\left(k,\omega\right)=c_{0}\left(\omega^{2}\right)1+c_{2}\left(\omega^{2}\right)k^{2}+c_{4}\left(\omega^{2}\right)k^{4}+c_{6}\left(\omega^{2}\right)k^{6}=p\left(k,\omega\right),
\]
where we assume that $c_{0},c_{2},c_{4},c_{6}:\R^{+}\fr\R^{+}$ are
continuous. We look for solutions $\widehat{\omega}=\widehat{\omega}\left(k\right)$
of 
\[
0=p\left(k,\widehat{\omega}\left(k\right)\right)
\]
$\Longleftrightarrow$
\begin{equation}
0=c_{0}\left(\left(\widehat{\omega}\left(k\right)\right)^{2}\right)+c_{2}\left(\left(\widehat{\omega}\left(k\right)\right)^{2}\right)k^{2}+c_{4}\left(\left(\widehat{\omega}\left(k\right)\right)^{2}\right)k^{4}+c_{6}\left(\left(\widehat{\omega}\left(k\right)\right)^{2}\right)k^{6}.\label{esempio1}
\end{equation}
Dividing (\ref{esempio1}) by $k^{6}$ we have equivalently
\begin{equation}
0=\frac{c_{0}\left(\left(\widehat{\omega}\left(k\right)\right)^{2}\right)}{k^{6}}+\frac{c_{2}\left(\left(\widehat{\omega}\left(k\right)\right)^{2}\right)}{k^{4}}+\frac{c_{4}\left(\left(\widehat{\omega}\left(k\right)\right)^{2}\right)}{k^{2}}+c_{6}\left(\left(\widehat{\omega}\left(k\right)\right)^{2}\right).\label{eq:kw}
\end{equation}
Since 
\[
\lim_{k\fr\infty}\frac{c_{0}\left(\left(\widehat{\omega}\left(k\right)\right)^{2}\right)}{k^{6}}=\lim_{k\fr\infty}\frac{c_{2}\left(\left(\widehat{\omega}\left(k\right)\right)^{2}\right)}{k^{4}}=\lim_{k\fr\infty}\frac{c_{4}\left(\left(\widehat{\omega}\left(k\right)\right)^{2}\right)}{k^{2}}=0,
\]
and
\begin{gather*}
0=\lim_{k\fr\infty}c_{6}\left(\left(\widehat{\omega}\left(k\right)\right)^{2}\right)=c_{6}\left(\omega_{*}^{2}\right)
\end{gather*}
we obtain the necessary condition
\begin{equation}
c_{6}\left(\omega_{*}^{2}\right)=0.\label{eq:cond necessaria}
\end{equation}

\begin{figure}[H]
\begin{centering}
\includegraphics[scale=0.4]{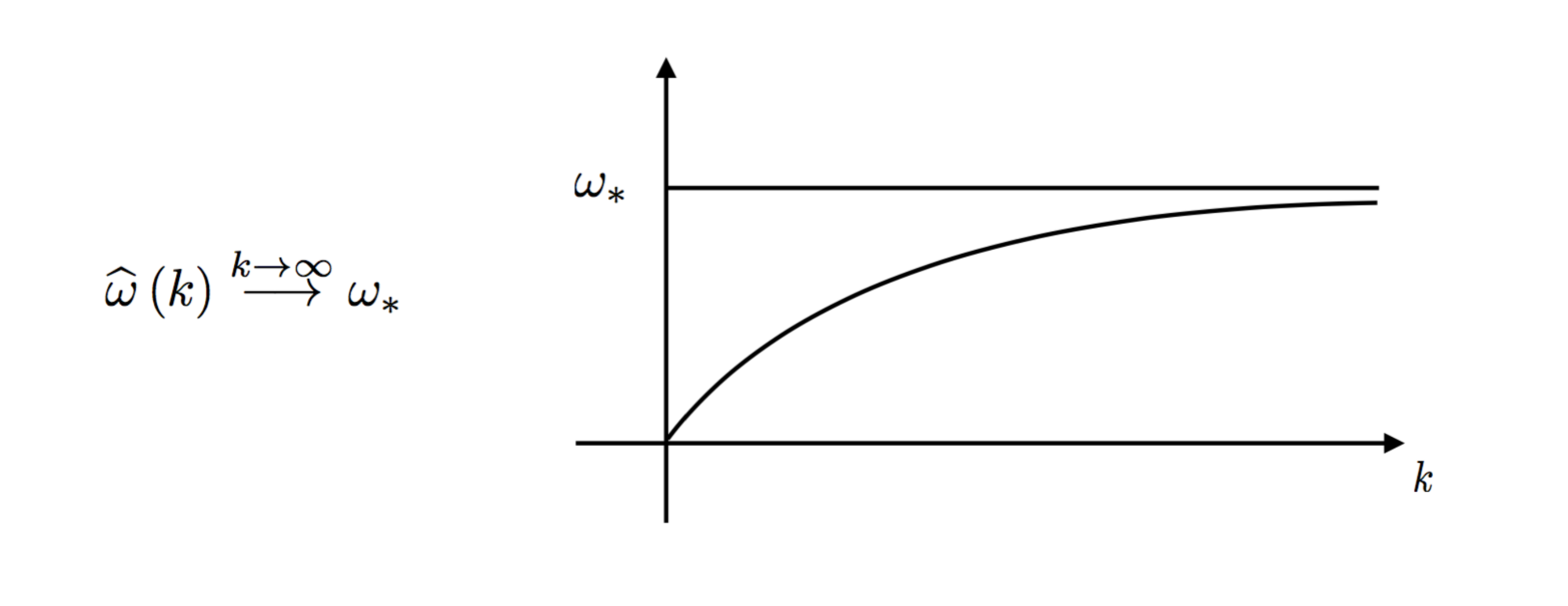}
\par\end{centering}
\caption{A bounded solution $\widehat{\omega}$ and its horizontal asymptote
$\omega_{*}$.}
\end{figure}
\end{example}
The condition (\ref{eq:cond necessaria}) is a necessary condition
that the horizontal asymptote has to satisfy. Because in our situation
we can not find an explicit expression for the dispersion curves,
the only possibility that we have to calculate the values of the horizontal
asymptote is to test the necessary condition (\ref{eq:cond necessaria}).
Adopting the notations proposed here, we can so finally prove the
following
\begin{prop}
\label{prop:CN}Let $\widehat{\omega}\left(k\right)$ be a bounded
solution of the problem ${\textstyle \det}\,\D\left(k,\omega\right)=0$
with horizontal asymptote $\omega_{*}$. Then $c_{18}\left(\omega_{*}^{2}\right)=0$. 
\end{prop}
\begin{proof}
Being $\widehat{\omega}\left(k\right)$ a solution of ${\textstyle \det}\,\D\left(k,\omega\right)=0$,
we have ${\textstyle \det}\,\D\left(k,\widehat{\omega}\left(k\right)\right)=0\;\forall\,k\in\left(0,\infty\right),$
i.e.
\[
\sum_{h=0}^{9}c_{2h}\left(\left(\widehat{\omega}\left(k\right)\right)^{2}\right)k^{2h}=0\qquad\forall\,k\in\left(0,\infty\right).
\]
Dividing by $k^{18}$ we find
\begin{equation}
\sum_{h=0}^{9}c_{2h}\left(\left(\widehat{\omega}\left(k\right)\right)^{2}\right)k^{2h-18}=0\quad\forall\,k\in\left(0,\infty\right).\label{eq:asintito orizz}
\end{equation}
For the continuity of the $c_{i}$ functions we have
\[
\lim_{k\fr+\infty}c_{2h}\left(\left(\widehat{\omega}\left(k\right)\right)^{2}\right)=c_{2h}\left(\omega_{*}^{2}\right)\qquad\forall\,h.
\]
So passing to the limit in (\ref{eq:asintito orizz}) we find
\[
\lim_{k\fr+\infty}\sum_{h=0}^{9}c_{2h}\left(\left(\widehat{\omega}\left(k\right)\right)^{2}\right)k^{2h-18}=c_{18}\left(\omega_{*}^{2}\right)=0.
\]
\end{proof}
\begin{cor}
If we have $\alpha_{1}=\alpha_{2}=\alpha_{3}=0$, and $\widehat{\omega}\left(k\right)$
is a solution of the problem $\mathrm{det}\,\D\left(k,\omega\right)=0$
with horizontal asymptote $\omega_{*}$, then $\widehat{c}_{6}\left(\omega_{*}\right)=0$. 
\end{cor}
Performing the calculation for the general relaxed micromorphic model
and the internal variable one, and considering only the positive roots,
we find the following possible values $\omega_{*}$ for the horizontal
asymptotes 
\begin{equation}
c_{18}\left(\omega_{*}\right)=0\quad\Leftrightarrow\quad\omega_{*}\in\left\{ \sqrt{\frac{2\:\mh}{\eta_{1}+\eta_{2}}},\:\sqrt{\frac{3\left(\lh+2\:\mh\right)}{2\:\eta_{1}+\eta_{3}}}\right\} \label{eq:asintoti orizzontali}
\end{equation}
and
\begin{gather}
\widehat{c}_{6}\left(\omega_{*}\right)=0\quad\Leftrightarrow\quad\omega_{*}\in\left\{ \sqrt{\frac{2\,\mc}{\eta_{2}}},\:\sqrt{\frac{2\left(\me+\mc\right)}{\eta_{1}}},\:\sqrt{\frac{q_{1}\pm\sqrt{q_{2}}}{\eta_{1}\eta_{2}\left(\mu_{c}+\mu_{e}\right)}},\:\sqrt{\frac{p_{1}\pm\sqrt{\left(p_{2}\right)^{2}-p_{3}}}{6\,\eta_{1}\eta_{3}\left(\lambda_{e}+2\mu_{e}\right)}}\right\} ,\label{eq:as oriz variab interne}
\end{gather}
where 
\begin{align*}
q_{1} & =\eta_{1}\mu_{c}\,\mu_{e}+\eta_{2}\left(\mu_{c}\left(\mu_{e}+\mh\right)+\mu_{e}\,\mh\right),\\
q_{2} & =\left(\left(\eta_{1}+\eta_{2}\right)\mu_{c}\,\mu_{e}+\eta_{2}\,\mh\left(\mu_{c}+\mu_{e}\right)\right){}^{2}-4\,\eta_{1}\,\eta_{2}\,\mu_{c}\,\mu_{e}\,\mh\left(\mu_{c}+\mu_{e}\right),
\end{align*}
and
\begin{align*}
p_{1} & =2\eta_{3}\left(3\lambda_{e}\left(\mu_{e}+\mh\right)+2\mu_{e}\left(\mu_{e}+3\mh\right)\right)+\eta_{1}\left(3\lambda_{e}\left(4\mu_{e}+3\lh+2\mh\right)\right)\\
 & \qquad+2\eta_{3}2\mu_{e}\left(4\mu_{e}+9\lh+6\mh\right),\\
p_{2} & =2\eta_{3}\left(3\lambda_{e}\left(\mu_{e}+\mh\right)+2\mu_{e}\left(\mu_{e}+3\mh\right)\right)+\eta_{1}\left(3\lambda_{e}\left(4\mu_{e}+3\lh+2\mh\right)\right)\\
 & \qquad+2\eta_{3}2\mu_{e}\left(4\mu_{e}+9\lh+6\mh\right),\\
p_{3} & =72\eta_{1}\eta_{3}\left(\lambda_{e}+2\mu_{e}\right)\big(\lambda_{e}\left(3\lh\left(\mu_{e}+\mh\right)+2\mh\left(3\mu_{e}+\mh\right)\right)\\
 & \qquad+2\mu_{e}\left(\lh\left(\mu_{e}+3\mh\right)+2\mh\left(\mu_{e}+\mh\right)\right)\big).
\end{align*}
We set 
\[
\omega_{l}^{int}=\sqrt{\frac{p_{1}-\sqrt{\left(p_{2}\right)^{2}-p_{3}}}{6\,\eta_{1}\eta_{3}\left(\lambda_{e}+2\mu_{e}\right)}},\quad{\displaystyle \omega_{t}^{int}=\sqrt{\frac{q_{1}-\sqrt{q_{2}}}{\eta_{1}\eta_{2}\left(\mu_{c}+\mu_{e}\right)}}},\quad\omega_{1}^{int}=\sqrt{\frac{p_{1}+\sqrt{\left(p_{2}\right)^{2}-p_{3}}}{6\,\eta_{1}\eta_{3}\left(\lambda_{e}+2\mu_{e}\right)}},\quad{\displaystyle \omega_{2}^{int}=\sqrt{\frac{q_{1}+\sqrt{q_{2}}}{\eta_{1}\eta_{2}\left(\mu_{c}+\mu_{e}\right)}}}.
\]
Even if we leave not explicitly proven that the dispersion curves
are monotonically increasing for all values of the constitutive parameters,
we checked that it is indeed the case for a large number of numerical
values of the parameters respecting positive definiteness of the strain
energy density. Moreover, for all the checked values of the parameters,
the values of $\omega_{*}$ computed by setting the coefficient of
the higher order of $k$ appearing in the polynomial (\ref{eq:poly})
to be equal to zero, (see (\ref{eq:asintoti orizzontali}) for the
relaxed micromorphic model and (\ref{eq:as oriz variab interne})
for the internal variable one) are always seen to be the values of
the horizontal asymptotes of the bounded dispersion curves. Hence,
even if we do not have an explicit proof that setting $c_{18}=0$
(or $c_{6}=0$ for the internal variable model) is also a sufficient
condition for horizontal asymptotes, this is indeed the case for all
combinations of the parameters which are sensible to be interesting
for applications. We explicitly remark that the horizontal asymptotes
shown in (\ref{eq:asintoti orizzontali}) are those found for the
relaxed micromorphic model with $\alpha_{1},\alpha_{2},\alpha_{3}\neq0$,
while those shown in (\ref{eq:as oriz variab interne}) are relative
to the internal variable case $\alpha_{1},\alpha_{2},\alpha_{3}=0$.
We notice that, as shown in \cite{madeo2016reflection}, the horizontal
asymptotes for the internal variable model significantly differ from
those obtained with the full non-local model (with non-vanishing $\alpha_{1}$,$\alpha_{2}$
and $\alpha_{3}$). This means that the internal variable model is
a pathological limit of the relaxed micromorphic model, in the sense
that setting to zero $\alpha_{1}$,$\alpha_{2}$ and $\alpha_{3}$
drastically changes the asymptotic properties of the dispersion curves.

\subsubsection{Tangents in 0 to the acoustic curves}

Another very important geometric characteristics of the dispersion
curves are the slopes at the origin of the acoustic branches. In this
way we can also directly compare our relaxed model to classical isotropic
linear elasticity. In the case that we are studying in this paper,
the direct computation of these quantities, given the great complexity
of the involved expressions, is impossible. Therefore we work with
the implicit function theorem applied to the expression of the determinant
equation. First of all, we remark that the matrix $\mathsf{E}_{4}$
given in (\ref{eq:E4}) cannot generate acoustic branches since ${\textstyle \det\,\mathsf{E}_{4}\left(0,0\right)}\neq0$
(when\footnote{If $\mc=0$, having that $\omega_{r}=0$, one of the uncoupled branches
becomes acoustic. } $\mc>0$). Thus the two independent acoustic branches are generated
by the matrices $\mathsf{E}_{1}$ and $\mathsf{E}_{2}$. The acoustic
branches are those solutions $\widehat{\omega}{}_{\textrm{aco},\alpha}\left(k\right)$
of the equations
\[
{\textstyle \det}\,\mathsf{E}_{\a}\left(k,\omega\right)=0,\qquad\alpha=1,2,
\]
such that $\widehat{\omega}{}_{\textrm{aco}}\left(0\right)=0$. It
can be checked that, for all $k\geq0$, the two independent acoustic
curves $\widehat{\omega}{}_{\textrm{aco};1}\left(k\right)$ and $\widehat{\omega}{}_{\textrm{aco};2}\left(k\right)$
verify the identities

\begin{equation}
0={\textstyle \det}\,\mathsf{E}_{\a}\left(k,\widehat{\omega}{}_{\textrm{aco};\alpha}\left(k\right)\right)=\sum_{p,q=1}^{3}\psi_{pq}^{\left(\alpha\right)}\left(\boldsymbol{m}\right)k^{2p}\widehat{\omega}_{\textrm{aco};\alpha}^{2q}\left(k\right)+\sum_{p=1}^{3}\varphi_{p}^{\left(\a\right)}\left(\boldsymbol{m}\right)k^{2p}+\sum_{q=1}^{3}\zeta_{q}^{\left(\a\right)}\left(\boldsymbol{m}\right)\widehat{\omega}_{\textrm{aco};\alpha}^{2q}\left(k\right)+\sigma^{\left(\a\right)}\left(\boldsymbol{m}\right)\label{eq:tang 0 2}
\end{equation}
for every $k\geq0$ and $\alpha\in\left\{ 1,2\right\} $, where $\psi_{pq}^{\left(\alpha\right)},\varphi_{p}^{\left(\a\right)},\zeta_{q}^{\left(\a\right)},\sigma^{\left(\a\right)}$
are real scalar functions of the vector of material parameters of
the model
\begin{equation}
\boldsymbol{m}=\left(\me,\mh,\le,\lh,\rho,\eta_{1},\eta_{2},\eta_{3},\a_{1},\a_{2},\a_{3},L_{c}\right).
\end{equation}
 In order to isolate the quantity $\widehat{\omega}'_{\textrm{aco};\a}\left(0\right)$,
which is the slope of the acoustic curves in $k=0$, we remark that
\begin{equation}
\forall\,k\geq0,\quad0=\frac{d}{dk}\left[{\textstyle \det}\,\mathsf{E}_{\a}\left(k,\widehat{\omega}{}_{\textrm{aco};\alpha}\left(k\right)\right)\right]=\frac{\partial}{\partial k}\left[{\textstyle \det}\,\mathsf{E}_{\a}\left(k,\widehat{\omega}{}_{\textrm{aco};\alpha}\left(k\right)\right)\right]+\frac{\partial}{\partial\omega}\left[{\textstyle \det}\,\mathsf{E}_{\a}\left(k,\widehat{\omega}{}_{\textrm{aco};\alpha}\left(k\right)\right)\right]\cdot\underbrace{\frac{d\,\widehat{\omega}{}_{\textrm{aco};\alpha}}{dk}\left(k\right)}_{=\widehat{\omega}{}_{\textrm{aco};\alpha}'\left(k\right)}.
\end{equation}
However, since we compute that $\frac{\partial}{\partial\omega}\left.{\textstyle \det}\,\mathsf{E}_{\a}\left(k,\widehat{\omega}{}_{\textrm{aco};\alpha}\left(k\right)\right)\right|_{k=0}=0$,
the latter relation does not give any condition on $\widehat{\omega}'_{\textrm{aco};\a}\left(0\right).$
For this reason we perform also the second derivative (the calculations
are given in Appendix 3) finding then
\begin{gather}
\widehat{\omega}'_{\textrm{aco};\a}\left(0\right)=\sqrt{\frac{-\,\varphi_{1}^{\left(\a\right)}\left(\boldsymbol{m}\right)}{\zeta_{1}^{\left(\a\right)}\left(\boldsymbol{m}\right)}}
\end{gather}
with
\begin{align}
\varphi_{1}^{\left(1\right)}\left(\boldsymbol{m}\right) & =2\lambda_{e}\left(3\lh\left(\mu_{e}+\mh\right)+2\mh\left(3\mu_{e}+\mh\right)\right),\\
 & \quad+4\mu_{e}\left(\lh\left(\mu_{e}+3\mh\right)+2\mh\left(\mu_{e}+\mh\right)\right)\\
\zeta_{1}^{\left(1\right)}\left(\boldsymbol{m}\right) & =-2\varrho\left(\mu_{e}+\mh\right)\left(2\left(\mu_{e}+\mh\right)+3\lambda_{e}+3\lh\right),
\end{align}
and
\begin{equation}
\zeta_{1}^{\left(2\right)}\left(\boldsymbol{m}\right)\varphi_{1}^{\left(2\right)}\left(\boldsymbol{m}\right)=-4\,\varrho\,\mu_{c}\left(\mu_{e}+\mh\right),\qquad=4\,\mu_{c}\,\mu_{e}\,\mh.
\end{equation}
Recalling, the final expressions for the two slopes are
\begin{align}
\widehat{\omega}'_{\textrm{aco};1}\left(0\right) & =\sqrt{\frac{\lambda_{e}\left(3\lh\left(\mu_{e}+\mh\right)+2\mh\left(3\mu_{e}+\mh\right)\right)+2\mu_{e}\left(\lh\left(\mu_{e}+3\mh\right)+2\mh\left(\mu_{e}+\mh\right)\right)}{\rho\left(\mu_{e}+\mh\right)\left(2\left(\mu_{e}+\mh\right)+3\left(\lambda_{e}+\lh\right)\right)}}\nonumber \\
\label{eq:slopes}\\
\widehat{\omega}'_{\textrm{aco};2}\left(0\right) & =\sqrt{\frac{\mu_{e}\,\mh}{\rho\left(\mu_{e}+\mh\right)}}.\nonumber 
\end{align}
Remembering the basic relations in \cite{barbagallo2016transparent,neff2007geometrically}
between our relaxed micromorphic parameters and that 
\begin{equation}
\mum=\frac{\mu_{e}\,\mh}{\mu_{e}+\mh},\qquad\lambda_{\textrm{macro}}=\frac{1}{3}\frac{\left(2\mu_{e}+3\lambda_{e}\right)\left(2\mh+3\lh\right)}{2\left(\mu_{e}+\mh\right)+3\left(\lambda_{e}+\lh\right)}-\frac{2}{3}\frac{\mu_{e}\,\mh}{\mu_{e}+\mh},\label{eq:micro macro}
\end{equation}
the equations in (\ref{eq:slopes}) can be neatly written as
\begin{equation}
\widehat{\omega}'_{\textrm{aco};1}\left(0\right)=\sqrt{\frac{2\mum+\lambda_{\textrm{macro}}}{\rho}},\qquad\qquad\widehat{\omega}'_{\textrm{aco};2}\left(0\right)=\sqrt{\frac{\mum}{\rho}}.\label{eq:acoustic macro}
\end{equation}
Based on this result we see that the tangents to the acoustic curves
fully recover the format of classical isotropic linear elasticity,
if the latter model is taken with parameters $\mum,\lambda_{\textrm{macro}}$.
This result should be compared with \cite[ eq. 8.13]{mind1963} where
Mindlin also obtained the tangents of the transverse and longitudinal
acoustic branches in 0. However, his results for the more general
micromorphic model do not support the transparency of (\ref{eq:acoustic macro}).

\section{Action of the material parameters on the behavior of the dispersion
curves}

In the relaxed micromorphic model presented in this work, we have
considered the splitting, following the Lie-Cartan decomposition of
$\gl$, of the micro-inertia and of the potential part related to
$\curl\,P$. This allows us to separate the governing behavior of
the deformation mechanisms associated to the pure deformation, the
volumetric expansion and the rotation of the microstructure. In this
way, we can directly explore how each of these deformation modes affects
the behavior of the dispersion curves. In the following, we show how
varying independently the material parameters $\eta_{1},\eta_{2},\eta_{3}$
and $\alpha_{1},\alpha_{2},\alpha_{3}$ we can act on the behavior
of the dispersion curves with more freedom with respect to the non-weighted
relaxed micromorphic model presented in \cite{madeo2015wave}.  We
have numerically solved the equation ${\textstyle \det}\,\D\left(k,\omega\right)=0$
looking for solutions of the type $\widehat{\omega}_{i}=\widehat{\omega}_{i}\left(k\right)$
which are curves of the considered medium. In this section we analyze
the obtained solutions for different choices of the material parameters
in order to highlight their effect on the behavior of the dispersion
curves. The material parameters in Table \ref{tab:Material-parameters1}
are those used in the simulations if not differently specified.

\begin{table}[H]
\centering{}%
\begin{tabular}{ccccccc}
\noalign{\vskip1mm}
$\mu_{e}$ & $\lambda_{e}$ & $\mh$ & $\lh$ & $\mu_{c}$ & $L_{c}$ & $\rho$\tabularnewline[1mm]
\hline 
\hline 
\noalign{\vskip2mm}
$200$ & 400 & 100 & 100 & 440 & $3$ & $2000$\tabularnewline[1mm]
\hline 
\noalign{\vskip1mm}
$\textrm{MPa}$ & $\mathrm{MPa}$ & $\textrm{MPa}$ & $\textrm{MPa}$ & $\textrm{MPa}$ & $\textrm{mm}$ & $\textrm{kg/}\textrm{m}^{3}$\tabularnewline[1mm]
\hline 
\end{tabular}\caption{Values of the material parameters used in the numerical determination
of the dispersion curves.\label{tab:Material-parameters1}}
\end{table}

\subsection{Classical results}

Classical linear elasticity and the Cosserat model can be obtained
as limit cases starting from the relaxed micromorphic one. In order
to directly compare the dispersion curves of these models with those
of the relaxed one, we present them in this sub-section. We briefly
recall that the strain energy density for the classical linear elasticity
is given by
\[
W_{\textrm{macro}}\left(\nabla u\right)=\mum\left\Vert \sym\,\grad\u\right\Vert ^{2}+\frac{\lam}{2}\left(\textrm{tr}\,\grad\u\right)^{2}.
\]
Moreover, the kinetic energy density used for the classical model
is clearly
\[
J_{\textrm{macro}}=\frac{1}{2}\,\rho\left\langle u_{,t},u_{,t}\right\rangle .
\]
Classical Cauchy linear elasticity gives rise to the dispersion curves
presented in Fig. \ref{fig:Cauchy - Cosserat} (left). The dispersion
curves reduce to straight lines, which means that the speed of propagation
of waves is independent of the frequency of the traveling waves. The
slopes of the dispersion curves are given by $\sqrt{\frac{\mum}{\rho}}$
for transversal waves and $\sqrt{\frac{2\mum+\lam}{\rho}}$ for longitudinal
waves ($\mum$ and $\lam$ are the Lamé parameters of the considered
Cauchy continuum).

As for the weighted Cosserat model, it features a strain energy density
of the type:
\begin{align}
W_{\textrm{cos}} & =\me\left\Vert \sym\,\nabla u\right\Vert ^{2}+\mc\left\Vert \skew\left(\nabla u-\P\right)\right\Vert ^{2}+\frac{\le}{2}\left(\textrm{tr}\,\nabla u\right)^{2}\label{eq:coss1}\\
 & \quad+\mu_{e}\,\frac{L_{c}^{2}}{2}\left(\alpha_{1}\left\Vert \ds\,\curl\,\skew\,\P\right\Vert ^{2}+\alpha_{2}\left\Vert \skew\,\curl\,\skew\,\P\right\Vert ^{2}+\frac{1}{3}\,\alpha_{3}\left(\tr\,\curl\,\skew\,\P\right)^{2}\right),\nonumber 
\end{align}
where $\P$ is constrained to be skew-symmetric. The kinetic energy
considered for the Cosserat model takes the form:
\begin{equation}
J_{\textrm{cos}}=\frac{\rho}{2}\left\Vert u_{,t}\right\Vert ^{2}+\frac{\eta_{2}}{2}\left\Vert \skew\,\P_{,t}\right\Vert .\label{eq:coss2}
\end{equation}
Following the same procedures shown in section 4 for the relaxed micromorphic
model, the study of plane wave propagation in Cosserat media (we give
the corresponding Euler-Lagrange equations in the appendix) gives
rise to dispersion curves of the type shown in Fig.\ref{fig:Cauchy - Cosserat}
(right). The values of the asymptotes are given by

\begin{equation}
c_{1}^{\textrm{cos}}=\sqrt{\frac{\alpha_{1}+2\,\alpha_{3}}{3\,\eta_{2}}\,\me\,L_{c}^{2}}=c_{\textrm{m}}^{\textrm{vd}},\qquad\qquad c_{2}^{\textrm{cos}}=\sqrt{\frac{\mc+\me}{\rho}}=c_{\textrm{s}}^{\textrm{}},
\end{equation}
and for the values of the slopes of the acoustic branches we find
\begin{equation}
c_{3}^{\textrm{cos}}=\frac{1}{2}\,\sqrt{\frac{\alpha_{1}+\alpha_{2}}{\eta_{2}}\,\me\,L_{c}^{2}}=\lim_{\eta_{1}\fr\infty}c_{\textrm{m}}^{\textrm{dr}},\qquad\qquad c_{4}^{\textrm{cos}}=\sqrt{\frac{\le+2\,\me}{\rho}}=c_{\textrm{p}}.
\end{equation}
We will show in what follows that even if the Cosserat model allows
to account for some dispersion at higher frequencies, the behavior
of dispersion curves related to the relaxed micromorphic model is
richer as it allows for the description of band-gaps and account for
more complex microstructure motions. We will finally show that both
the classical cases presented in this subsection (Cauchy and Cosserat)
can be obtained as degenerate limit cases of the relaxed micromorphic
model.
\begin{figure}[H]
\begin{centering}
\begin{tabular}{ccc}
\includegraphics[scale=0.5]{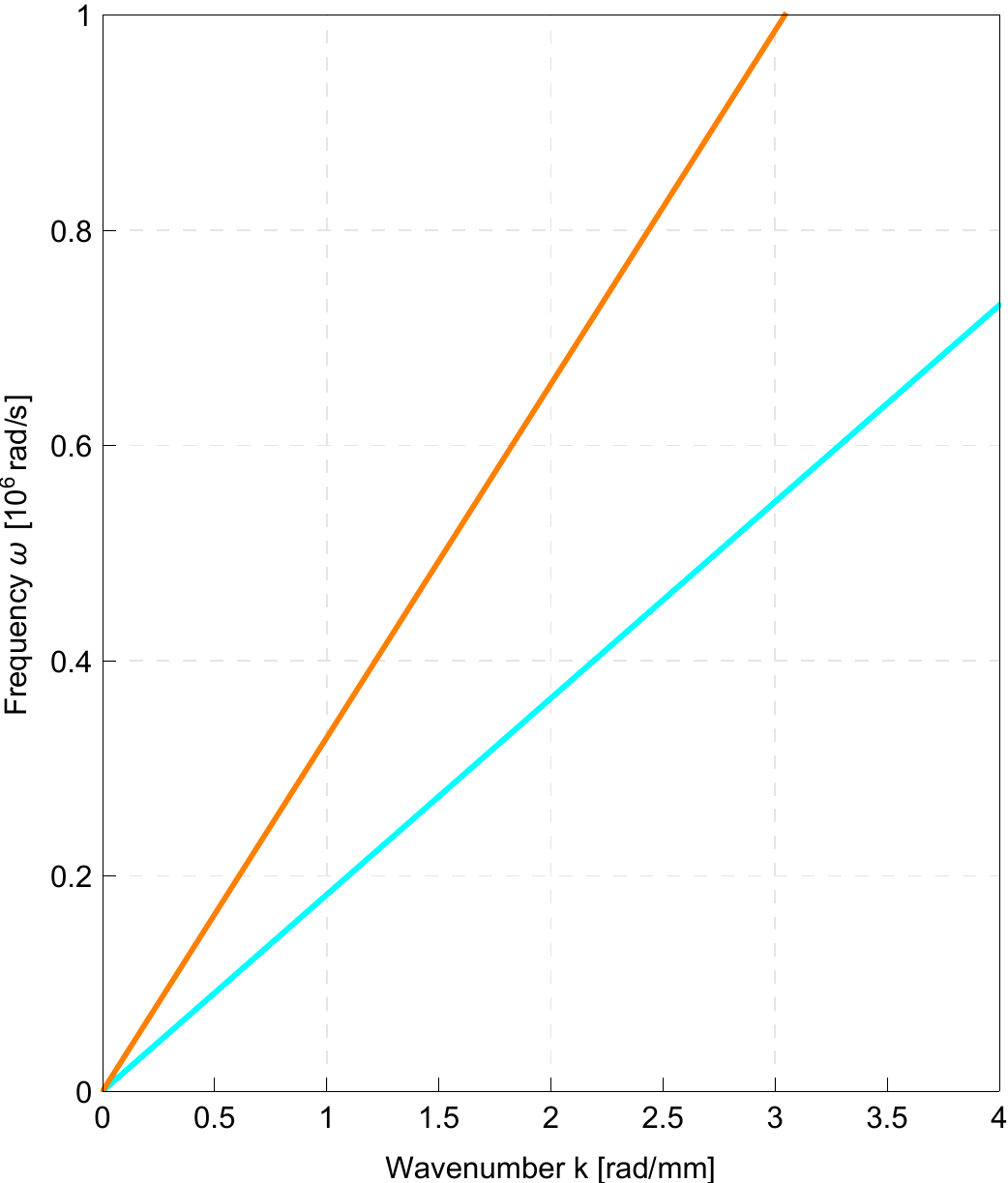} &  & \includegraphics[scale=0.5]{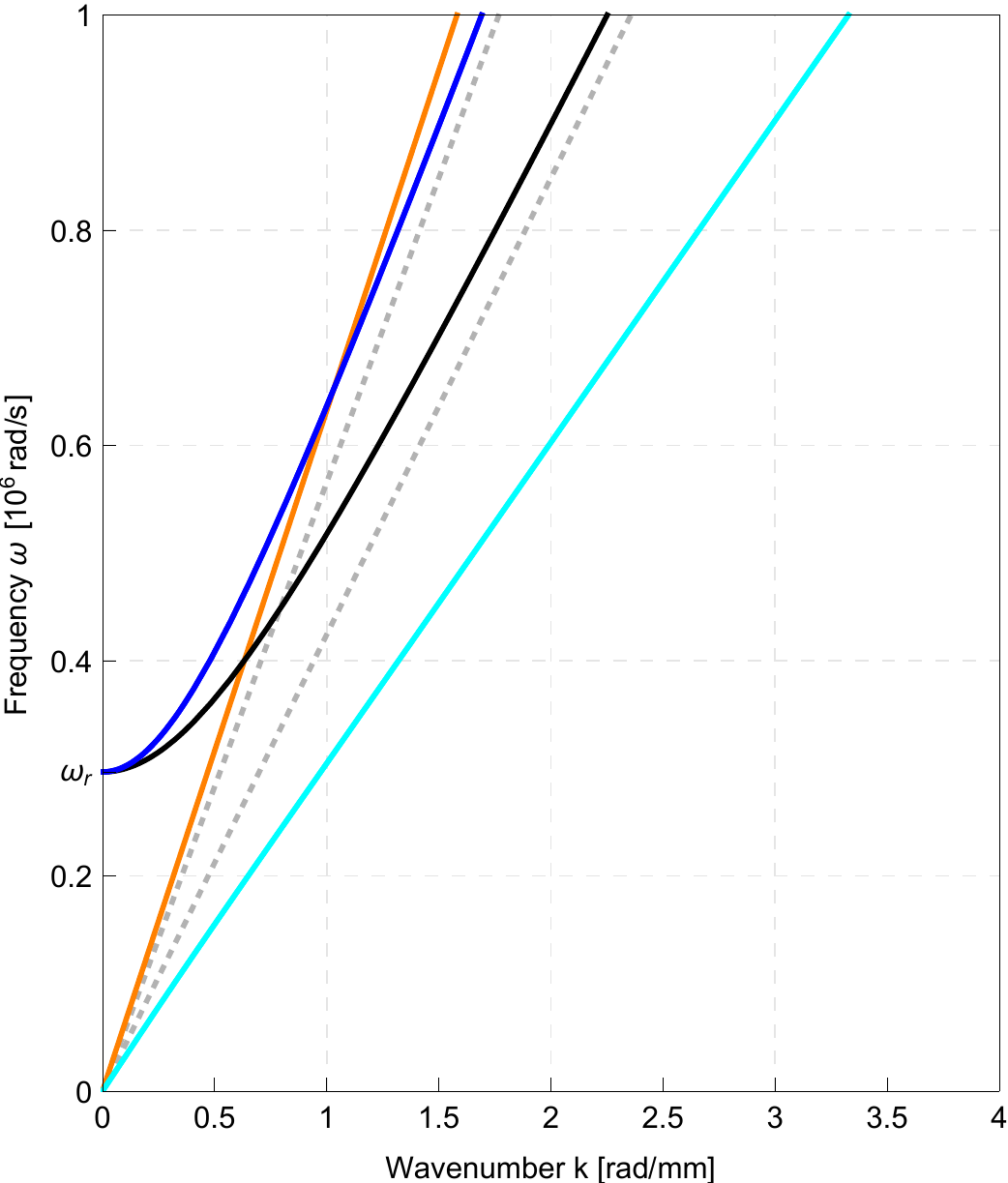}\tabularnewline
(a) &  & (b)\tabularnewline
\end{tabular}
\par\end{centering}
\caption{\label{fig:Cauchy - Cosserat}Dispersion curves for linear elasticity
(a) and the Cosserat model with $\eta_{2}=10^{-2}\left[\textrm{Kg}/\textrm{m}\right]$
(b).}
\end{figure}

\subsection{The classical relaxed micromorphic model and the internal variable
model}

Before proceeding, we report some results already available in the
literature \cite{madeo2015wave,madeo2016reflection} which we obtain
from our weighted model setting $\eta_{1}=\eta_{2}=\eta_{3}$ and
$\alpha_{1}=\alpha_{2}=\alpha_{3}$ (classical relaxed micromorphic
case). Moreover, we also study the limit case that is obtained setting
$\alpha_{1}=\alpha_{2}=\alpha_{3}=0$, which is also known as internal
variable model (no derivatives of the micro-distortion $\P$ appearing
in the strain energy density). The dispersion curves for the two models
are shown in Fig. \ref{fig:classical}.

\begin{figure}[h]
\centering{}%
\begin{tabular}{ccc}
\includegraphics[scale=0.55]{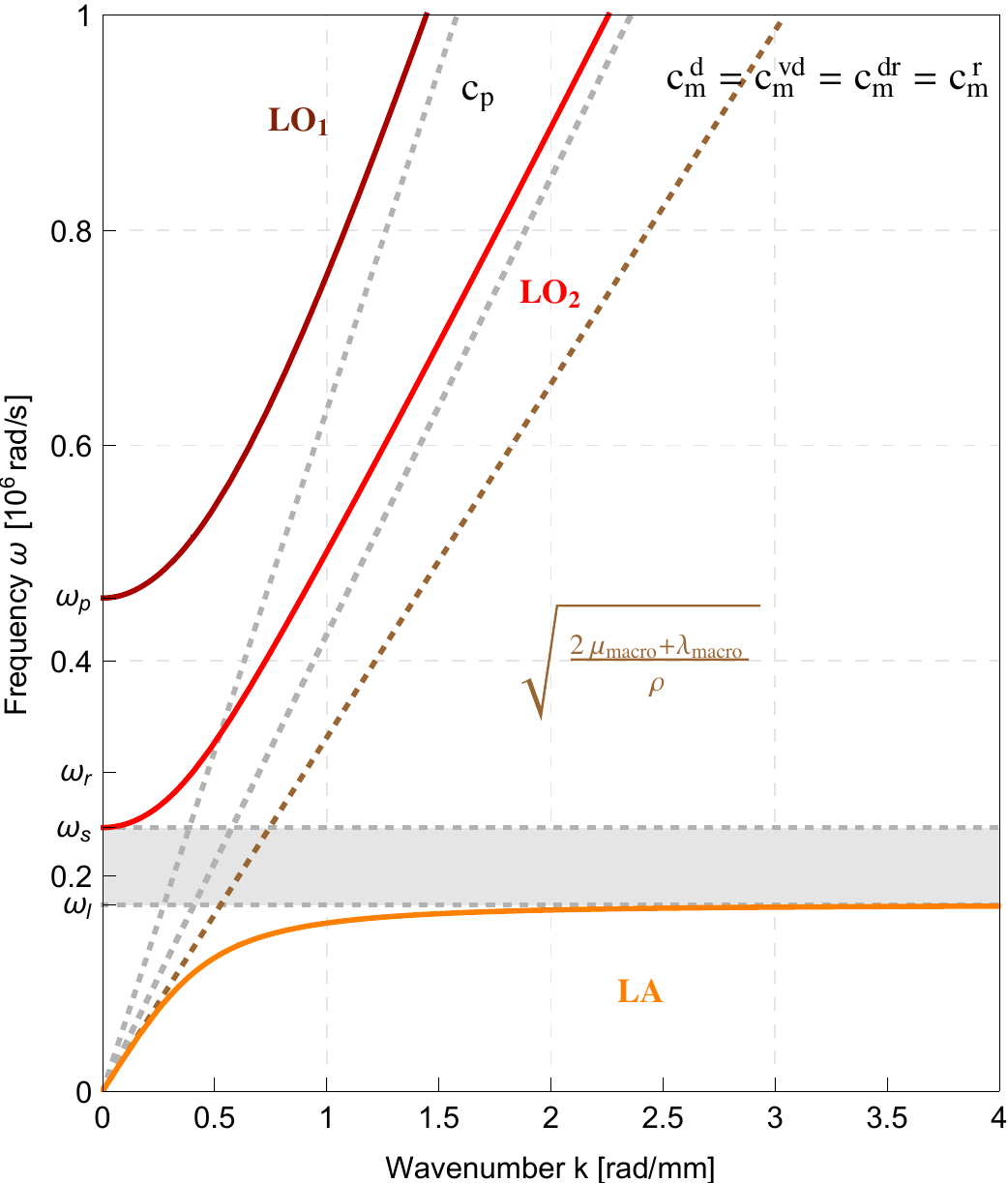} & \includegraphics[scale=0.55]{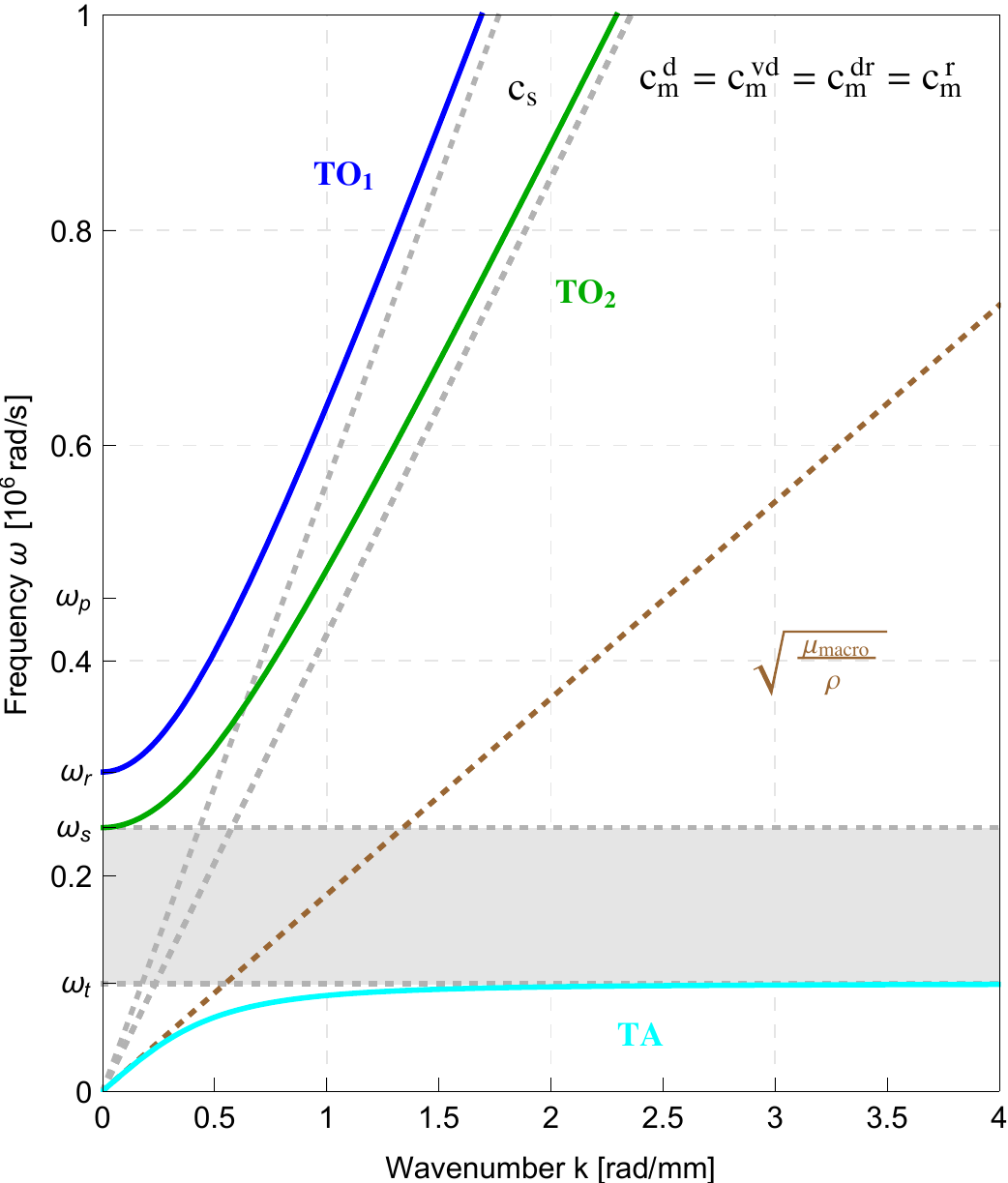} & \includegraphics[scale=0.55]{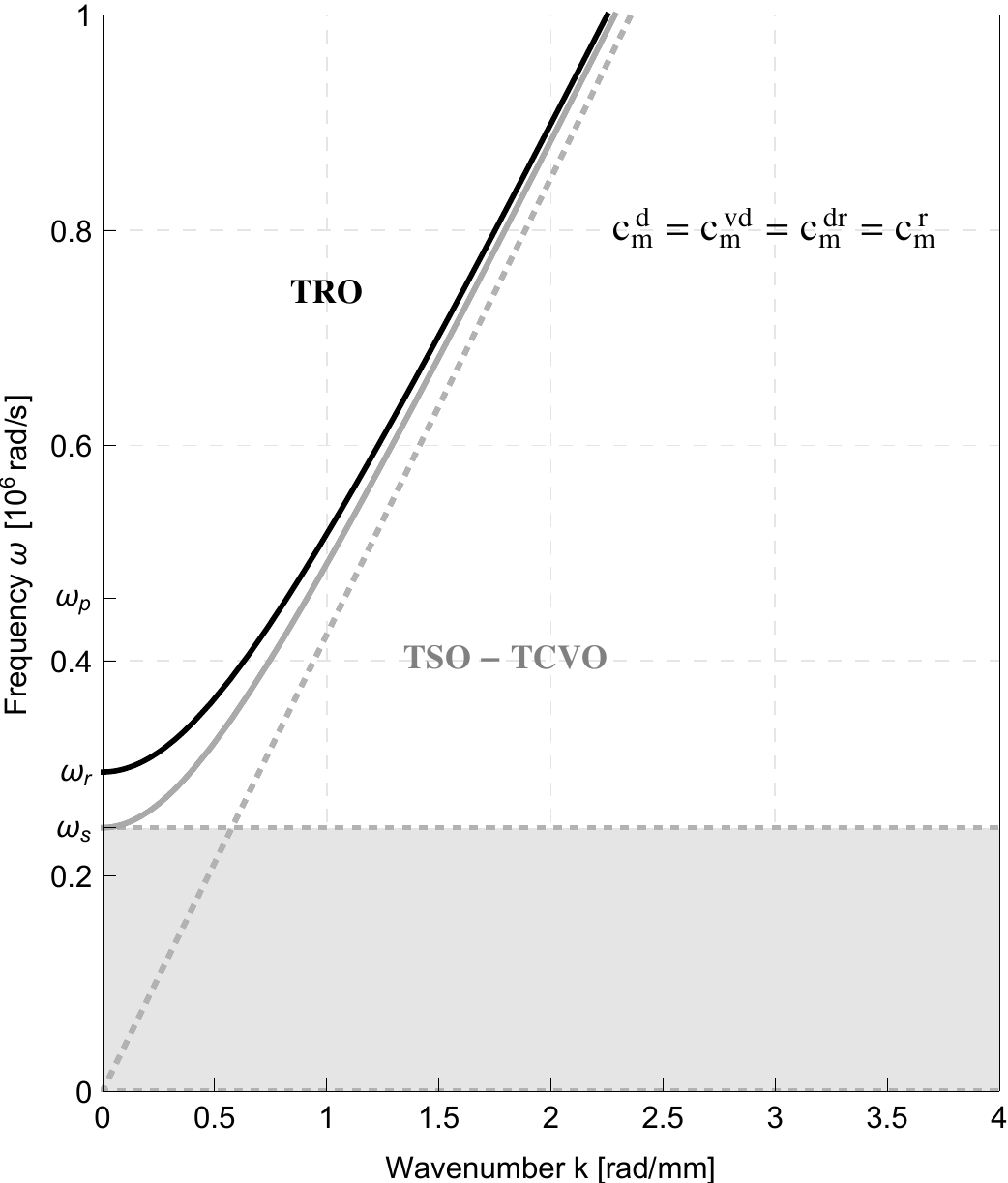}\tabularnewline
(a1) longitudinal dispersion curves & (a2) transversal dispersion curves & (a3) uncoupled dispersion curves\tabularnewline
\includegraphics[scale=0.55]{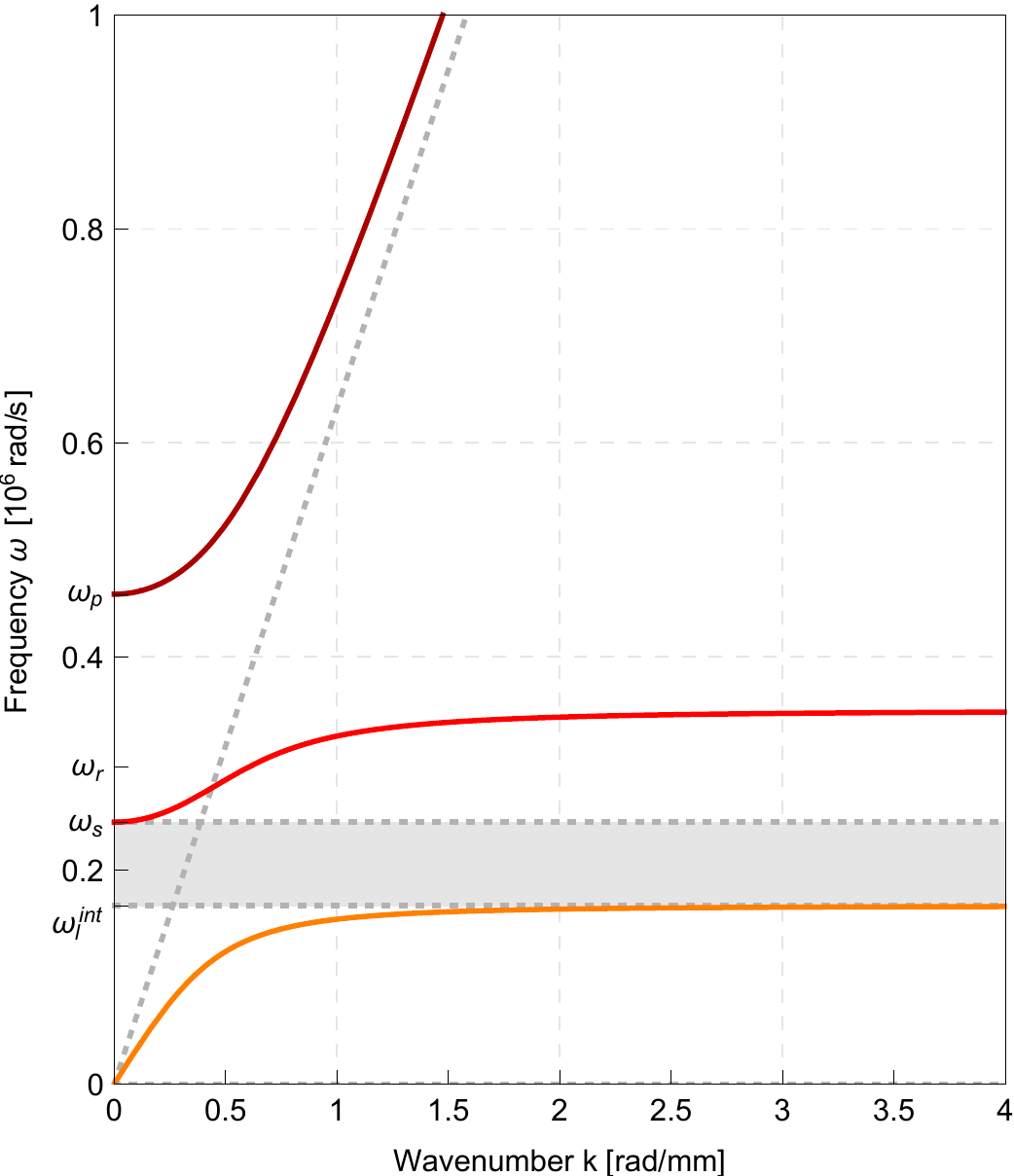} & \includegraphics[scale=0.55]{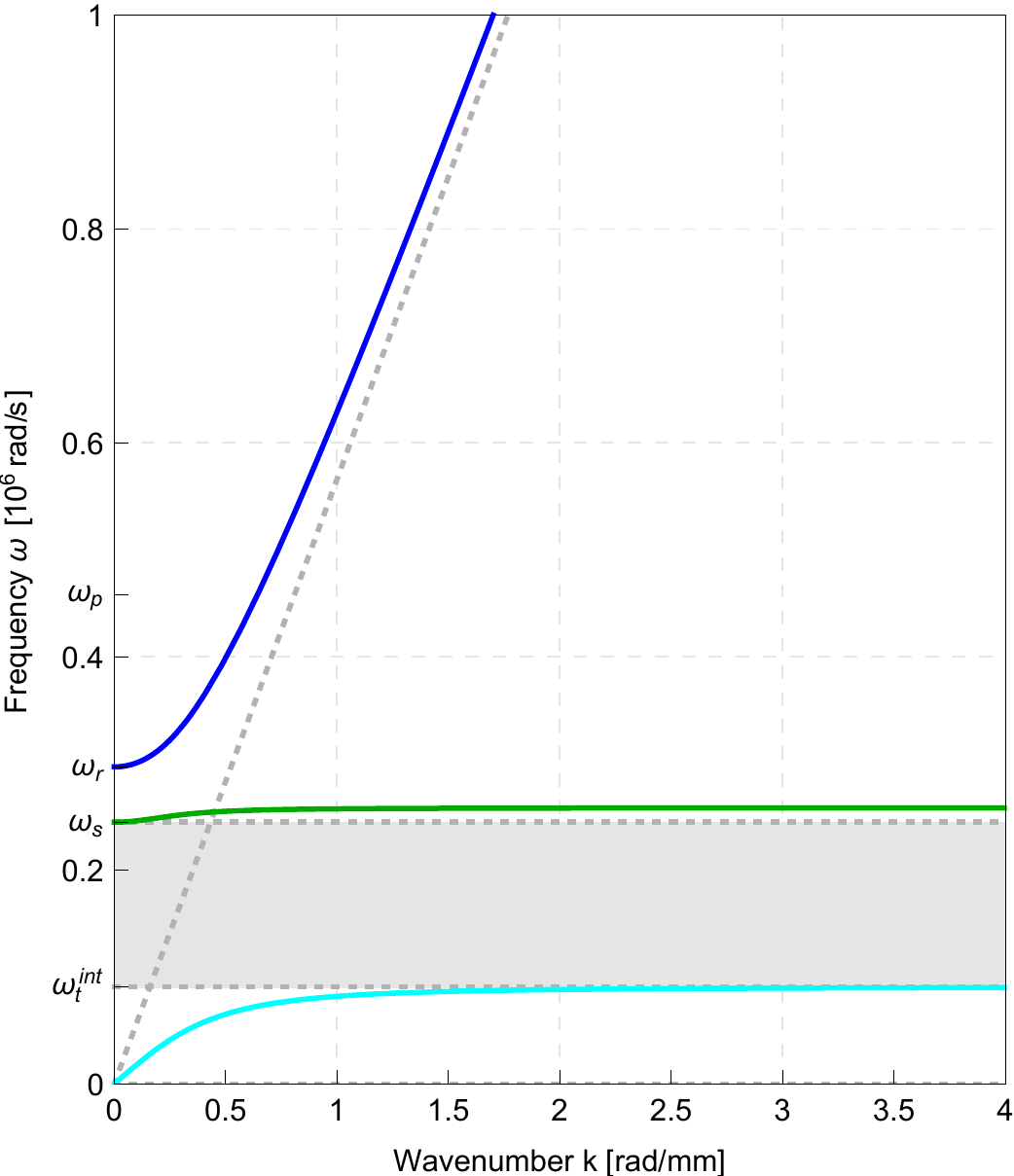} & \includegraphics[scale=0.55]{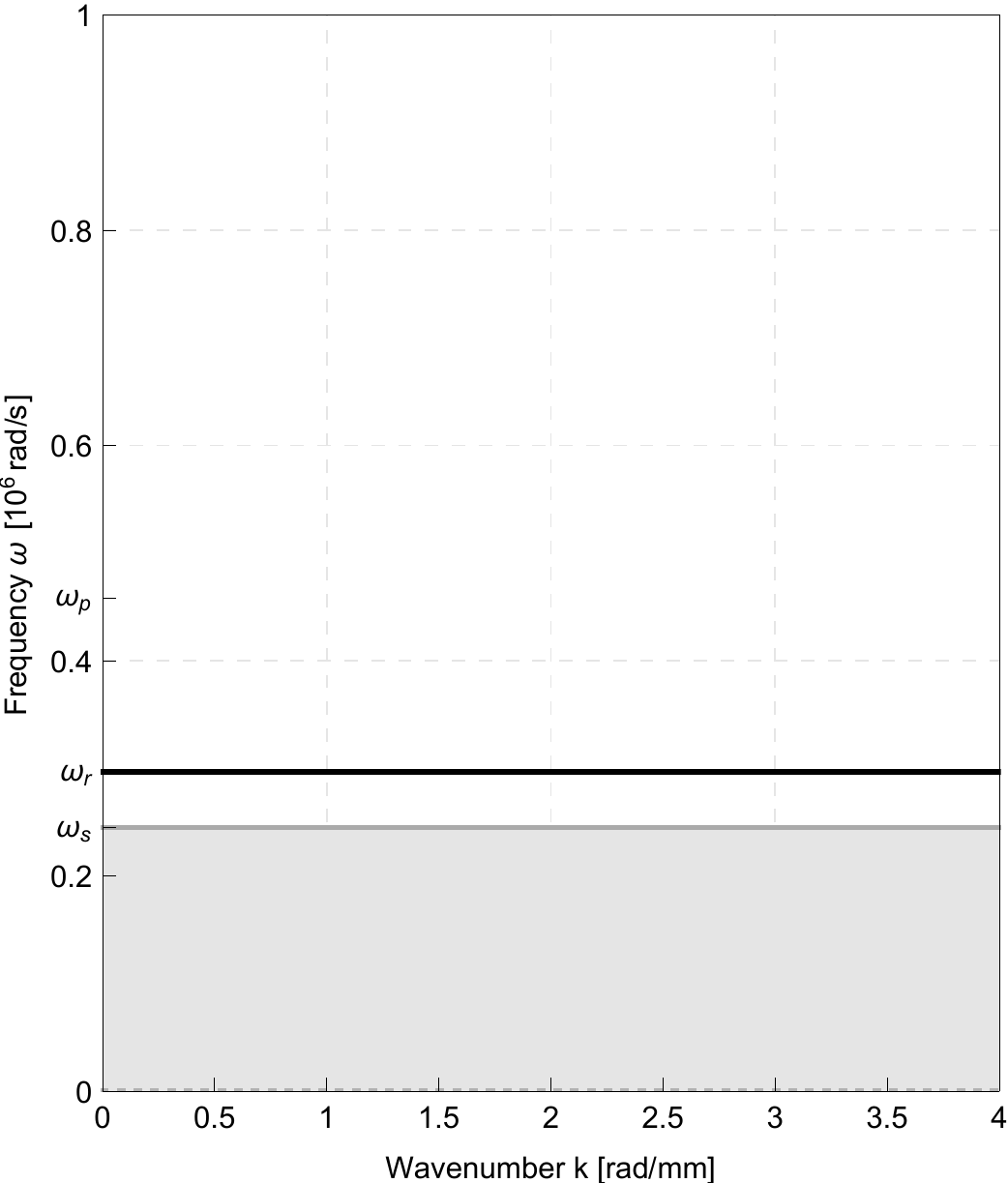}\tabularnewline
(b1) longitudinal dispersion curves & (b2) transversal dispersion curves & (b3) uncoupled dispersion curves\tabularnewline
\end{tabular}\caption{\label{fig:classical}Dispersion curves for the classical relaxed
micromorphic case $\alpha_{1}=\alpha_{2}=\alpha_{3}=1\quad\textrm{and}\quad\eta_{1}=\eta_{2}=\eta_{3}=10^{-2}\,Kg/m$
(a1,a2,a3), and for the internal variable case: $\alpha_{1}=\alpha_{2}=\alpha_{3}=0\quad\textrm{and}\quad\eta_{1}=\eta_{2}=\eta_{3}=10^{-2}Kg/m$
(b1,b2,b3).\label{fig:Dispersion-curves-for}}
\end{figure}

In both cases, we find 8 different dispersion curves instead of 12.
Indeed, we would expect 12 curves since the kinematics is given by
the 3 components of the displacement field $u$, plus the 9 components
of the micro-distortion field $P$. This means that there are overlapped
curves (our check exhibits 4 couples of overlapping curves). The band-gap
region is clearly present in both models. These results are fully
compatible with what has been found in \cite{madeo2015wave} and in
\cite{madeo2016reflection}. In particular, we find the same behavior
presented in \cite{madeo2015wave} and \cite{madeo2016reflection},
where sufficient conditions on the constitutive parameters have been
found that guarantee the existence of band gaps. Moreover, for the
internal-variable case, it has been shown in \cite{madeo2016reflection}
that two optic branches become horizontal and four horizontal asymptotes
can be found which are not related with the two horizontal asymptotes
of the case $\alpha_{1}=\alpha_{2}=\alpha_{3}\neq0$. 

In the present work, in order to present a parametric study on the
behavior of the dispersion curves varying a great number of material
parameters, we have decided to represent, in the newly studied cases,
all the dispersion curves in the same picture associating to every
curve a specific color. With respect to the nomenclature states in
\cite{madeo2015wave} the correspondence with our choice of colors
is as follows (Table \ref{tab:Nomenclature-and-colors}):
\begin{table}[H]
\begin{centering}
\begin{tabular}{lllllllll}
\noalign{\vskip2mm}
longitudinal curves &  &  &  & transversal curves &  &  &  & uncoupled curves\tabularnewline[2mm]
\hline 
\hline 
\noalign{\vskip2mm}
$\textcolor{DarkRed}{\LARGE{\newmoon}}$$\;$Dark Red $\rightarrow$
$\mathbf{LO}_{\boldsymbol{1}}$ &  &  &  & ${\color{blue}\LARGE{}\newmoon}\;$Blue $\rightarrow$ $\mathbf{TO}_{\boldsymbol{1}}$ &  &  &  & $\LARGE{}\newmoon\;$Black $\rightarrow$ $\mathbf{TRO}$\tabularnewline[2mm]
\hline 
\noalign{\vskip2mm}
${\color{red}\LARGE{}\CIRCLE}\;$Red $\rightarrow$ $\mathbf{LO}_{\boldsymbol{2}}$ &  &  &  & $\textcolor{Green}{\LARGE{\newmoon}}$$\;$Green $\rightarrow$ $\mathbf{TO}_{\boldsymbol{2}}$ &  &  &  & ${\color{gray}\LARGE{}\newmoon}\;$Gray $\rightarrow$ $\mathbf{LSO}$\tabularnewline[2mm]
\hline 
\noalign{\vskip2mm}
${\color{orange}\LARGE{}\newmoon}\;$Orange $\rightarrow$ $\mathbf{LA}$ &  &  &  & $\textcolor{Cyan}{\LARGE{\newmoon}}$$\;$Cyan $\rightarrow$ $\mathbf{TA}$ &  &  &  & ${\color{gray}\LARGE{}\newmoon}\;$Gray $\rightarrow$ $\mathbf{TCVO}$\tabularnewline[2mm]
\end{tabular}
\par\end{centering}
\caption{Nomenclature and colors for dispersion curves\label{tab:Nomenclature-and-colors}}

\end{table}

\subsection{A panorama of dispersion curves for the weighted relaxed micromorphic
model.}

In this subsection we present a complete panorama of the dispersion
curves associated to the weighted relaxed micromorphic model highlighting
the effect of each of the weights $\eta_{1},\eta_{2},\eta_{3},\alpha_{1},\alpha_{2}$
and $\alpha_{3}$ on the dispersion curves themselves. If not differently
specified, the reference values of the weights are those given in
the following table:
\begin{table}[H]
\centering{}%
\begin{tabular}{cccccc}
\noalign{\vskip1mm}
$\alpha_{1}$ & $\alpha_{2}$ & $\alpha_{3}$ & $\eta_{1}$ & $\eta_{2}$ & $\eta_{3}$\tabularnewline[1mm]
\hline 
\hline 
\noalign{\vskip2mm}
$1$ & $1$ & $1$ & $10^{-2}$ & $10^{-2}$ & $10^{-2}$\tabularnewline[1mm]
\hline 
\noalign{\vskip1mm}
- & - & - & $\left[\textrm{Kg}/\textrm{m}\right]$ & $\left[\textrm{Kg}/\textrm{m}\right]$ & $\left[\textrm{Kg}/\textrm{m}\right]$\tabularnewline[1mm]
\hline 
\end{tabular}\caption{Values of the material parameters used in the numerical determination
of the dispersion curves.\label{tab:weights-parameters}}
\end{table}

\subsubsection{Case ${\displaystyle \mu_{c}>0,\lim_{\alpha_{1}\protect\fr0}}$}

Characteristic limit elastic energy $\left\Vert \nabla u-P\right\Vert ^{2}+\left\Vert \sym\,P\right\Vert ^{2}+\left\Vert \skew\,\curl\,P\right\Vert ^{2}+\frac{1}{3}\left(\textrm{tr}\,\curl\,P\right)^{2}$.
\\
Characteristic limit kinetic energy $\left\Vert u_{,t}\right\Vert ^{2}+\left\Vert P_{,t}\right\Vert ^{2}$.

\begin{figure}[H]
\centering{}%
\begin{tabular}{ccc}
\includegraphics[scale=0.5]{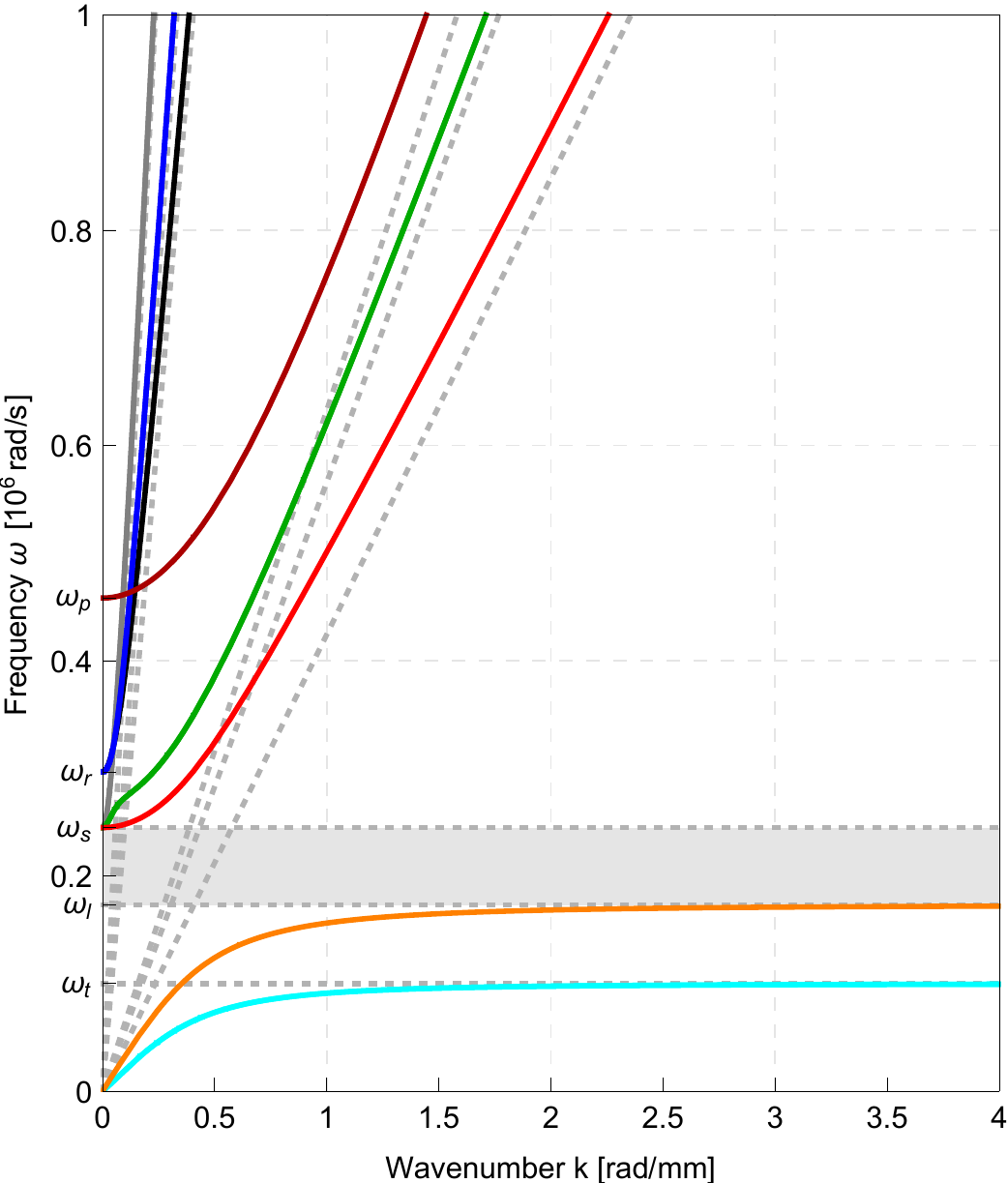} & \includegraphics[scale=0.5]{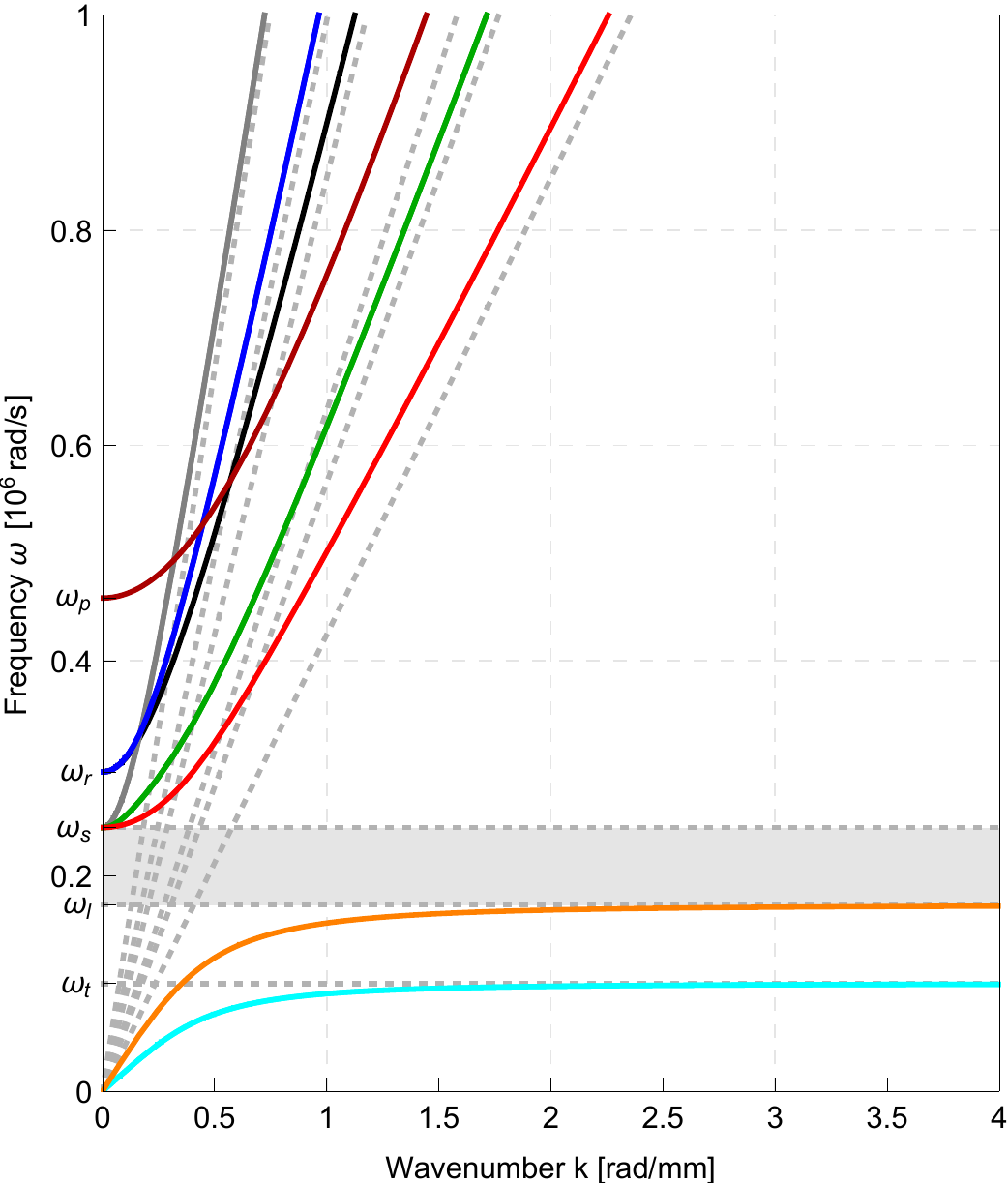} & \includegraphics[scale=0.5]{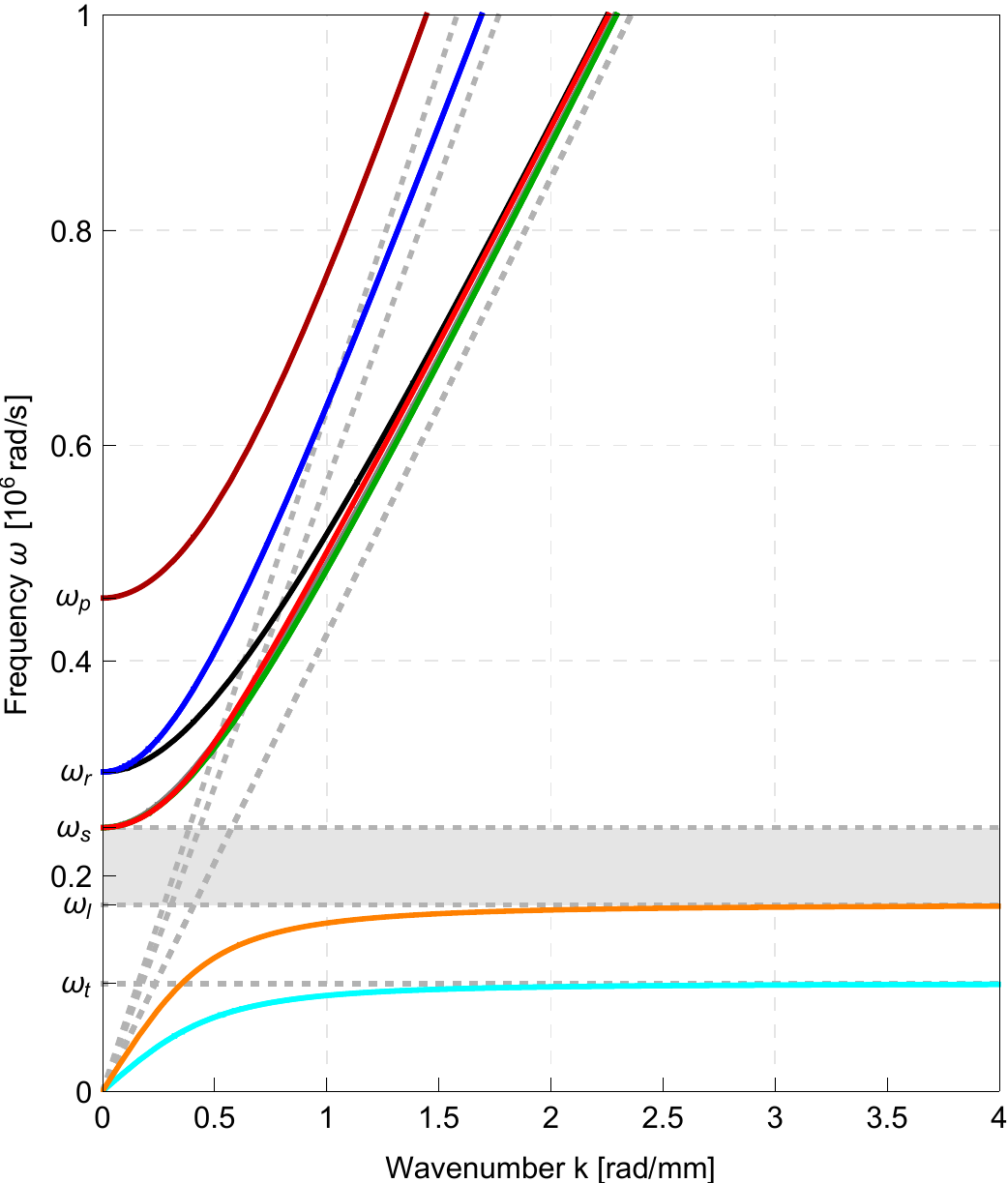}\tabularnewline
$\alpha_{1}=100$ & $\alpha_{1}=10$ & $\alpha_{1}=1$\tabularnewline
\end{tabular}
\end{figure}
\begin{figure}[H]
\centering{}%
\begin{tabular}{ccc}
\includegraphics[scale=0.5]{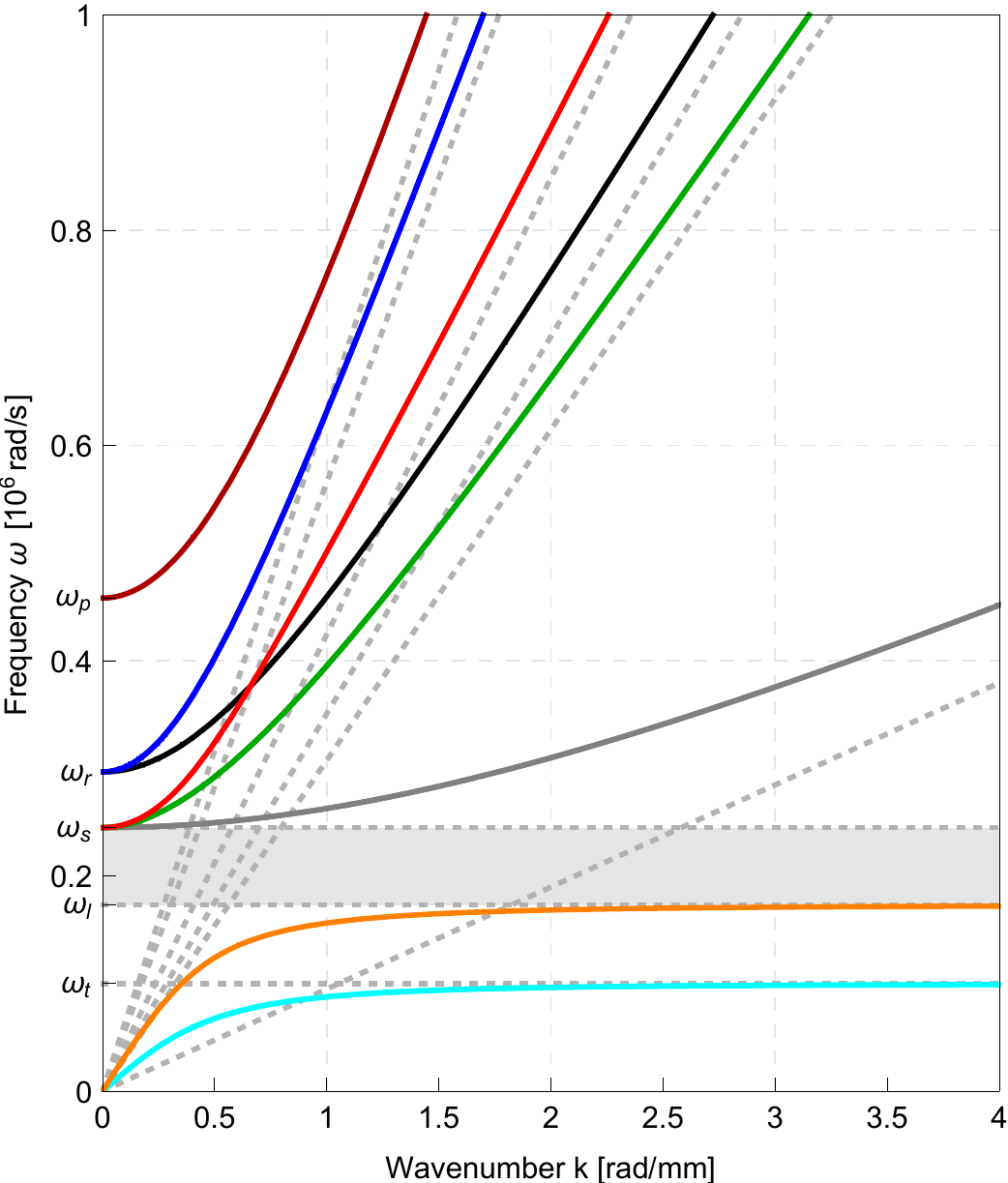} & \includegraphics[scale=0.5]{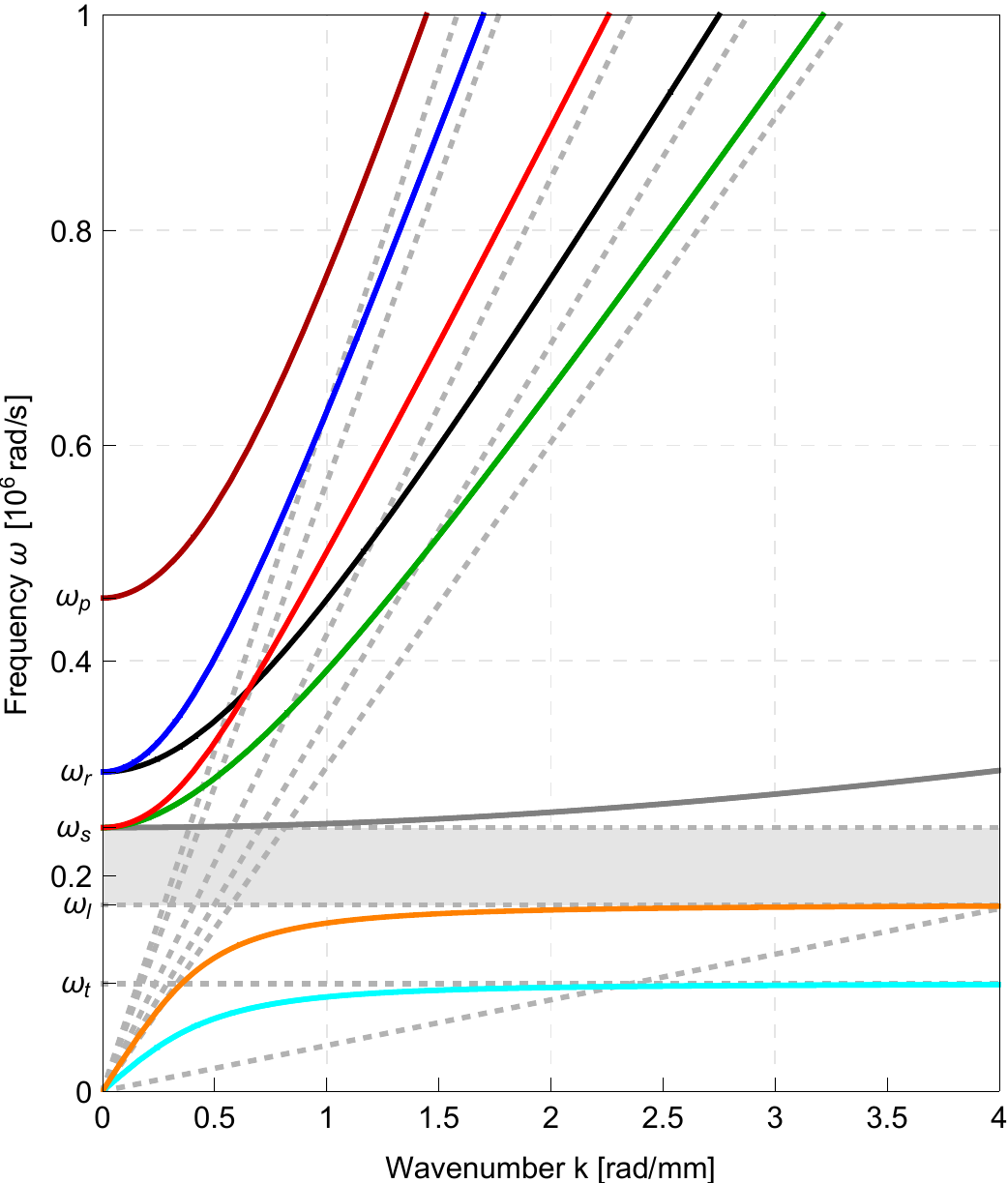} & \includegraphics[scale=0.5]{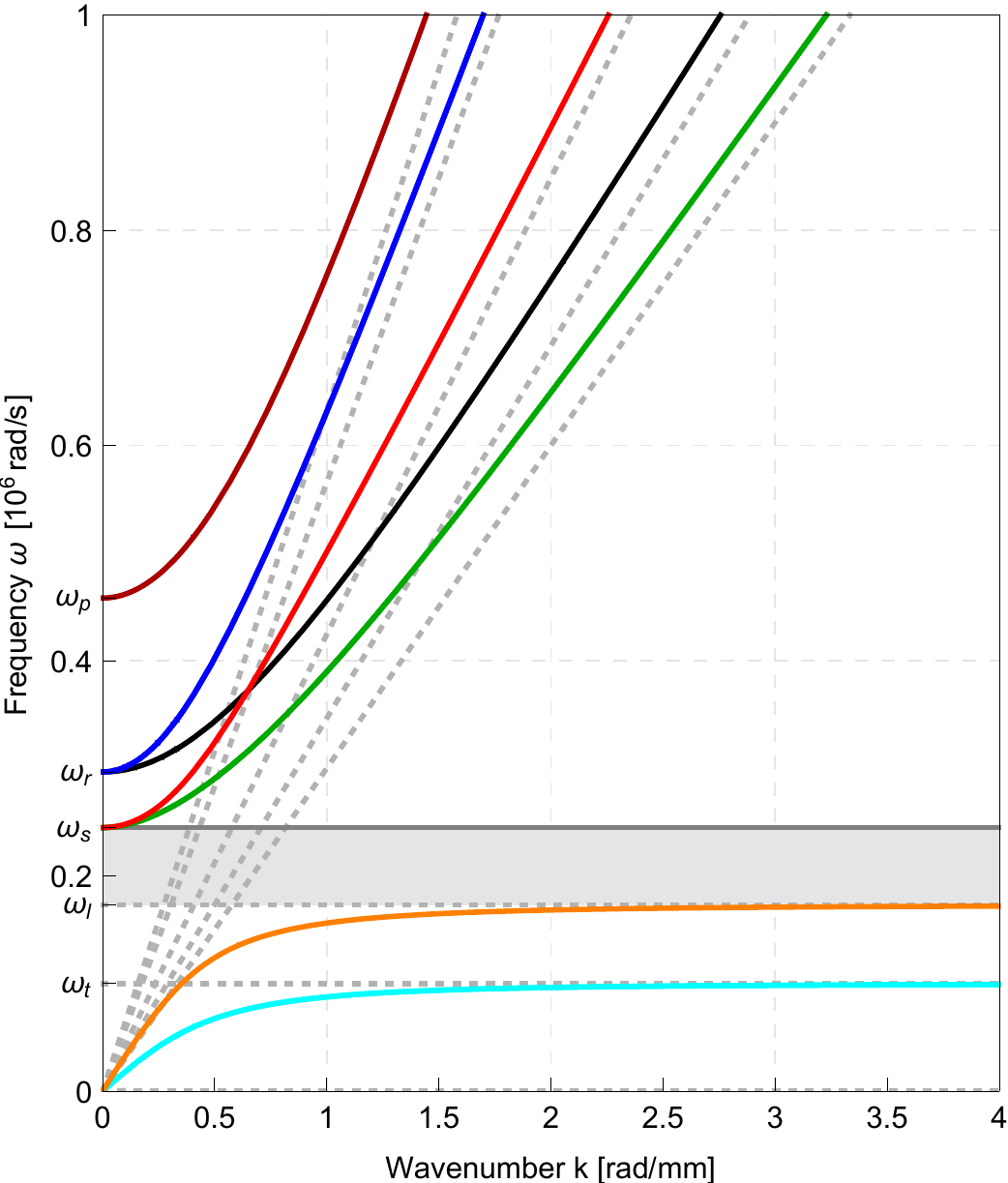}\tabularnewline
$\alpha_{1}=0.05$ & $\alpha_{1}=0.01$ & $\alpha_{1}=0$\tabularnewline
\end{tabular}\caption{Effect of the parameter $\alpha_{1}$ on the dispersion curves. We
set $\eta_{1}=\eta_{2}=\eta_{3}=10^{-2}\,Kg/m$\label{fig:alpha_1}}
\end{figure}
On the basis of the picture, it is clear that the action of the parameter
$\alpha_{1}$ preserves the presence of the band gap. This parameter
does not act on the curves with a horizontal asymptote (cyan, orange)
and it does not act on the optic curves in dark red and red. It is
also possible to remark that with the variation of $\alpha_{1}$ some
curves change their oblique asymptotes. Finally, one of the dispersion
curves becomes completely horizontal when setting $\alpha_{1}=0$.
This feature is peculiar to the parameter $\alpha_{1}$ because no
completely horizontal curves are produced by setting $\alpha_{2}=0$
or $\alpha_{3}=0$ (see subsequent pictures). This means that \textit{it
is mainly the parameter $\alpha_{1}$ which governs non-localities
in metamaterials}. Indeed, horizontal dispersion curves are peculiar
of metamaterials in which adjacent unit cells do not affect the behavior
of each other based on the hypothesis of separation of scales (see
\cite{sridhar2016homogenization,pham2013transient}). We explicitly
remark that the picture obtained with $\alpha_{1}=1$ is the one relative
to the classical relaxed micromorphic model (see also Fig. \ref{fig:classical}
(a)).

\newpage{}

\subsubsection{Case ${\displaystyle \mu_{c}>0,\lim_{\alpha_{2}\protect\fr0}}$}

Characteristic limit elastic energy $\left\Vert \nabla u-P\right\Vert ^{2}+\left\Vert \sym\,P\right\Vert ^{2}+\left\Vert \sym\,\curl\,P\right\Vert ^{2}$.\\
Characteristic limit kinetic energy $\left\Vert u_{,t}\right\Vert ^{2}+\left\Vert P_{,t}\right\Vert ^{2}$.

\begin{figure}[H]
\begin{centering}
\begin{tabular}{ccc}
\includegraphics[scale=0.5]{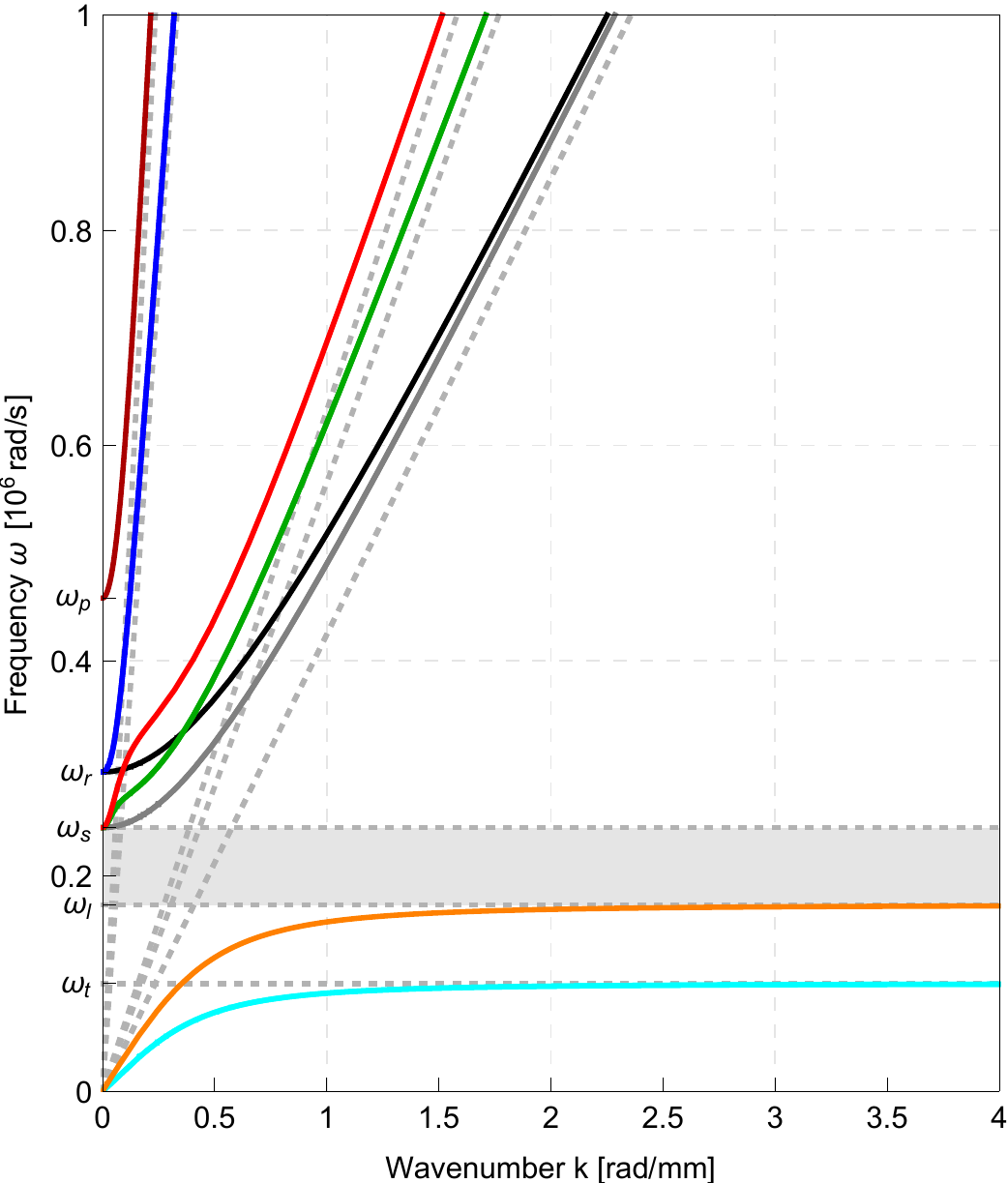} & \includegraphics[scale=0.5]{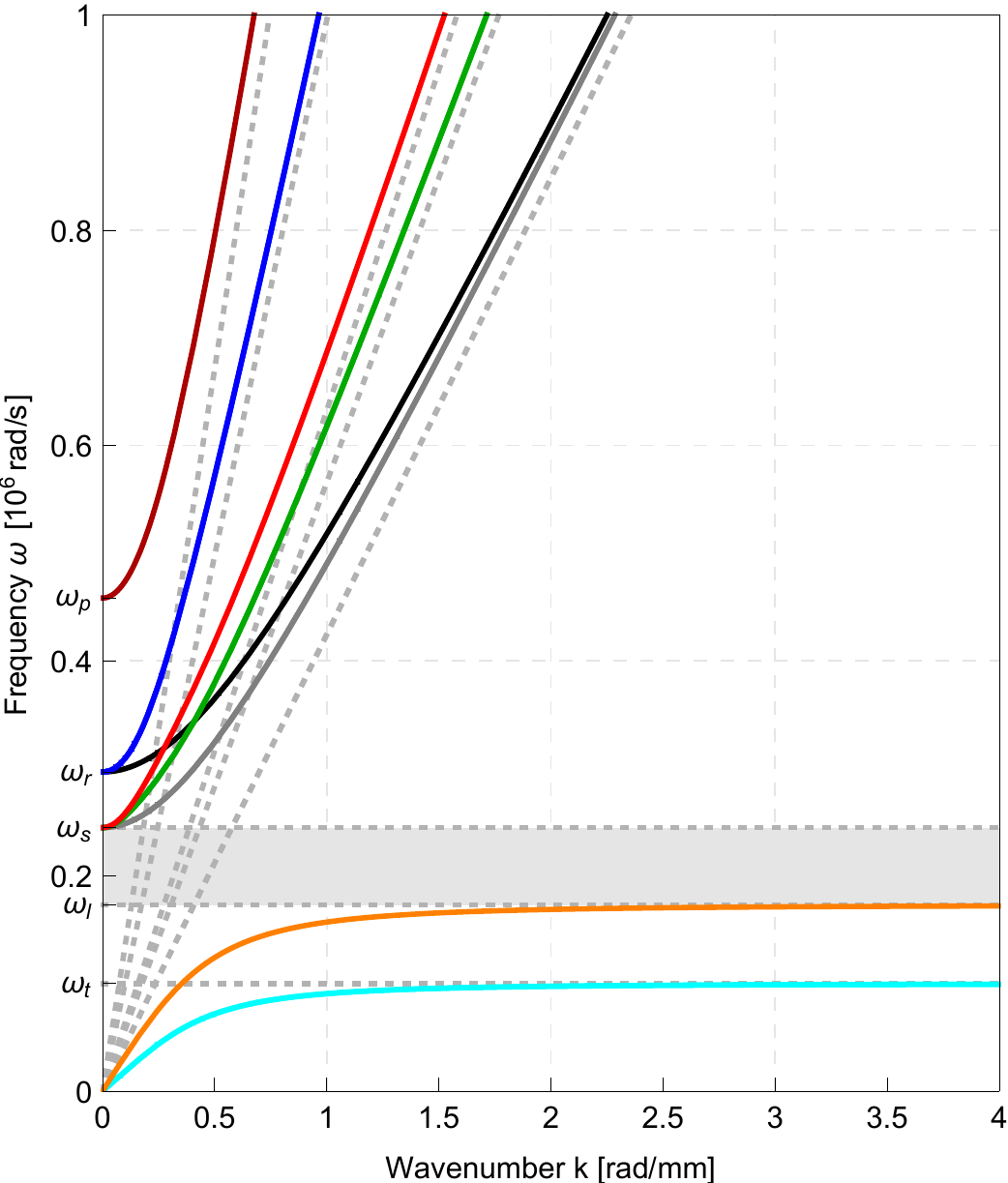} & \includegraphics[scale=0.5]{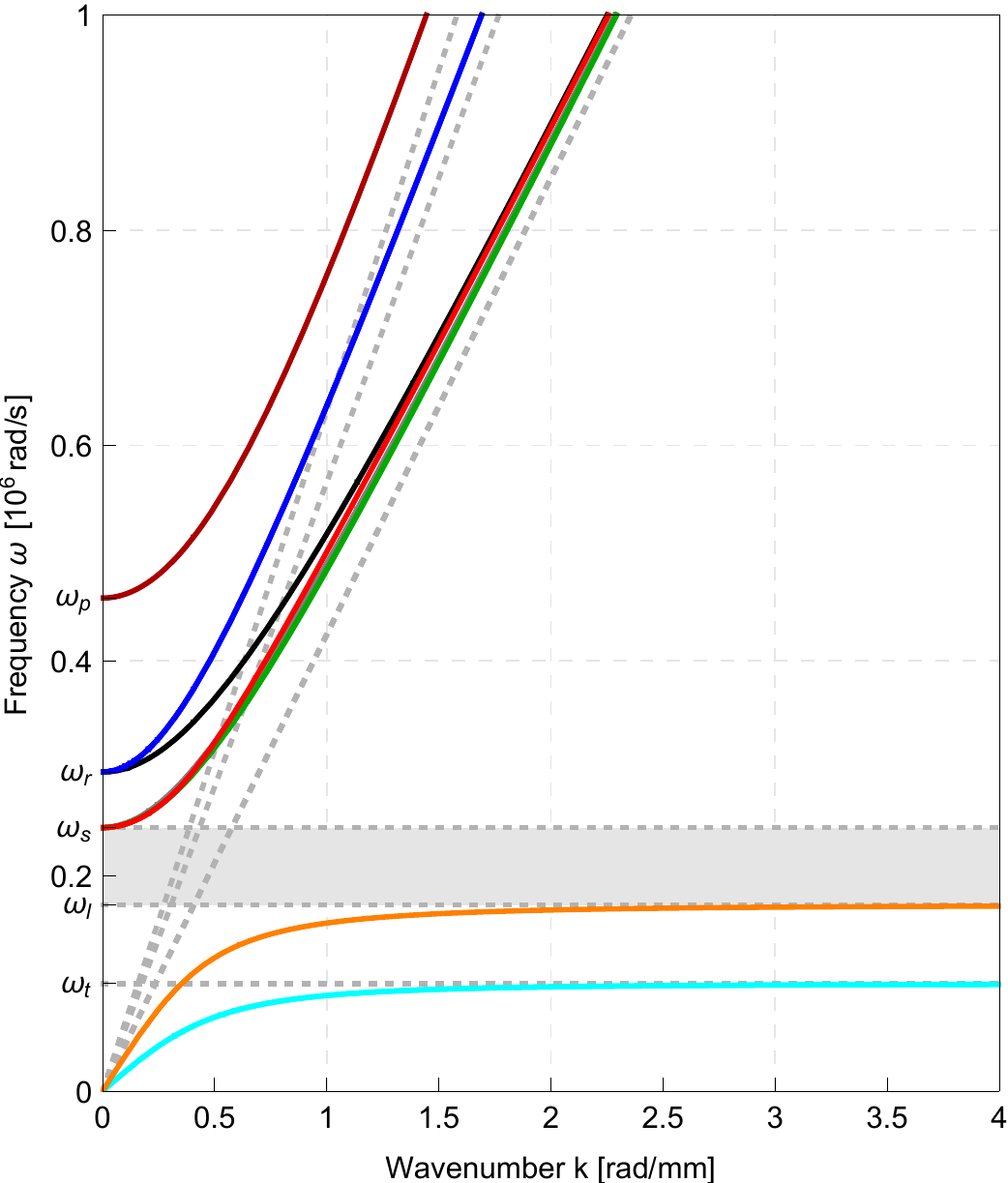}\tabularnewline
$\alpha_{2}=100$ & $\alpha_{2}=10$ & $\alpha_{2}=1$\tabularnewline
\includegraphics[scale=0.5]{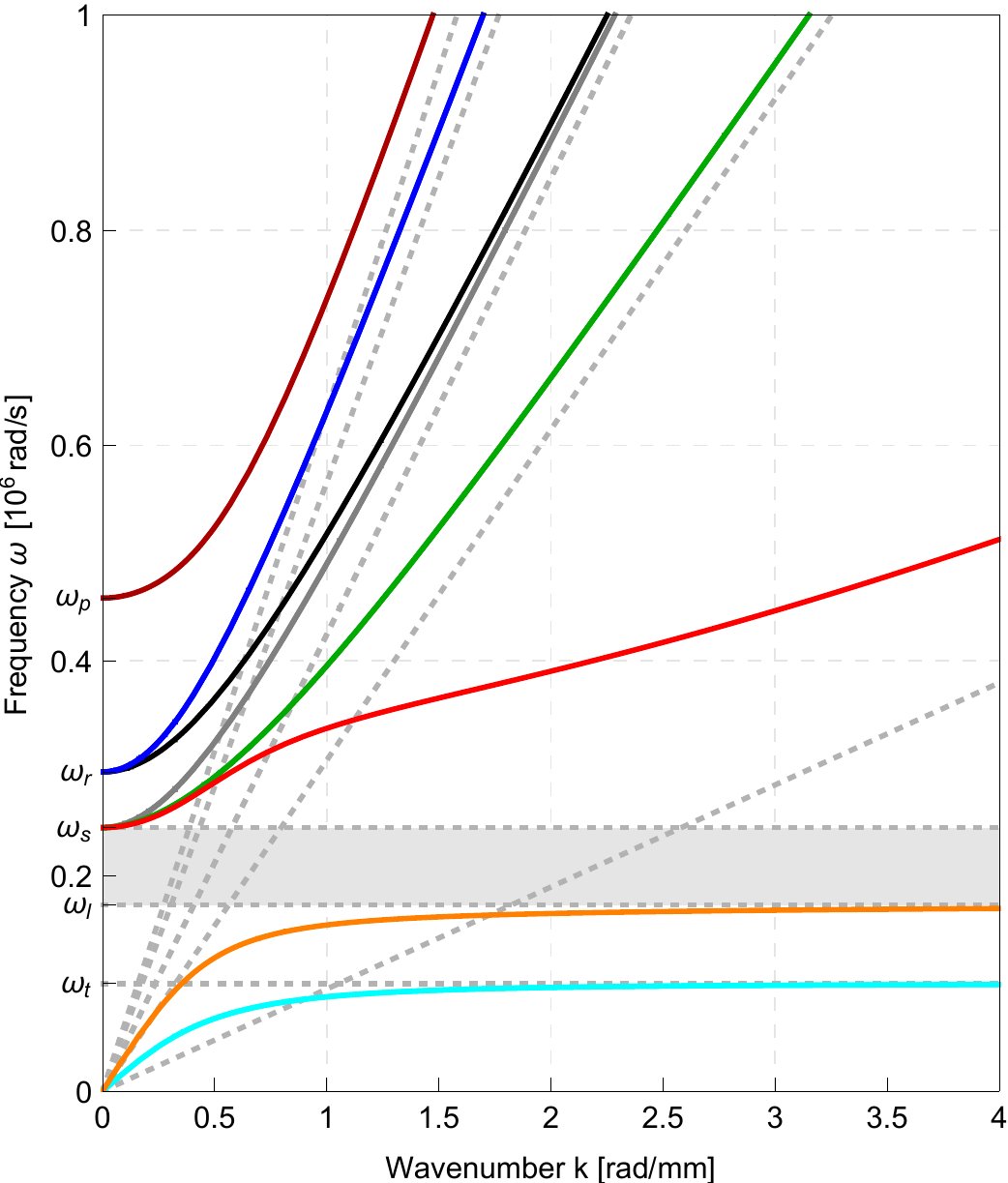} & \includegraphics[scale=0.5]{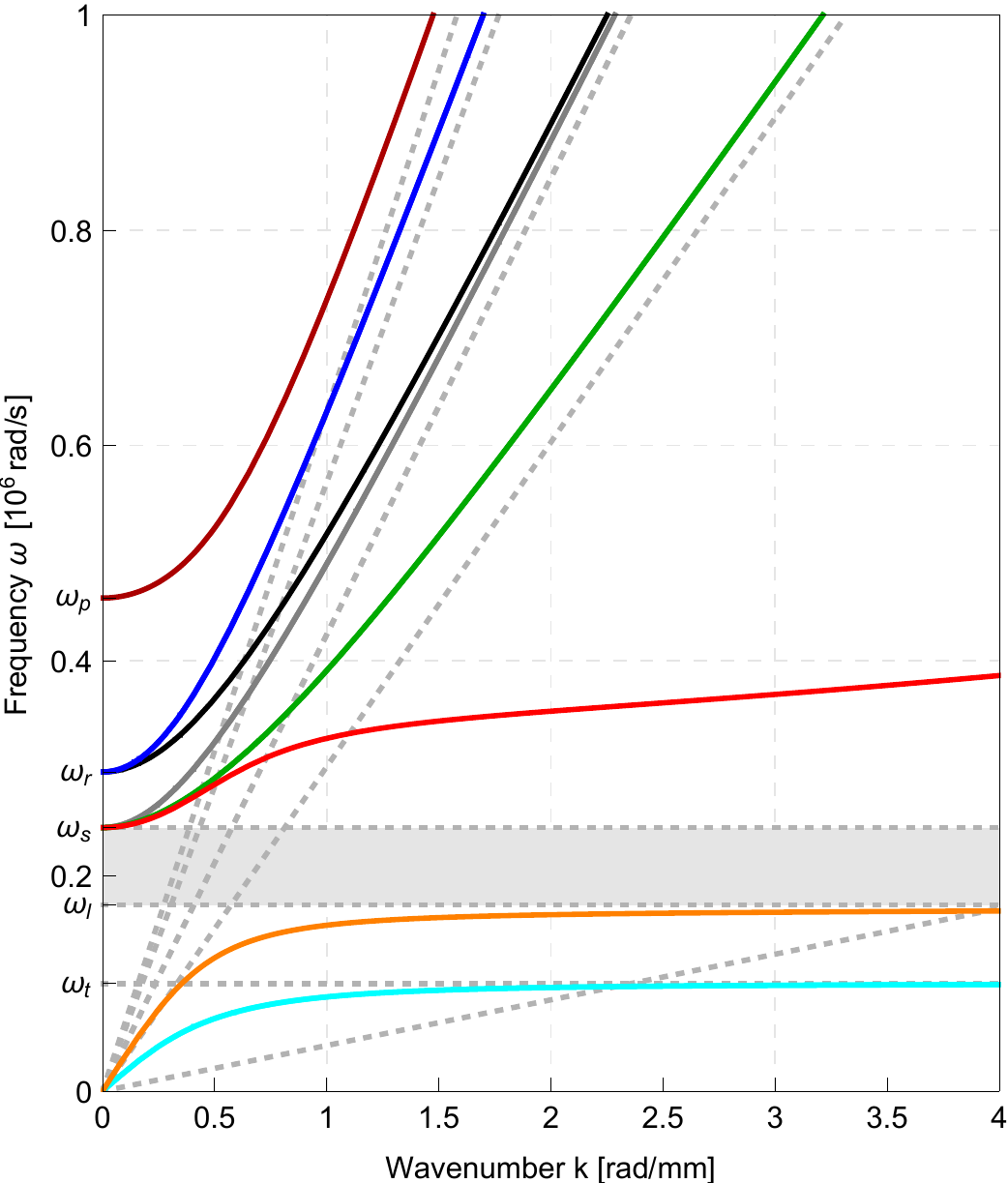} & \includegraphics[scale=0.5]{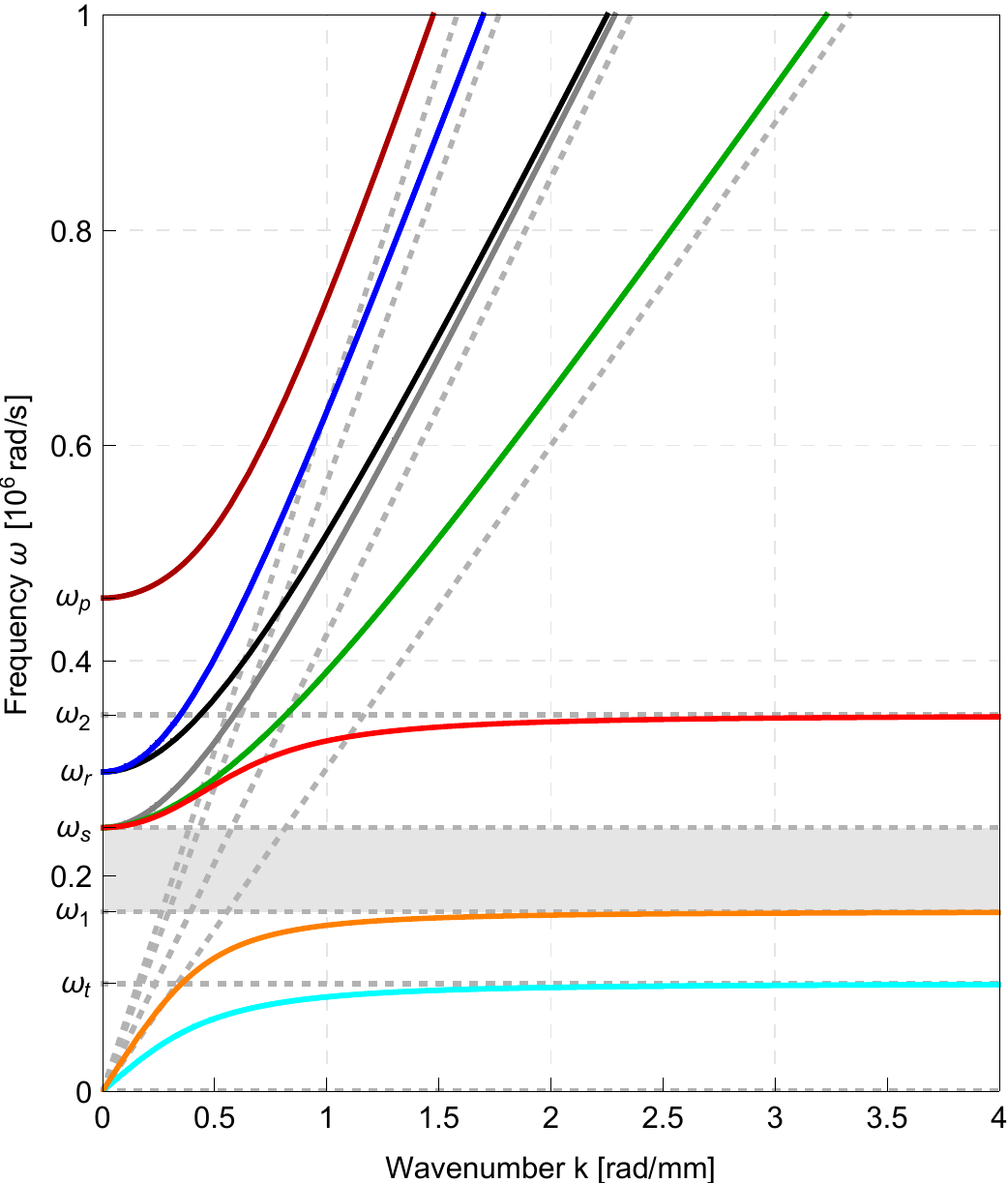}\tabularnewline
$\alpha_{2}=0.05$ & $\alpha_{2}=0.01$ & $\alpha_{2}=0$\tabularnewline
\end{tabular}
\par\end{centering}
\caption{Effect of the parameter $\alpha_{2}$ on the dispersion curves.\label{fig:alpha_2}}
\end{figure}

On the basis of the picture, it is clear that the action of the parameter
$\alpha_{2}$ preserves the presence of the band gap. Varying the
values of this parameter, we have an action only on the curves in
black and gray. The parameter $\alpha_{2}$ is also seen to have some
direct influence on the horizontal asymptotes for the orange optic
wave. In fact if $\alpha_{2}=0$ the value of this horizontal asymptote
changes and we can calculate it solving the equation 
\begin{equation}
c_{16}\left(\omega_{*}^{2}\right)=0\label{eq:c_16}
\end{equation}
with respect to $\omega_{*}$. This is exactly the direct application
of the proposition \ref{prop:CN} in the case in which the coefficient
$c_{18}\left(\omega^{2}\right)$ is not present. Between the solutions
of the eq.(\ref{eq:c_16}) we find 
\begin{align*}
\omega_{1} & =\sqrt{\frac{\widetilde{q}_{1}}{3\eta_{1}\eta_{3}\left(\lambda_{e}+2\mu_{e}\right)}-\frac{\sqrt{\left(\widetilde{q}_{2}\right)^{2}-4\left(3\eta_{1}\eta_{3}\lambda_{e}+6\eta_{1}\eta_{3}\mu_{e}\right)\widetilde{q}_{3}}}{6\eta_{1}\eta_{3}\left(\lambda_{e}+2\mu_{e}\right)}+\frac{9\eta_{1}\lambda_{e}\lambda_{h}}{6\eta_{1}\eta_{3}\left(\lambda_{e}+2\mu_{e}\right)}}\\
\omega_{2} & =\sqrt{\frac{\widetilde{q}_{1}}{3\eta_{1}\eta_{3}\left(\lambda_{e}+2\mu_{e}\right)}+\frac{\sqrt{\left(\widetilde{q}_{2}\right)^{2}-4\left(3\eta_{1}\eta_{3}\lambda_{e}+6\eta_{1}\eta_{3}\mu_{e}\right)\widetilde{q}_{3}}}{6\eta_{1}\eta_{3}\left(\lambda_{e}+2\mu_{e}\right)}+\frac{9\eta_{1}\lambda_{e}\lambda_{h}}{6\eta_{1}\eta_{3}\left(\lambda_{e}+2\mu_{e}\right)}}
\end{align*}
where
\begin{align*}
\widetilde{q}_{1} & =4\eta_{1}\mu_{e}^{2}+2\eta_{3}\mu_{e}^{2}+6\eta_{1}\lambda_{e}\mu_{e}+3\eta_{3}\lambda_{e}\mu_{e}+6\eta_{1}\mu_{e}\mh+6\eta_{3}\mu_{e}\mh+9\eta_{1}\mu_{e}\lambda_{h}+3\eta_{1}\lambda_{e}\mh+3\eta_{3}\lambda_{e}\mh\\
\widetilde{q}_{2} & =-12\eta_{1}\lambda_{e}\mu_{e}-6\eta_{3}\lambda_{e}\mu_{e}-8\eta_{1}\mu_{e}^{2}-4\eta_{3}\mu_{e}^{2}-18\eta_{1}\mu_{e}\lh-6\eta_{1}\lambda_{e}\mh-6\eta_{3}\lambda_{e}\mh\mh\\
 & \quad-9\eta_{1}\lambda_{e}\lh-12\eta_{1}\mu_{e}\mh-12\eta_{3}\mu_{e},\\
\widetilde{q}_{3} & =12\mu_{e}^{2}\lh+18\lambda_{e}\mu_{e}\lh+36\lambda_{e}\mu_{e}\mh+36\mu_{e}\lambda_{h}\mh+12\lambda_{e}\mh^{2}\\
 & \quad+18\lambda_{e}\lh\mh+24\mu_{e}^{2}\mh+24\mu_{e}\mh^{2}.
\end{align*}

\subsubsection{Case ${\displaystyle \mu_{c}>0,\lim_{\alpha_{3}\protect\fr0}}$}

Characteristic limit elastic energy $\left\Vert \nabla u-P\right\Vert ^{2}+\left\Vert \sym\,P\right\Vert ^{2}+\left\Vert \dev\,\curl\,P\right\Vert ^{2}$\\
Characteristic limit kinetic energy $\left\Vert u_{,t}\right\Vert ^{2}+\left\Vert P_{,t}\right\Vert ^{2}$.

\begin{table}[H]
\centering{}%
\begin{tabular}{ccc}
\includegraphics[scale=0.5]{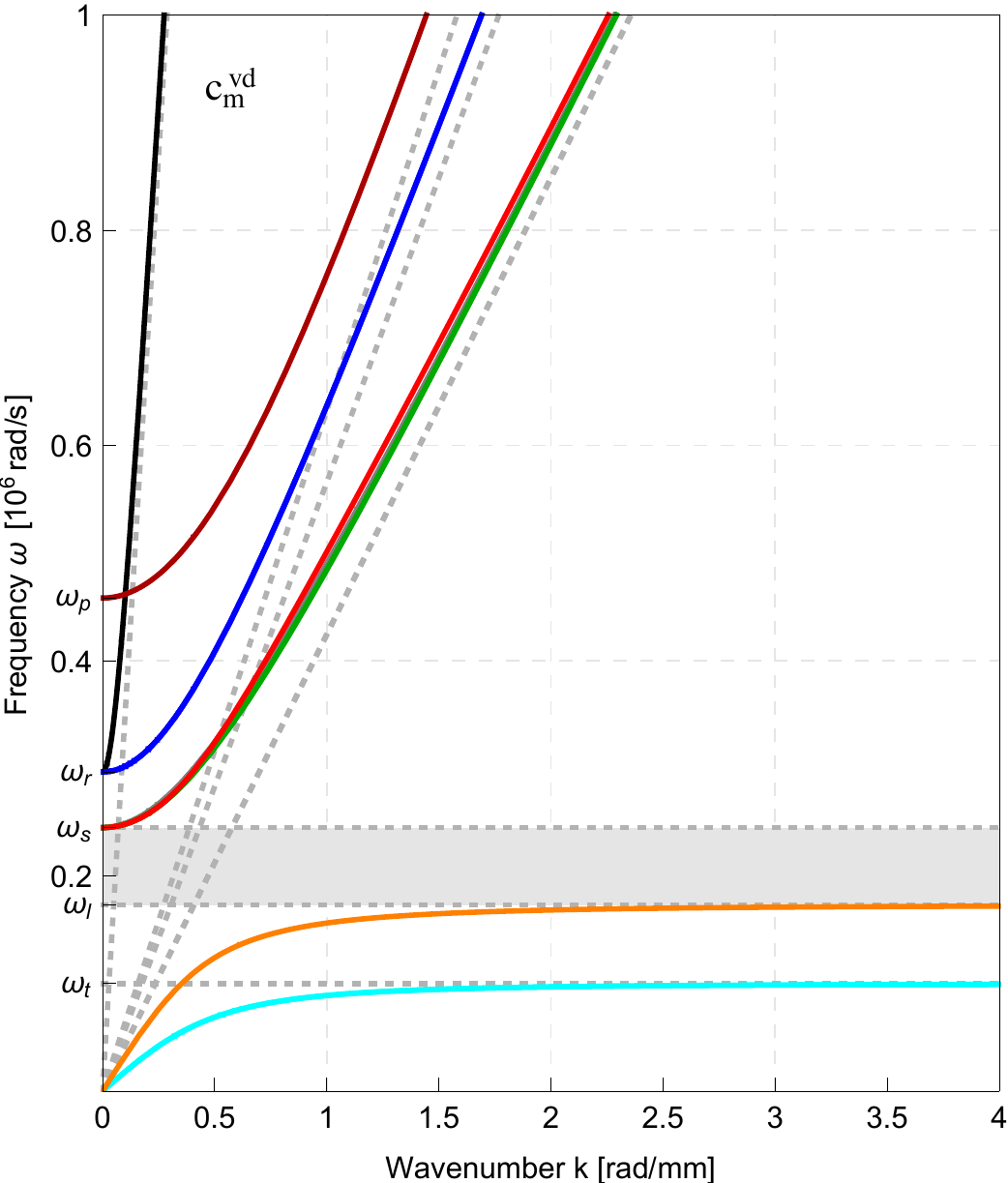} & \includegraphics[scale=0.5]{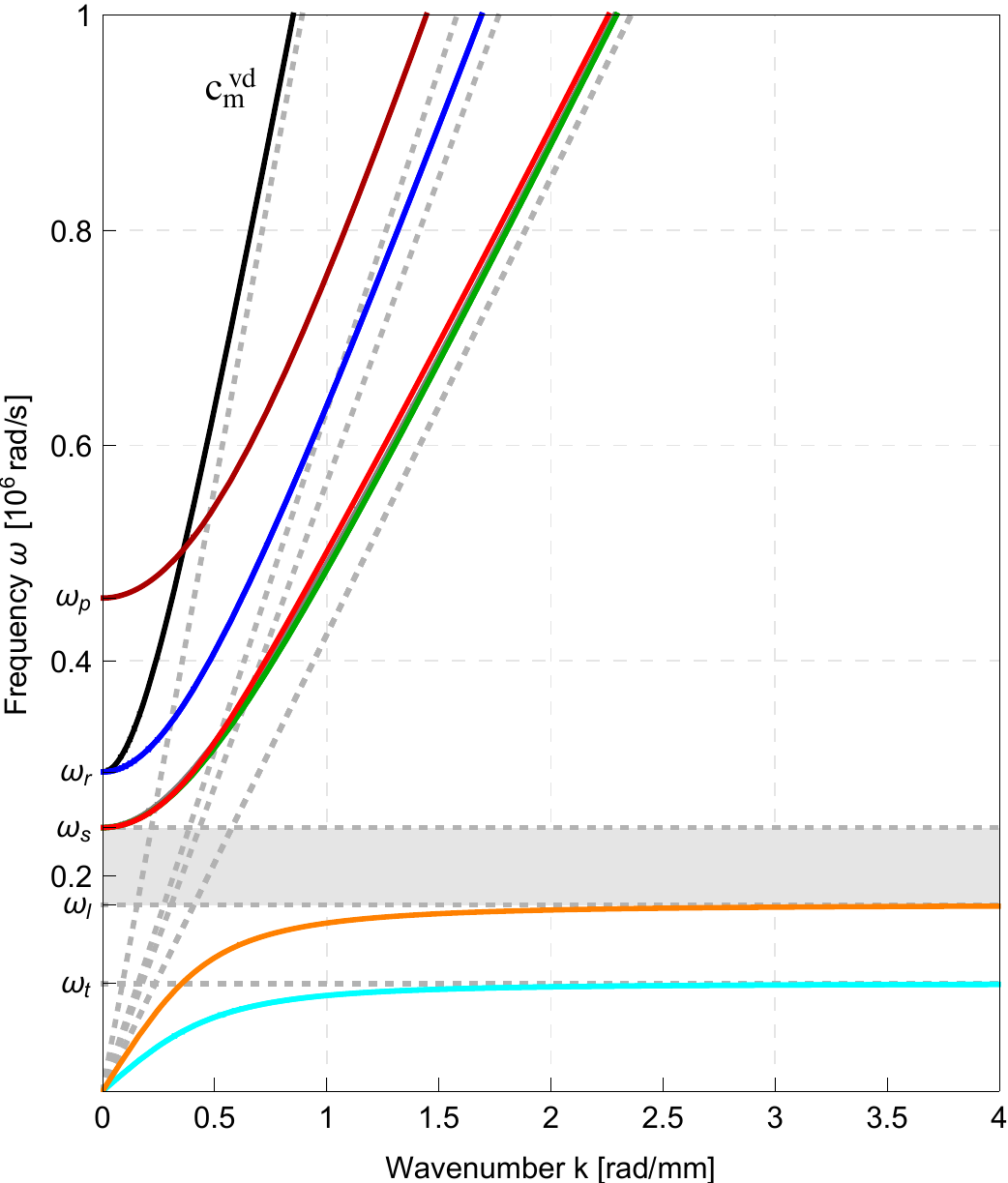} & \includegraphics[scale=0.5]{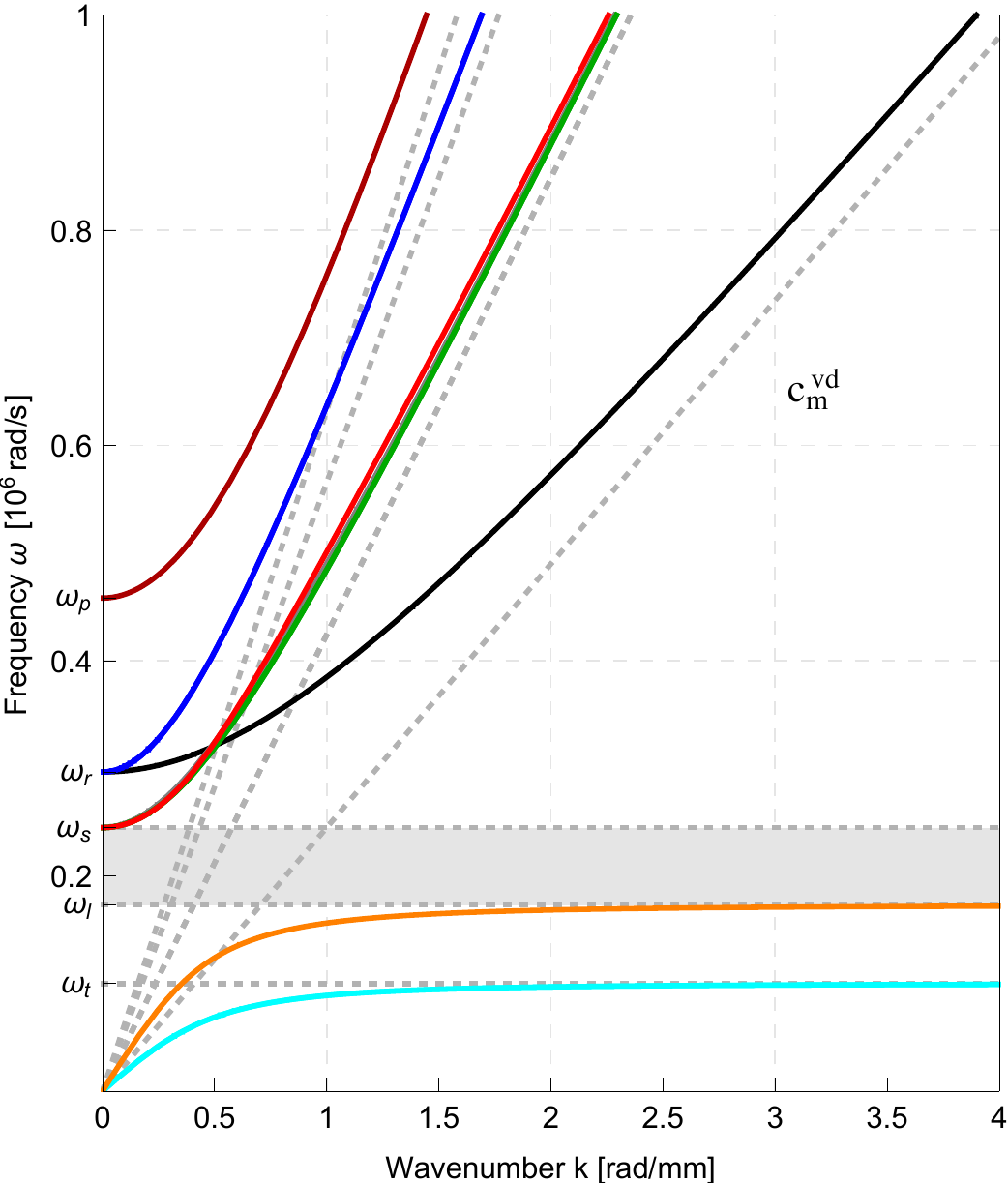}\tabularnewline
$\alpha_{3}=100$ & $\alpha_{3}=10$ & $\alpha_{3}=0$\tabularnewline
\end{tabular}\caption{Effect of the parameter $\alpha_{3}$ on the dispersion curves.}
\end{table}

The action of the parameter $\alpha_{3}$ preserves the presence
of the band gaps. It acts only on the uncoupled curve in black leaving
all the others fixed. The parameter $\alpha_{3}$ has no direct effect
neither on horizontal asymptotes, nor on the creation of purely horizontal
curves.The oblique asymptote of the black branch is $\mathrm{c}_{\textrm{m}}^{\textrm{vd}}$,
and we explicitly remark that
\[
\lim_{\alpha_{3}\fr0}\mathrm{c}_{\textrm{m}}^{\textrm{vd}}=\frac{\alpha_{1}}{\eta_{3}}\frac{\me\,L_{c}^{2}}{3}\qquad\textrm{and}\qquad\lim_{\alpha_{3}\fr0}\mathrm{c}_{\textrm{m}}^{\textrm{vd}}=+\infty.
\]

\newpage{}

\subsubsection{Case ${\displaystyle \mu_{c}>0,\lim_{\alpha_{2},\alpha_{3}\protect\fr0}}$}

Characteristic limit elastic energy $\left\Vert \nabla u-P\right\Vert ^{2}+\left\Vert \sym\,P\right\Vert ^{2}+\left\Vert \dev\,\sym\,\curl\,P\right\Vert ^{2}$.\\
Characteristic limit kinetic energy $\left\Vert u_{,t}\right\Vert ^{2}+\left\Vert P_{,t}\right\Vert ^{2}$.

\begin{figure}[H]
\centering{}%
\begin{tabular}{ccc}
\includegraphics[scale=0.5]{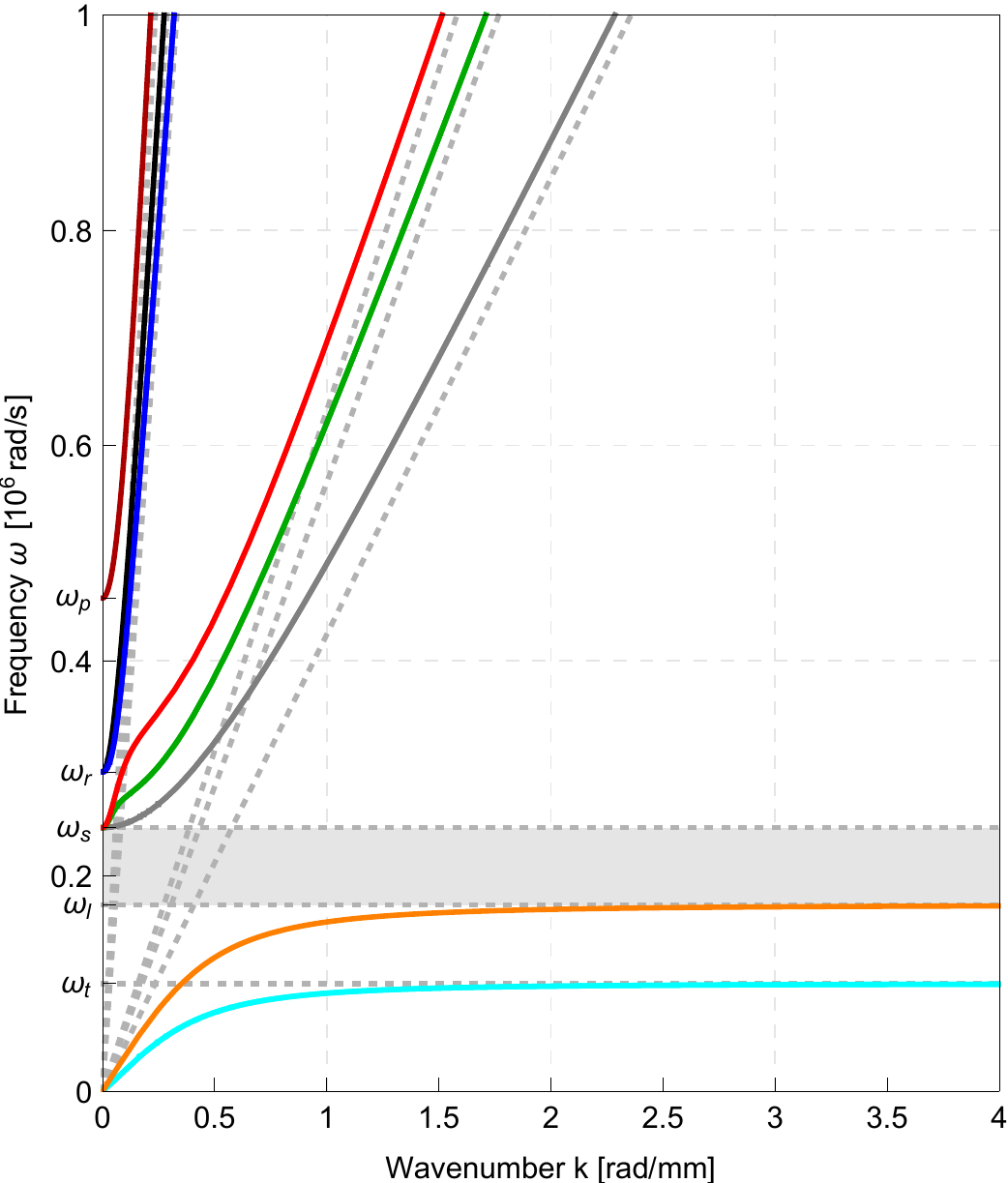} & \includegraphics[scale=0.5]{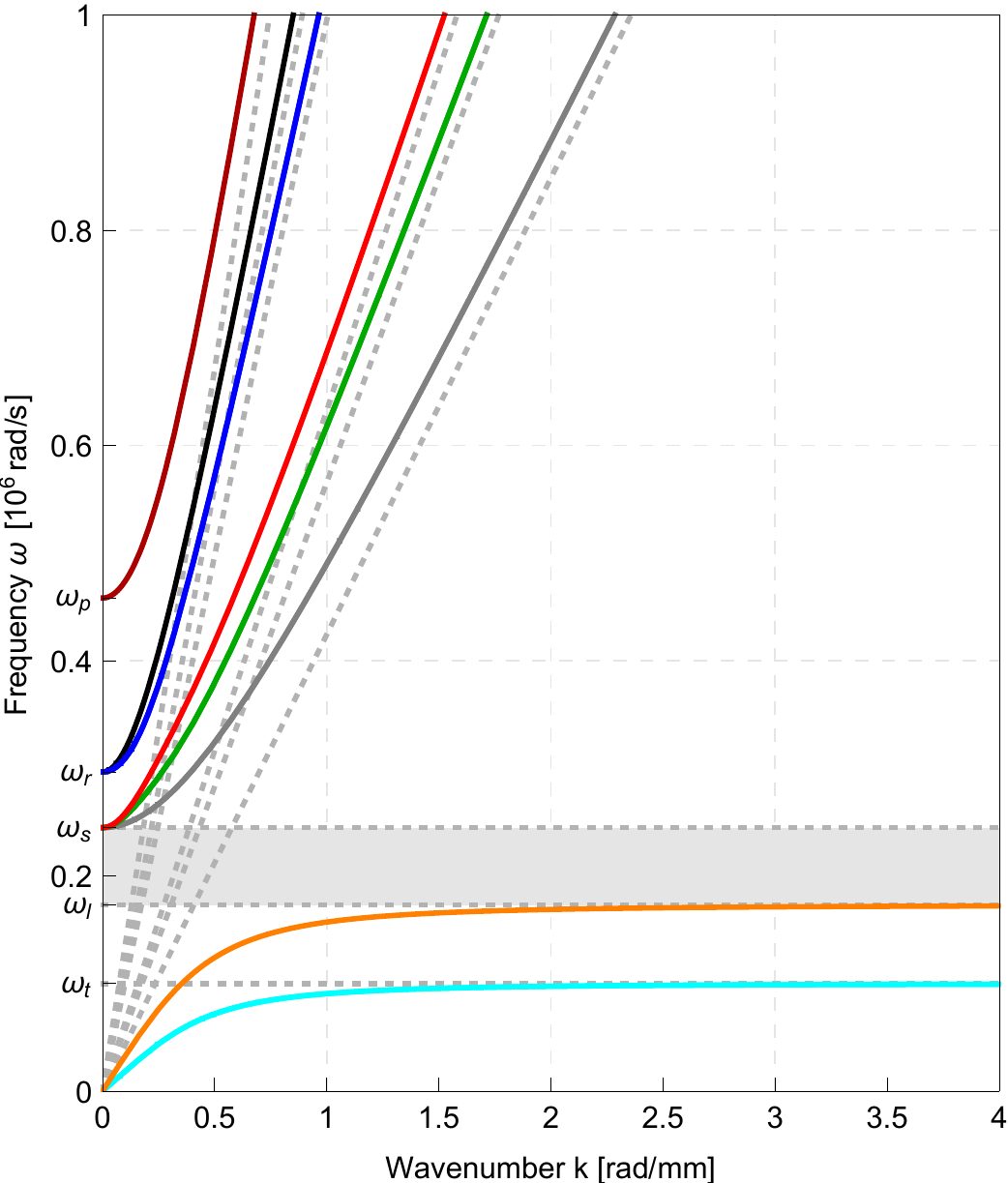} & \includegraphics[scale=0.5]{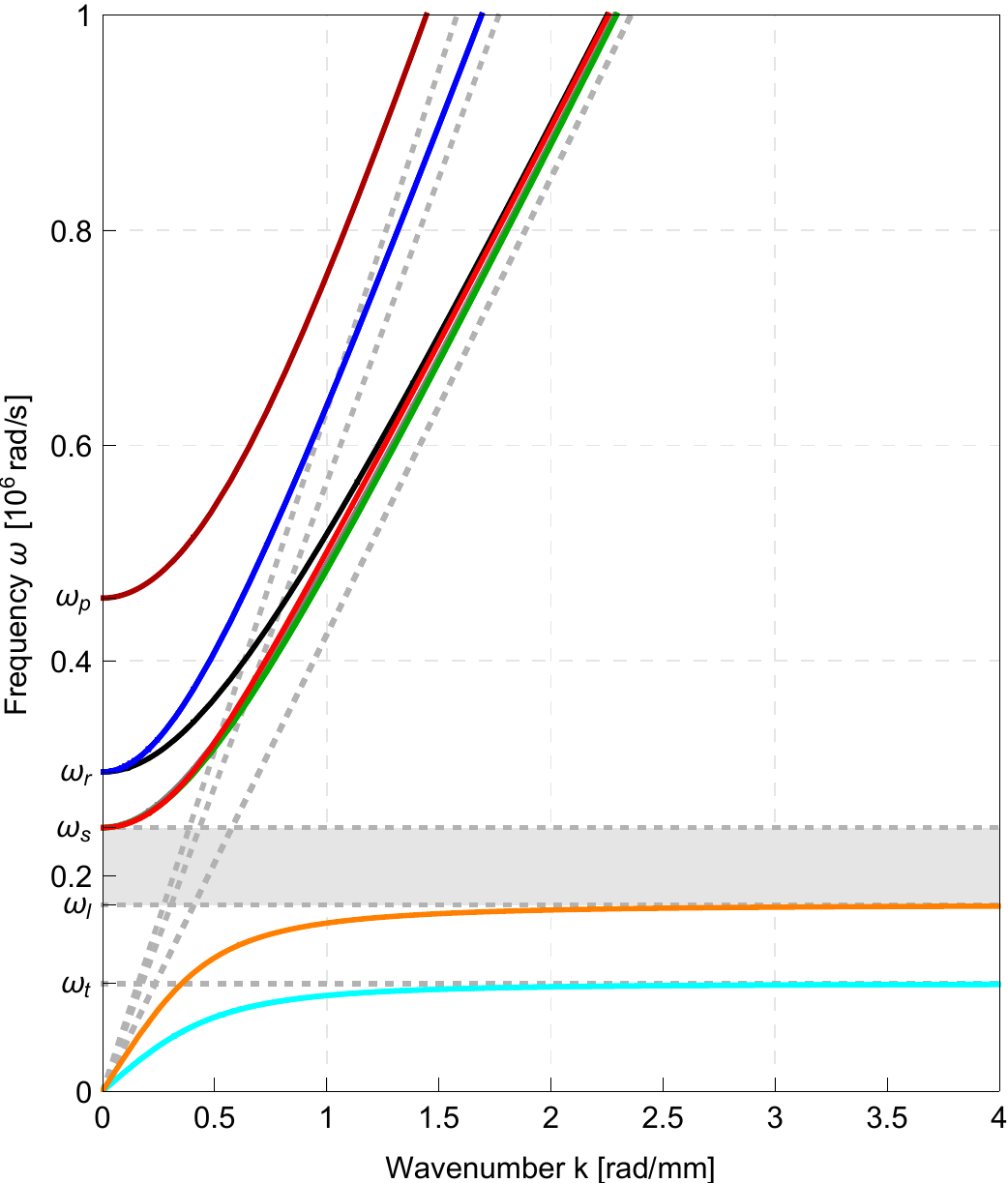}\tabularnewline
$\alpha_{2}=\alpha_{3}=100$ & $\alpha_{2}=\alpha_{3}=10$ & $\alpha_{2}=\alpha_{3}=1$\tabularnewline
\includegraphics[scale=0.5]{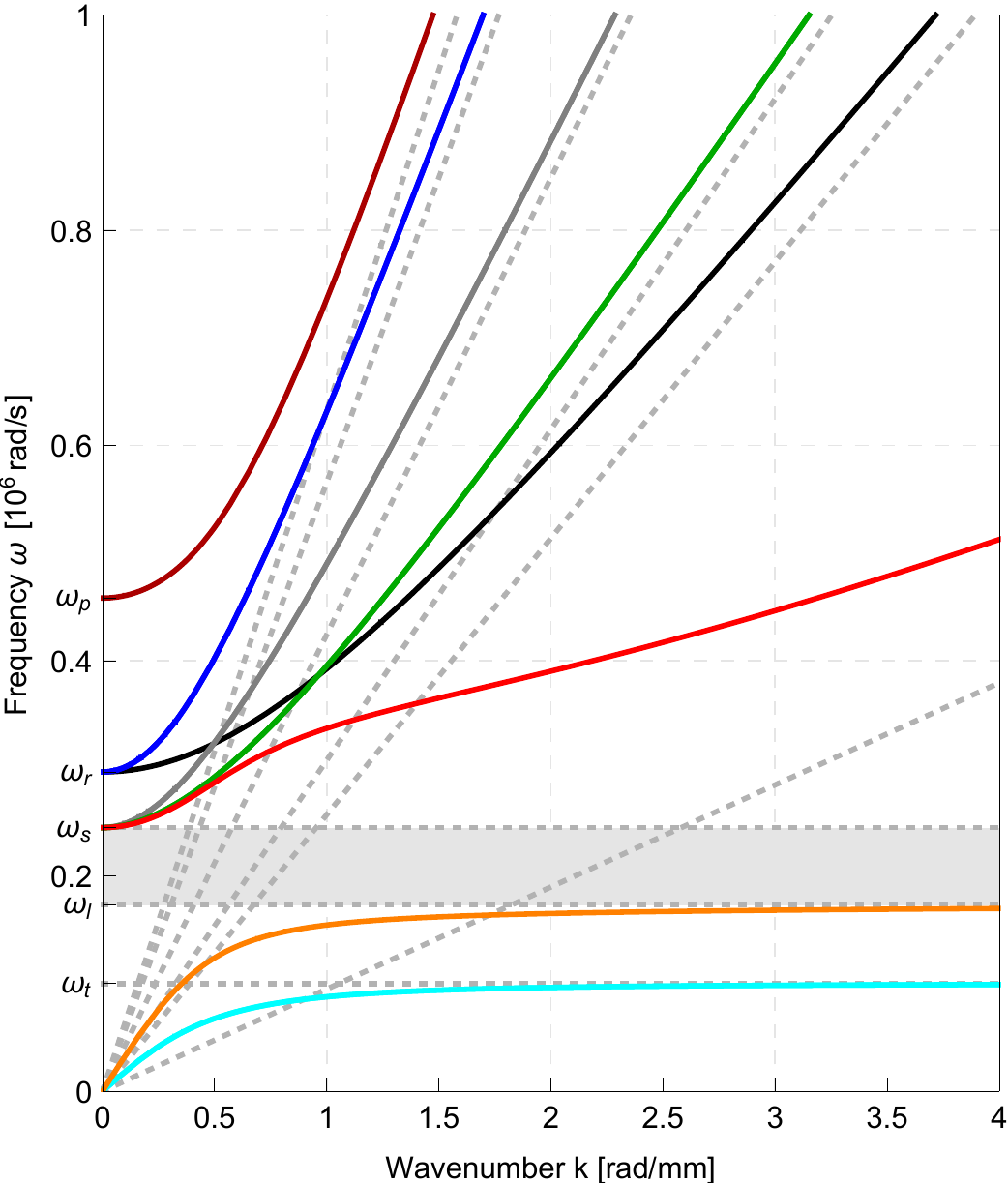} & \includegraphics[scale=0.5]{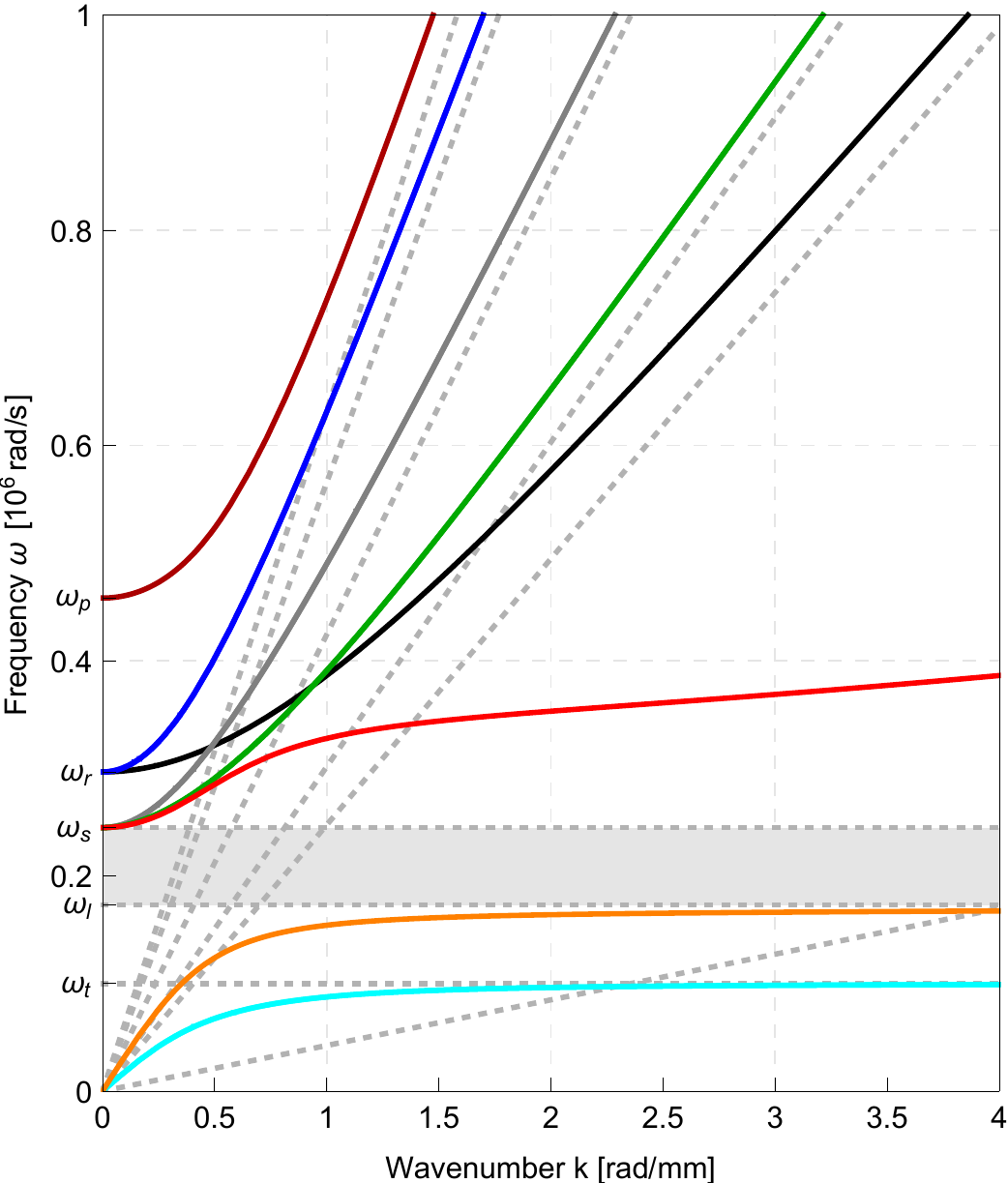} & \includegraphics[scale=0.5]{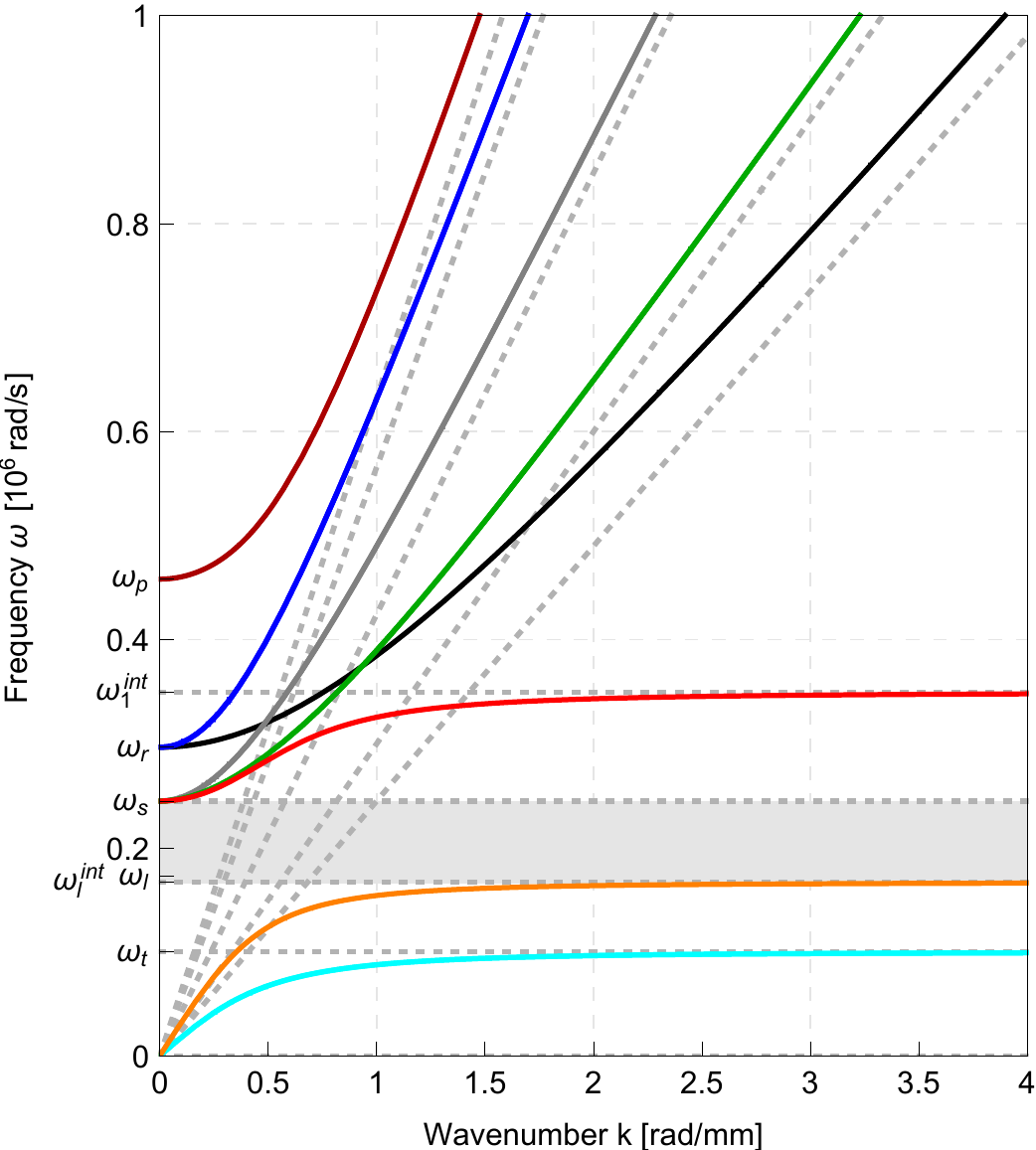}\tabularnewline
$\alpha_{2}=\alpha_{3}=0.05$ & $\alpha_{2}=\alpha_{3}=0.01$ & $\alpha_{2}=\alpha_{3}=0$\tabularnewline
\end{tabular}\caption{Combined effect of the parameter $\alpha_{2}$ and $\alpha_{3}$ on
the dispersion curves.}
\end{figure}
The combined action of the parameters $\alpha_{2}$ and $\alpha_{3}$
is given by the superposition of the effects observed in subsections
5.3.2 and 5.3.3. Only three curves (gray, orange, cyan) remain fixed.

\newpage{}

\subsubsection{Vanishing Cosserat couple modulus ${\displaystyle \mu_{c}=0}$ and
$\lim_{\alpha_{1}\protect\fr0}$}

Characteristic limit elastic energy $\left\Vert \sym\left(\nabla u-P\right)\right\Vert ^{2}+\left\Vert \sym\,P\right\Vert ^{2}+\left\Vert \skew\,\curl\,P\right\Vert ^{2}+\frac{1}{3}\left(\textrm{tr}\,\curl\,P\right)^{2}.$\\
Characteristic limit kinetic energy $\left\Vert u_{,t}\right\Vert ^{2}+\left\Vert P_{,t}\right\Vert ^{2}$.

\begin{figure}[H]
\centering{}%
\begin{tabular}{ccc}
\includegraphics[scale=0.5]{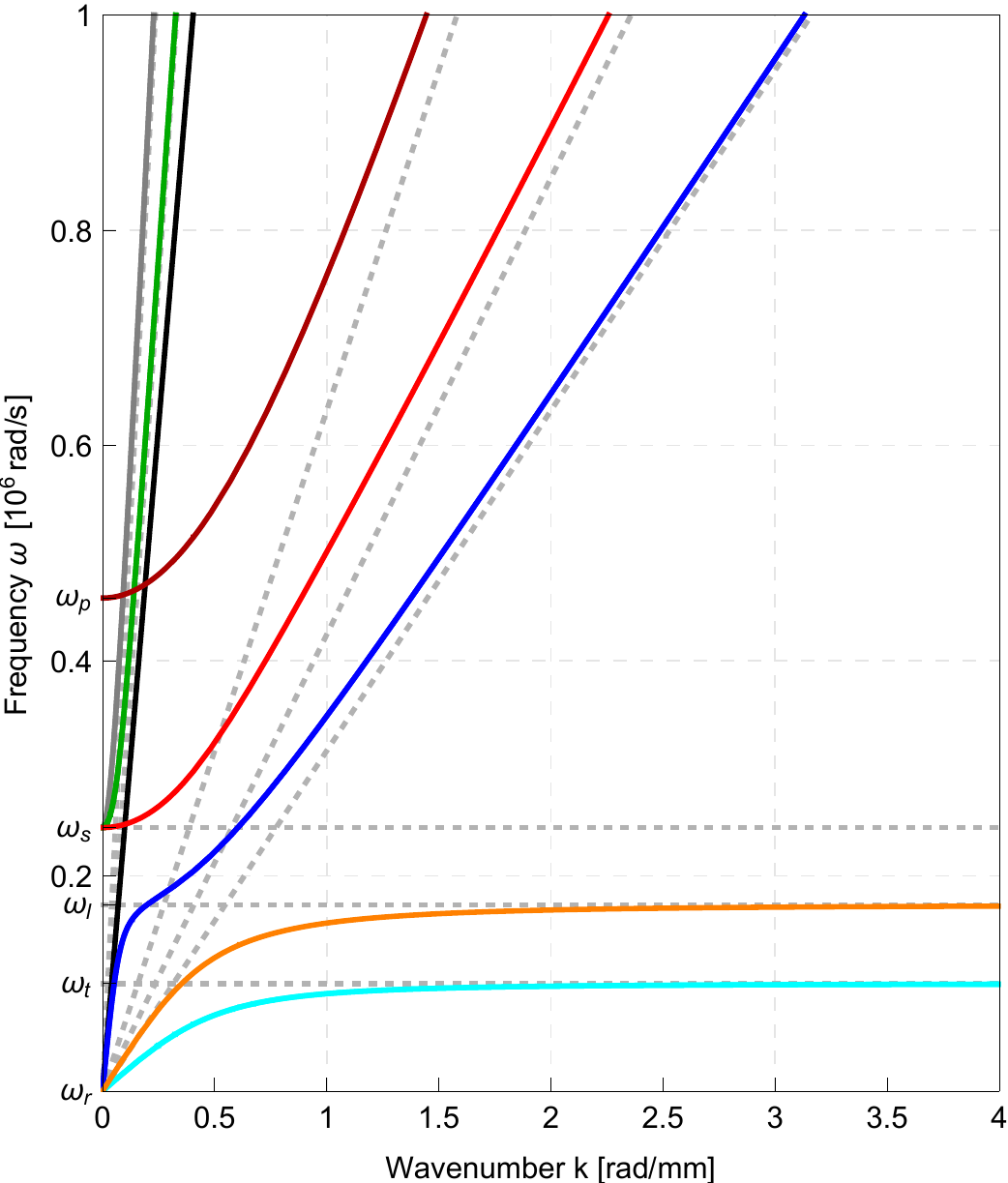} & \includegraphics[scale=0.5]{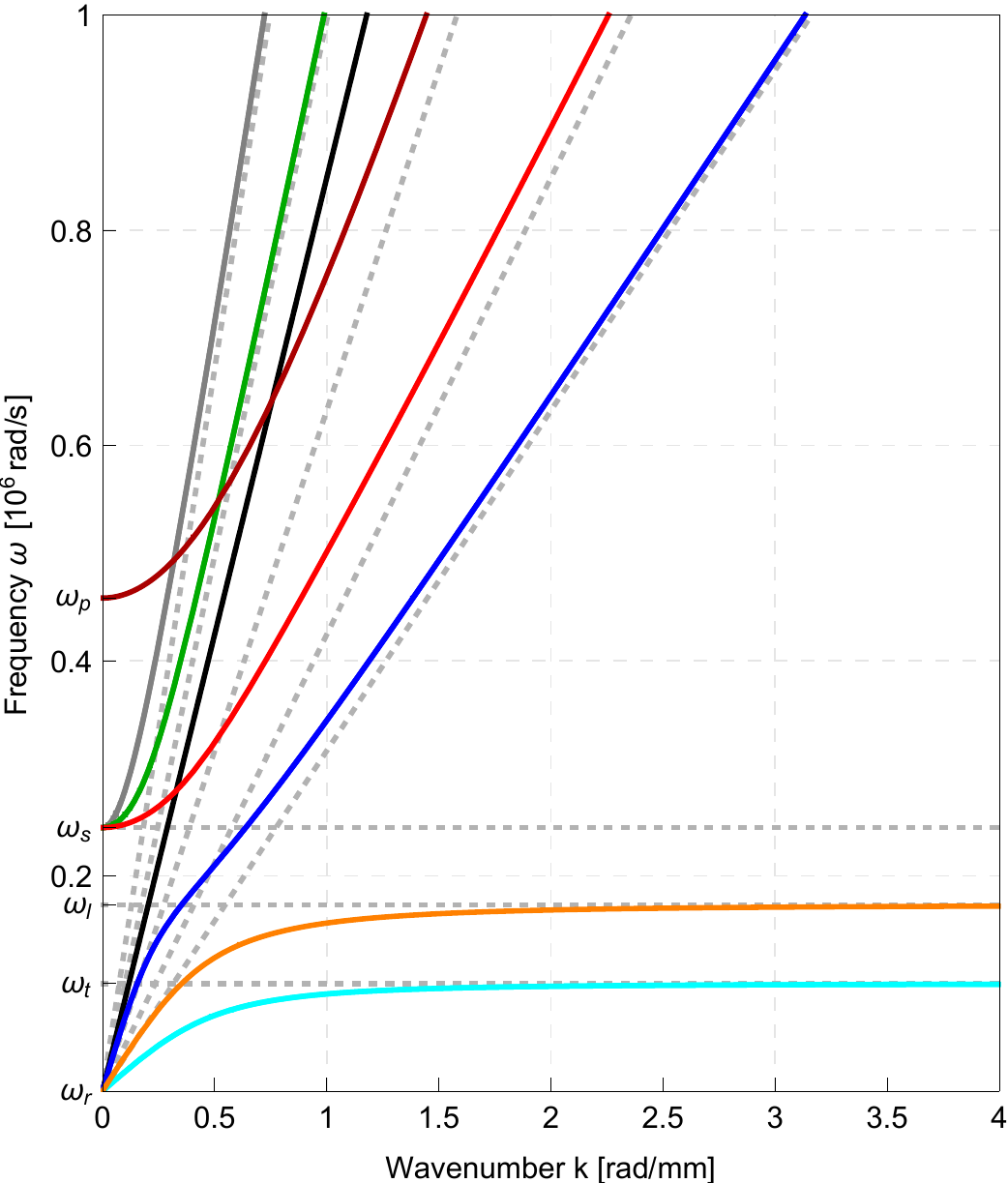} & \includegraphics[scale=0.5]{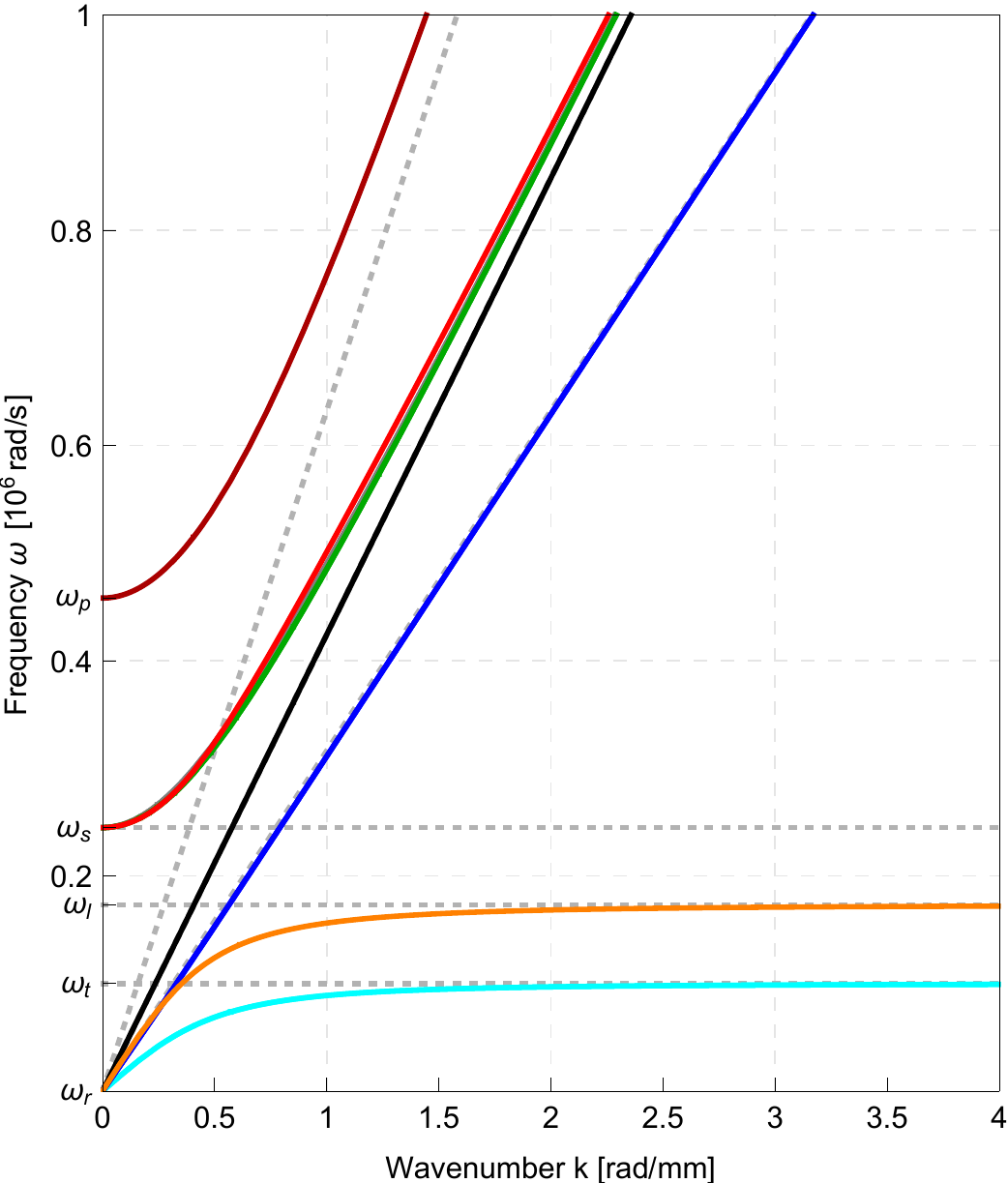}\tabularnewline
$\alpha_{1}=100$ & $\alpha_{1}=10$ & $\alpha_{1}=1$\tabularnewline
\includegraphics[scale=0.5]{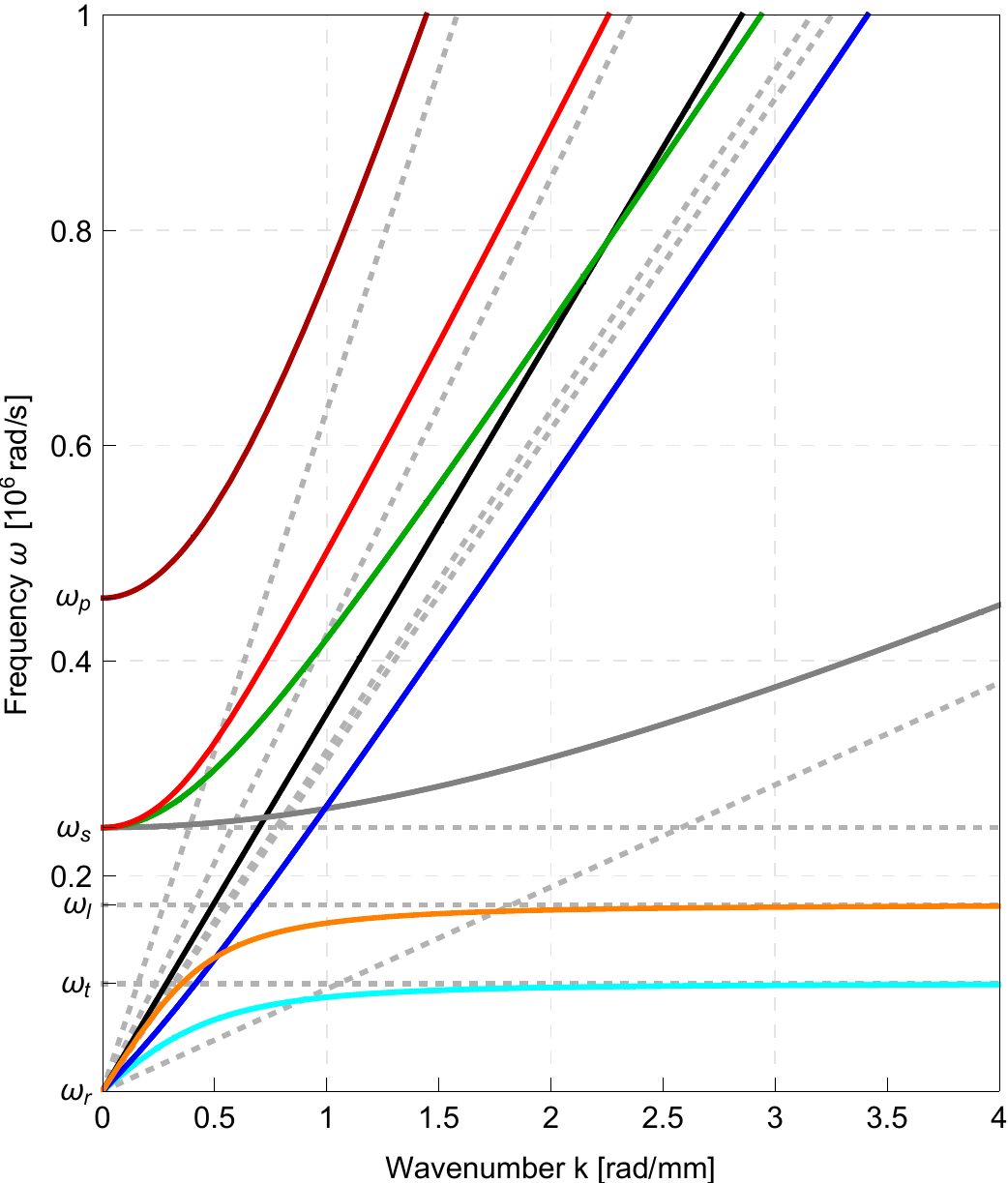} & \includegraphics[scale=0.5]{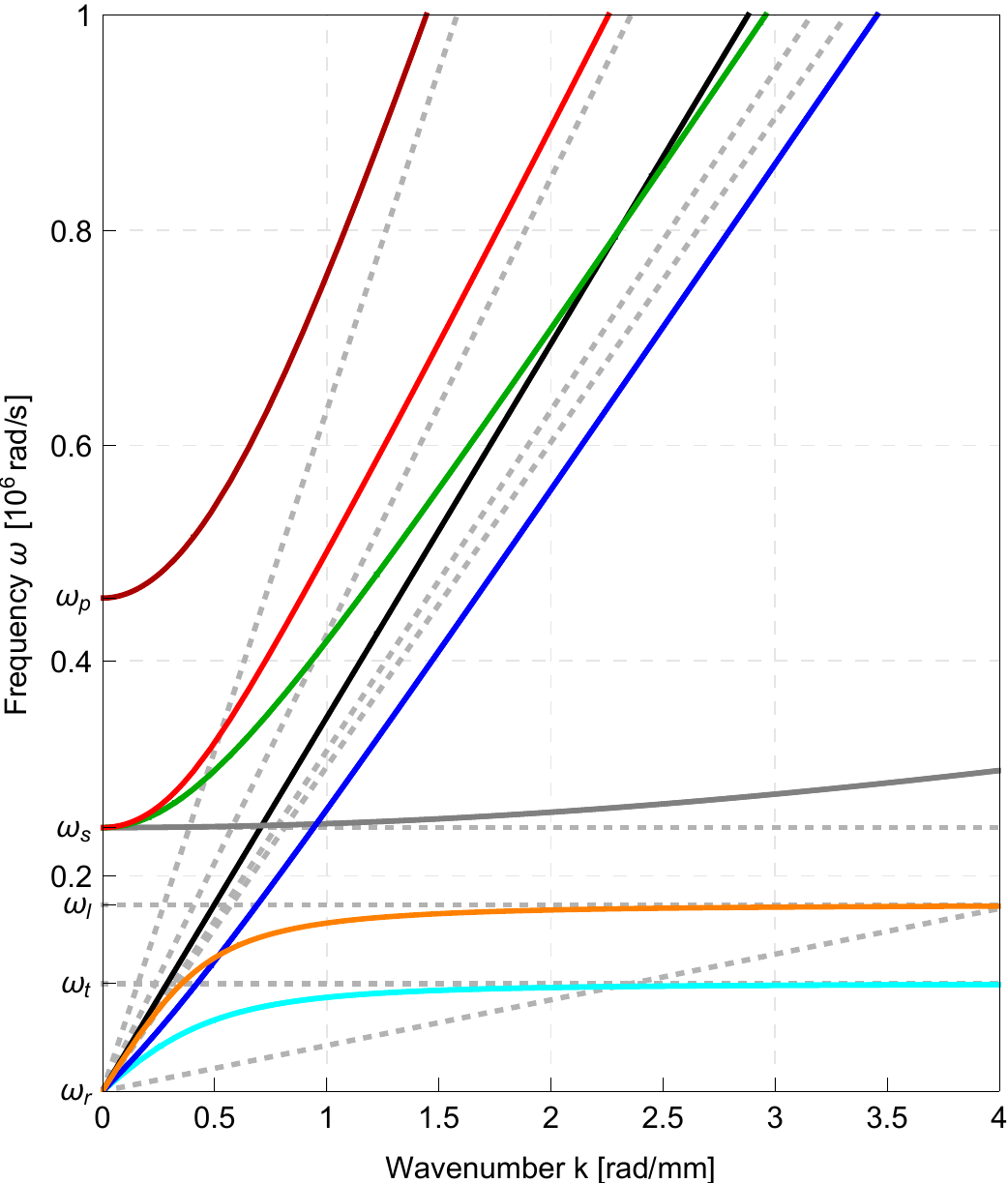} & \includegraphics[scale=0.5]{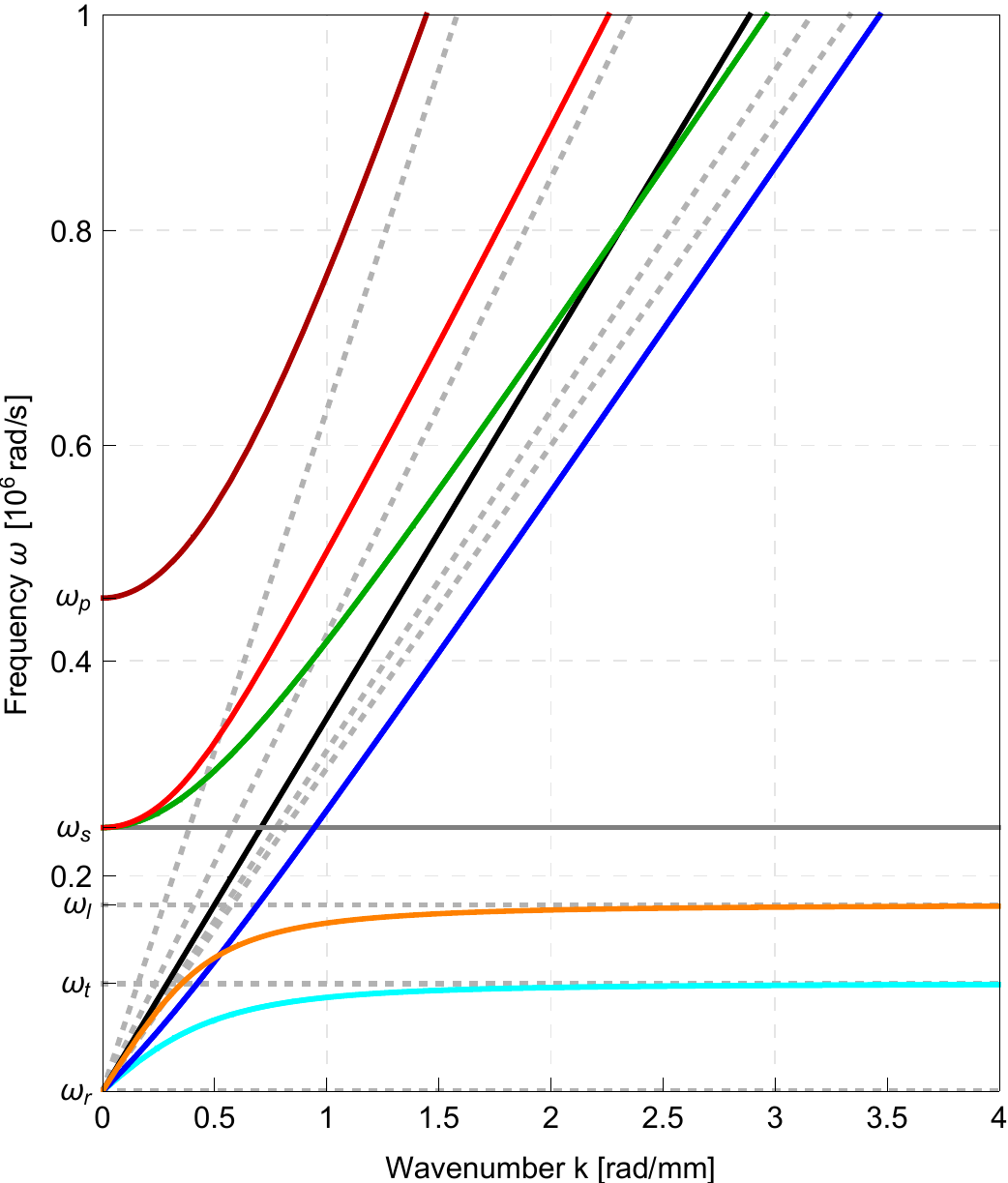}\tabularnewline
$\alpha_{1}=0.05$ & $\alpha_{1}=0.01$ & $\alpha_{1}=0$\tabularnewline
\end{tabular}\caption{Effect of the parameter $\alpha_{1}$ on the dispersion curves for
the case $\mu_{c}=0$. Higher values of $\alpha_{1}$ have some non-negligible
effects on the new extra acoustic curves.}
\end{figure}
In this case we see that two curves (black and green) become acoustic.
As a consequence, there is no complete band gap. This is coherent
with the results of \cite{madeo2015wave} in which the existence of
2 complete band-gaps is directly related to a non-vanishing Cosserat
couple modulus $\mu_{c}>0$. The particular effect of the parameter
$\alpha_{1}=0$ on the existence of a horizontal curve is preserved
(see also Fig.\ref{fig:alpha_1}).

We explicitly mention that the presence of 4 acoustic curves is not
observed in any known pattern of dispersion curves for real metamaterials.
This means that such metamaterials need to have a non-vanishing Cosserat
couple modulus $\mc>0$ which allows for the description of rotational
micro-motions at higher frequencies.

\newpage{}

\subsubsection{Vanishing Cosserat couple modulus ${\displaystyle \mu_{c}=0}$ and
$\lim_{\alpha_{2}\protect\fr0}$}

Characteristic limit elastic energy $\left\Vert \sym\left(\nabla u-P\right)\right\Vert ^{2}+\left\Vert \sym\,P\right\Vert ^{2}+\left\Vert \sym\,\curl\,P\right\Vert ^{2}$.\\
Characteristic limit kinetic energy $\left\Vert u_{,t}\right\Vert ^{2}+\left\Vert P_{,t}\right\Vert ^{2}$.

\begin{figure}[H]
\centering{}%
\begin{tabular}{ccc}
\includegraphics[scale=0.5]{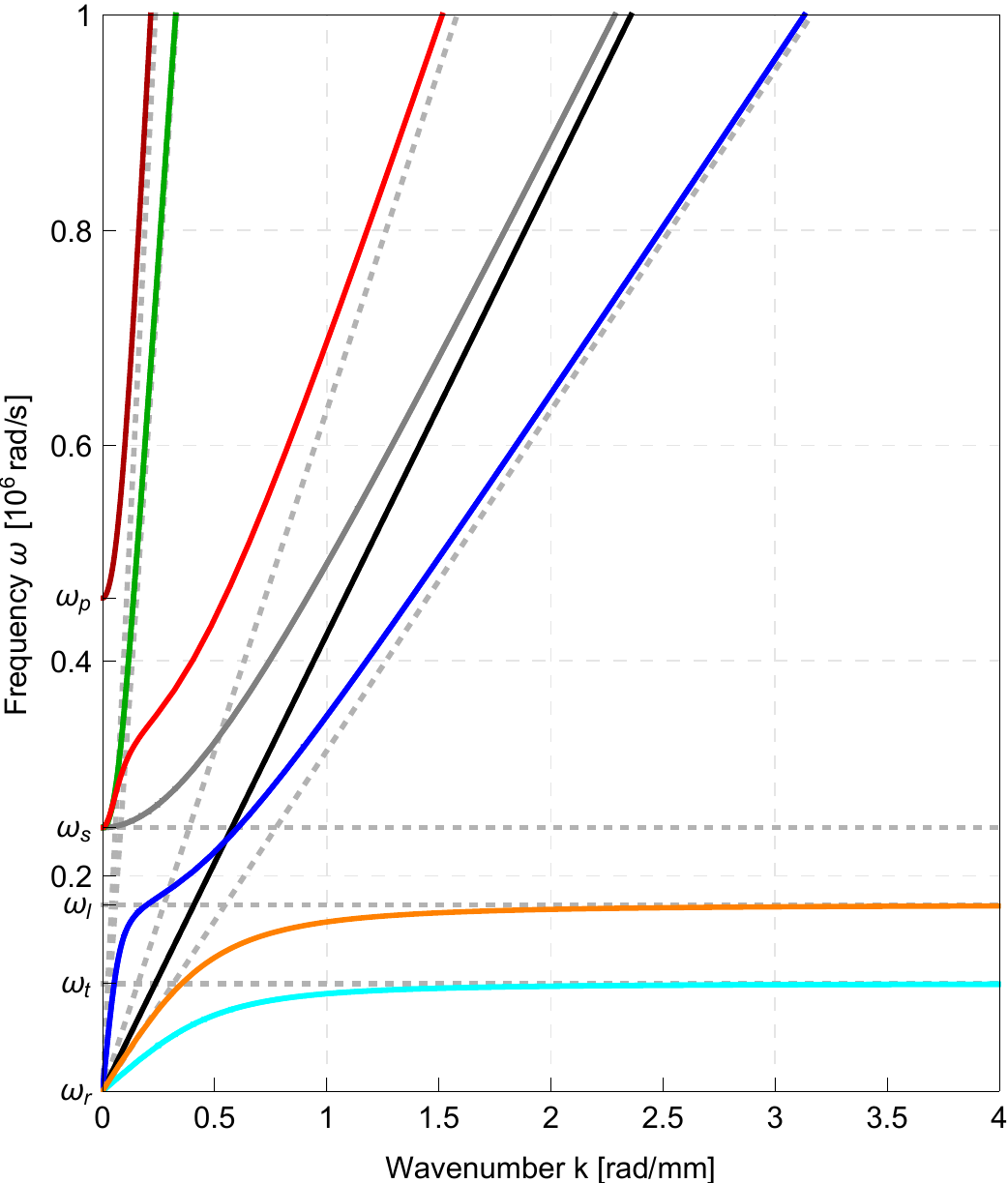} & \includegraphics[scale=0.5]{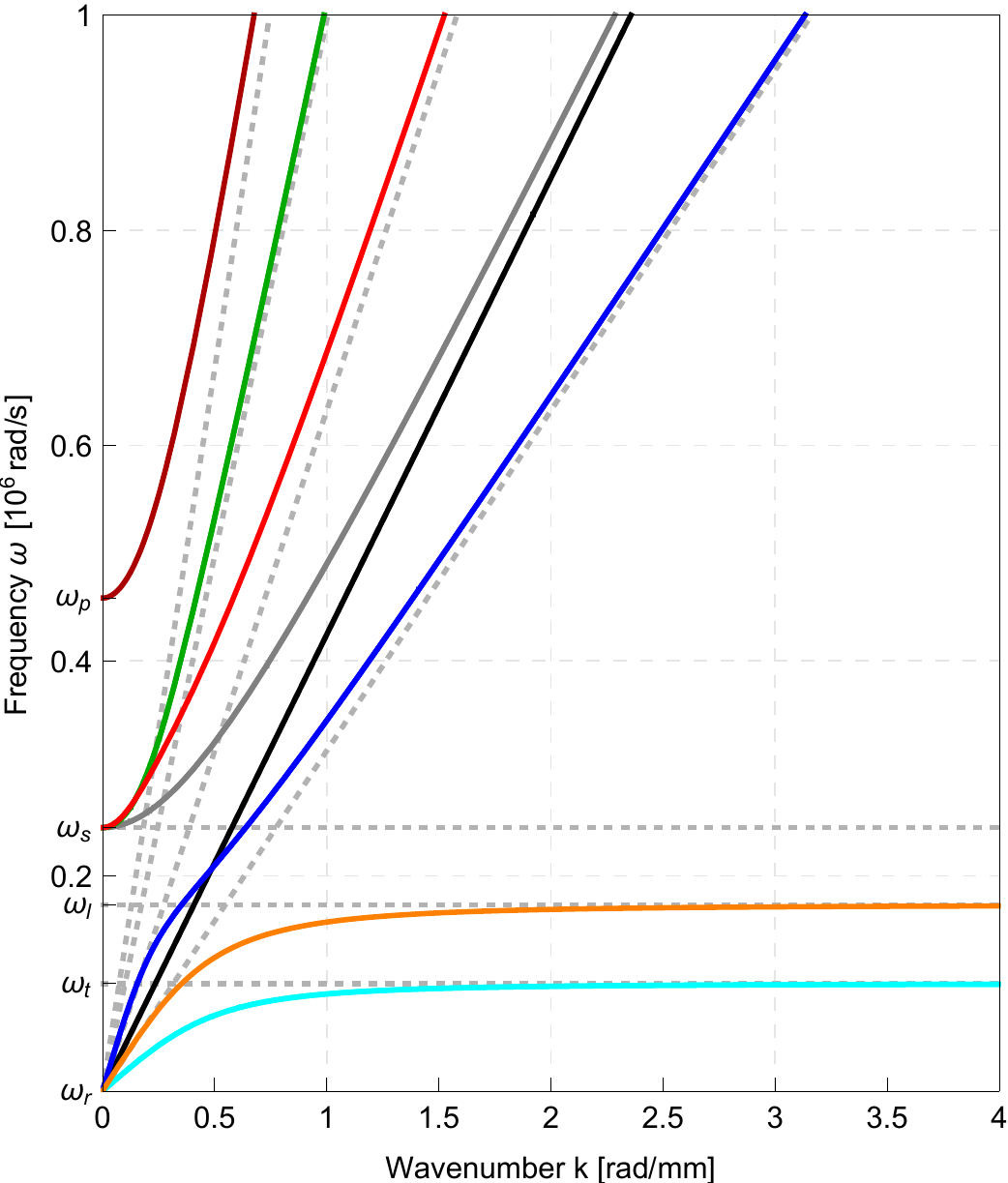} & \includegraphics[scale=0.5]{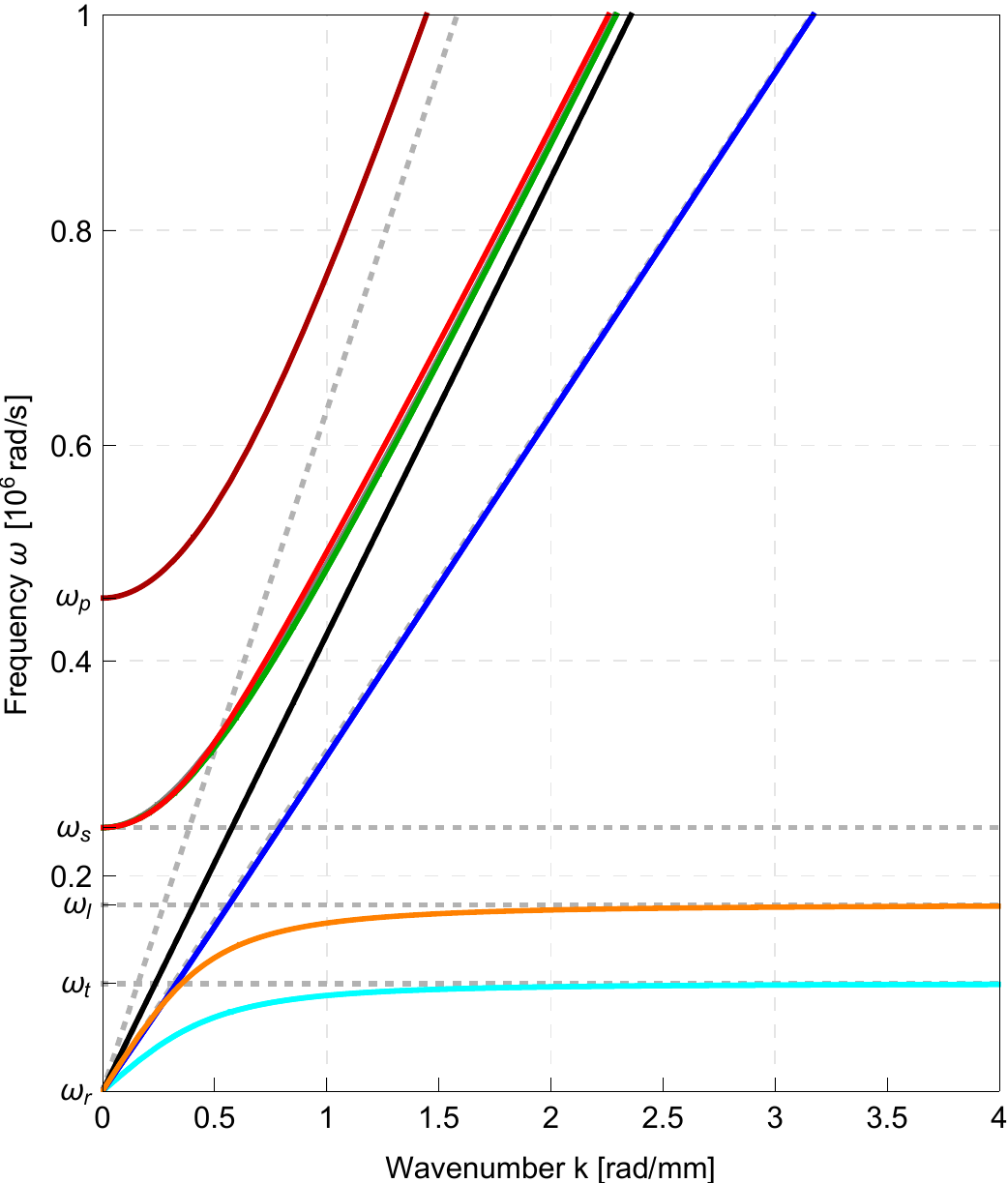}\tabularnewline
$\alpha_{2}=100$ & $\alpha_{2}=10$ & $\alpha_{2}=1$\tabularnewline
\includegraphics[scale=0.5]{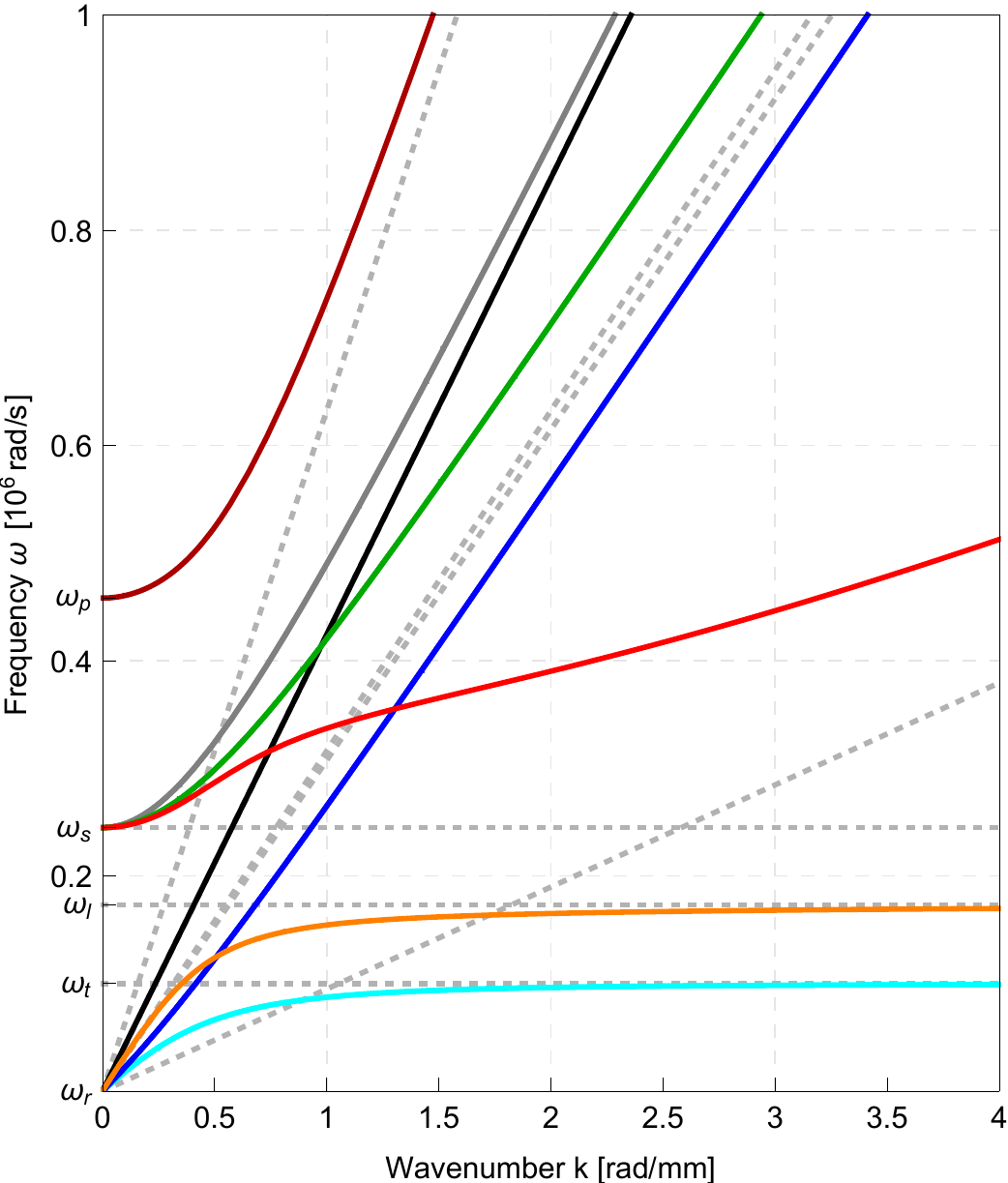} & \includegraphics[scale=0.5]{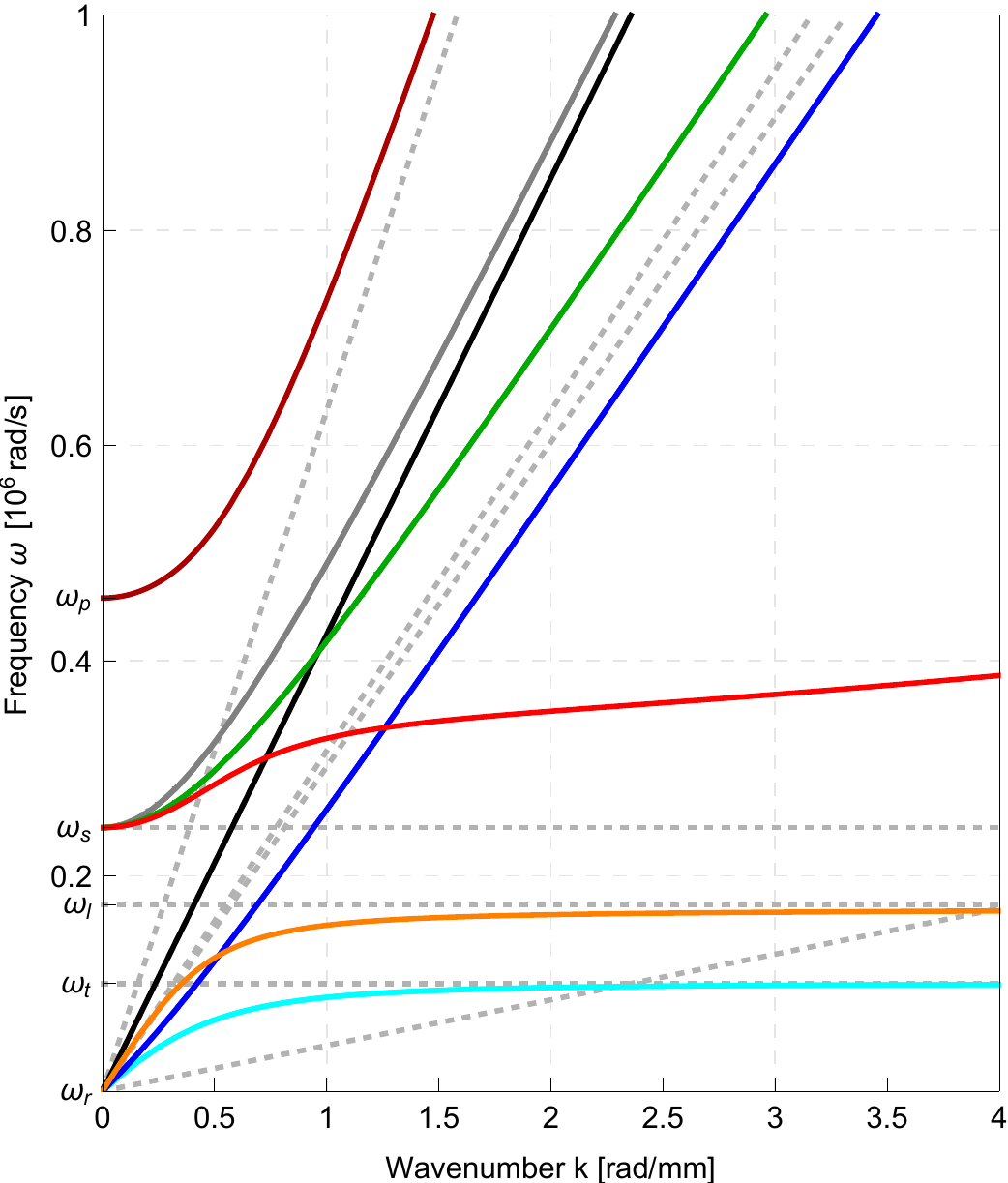} & \includegraphics[scale=0.5]{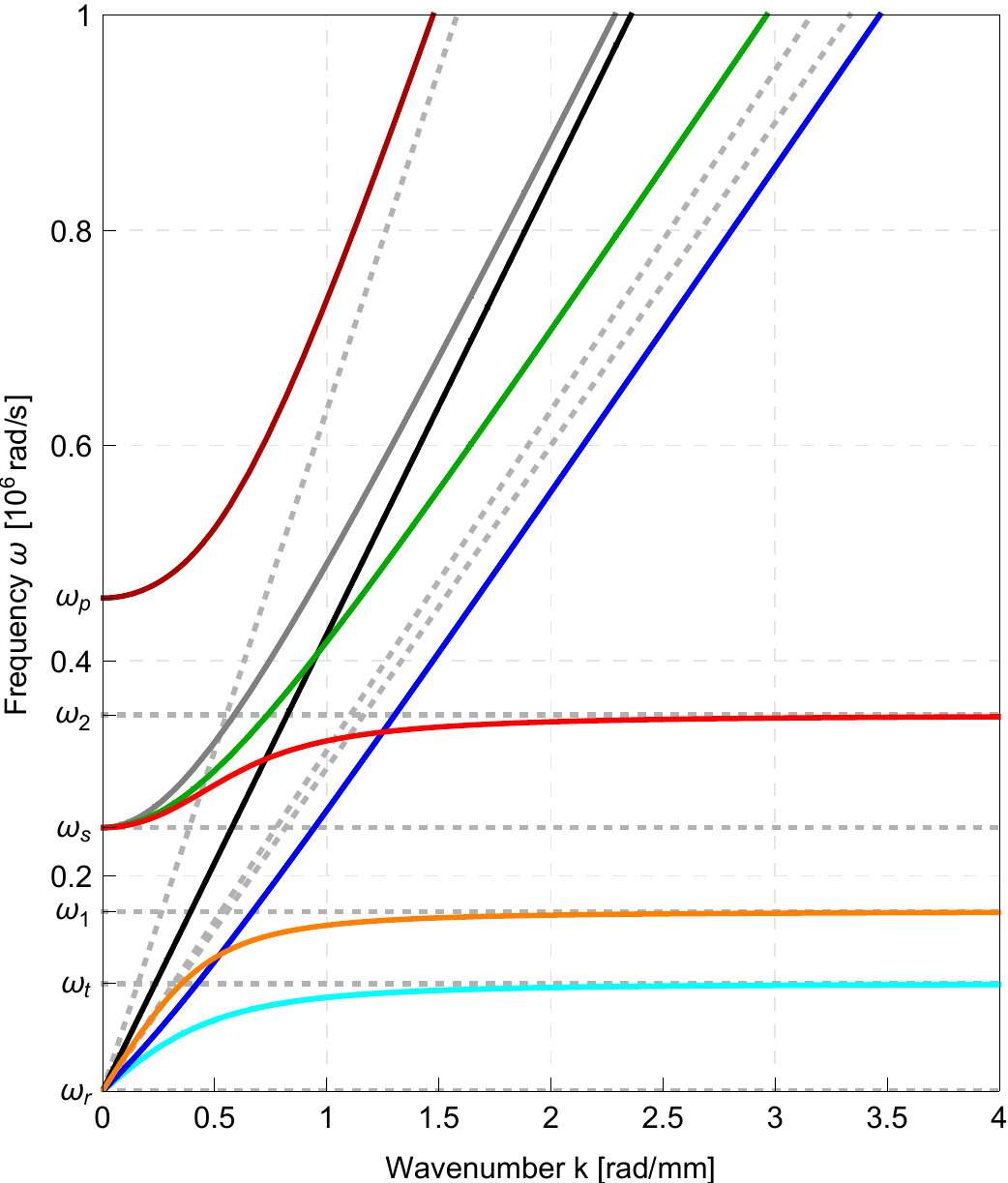}\tabularnewline
$\alpha_{2}=0.05$ & $\alpha_{2}=0.01$ & $\alpha_{2}=0$\tabularnewline
\end{tabular}\caption{Effect of the parameter $\alpha_{2}$ on the dispersion curves for
the case $\mu_{c}=0$.}
\end{figure}

Again, the two extra characteristic acoustic curves that arise when
setting $\mc=0$ are recovered again. An effect of the parameter $\alpha_{2}$
similar to the one shown in Fig. \ref{fig:alpha_2} is also found
for the optic wave which becomes horizontal. A high value of $\alpha_{2}$
has also a visible effect on one of the two extra acoustic curves.

\newpage{}

\subsubsection{Vanishing Cosserat couple modulus ${\displaystyle \mu_{c}=0}$ and
$\lim_{\alpha_{3}\protect\fr0}$}

Characteristic limit elastic energy $\left\Vert \sym\left(\nabla u-P\right)\right\Vert ^{2}+\left\Vert \sym\,P\right\Vert ^{2}+\left\Vert \dev\,\curl\,P\right\Vert ^{2}$.\\
Characteristic limit kinetic energy $\left\Vert u_{,t}\right\Vert ^{2}+\left\Vert P_{,t}\right\Vert ^{2}$.

\begin{table}[H]
\centering{}%
\begin{tabular}{ccc}
\includegraphics[scale=0.5]{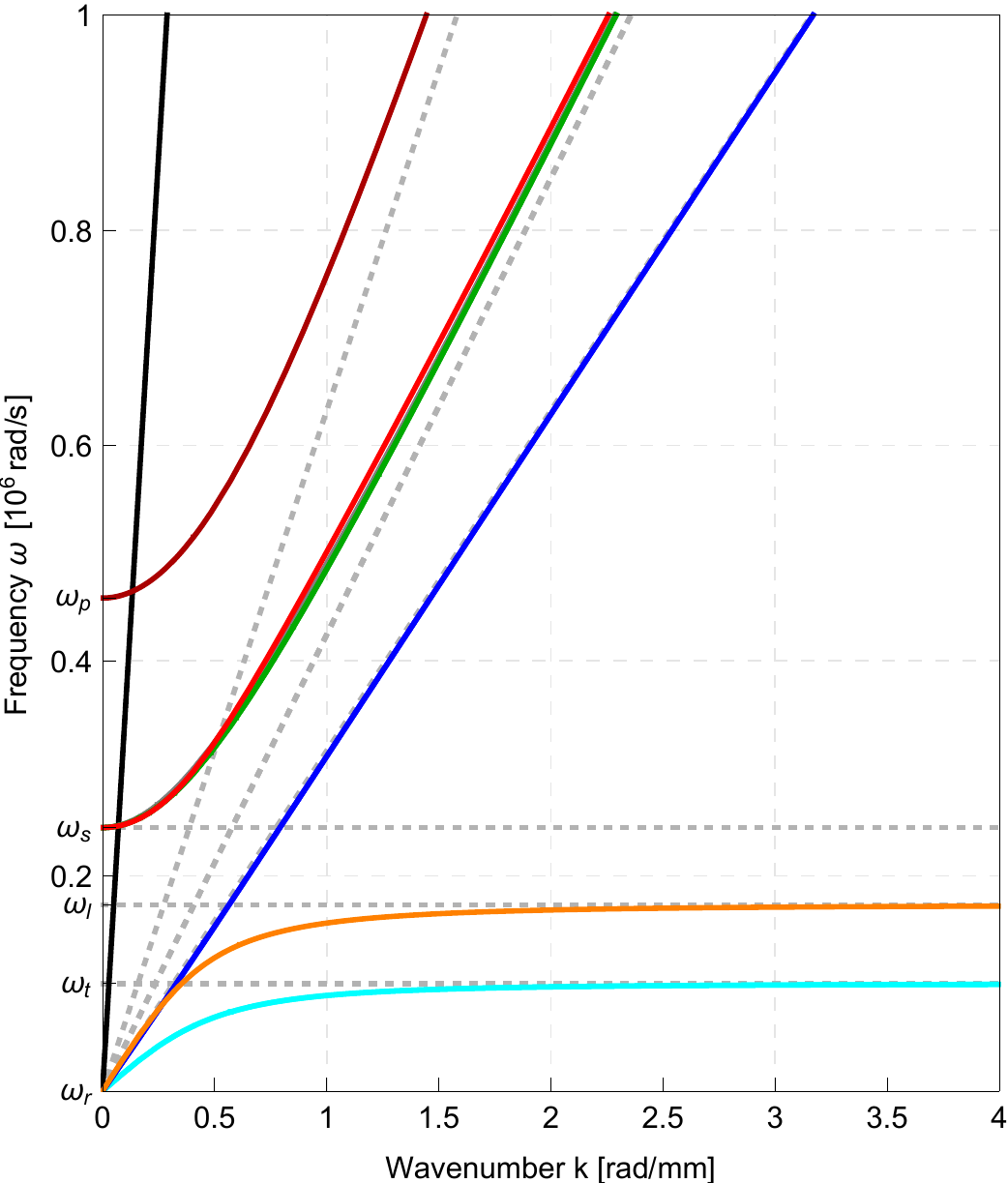} & \includegraphics[scale=0.5]{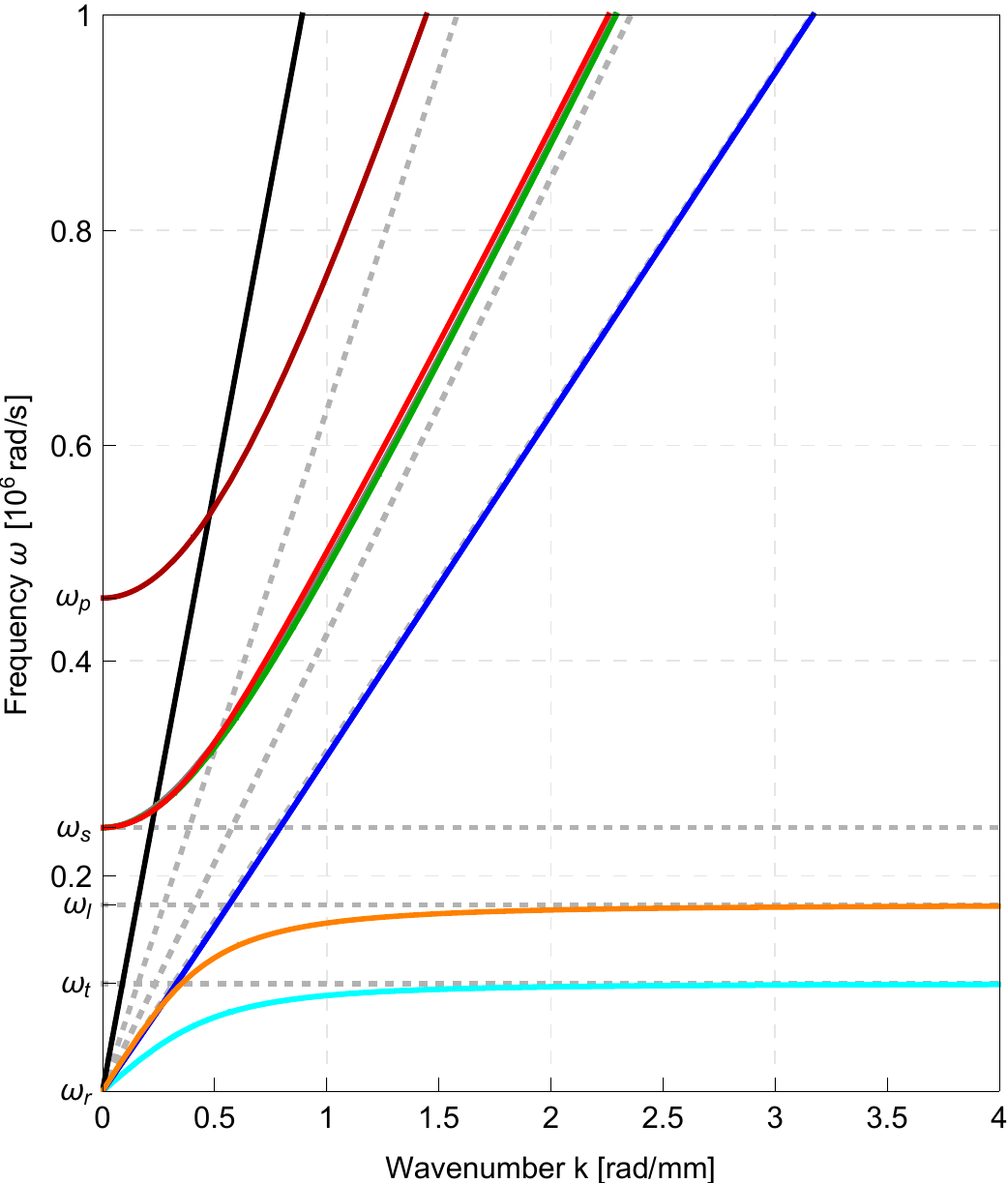} & \includegraphics[scale=0.5]{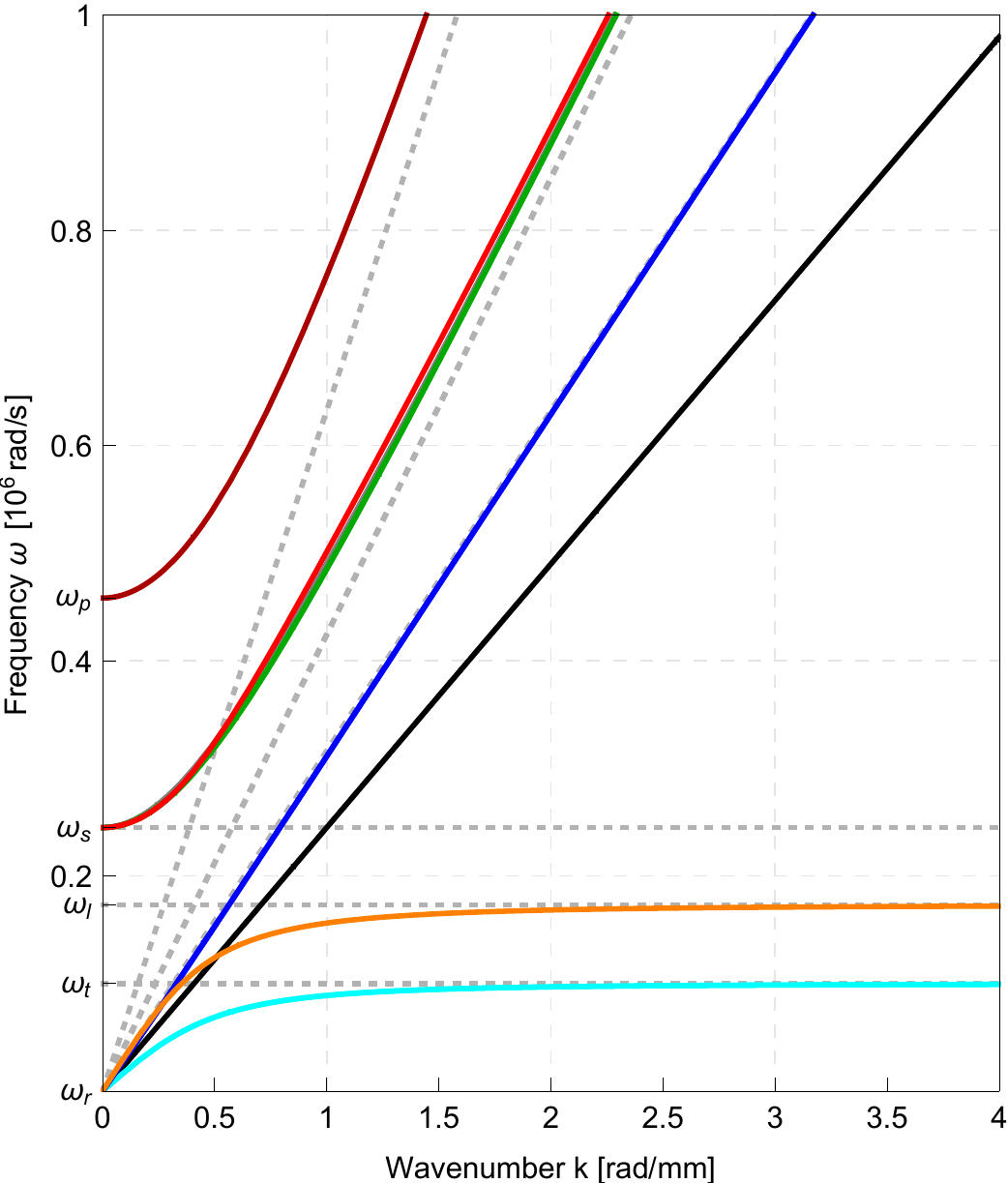}\tabularnewline
$\alpha_{3}=100$ & $\alpha_{3}=10$ & $\alpha_{3}=0$\tabularnewline
\end{tabular}\caption{Effect of the parameter $\alpha_{3}$ on the dispersion curves for
the case $\mu_{c}=0$.}
\end{table}

In this case we see again that two curves (blue and black) become
acoustic. There is no complete band gap. This confirms once again
the need of having $\mu_{c}>0$ as a necessary condition for the existence
of complete band-gaps. The effect of the parameter $\alpha_{3}$
is limited to the control of the slope of one of the two extra acoustic
curves whose onset is related to the fact of setting $\mc=0$.

\newpage{}

\subsection{Variation of the micro-inertia weighting}

As shown in section 4.3.1, the three weights of the micro-inertia
have a fundamental role on the definition of the cut-off frequencies
of the optic waves. Indeed, this is one of the main results of the
present paper: the split of the micro-inertia allows to control separately
the starting point of the optic curves which can be translated along
the y - axis by simply varying the value of each of the parameters
$\eta_{1},\eta_{2},\eta_{3}$. Such possibility of independent control
of the optic branches is a major characteristic for an effective calibration
of the material parameters of the relaxed micromorphic model on the
dispersion patterns of real metamaterials.

\subsubsection{Case ${\displaystyle \eta_{1}\protect\fr0}$}

Characteristic limit elastic energy $\left\Vert \nabla u-P\right\Vert ^{2}+\left\Vert \sym\,P\right\Vert ^{2}+\left\Vert \curl\,P\right\Vert ^{2}$.\\
Characteristic limit kinetic energy $\left\Vert u_{,t}\right\Vert ^{2}+\left\Vert \skew P_{,t}\right\Vert ^{2}+\frac{1}{3}\left(\textrm{tr}\,P_{,t}\right)^{2}$.

\begin{figure}[H]
\centering{}%
\begin{tabular}{ccc}
\includegraphics[scale=0.5]{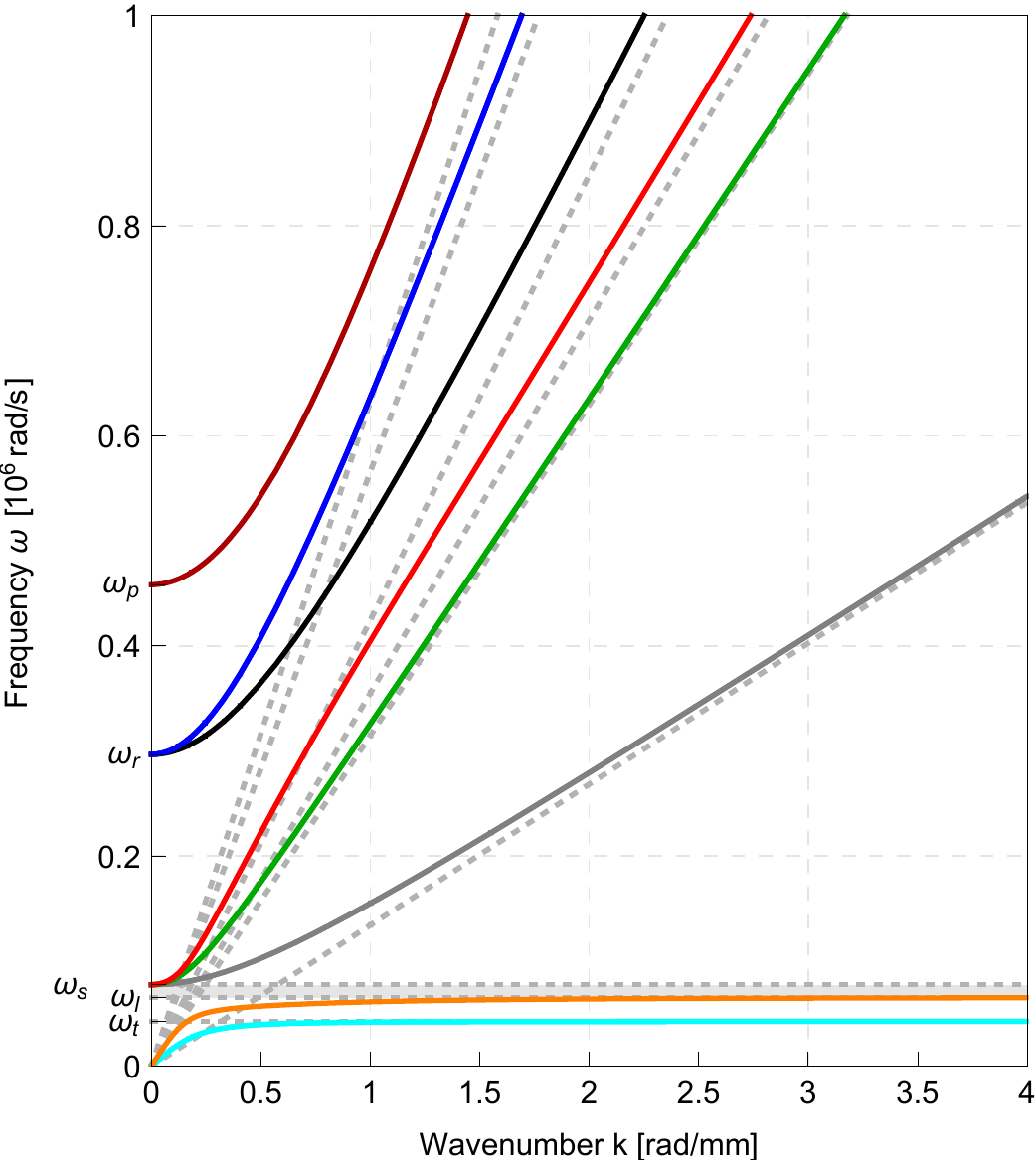} & \includegraphics[scale=0.5]{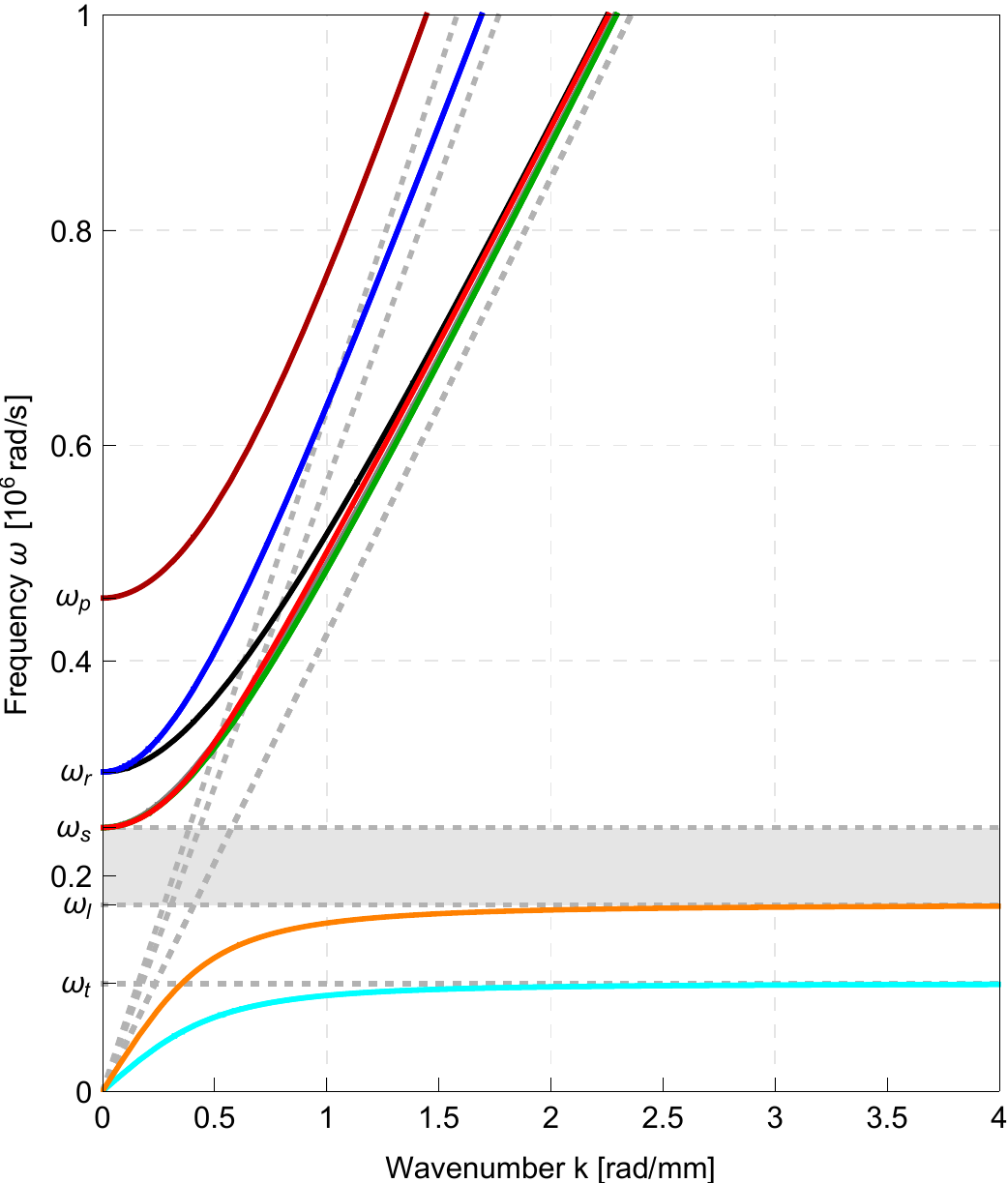} & \includegraphics[scale=0.5]{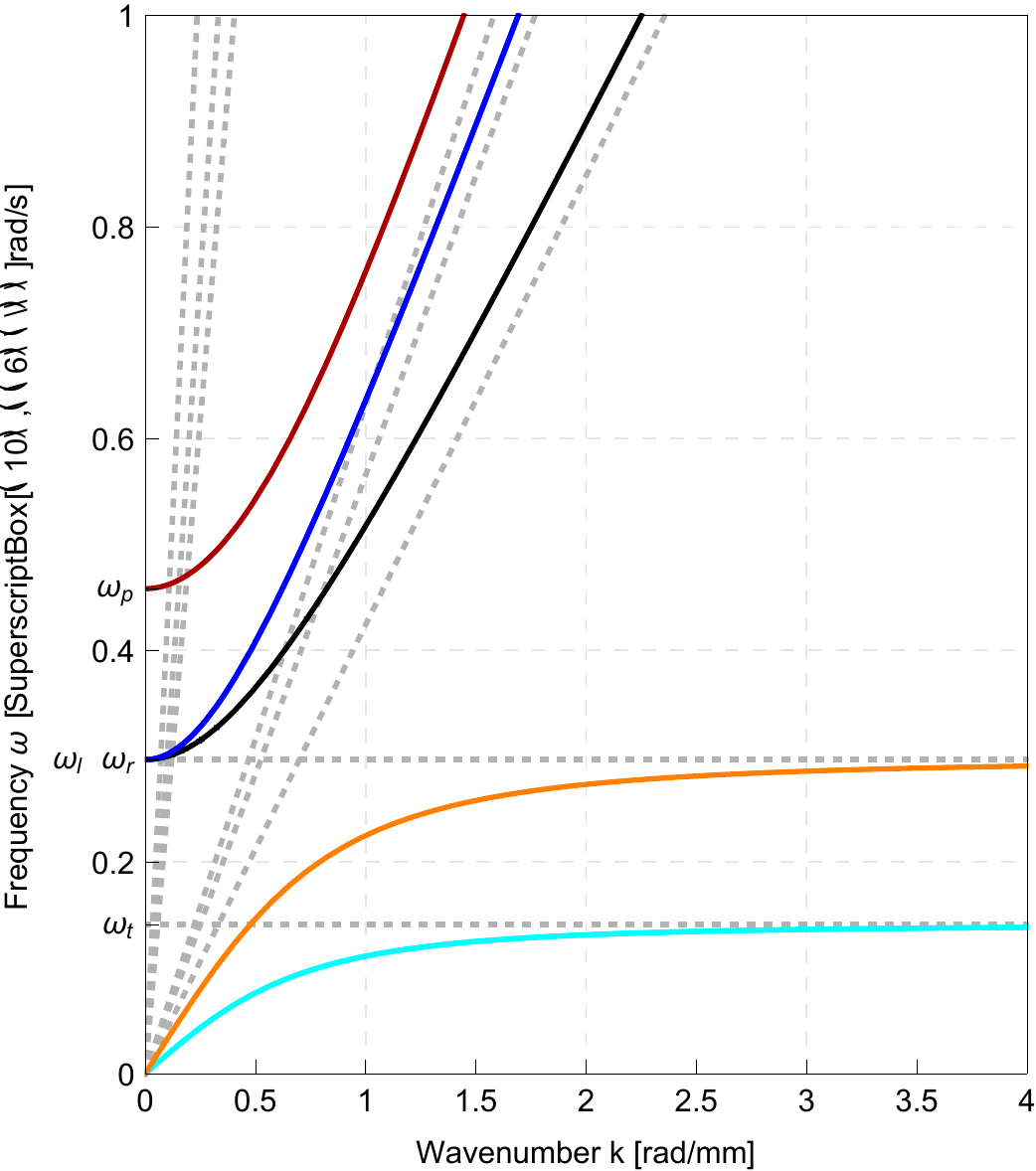}\tabularnewline
$\eta_{1}=10^{-1}$ & $\eta_{1}=10^{-2}$ & $\eta_{1}=10^{-4}$\tabularnewline
\end{tabular}\caption{Effect of the parameter $\eta_{1}$ on the dispersion curves.\label{fig:eta1}}
\end{figure}

In this case we can see that the band gap is preserved when $\eta_{1}\in\left(0,10^{-2}\right)$.
For values $\eta_{1}\in\left[10^{-2},10^{2}\right]$ the band gap
is always present but it becomes smaller. For values of $\eta_{1}$
smaller than $10^{-4}$ the behavior of the dispersion curves is unchanged
with respect to the case with $\eta_{1}=10^{-4}$. This characteristic
behavior is directly related to the definition of the cut-off frequency
$\omega_{s}=\sqrt{\frac{2\,\me+\mh}{\eta_{1}}}$. For $\eta_{1}\fr0$
some of the optic branches go to infinity and do not appear in the
dispersion diagram (Fig. \ref{fig:eta1} right). For smaller values
of $\eta_{1}$ the optic branches starting from the cut-off frequency
$\omega_{s}$ appear in the dispersion diagram (Fig. \ref{fig:eta1}
center). For higher values of $\eta_{1}$, the optic curves originating
from $\omega_{s}$ start from a lower value and the slope of one of
such curves becomes smaller (Fig. \ref{fig:eta1} right). In the limit
$\eta_{1}\fr\infty$ one would expect that two optic branches related
to $\omega_{s}$ become acoustic. This is indeed the case as it will
be shown in subsection 5.4.5.

\newpage{}

\subsubsection{Case $\eta_{2}\protect\fr0$}

Characteristic limit elastic energy $\left\Vert \nabla u-P\right\Vert ^{2}+\left\Vert \sym\,P\right\Vert ^{2}+\left\Vert \curl\,P\right\Vert ^{2}$.\\
Characteristic limit kinetic energy $\left\Vert u_{,t}\right\Vert ^{2}+\left\Vert \sym\,P_{,t}\right\Vert ^{2}$.

\begin{figure}[H]
\centering{}%
\begin{tabular}{ccc}
\includegraphics[scale=0.5]{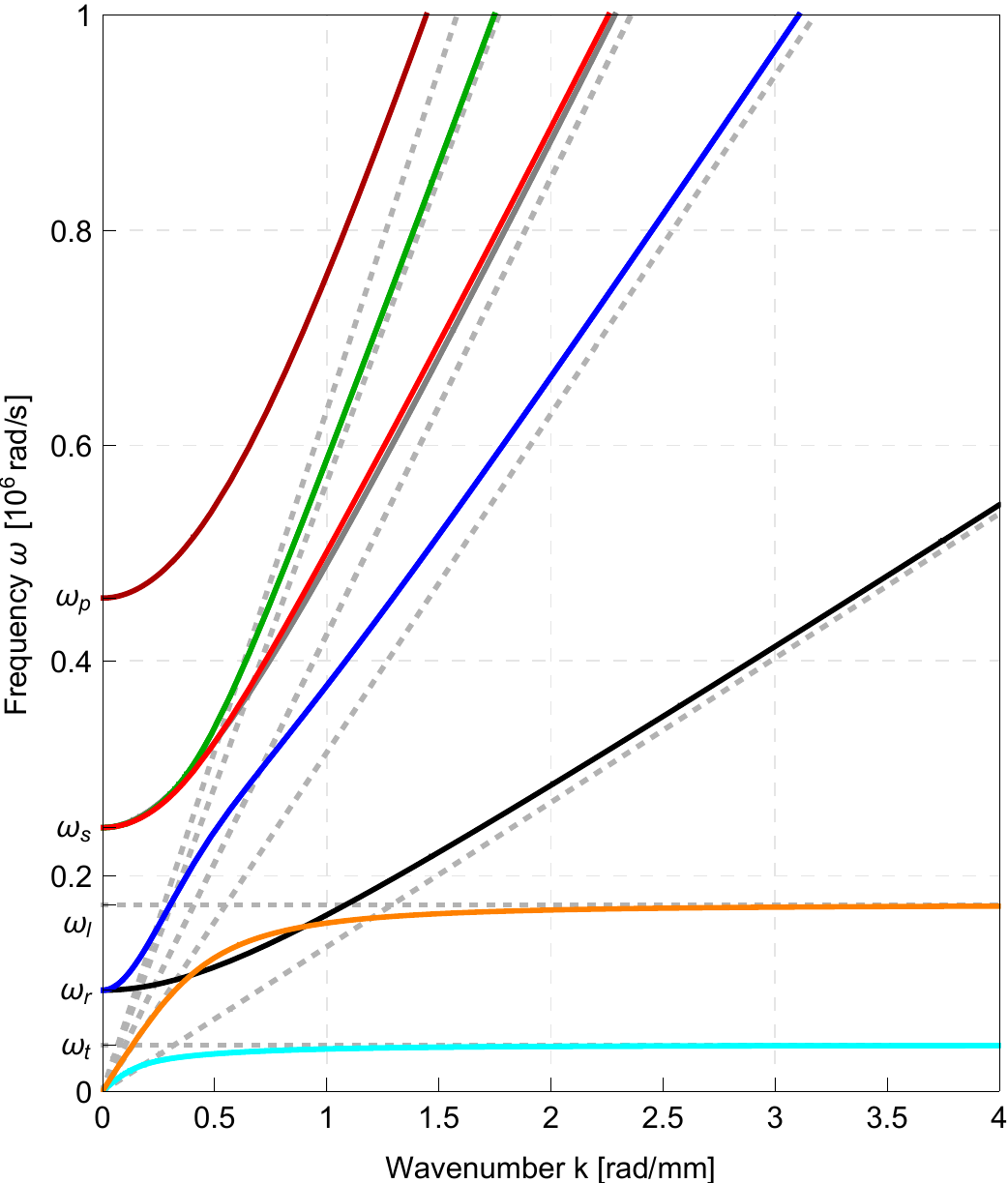} & \includegraphics[scale=0.5]{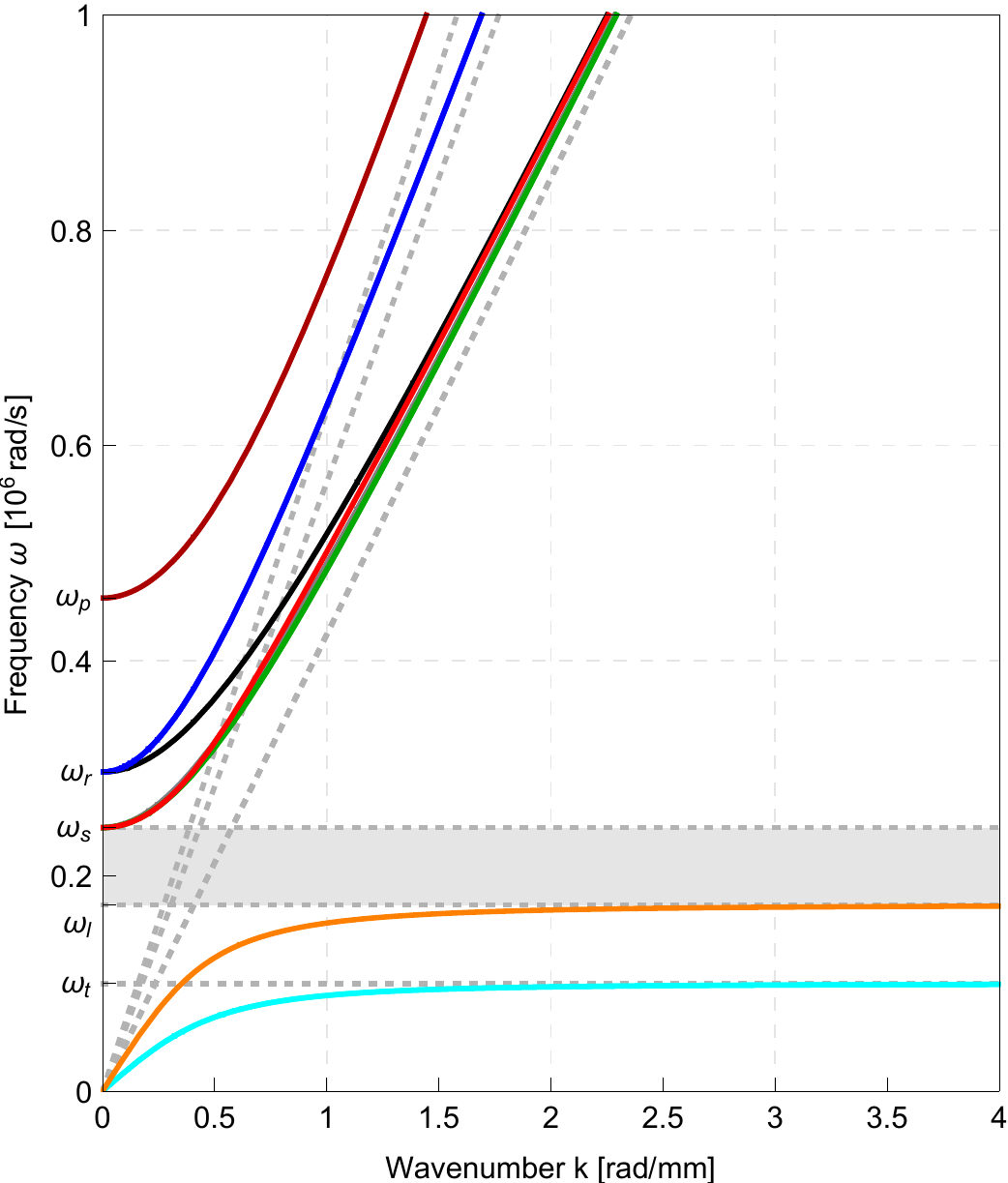} & \includegraphics[scale=0.5]{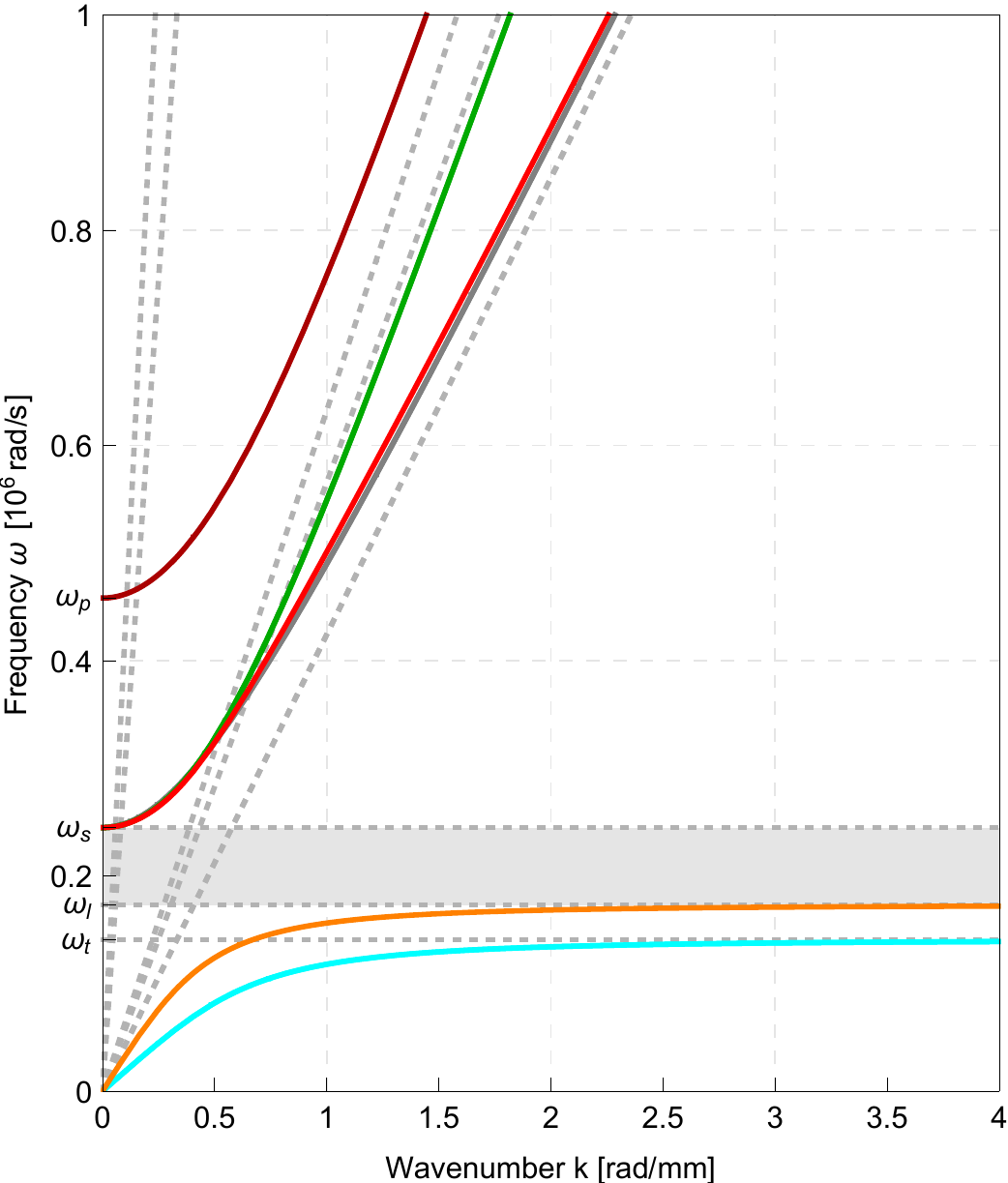}\tabularnewline
$\eta_{2}=10^{-1}$ & $\eta_{2}=10^{-2}$ & $\eta_{2}=10^{-4}$\tabularnewline
\end{tabular}\caption{Effect of the parameter $\eta_{2}$ on the dispersion curves.}
\end{figure}

In this case we can see that the band gap is preserved when $\eta_{2}\in\left(0,10^{-2}\right)$.
There is a value $\eta_{\textrm{crit}}\in\left(10^{-2},10^{2}\right)$
such that for every $\eta_{2}>\eta_{\textrm{crit}}$ the band gap
is absent. This critical value can be related to the definition of
the cut-off frequency $\omega_{r}=\sqrt{\frac{2\,\mu_{c}}{\eta_{2}}}$.

Analogous consideration can be made with respect to the preceding
case. In the present case the optic branches originating from $\omega_{r}$
are involved in the translations of the cut-offs associated to the
variation of $\eta_{2}$. The limit case $\eta_{2}\fr\infty$ will
be discussed in subsection 5.4.6.

\subsubsection{Case $\eta_{3}\protect\fr0$}

Characteristic limit elastic energy $\left\Vert \nabla u-P\right\Vert ^{2}+\left\Vert \sym\,P\right\Vert ^{2}+\left\Vert \curl\,P\right\Vert ^{2}$.\\
Characteristic limit kinetic energy $\left\Vert u_{,t}\right\Vert ^{2}+\left\Vert \dev\,P_{,t}\right\Vert ^{2}$.

\begin{figure}[H]
\centering{}%
\begin{tabular}{ccc}
\includegraphics[scale=0.5]{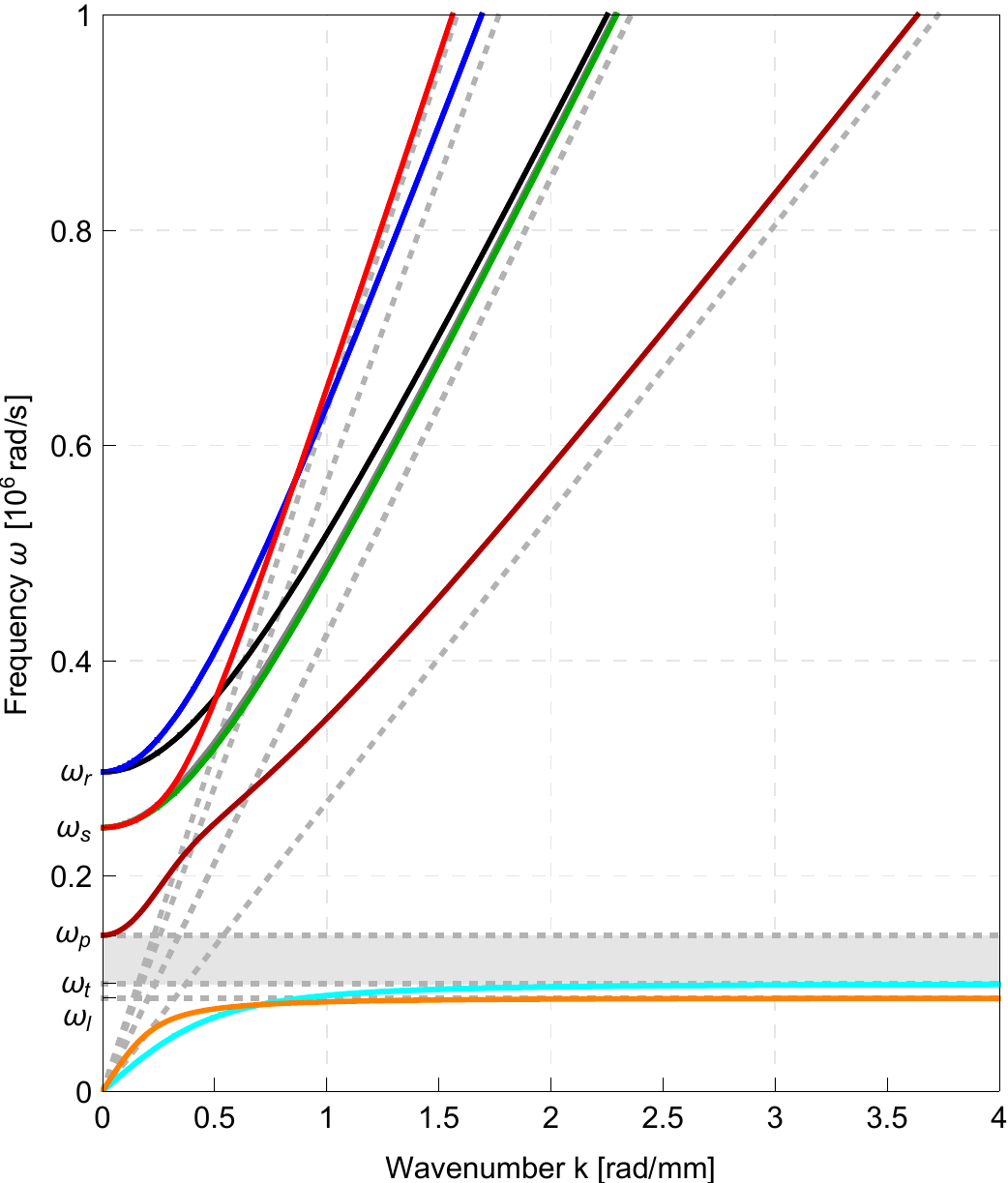} & \includegraphics[scale=0.5]{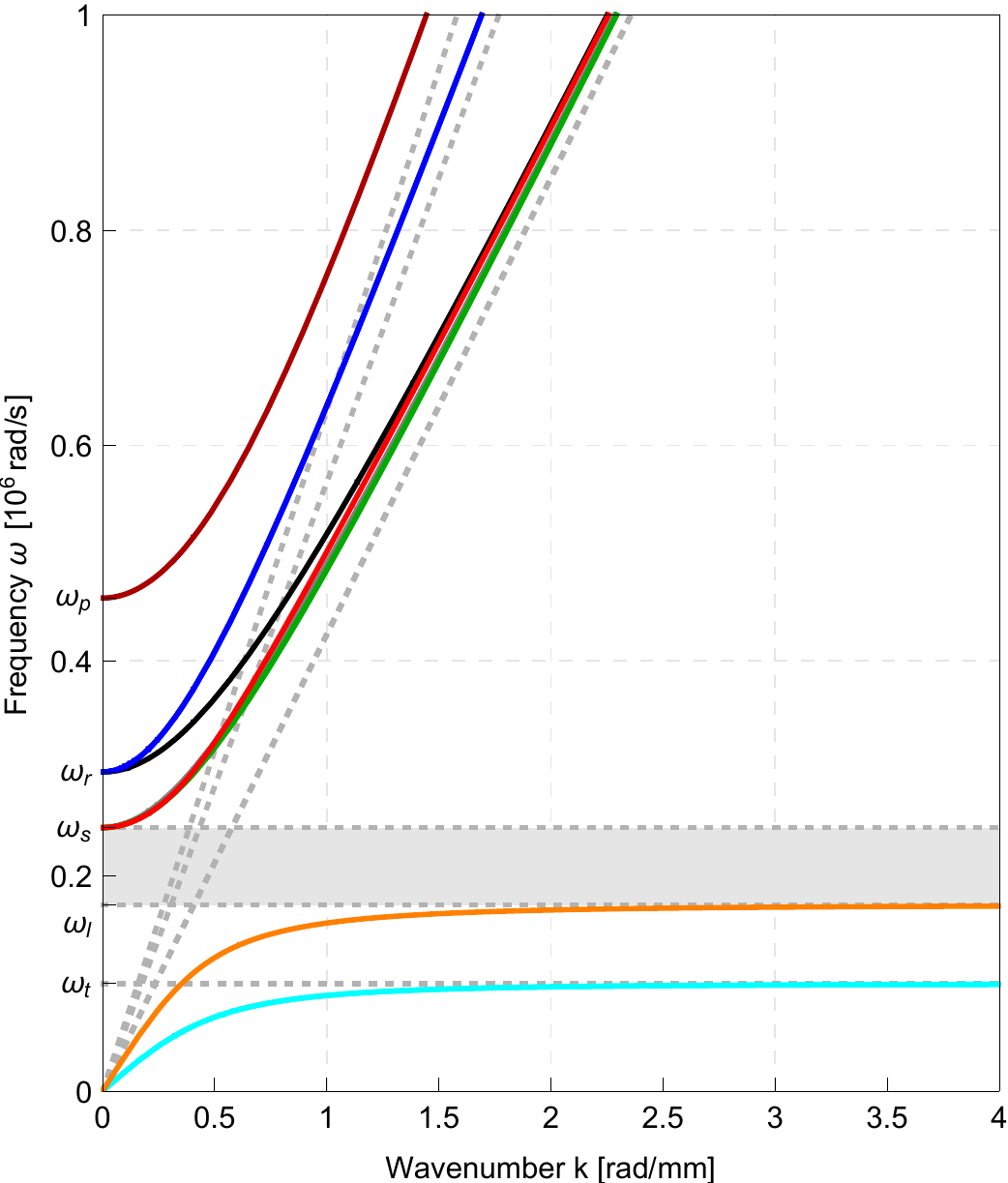} & \includegraphics[scale=0.5]{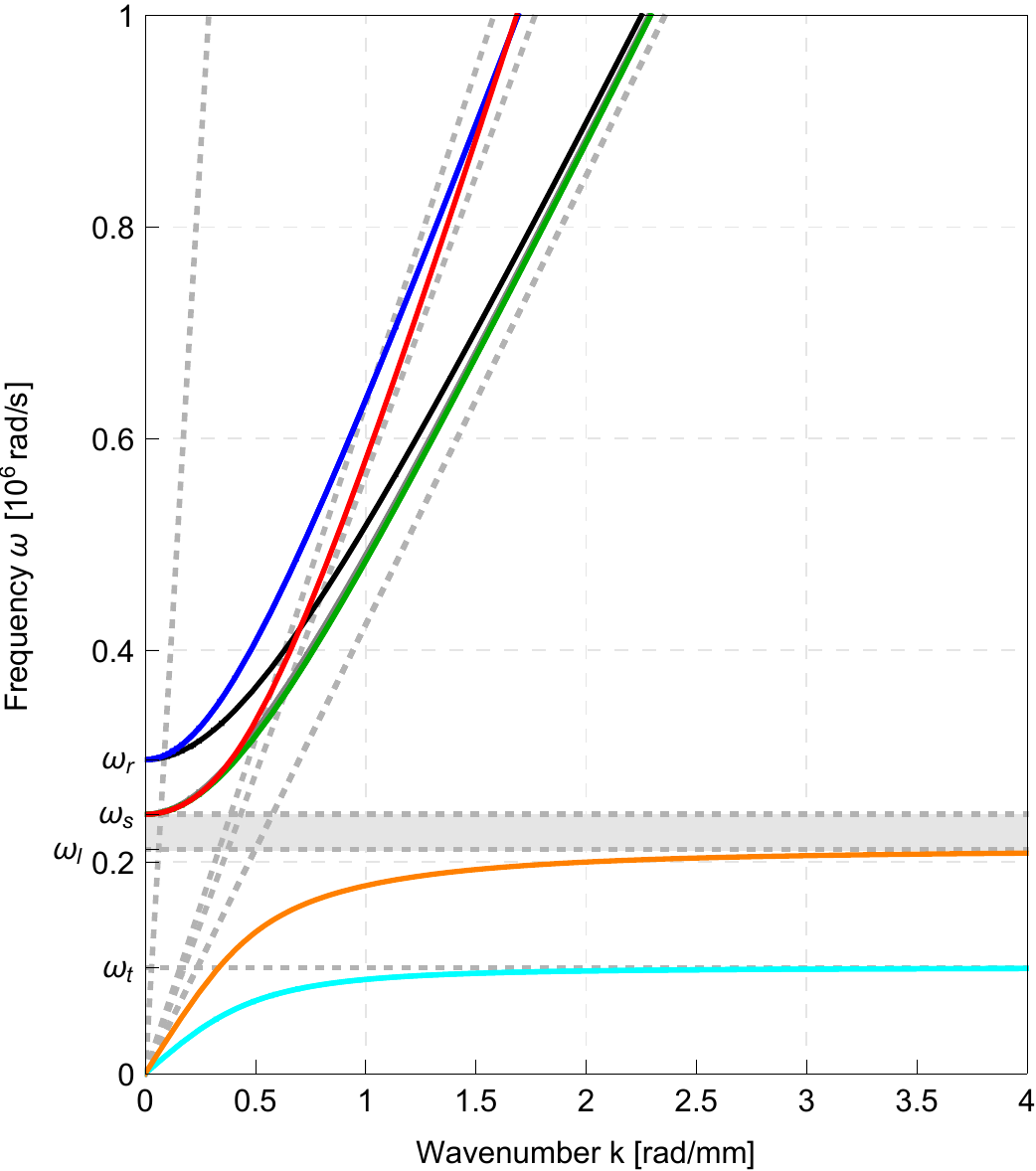}\tabularnewline
$\eta_{3}=10^{-1}$ & $\eta_{3}=10^{-2}$ & $\eta_{3}=10^{-4}$\tabularnewline
\end{tabular}\caption{Effect of the parameter $\eta_{3}$ on the dispersion curves.}
\end{figure}

In this case we can see that the band gap is preserved when $\eta_{3}\in\left(0,10^{-2}\right)$.
We will see in section 5.4.7 that there is a value $\eta_{\textrm{crit}}\in\left(10^{-2},10^{2}\right)$
such that for every $\eta_{3}>\eta_{\textrm{crit}}$ the band gap
is absent. A similar behavior with respect to the previous two cases
is observed for what concerns the optic wave originating from $\omega_{p}=\sqrt{\frac{3\left(\le+\lh\right)+2\left(\me+\mh\right)}{\eta_{3}}}$.
The limit case $\eta_{3}\fr\infty$ will be described in section 5.4.7. 

\subsubsection{Cases ${\displaystyle \eta_{1,}\eta_{2},\eta_{3}\protect\fr0}$:
the fundamental role of the micro-inertia for enriched continuum mechanics}

Characteristic limit elastic energy $\left\Vert \nabla u-P\right\Vert ^{2}+\left\Vert \sym\,P\right\Vert ^{2}+\left\Vert \curl\,P\right\Vert ^{2}$.\\
Characteristic limit kinetic energy $\left\Vert u_{,t}\right\Vert ^{2}$.

\begin{figure}[H]
\centering{}%
\begin{tabular}{ccc}
\includegraphics[scale=0.5]{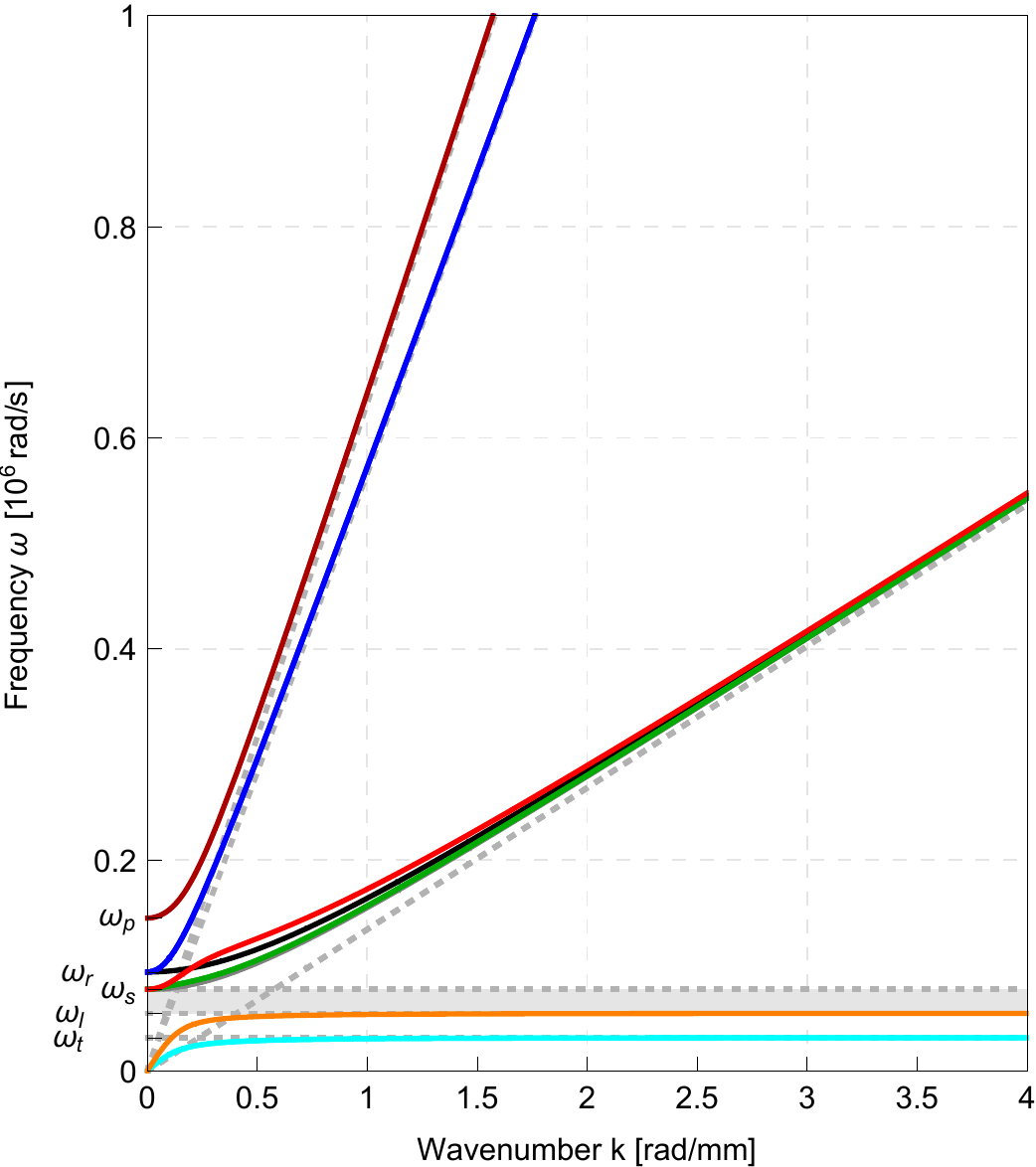} & \includegraphics[scale=0.5]{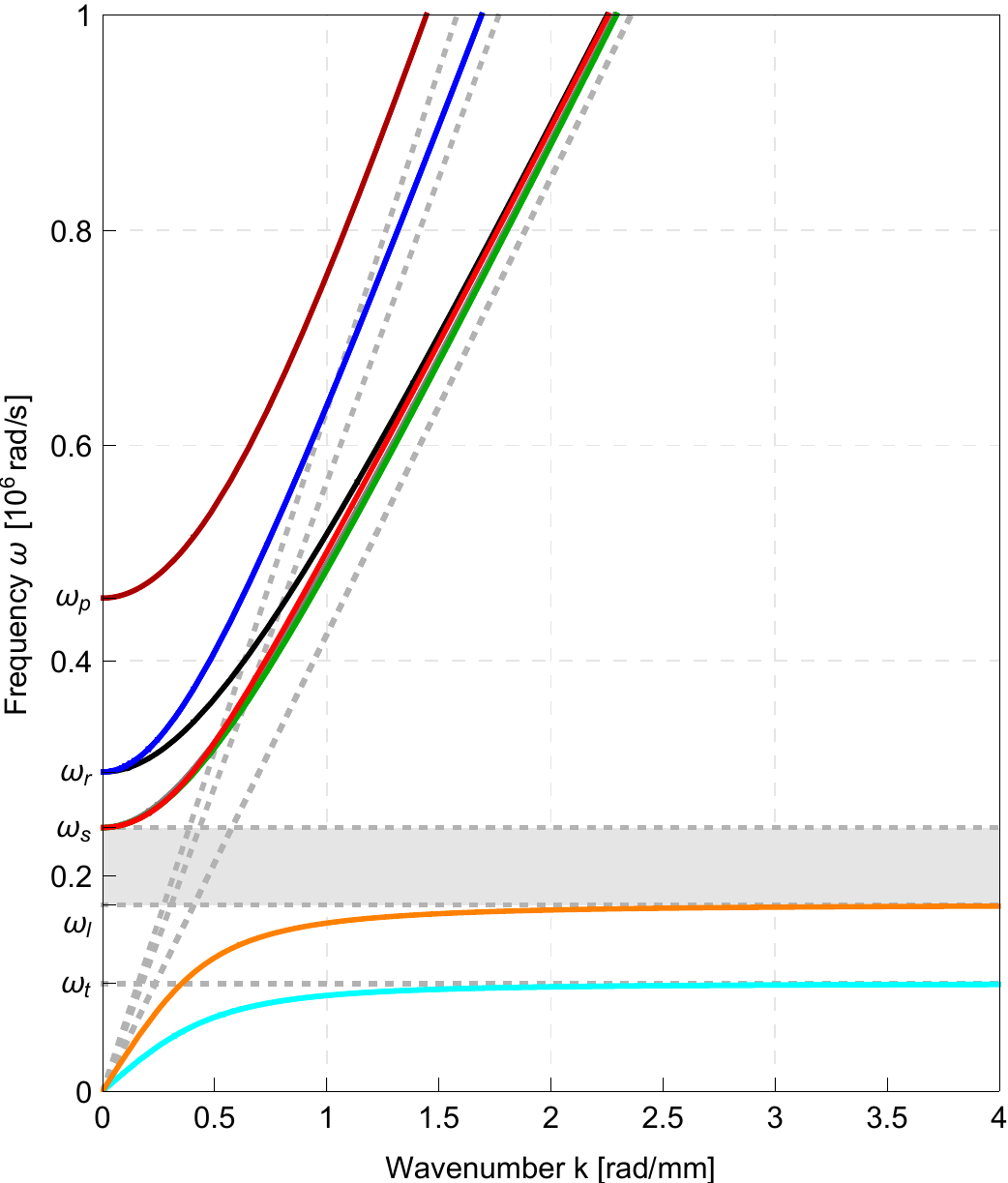} & \includegraphics[scale=0.5]{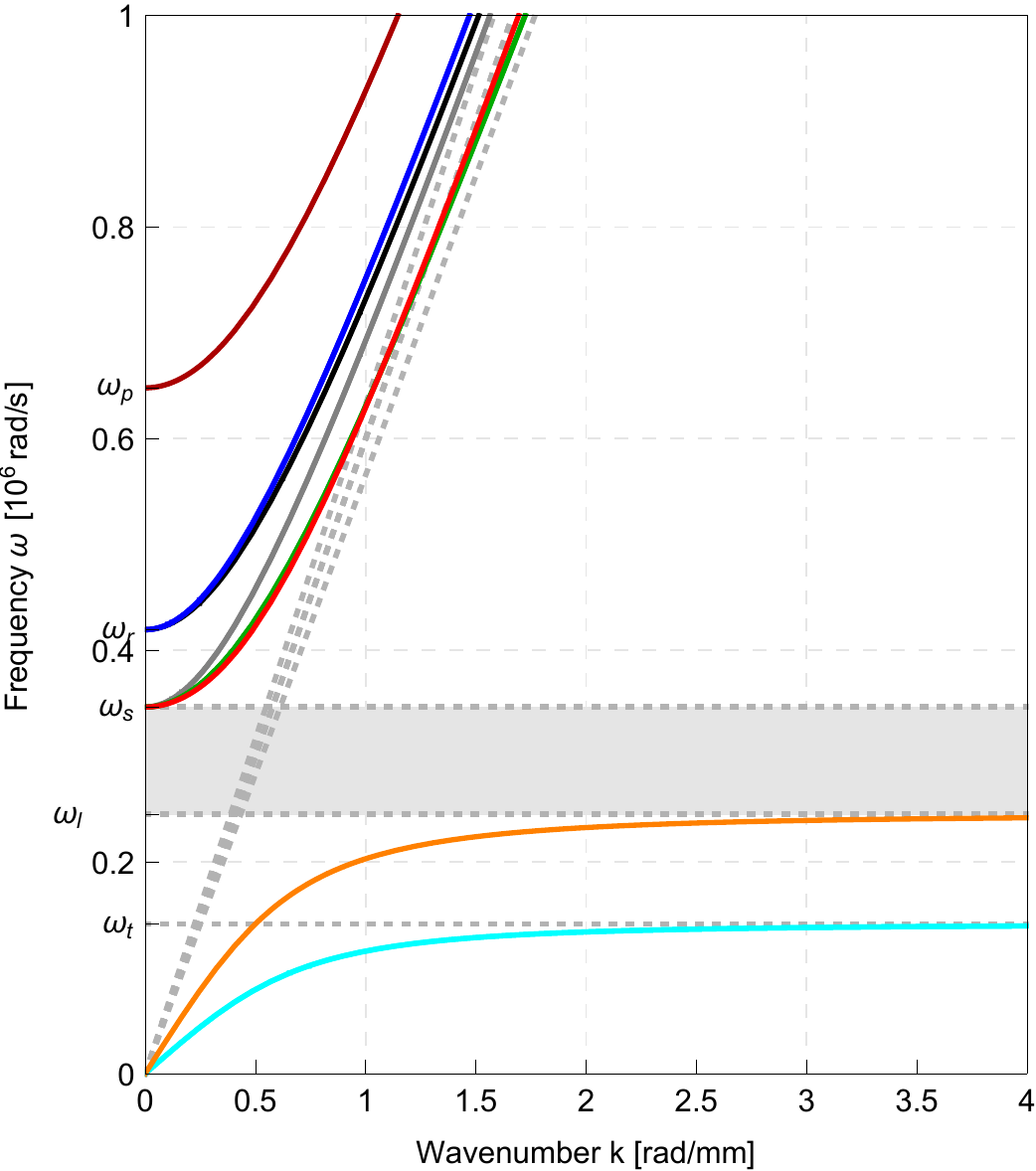}\tabularnewline
$\eta_{1}=\eta_{2}=\eta_{3}=10^{-1}$ & $\eta_{1}=\eta_{2}=\eta_{3}=10^{-2}$ & $\eta_{1}=\eta_{2}=\eta_{3}=5\cdot10^{-3}$\tabularnewline
\includegraphics[scale=0.5]{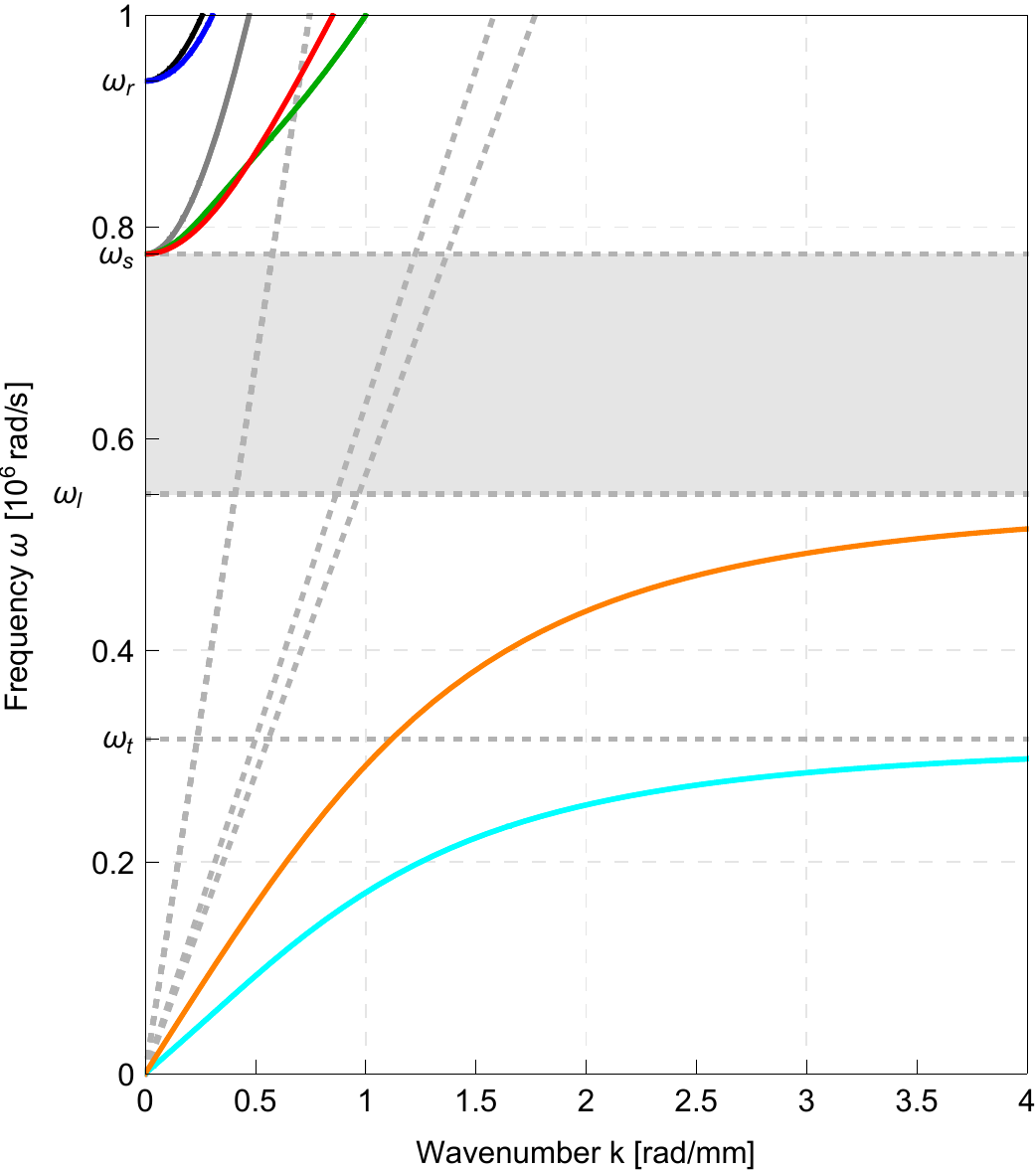} & \includegraphics[scale=0.5]{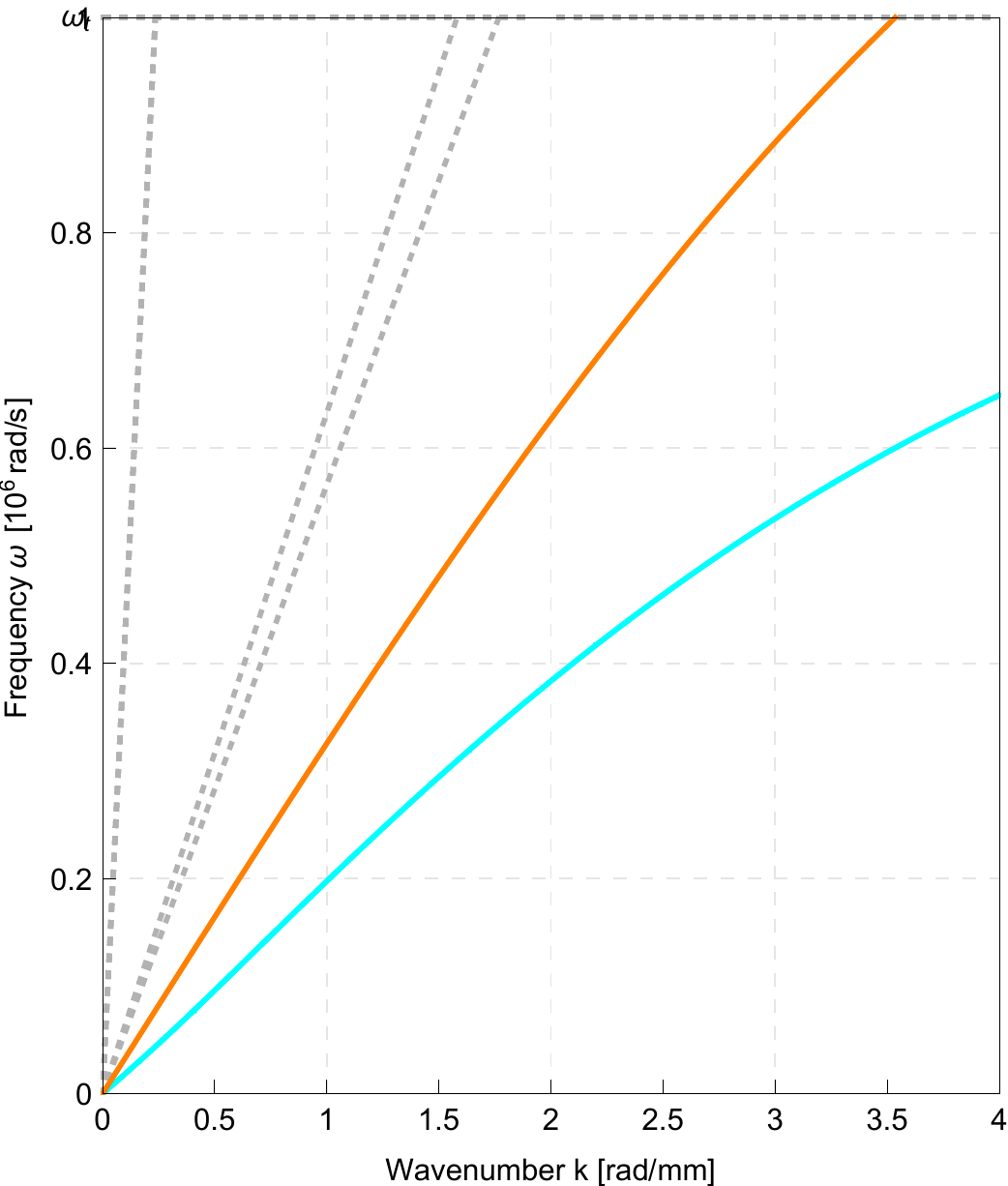} & \includegraphics[scale=0.5]{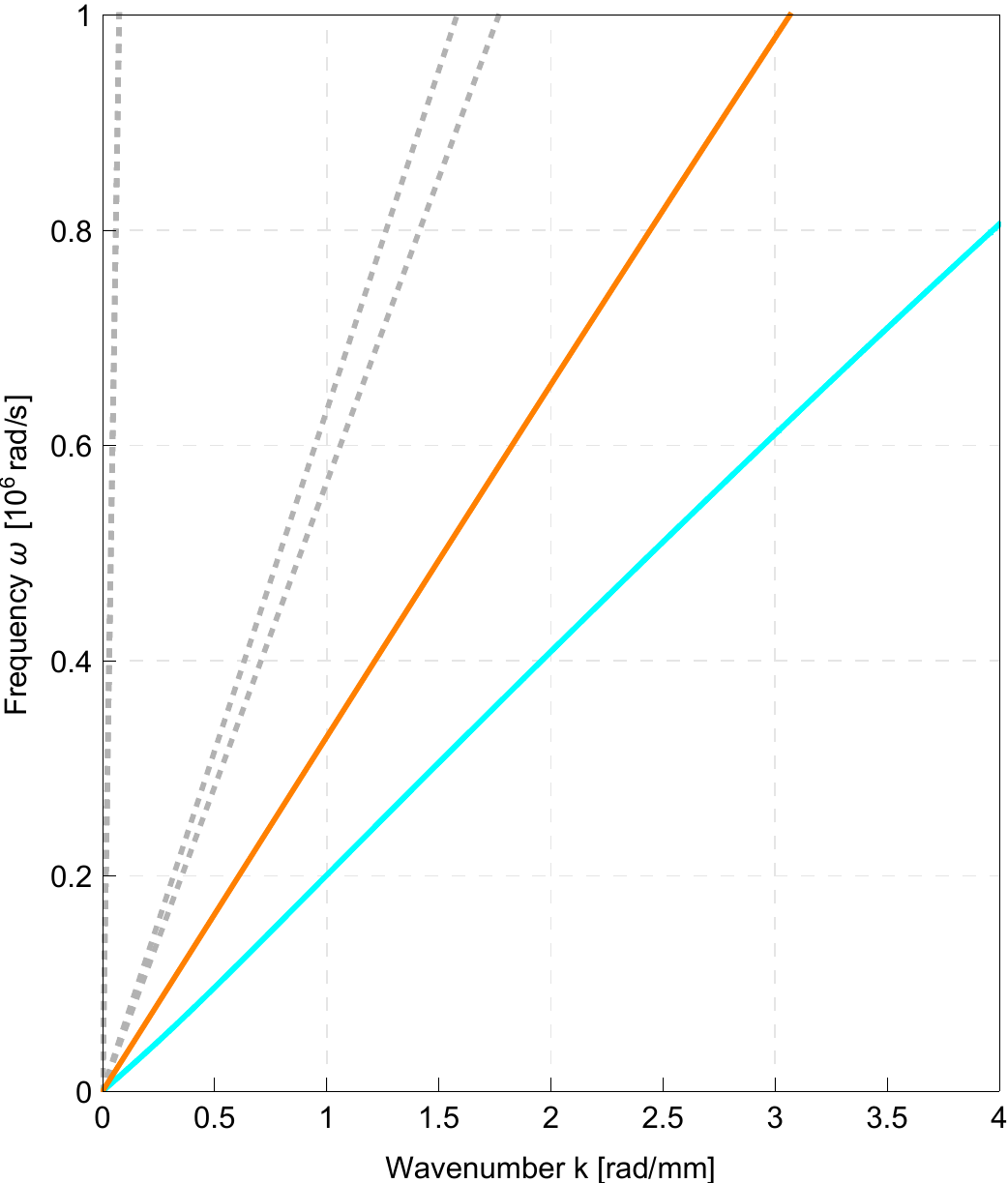}\tabularnewline
$\eta_{1}=\eta_{2}=\eta_{3}=10^{-3}$ & $\eta_{1}=\eta_{2}=\eta_{3}=10^{-4}$ & $\eta_{1}=\eta_{2}=\eta_{3}=10^{-5}$\tabularnewline
\end{tabular}\caption{\label{fig:Combined-effect}Combined effect of the parameter $\eta_{1},\eta_{2},\eta_{3}$
on the dispersion curves.}
\end{figure}

In Figure \ref{fig:Combined-effect} we show the combined effect of
the micro-inertia parameters on the behavior of the dispersion curves.
When letting the three parameters tend to zero with the same speed,
one ends up with a dispersion diagram which is peculiar of the classical
linear elastic Cauchy media (Fig.\ref{fig:Combined-effect} bottom
right). 

This is a fundamental result of the present study which is not exhaustively
treated in the literature: if one considers a continuum with enriched
kinematic ($u,\P$), if no micro-inertia is considered to complement
the macro-inertia $\rho\left\Vert u_{,t}\right\Vert ^{2}$, then the
dispersion curves will not be different from those of the classical
Cauchy continuum (Fig. \ref{fig:Cauchy - Cosserat} (a)). In order
to activate the micro-motions associated to the micro-distortion tensor
$\P$ a micro-inertia $\eta\left\Vert \P_{,t}\right\Vert ^{2}$ is
needed.

\newpage{}

\subsubsection{Case ${\displaystyle \eta_{1}\protect\fr+\infty}$}

Characteristic limit elastic energy $\left\Vert \sym\left(\nabla u-P\right)\right\Vert ^{2}+\left\Vert \sym\,P\right\Vert ^{2}+\left\Vert \curl\,P\right\Vert ^{2}$.\\
Characteristic limit kinetic energy $\left\Vert u_{,t}\right\Vert ^{2}+\left\Vert \skew\,P_{,t}\right\Vert ^{2}+\frac{1}{3}\left(\textrm{tr}\,P_{,t}\right)^{2},\quad\dev\,\sym\,\P_{,t}=0$.

\begin{figure}[H]
\centering{}%
\begin{tabular}{ccc}
\includegraphics[scale=0.5]{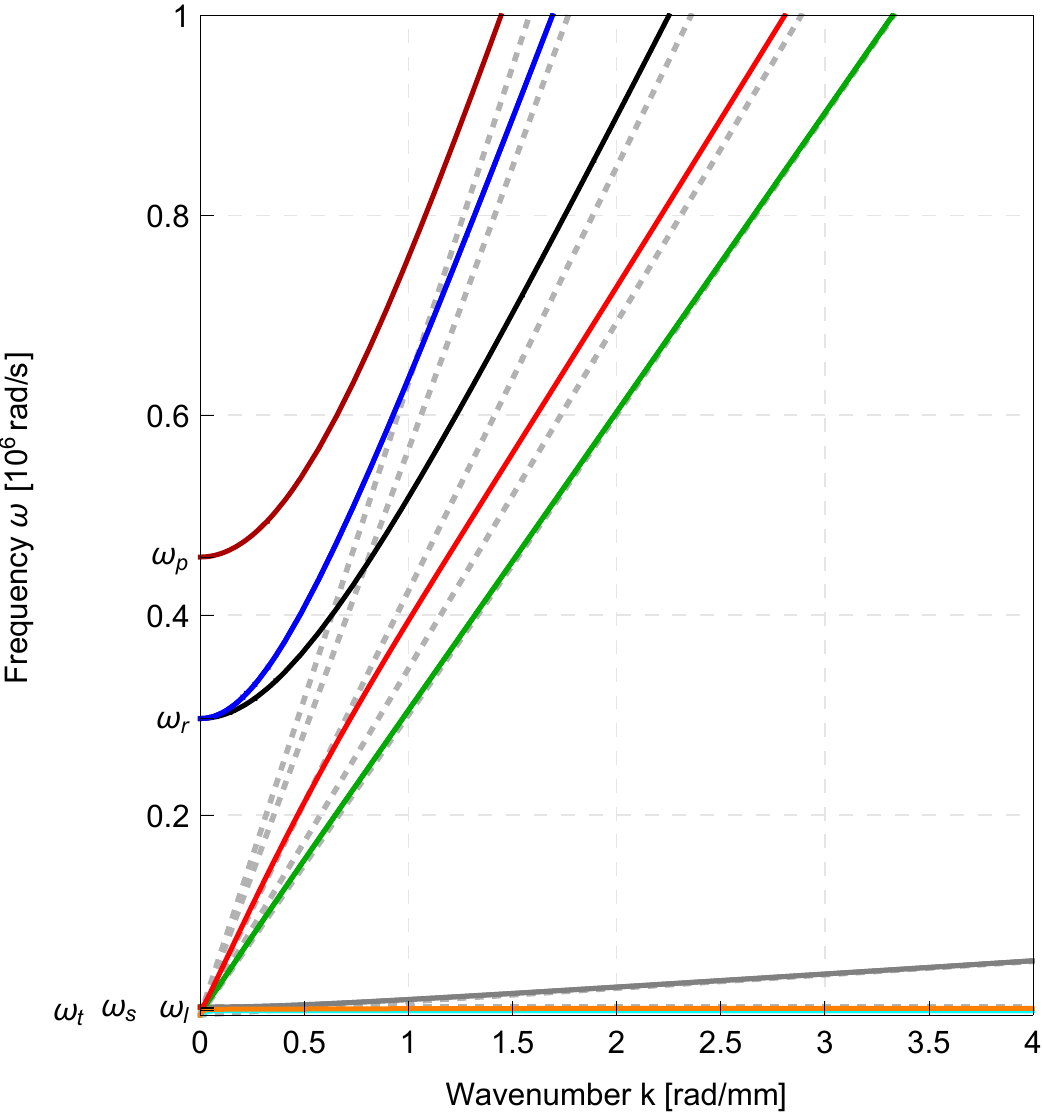} & \includegraphics[scale=0.5]{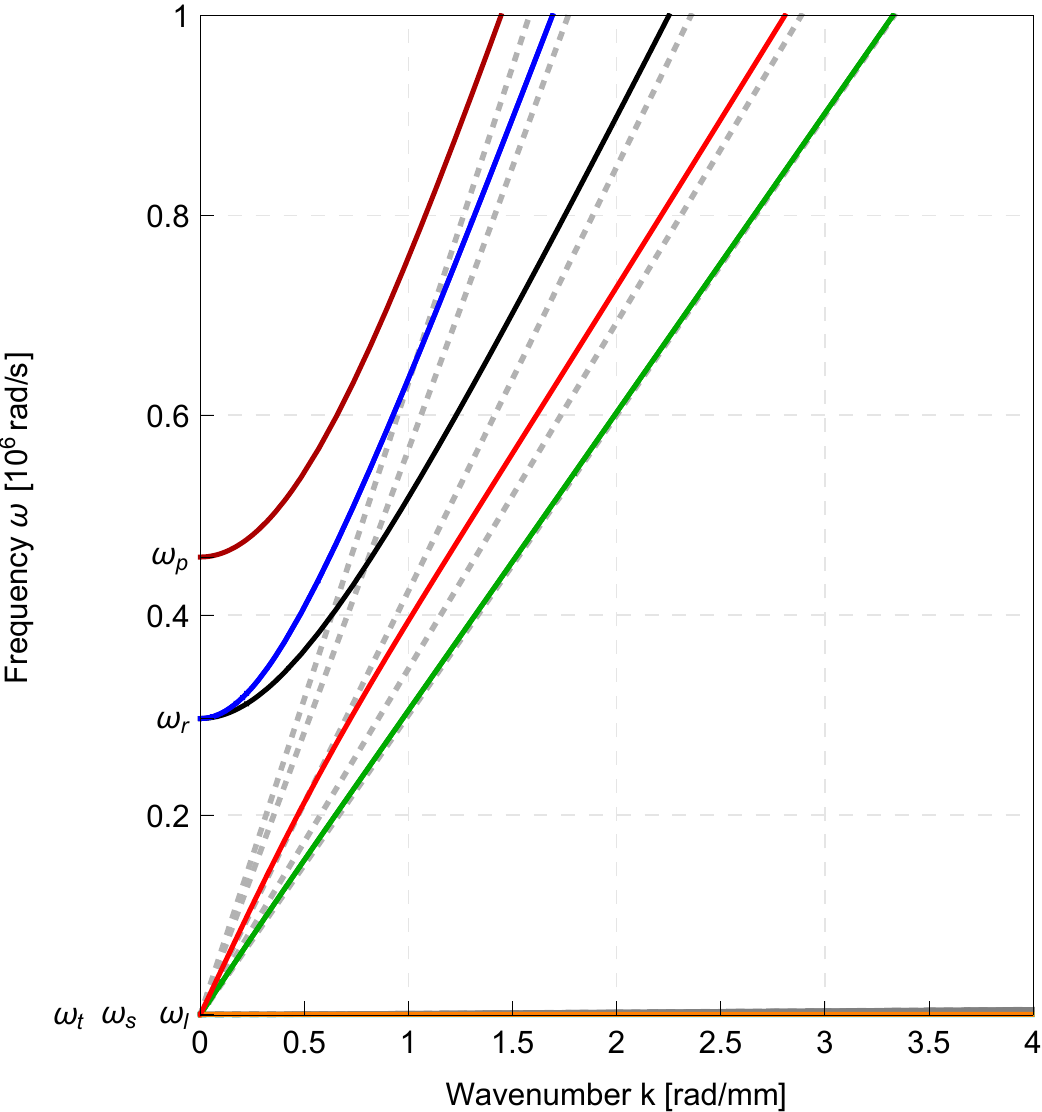} & \includegraphics[scale=0.5]{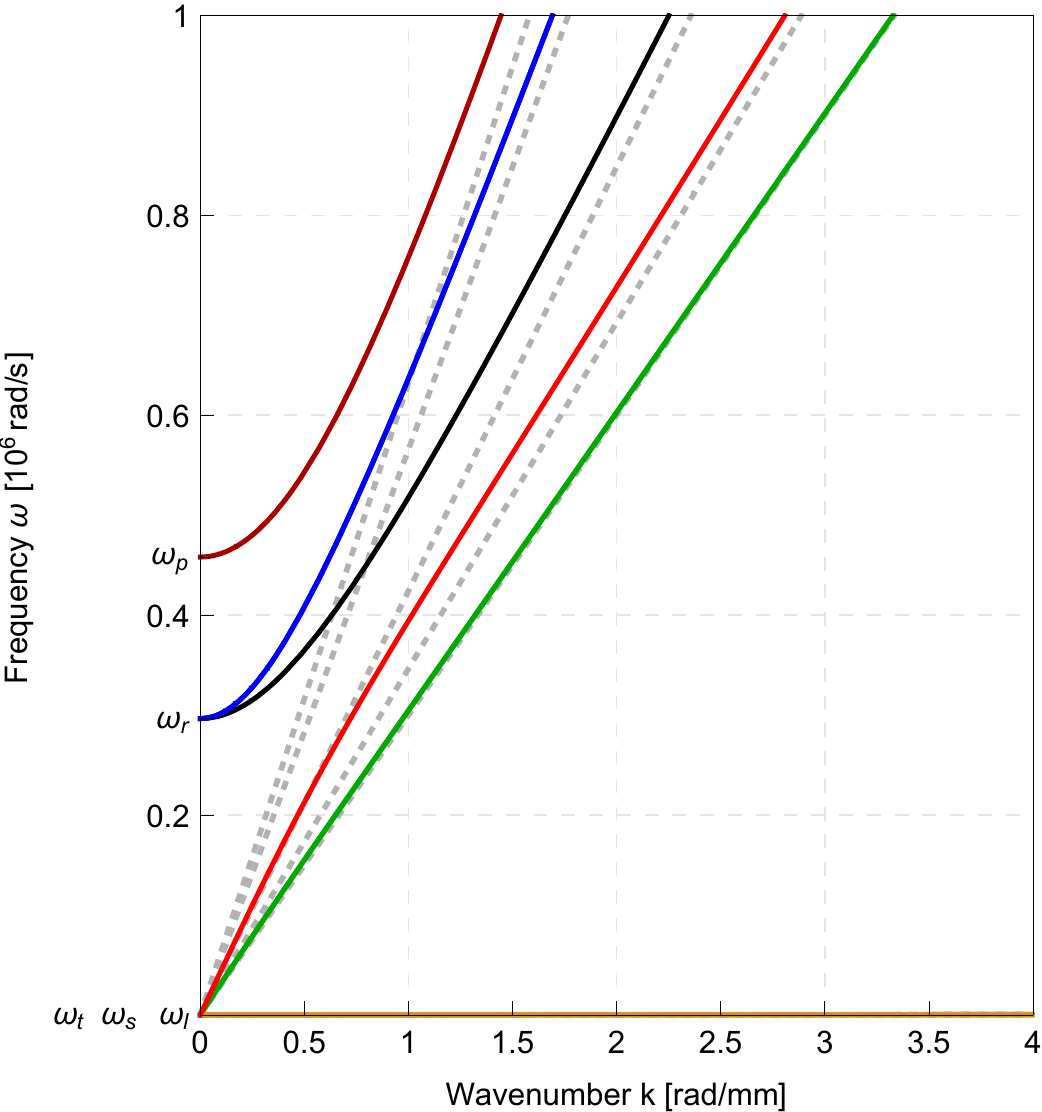}\tabularnewline
$\eta_{1}=10$ & $\eta_{1}=10^{3}$ & $\eta_{1}=10^{5}$\tabularnewline
\end{tabular}\caption{Effect of the parameter $\eta_{1}$ on the dispersion curves.}
\end{figure}
We complete here the case treated in subsection 3.4.1, by describing
the behavior of the dispersion curves when letting $\eta_{1}\fr\infty$. 

As expected, the optic branches originating from the cut-off $\omega_{s}$
become acoustic and, moreover, they are non-dispersive. What is more
surprising is that the original acoustic branches flatten to zero
and hence disappear from the dispersion diagram.

This means that, in the limit, we are constraining the system to have
less degrees of freedom by artificially imposing an infinite inertia
that does not allow some specific micro-vibrations.

\begin{figure}[H]
\centering{}%
\begin{tabular}{ccc}
\includegraphics[scale=0.5]{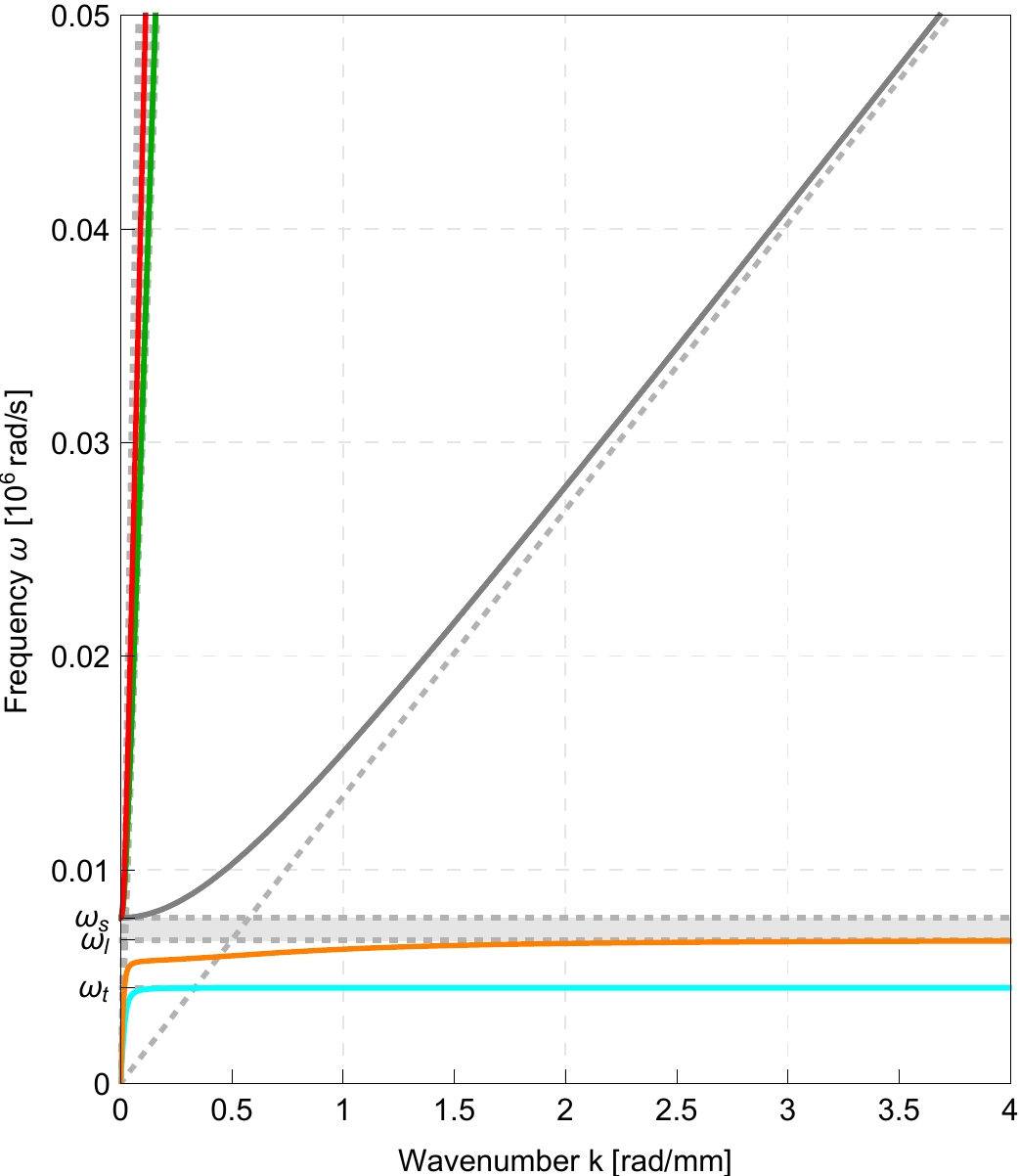} & \includegraphics[scale=0.5]{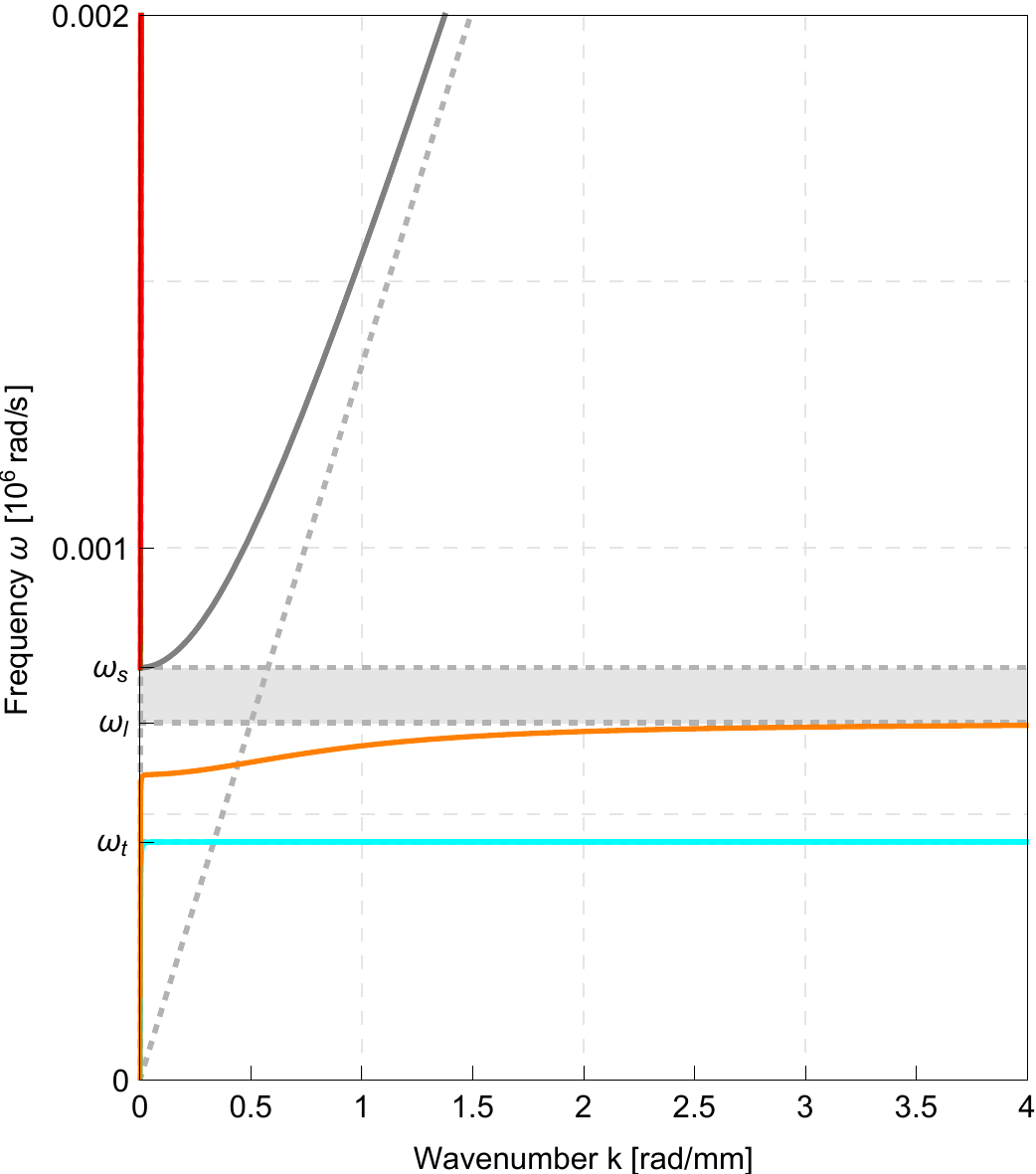} & \includegraphics[scale=0.5]{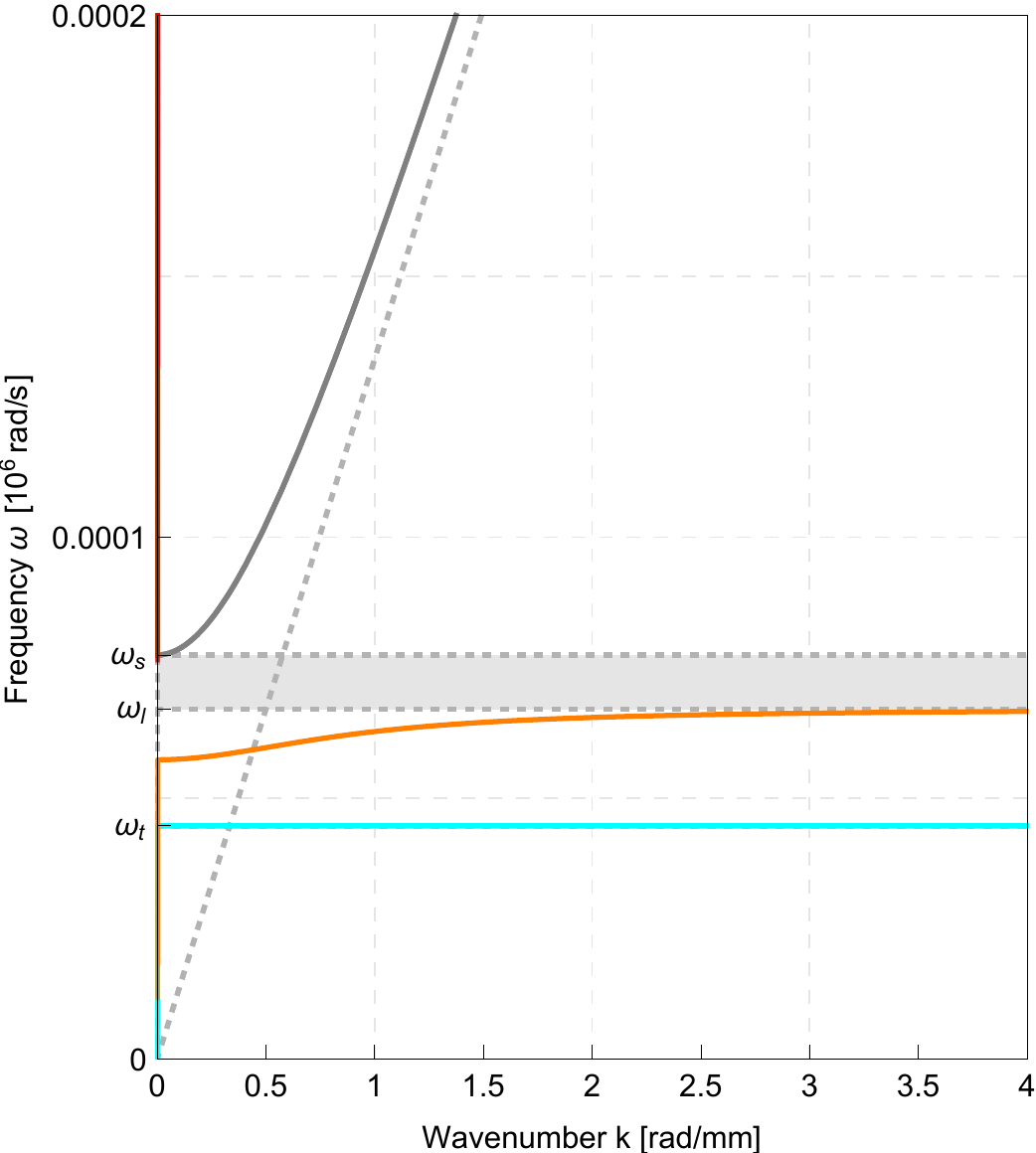}\tabularnewline
$\eta_{1}=10^{1}$ & $\eta_{1}=10^{3}$ & $\eta_{1}=10^{5}$\tabularnewline
\end{tabular}\caption{Zoom on the acoustic branches. \label{fig:Zoom1}}
\end{figure}

In Fig. \ref{fig:Zoom1} we make a zoom on the acoustic curves that
are flattening to zero.

\newpage{}

\subsubsection{Case $\eta_{2}\protect\fr+\infty$}

Characteristic limit elastic energy $\left\Vert \sym\left(\nabla u-P\right)\right\Vert ^{2}+\left\Vert \sym\,P\right\Vert ^{2}+\left\Vert \curl\,P\right\Vert ^{2}$.\\
Characteristic limit kinetic energy $\left\Vert u_{,t}\right\Vert ^{2}+\left\Vert \sym\,P_{,t}\right\Vert ^{2},\quad\skew\P_{,t}=0$.

\begin{figure}[H]
\centering{}%
\begin{tabular}{ccc}
\includegraphics[scale=0.5]{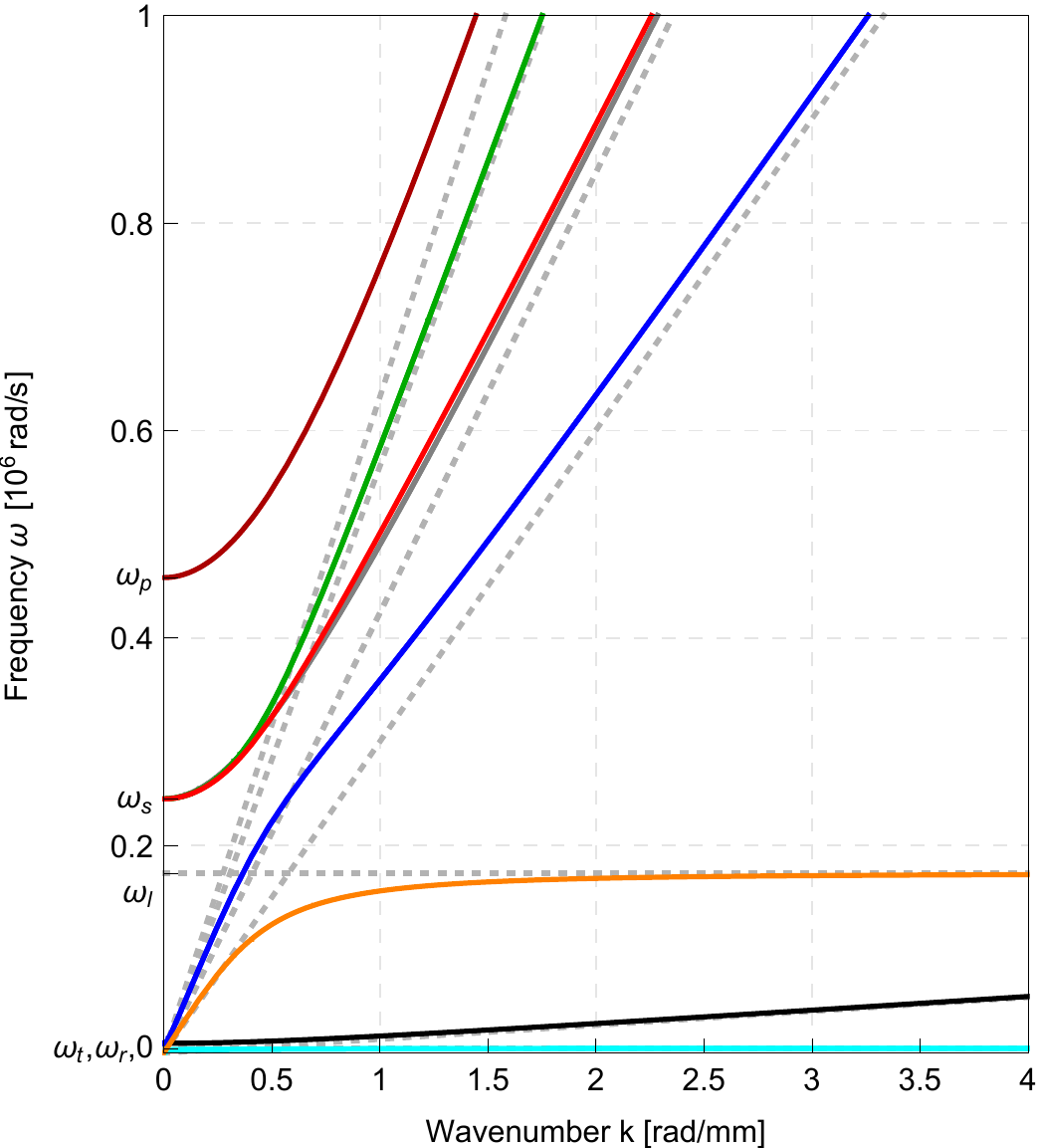} & \includegraphics[scale=0.5]{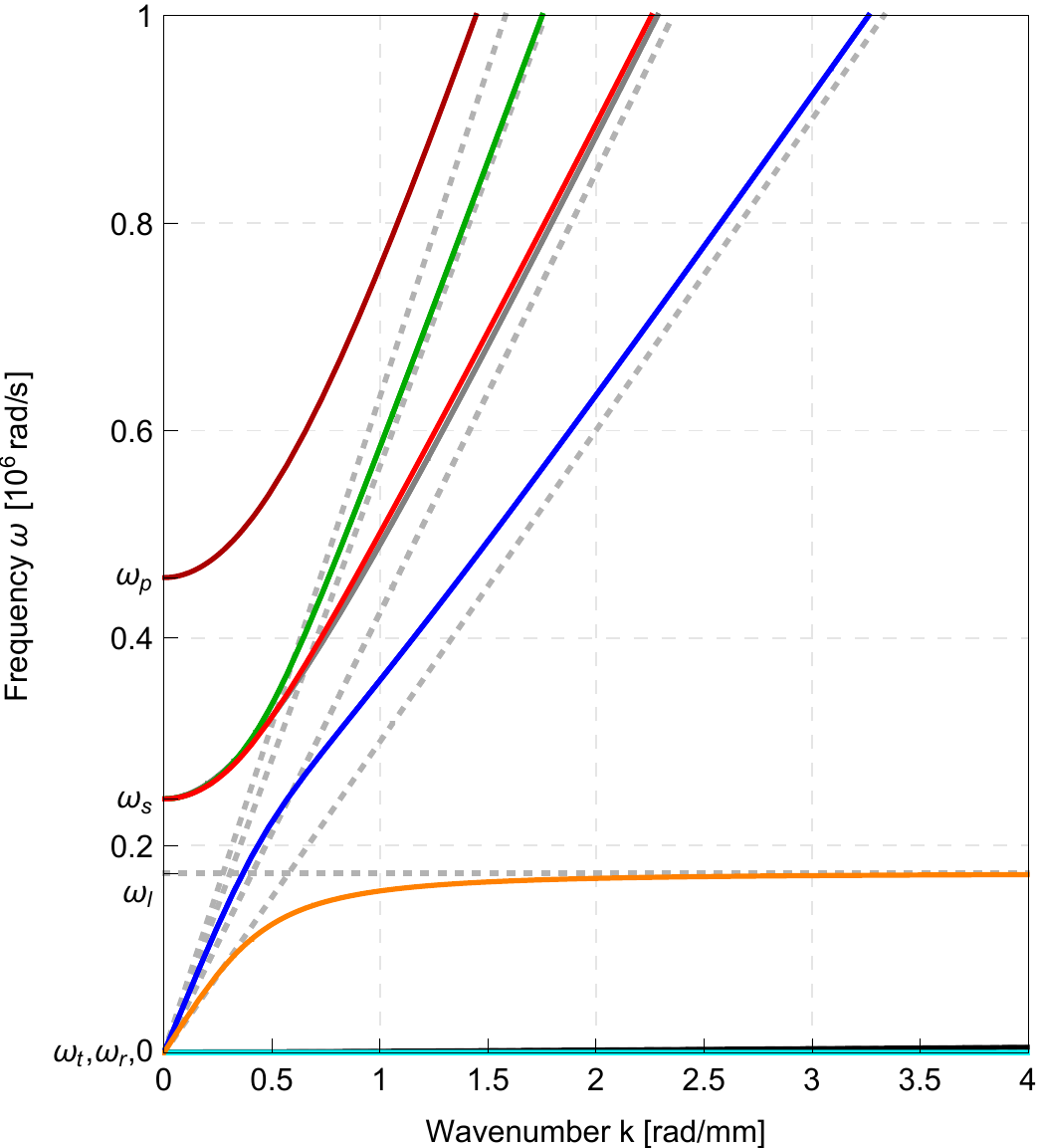} & \includegraphics[scale=0.5]{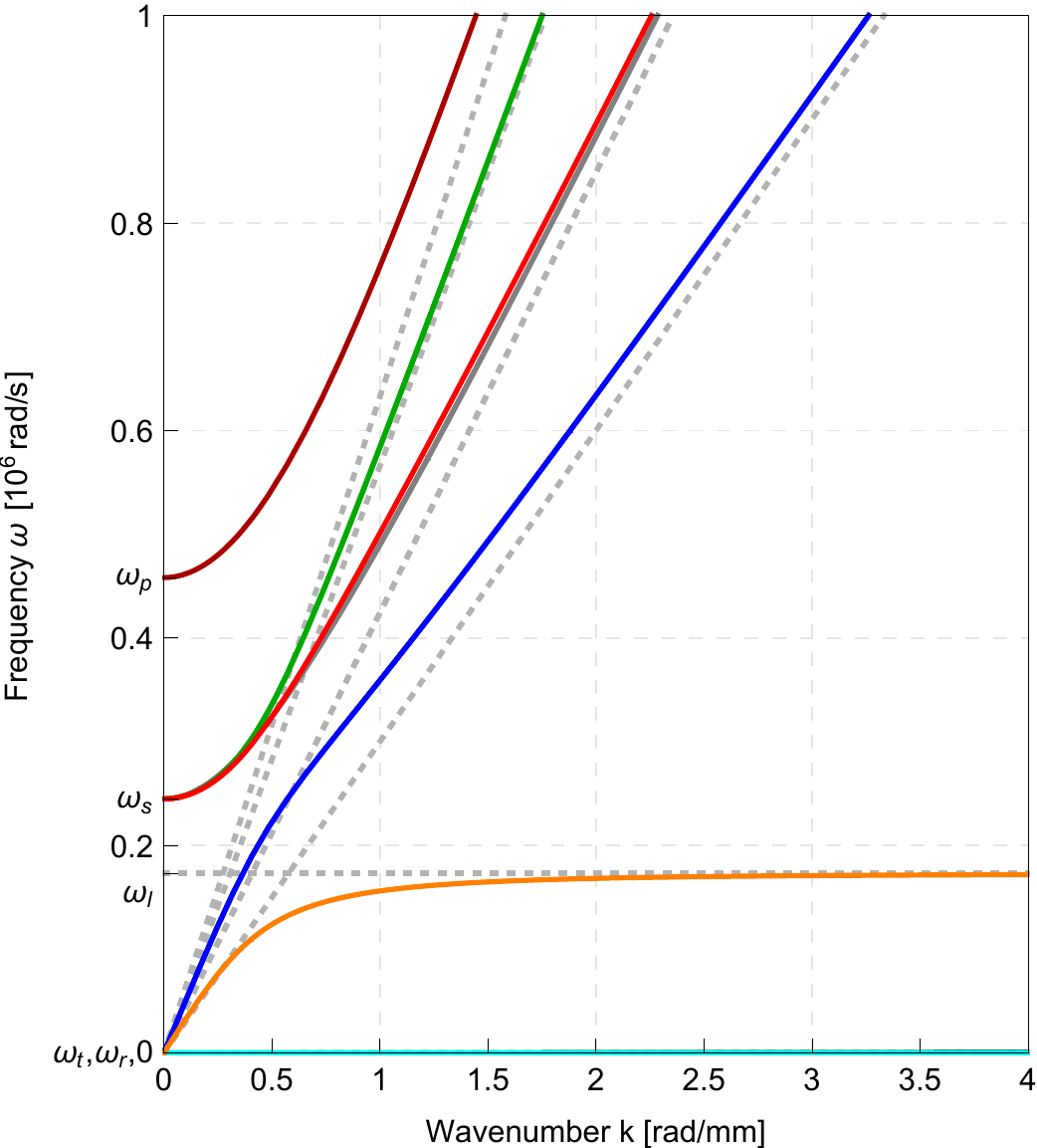}\tabularnewline
$\eta_{2}=10$ & $\eta_{2}=10^{3}$ & $\eta_{2}=10^{5}$\tabularnewline
\end{tabular}\caption{Effect of the parameter $\eta_{2}$ on the dispersion curves.}
\end{figure}
The same reasoning of subsection 5.4.5 can be repeated here for the
two optic curves originating from the cut-off $\omega_{r}$.

\begin{figure}[H]
\centering{}%
\begin{tabular}{ccc}
\includegraphics[scale=0.5]{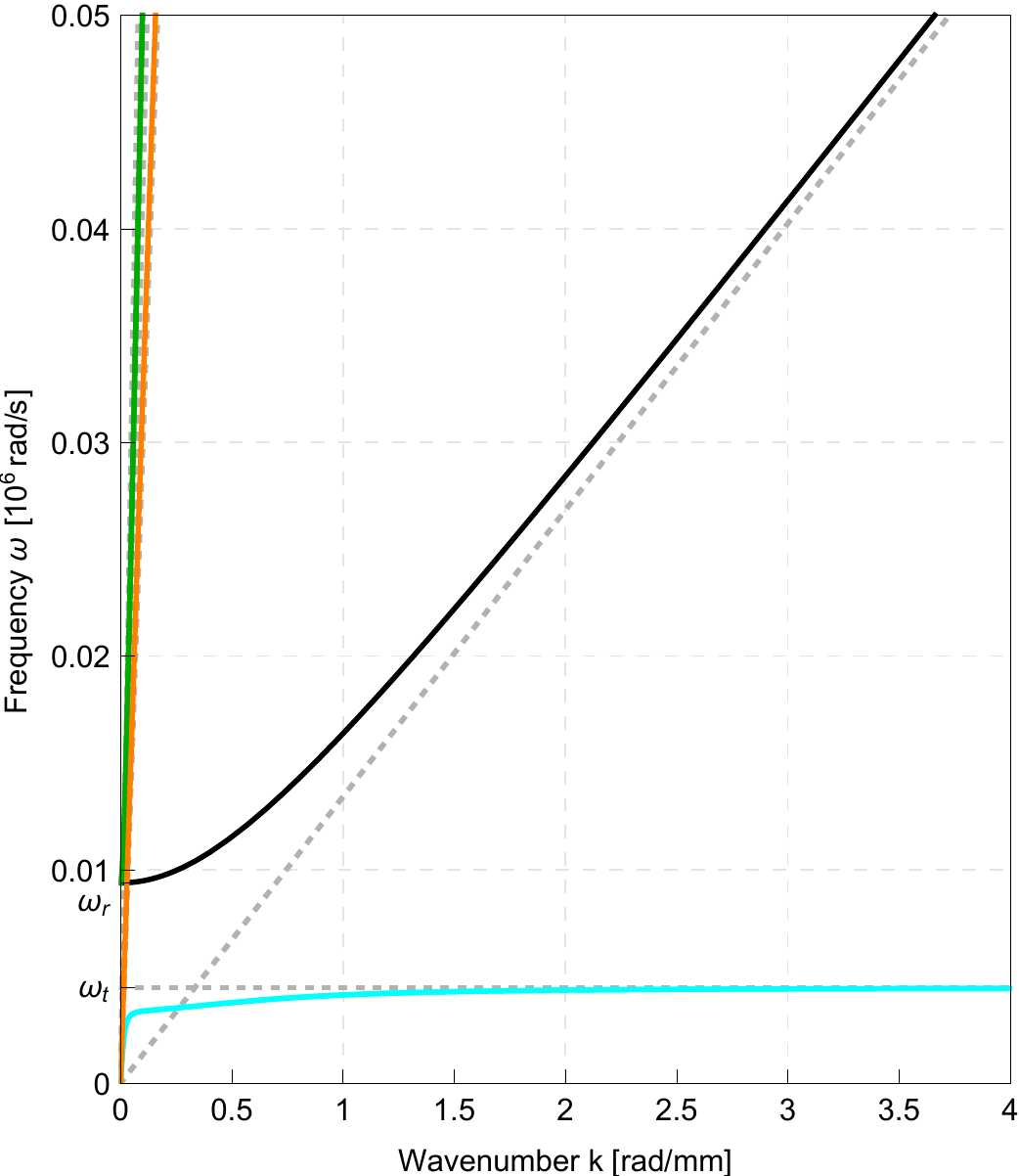} & \includegraphics[scale=0.5]{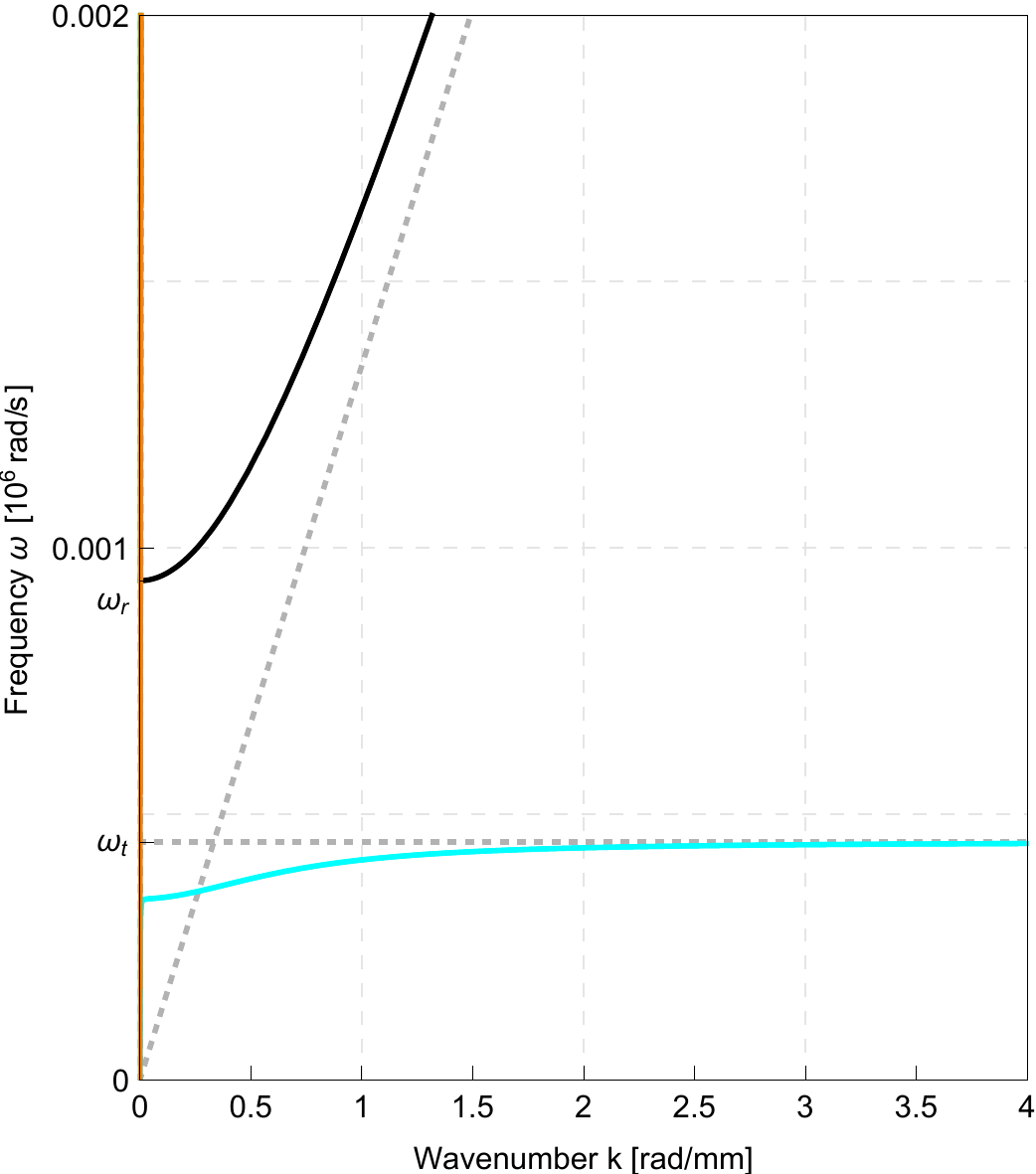} & \includegraphics[scale=0.5]{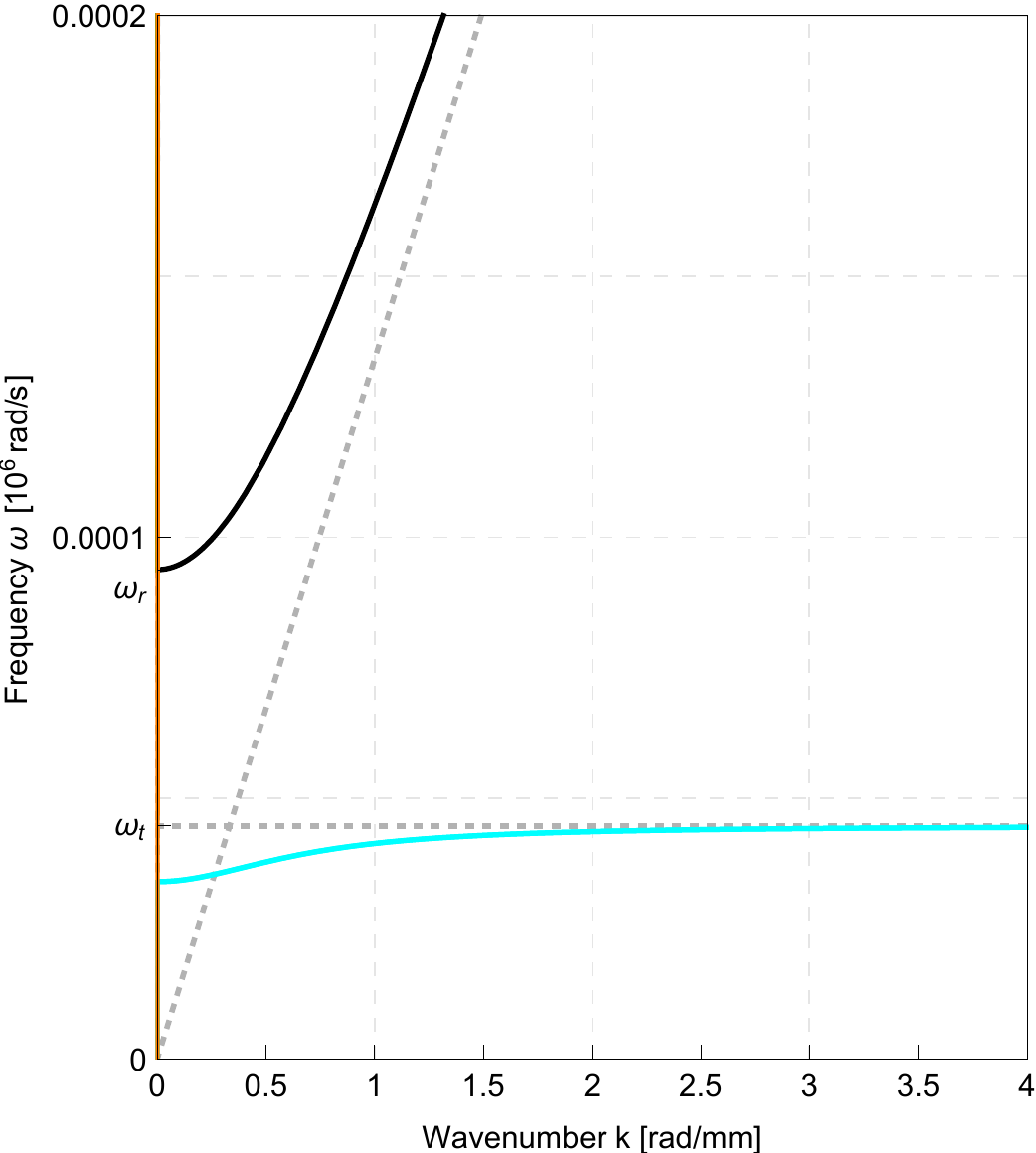}\tabularnewline
$\eta_{2}=10$ & $\eta_{2}=10^{3}$ & $\eta_{2}=10^{5}$\tabularnewline
\end{tabular}\caption{Zoom on the acoustic branches. \label{fig:Zoom2}}
\end{figure}
In Fig. \ref{fig:Zoom2} we show again the zoom on the dispersion
curves that are flattening to zero.

\newpage{}

\subsubsection{Case $\eta_{3}\protect\fr+\infty$}

Characteristic limit elastic energy $\left\Vert \sym\left(\nabla u-P\right)\right\Vert ^{2}+\left\Vert \sym\,P\right\Vert ^{2}+\left\Vert \curl\,P\right\Vert ^{2}$.\\
Characteristic limit kinetic energy $\left\Vert u_{,t}\right\Vert ^{2}+\left\Vert \dev\,\sym\,P_{,t}\right\Vert ^{2}+\left\Vert \skew\,P_{,t}\right\Vert ^{2},\quad\textrm{tr}\,P_{,t}=0$.

\begin{figure}[H]
\centering{}%
\begin{tabular}{ccc}
\includegraphics[scale=0.5]{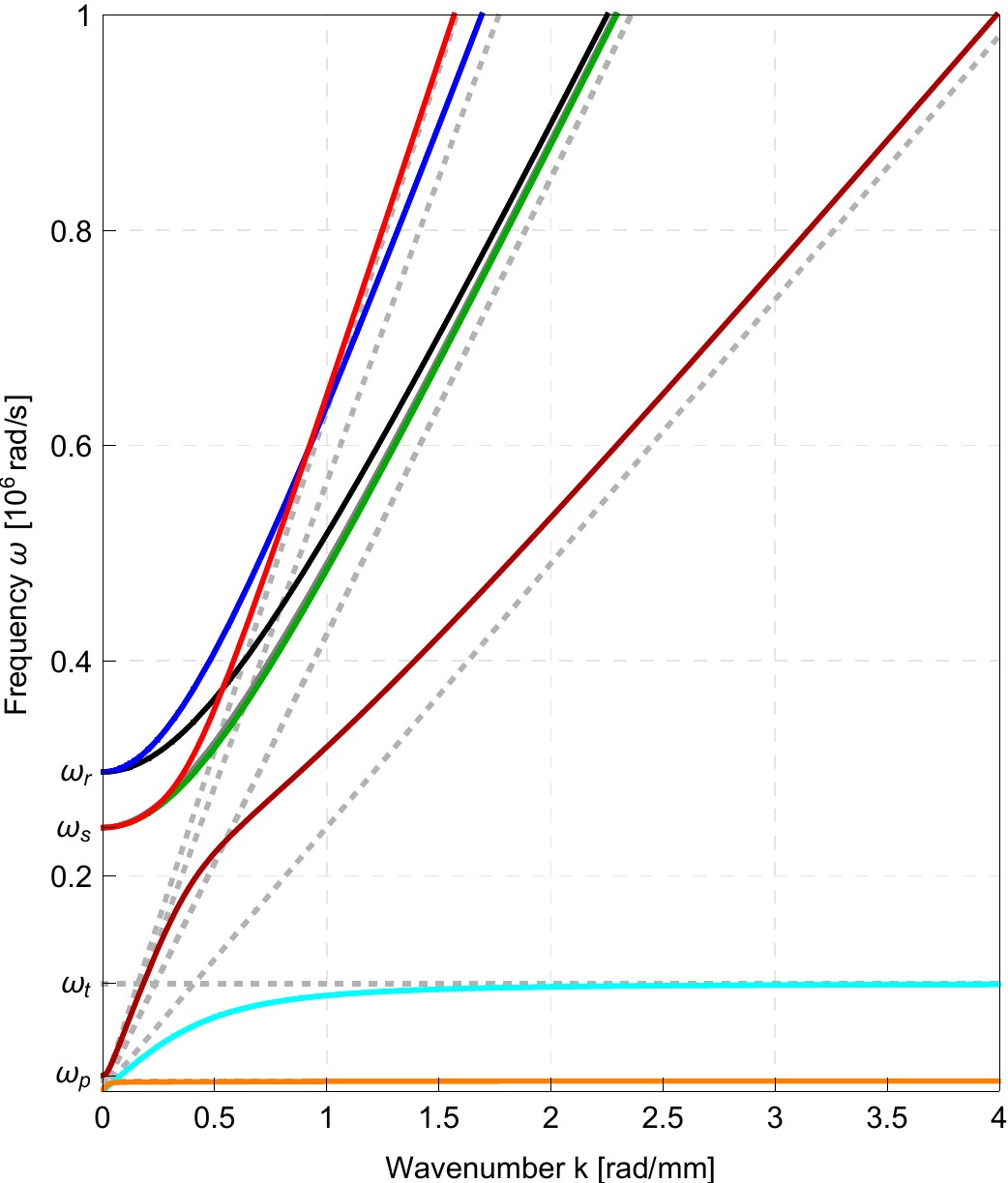} & \includegraphics[scale=0.5]{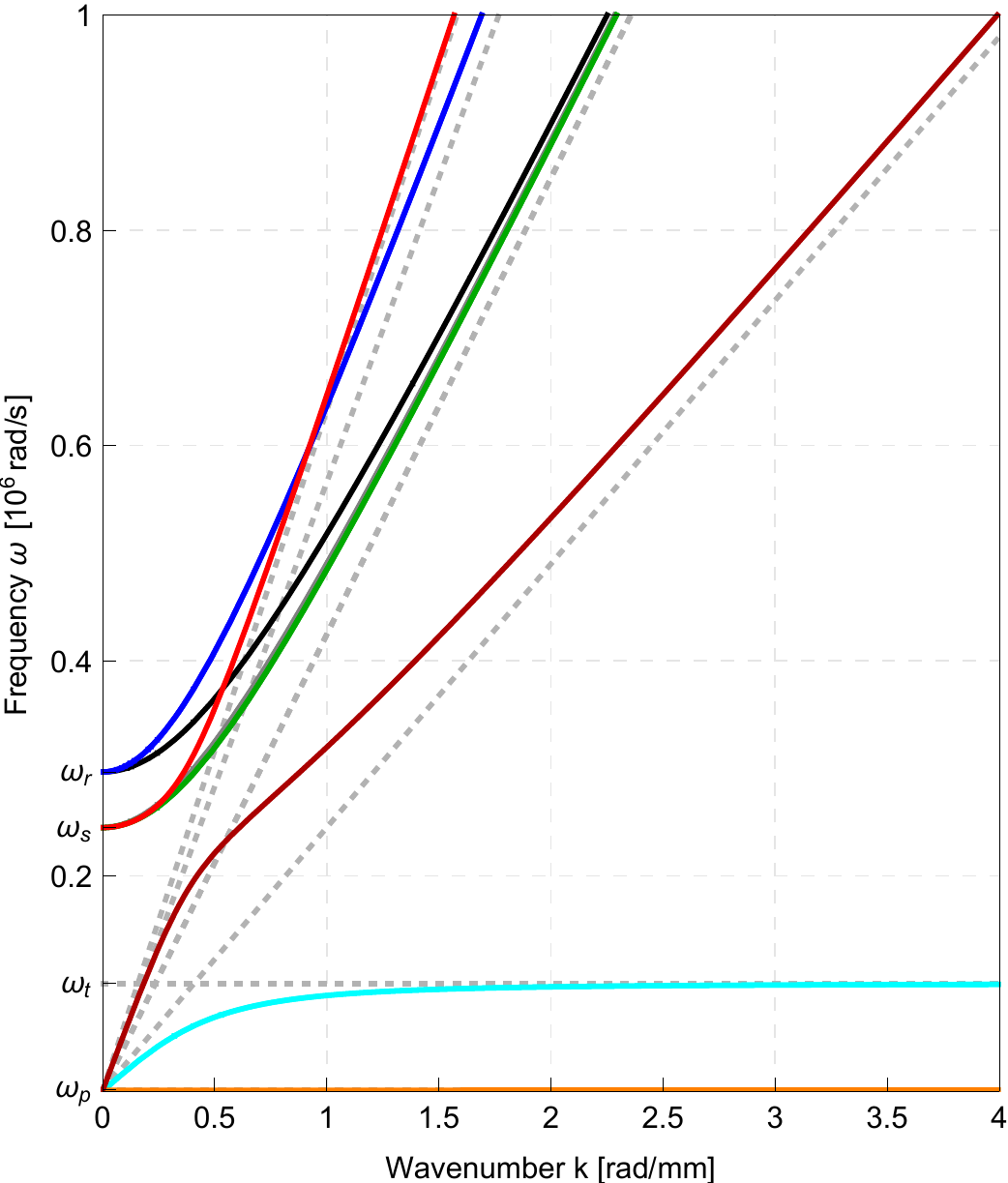} & \includegraphics[scale=0.5]{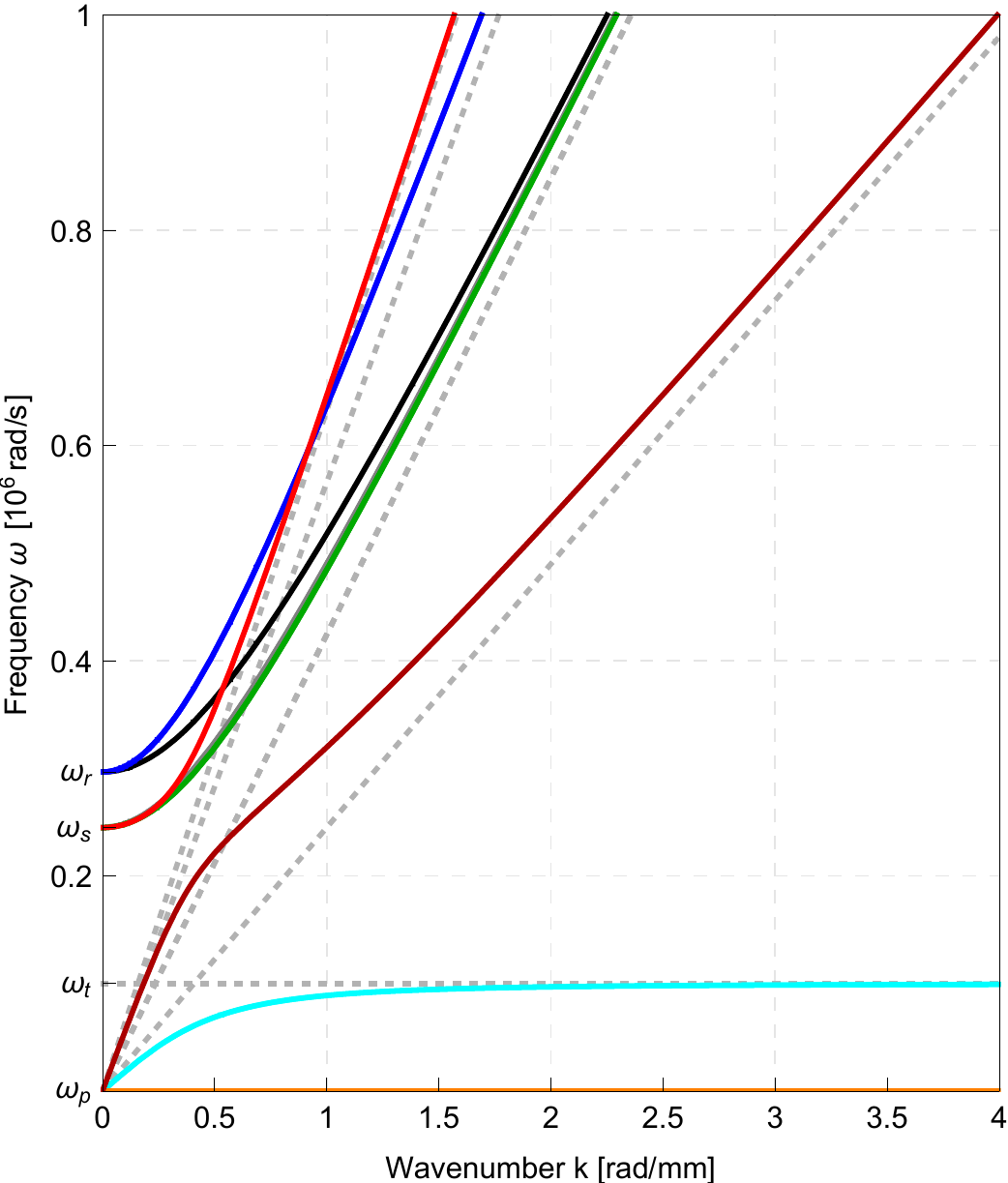}\tabularnewline
$\eta_{3}=10$ & $\eta_{3}=10^{3}$ & $\eta_{3}=10^{5}$\tabularnewline
\end{tabular}\caption{Effect of the parameter $\eta_{3}$ on the dispersion curves.}
\end{figure}
The same reasoning of subsection 5.4.5 can be repeated here for the
two optic curves originating from the cut-off $\omega_{p}$.
\begin{figure}[H]
\centering{}%
\begin{tabular}{ccc}
\includegraphics[scale=0.5]{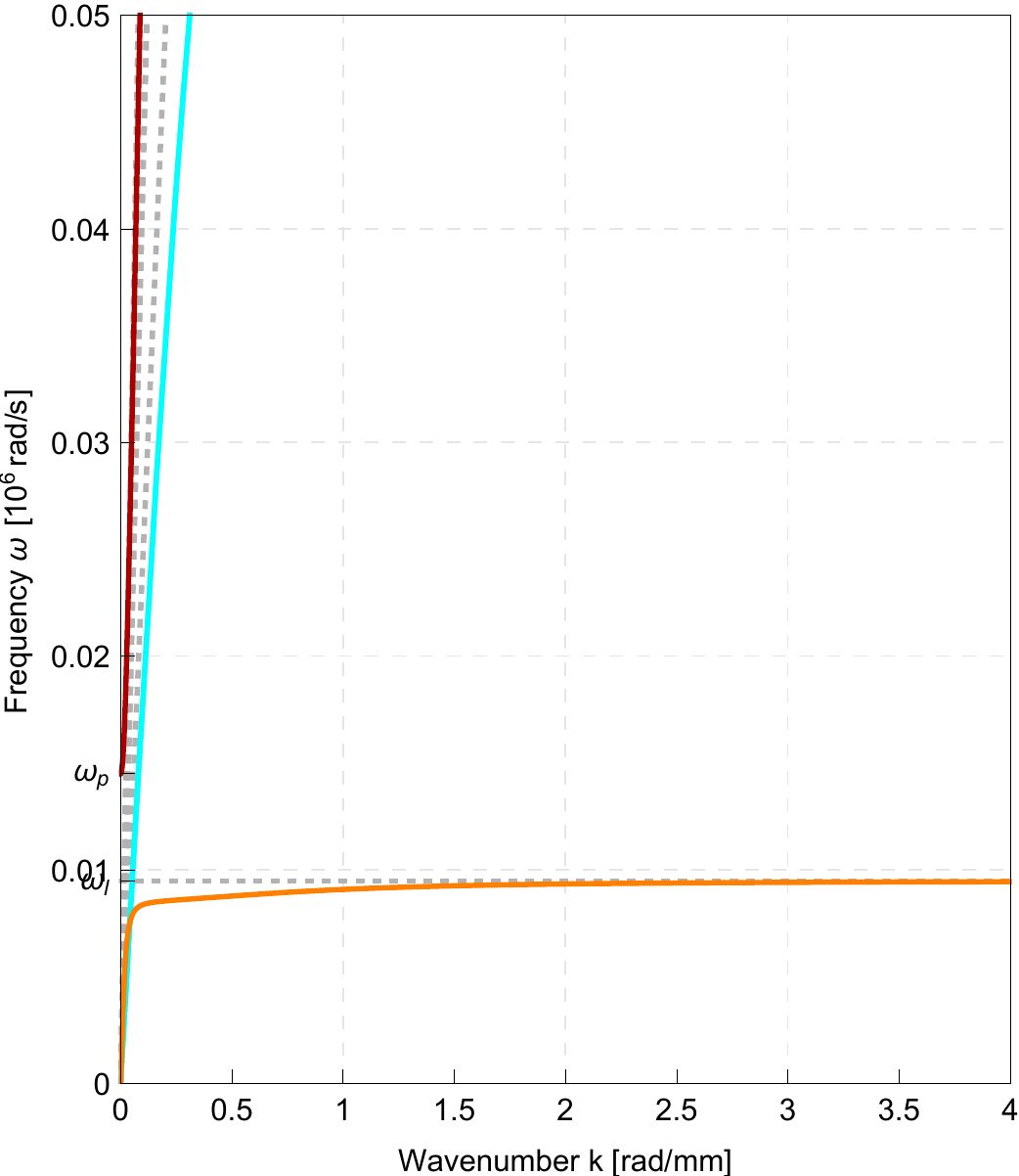} & \includegraphics[scale=0.5]{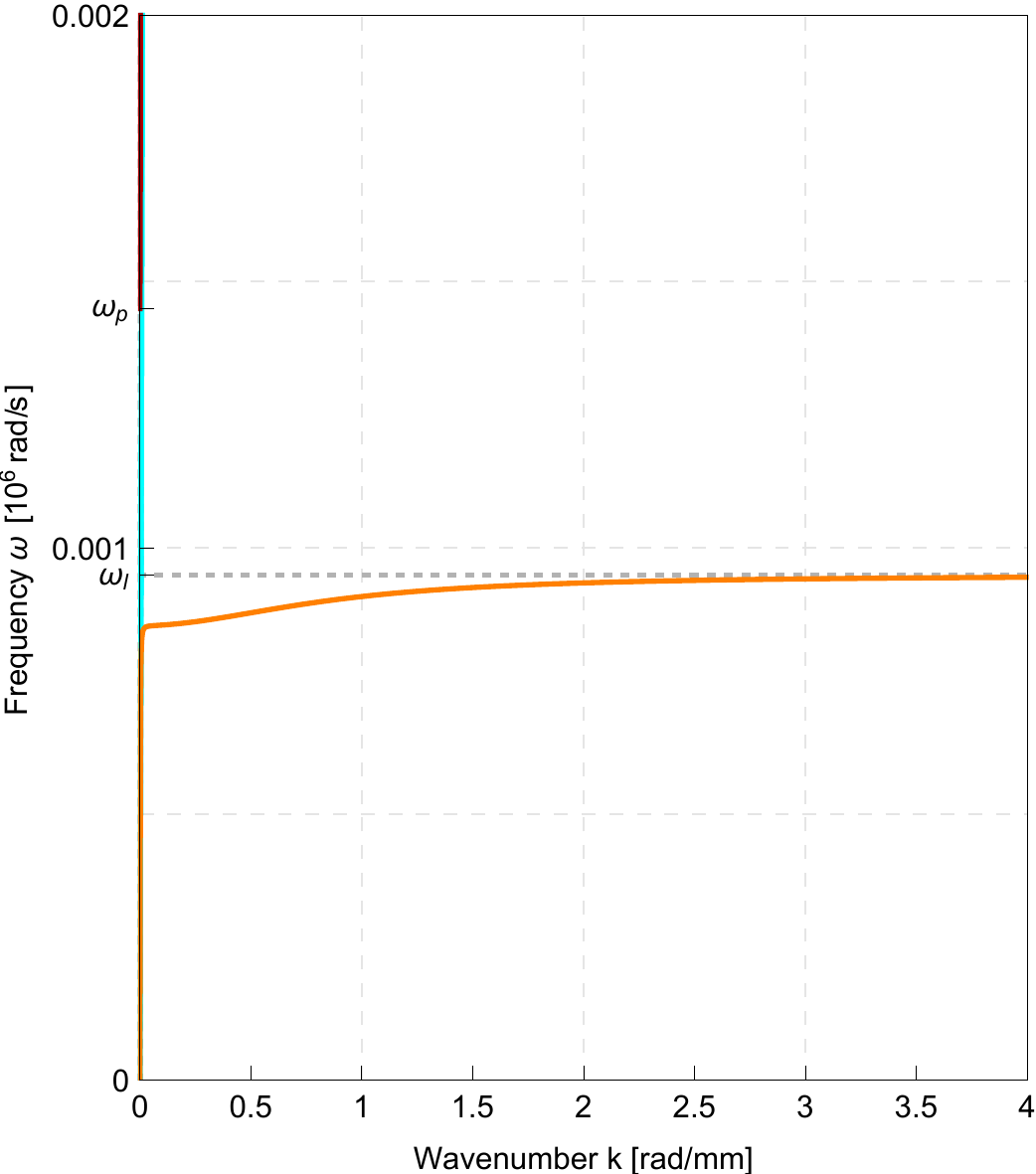} & \includegraphics[scale=0.5]{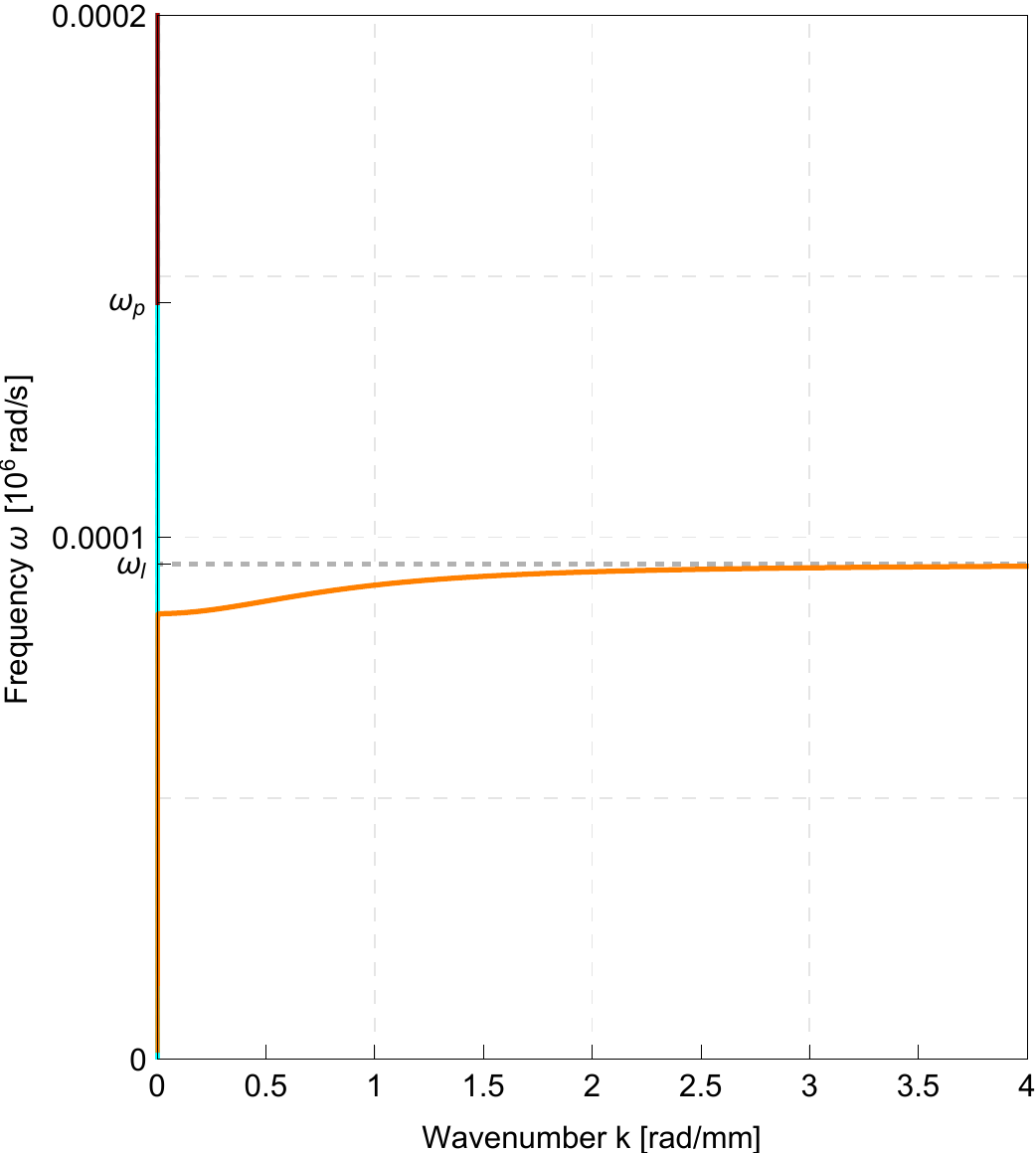}\tabularnewline
$\eta_{3}=10$ & $\eta_{3}=10^{3}$ & $\eta_{3}=10^{5}$\tabularnewline
\end{tabular}\caption{Zoom on the acoustic branches. \label{fig:Zoom3}}
\end{figure}
 Fig. \ref{fig:Zoom3} shows the behavior of the dispersion curves
flattening to zero.

\newpage{}

\subsubsection{Cases ${\displaystyle \eta_{1,}\eta_{2},\eta_{3}\protect\fr+\infty}$:
a rigidified Cauchy material}

Characteristic limit elastic energy $\left\Vert \sym\left(\nabla u-P\right)\right\Vert ^{2}+\left\Vert \sym\,P\right\Vert ^{2}+\left\Vert \dev\,\curl\,P\right\Vert ^{2}$.\\
Characteristic limit kinetic energy $\left\Vert u_{,t}\right\Vert ^{2},\quad P_{,t}=0$.

\begin{figure}[H]
\centering{}%
\begin{tabular}{ccc}
\includegraphics[scale=0.5]{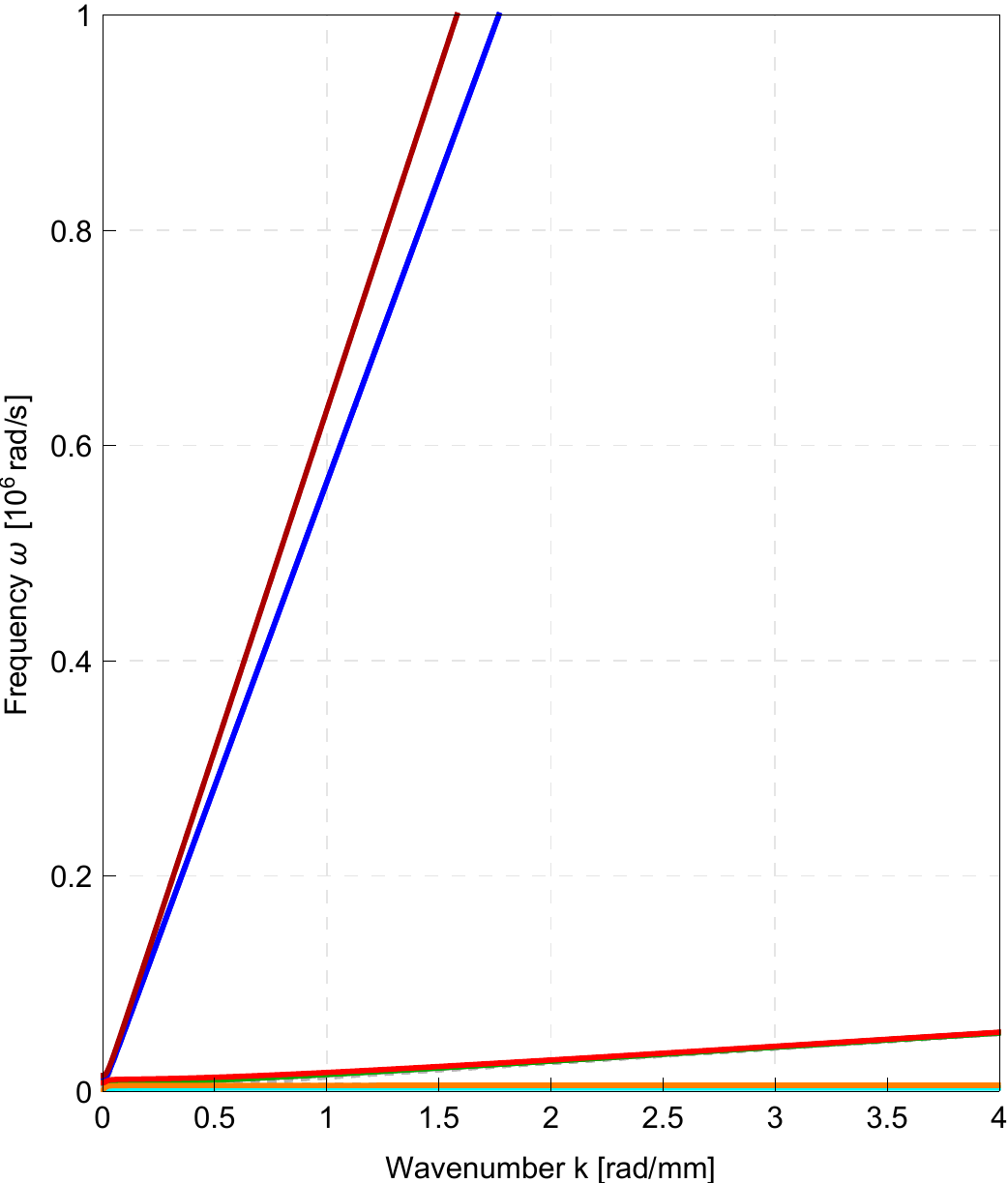} & \includegraphics[scale=0.5]{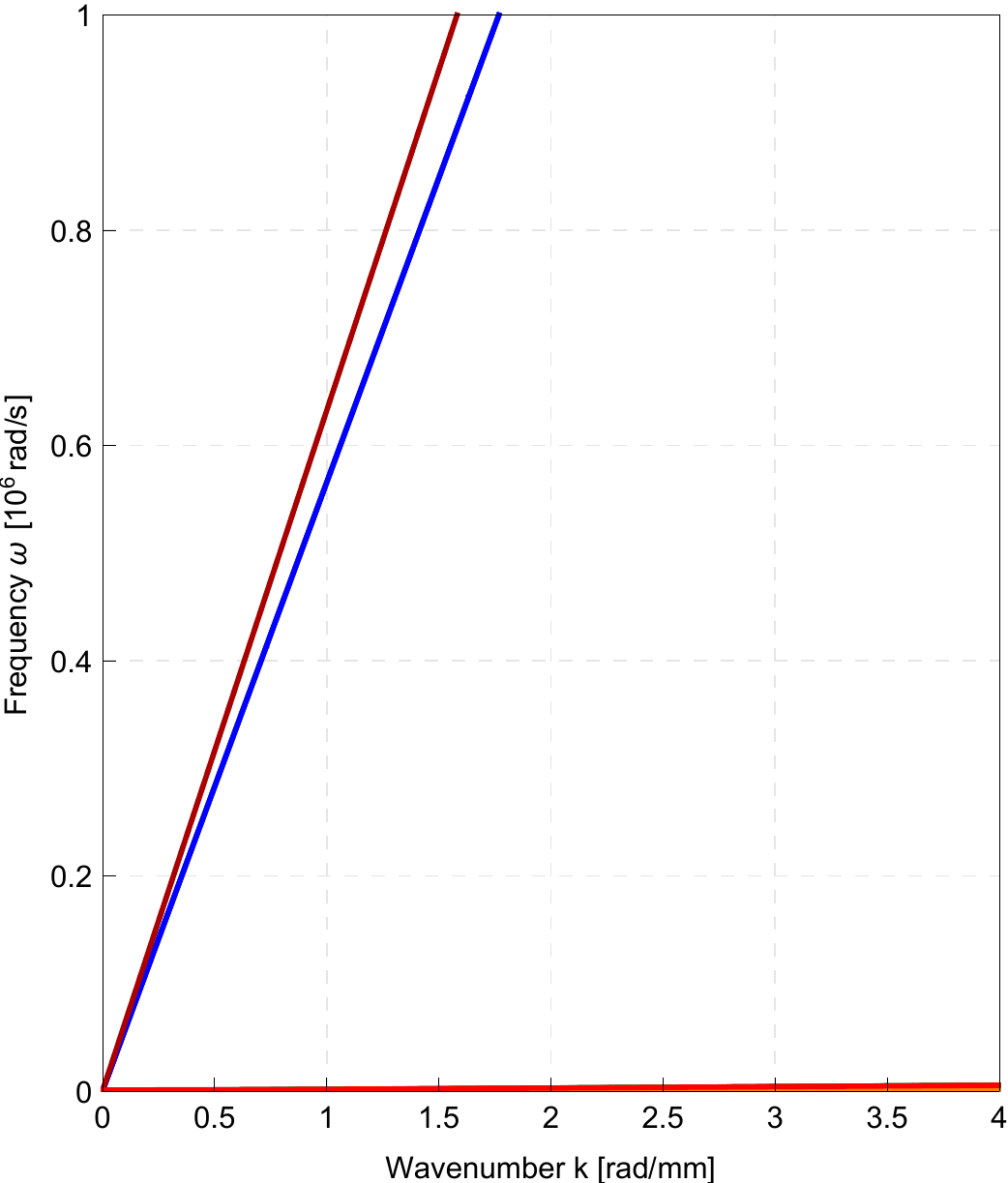} & \includegraphics[scale=0.5]{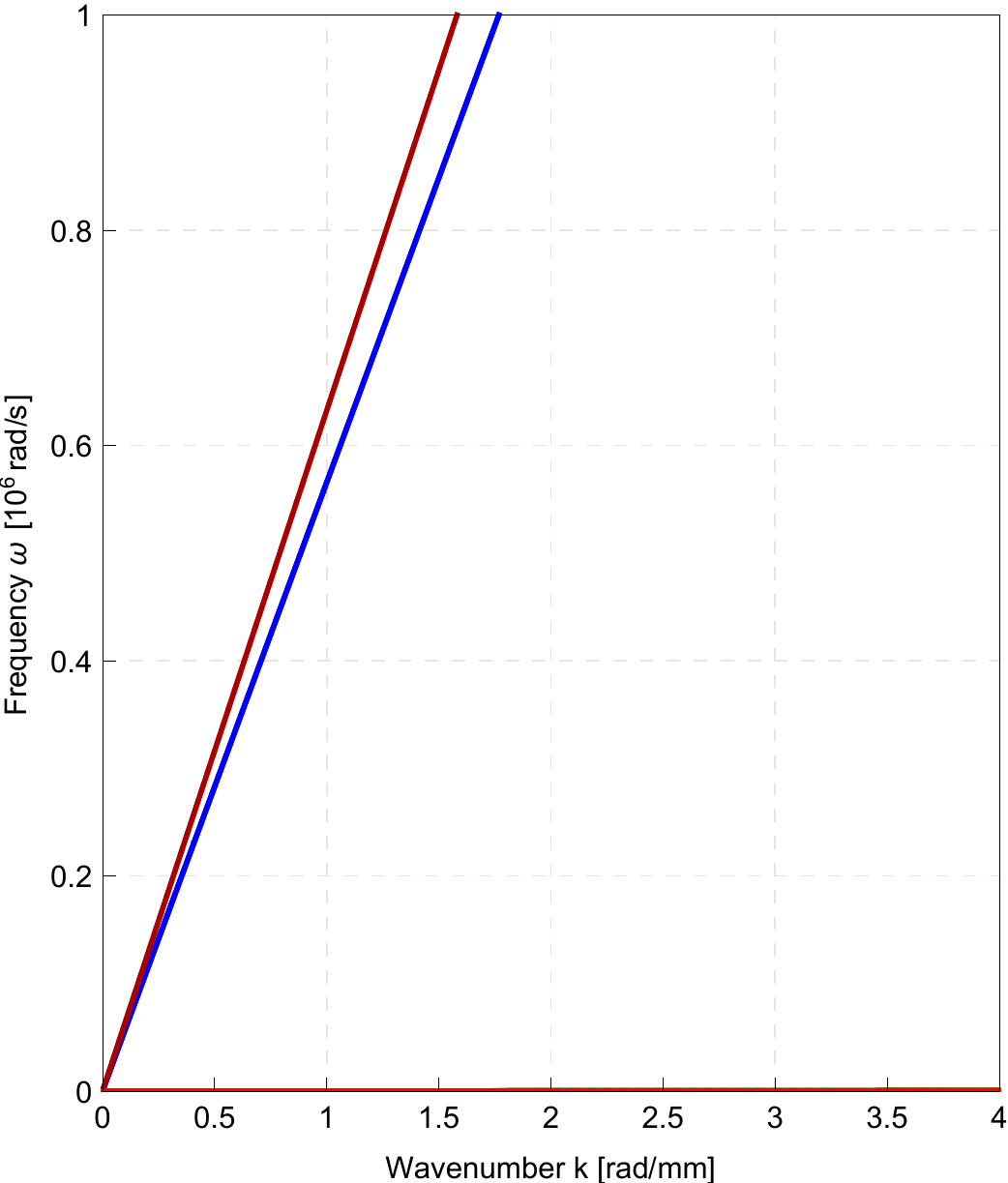}\tabularnewline
$\eta_{1}=\eta_{2}=\eta_{3}=10$ & $\eta_{1}=\eta_{2}=\eta_{3}=10^{3}$ & $\eta_{1}=\eta_{2}=\eta_{3}=10^{5}$\tabularnewline
\end{tabular}\caption{Combined effect of the parameters $\eta_{1},\eta_{2}\eta_{3}\protect\fr0$
on the dispersion curves.}
\end{figure}

This particular case, obtained letting simultaneously $\eta_{1},\eta_{2}$
and $\eta_{3}$ tend to infinity, gives rise to a Cauchy-like material
behavior. Nevertheless, the physical meaning attached to this phenomenon
is drastically different from the result obtained in section 5.4.4
when setting $\eta_{1}=\eta_{2}=\eta_{3}=0.$

Indeed, in that case, considering an enriched kinematics $\left(u,\P\right)$
without the micro-inertia $\left\Vert \P_{,t}\right\Vert ^{2}$ did
not allow to such microstructure to manifest itself. It is as if one
introduces a complex constitutive behavior for a metamaterial, but
does not allow to investigate its dynamical behavior. Indeed, the
result was the same obtained for the classical Cauchy medium as if
it did not have any underlying microstructure. 

On the other hand, the case considered here is quite different: we
are indeed introducing the inertia of the microstructure in the model,
but such inertia is so high that the microstructure is ``frozen''
and cannot vibrate locally. We thus end-up with a Cauchy material
which is more rigid than the original one (slope of the acoustic curves
is bigger than that in Fig. 1(a)). 

\newpage{}

\subsection*{Case ${\displaystyle L_{c}\protect\fr0\quad\simeq\quad\alpha_{1},\alpha_{2},\alpha_{3}\protect\fr0}$
(internal variable model)}

Characteristic limit elastic energy $\left\Vert \nabla u-P\right\Vert ^{2}+\left\Vert \sym\,P\right\Vert ^{2}$.\\
Characteristic limit kinetic energy $\left\Vert u_{,t}\right\Vert ^{2}+\left\Vert P_{,t}\right\Vert ^{2}$.

\begin{figure}[H]
\centering{}%
\begin{tabular}{ccc}
\includegraphics[scale=0.5]{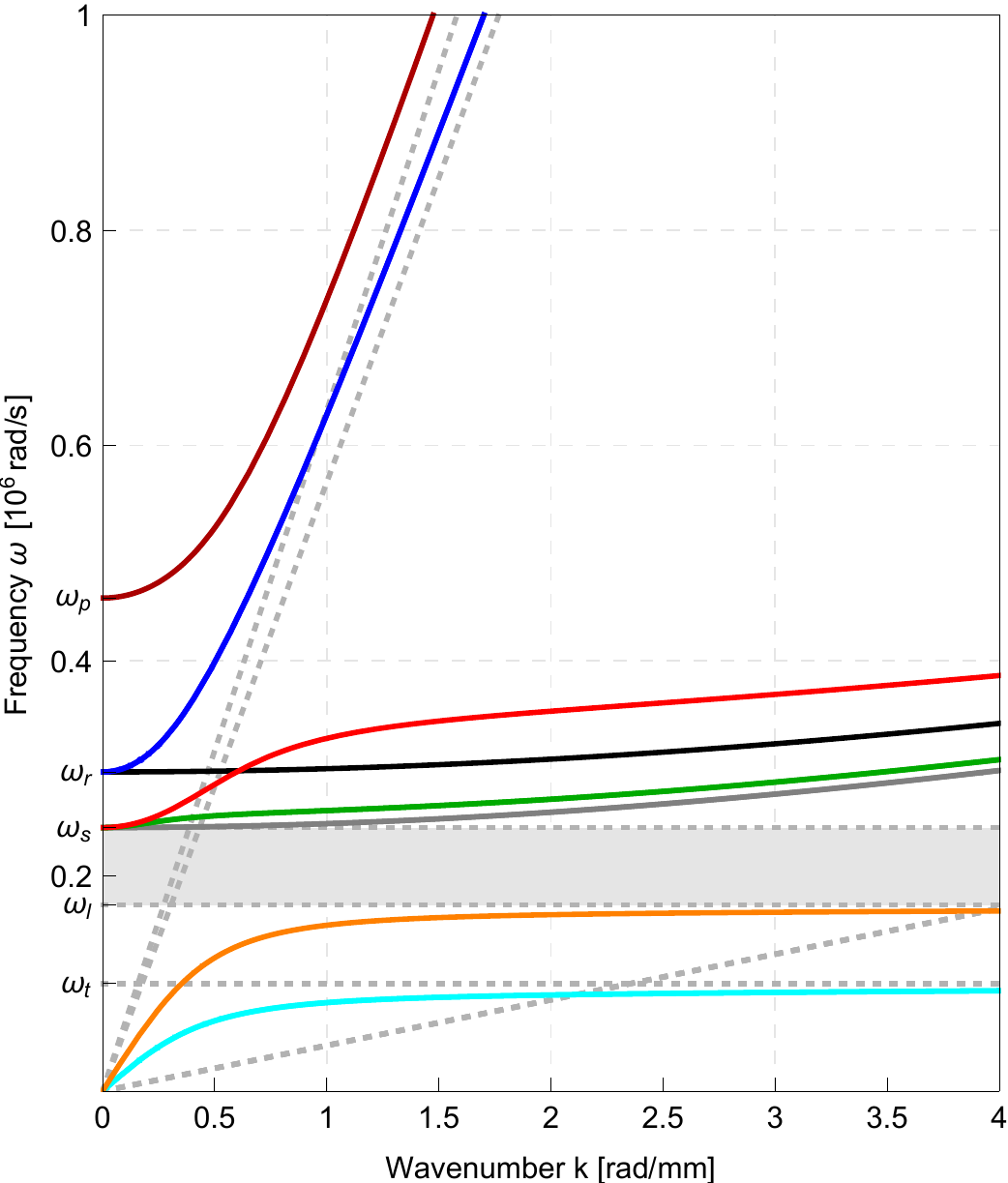} & \includegraphics[scale=0.5]{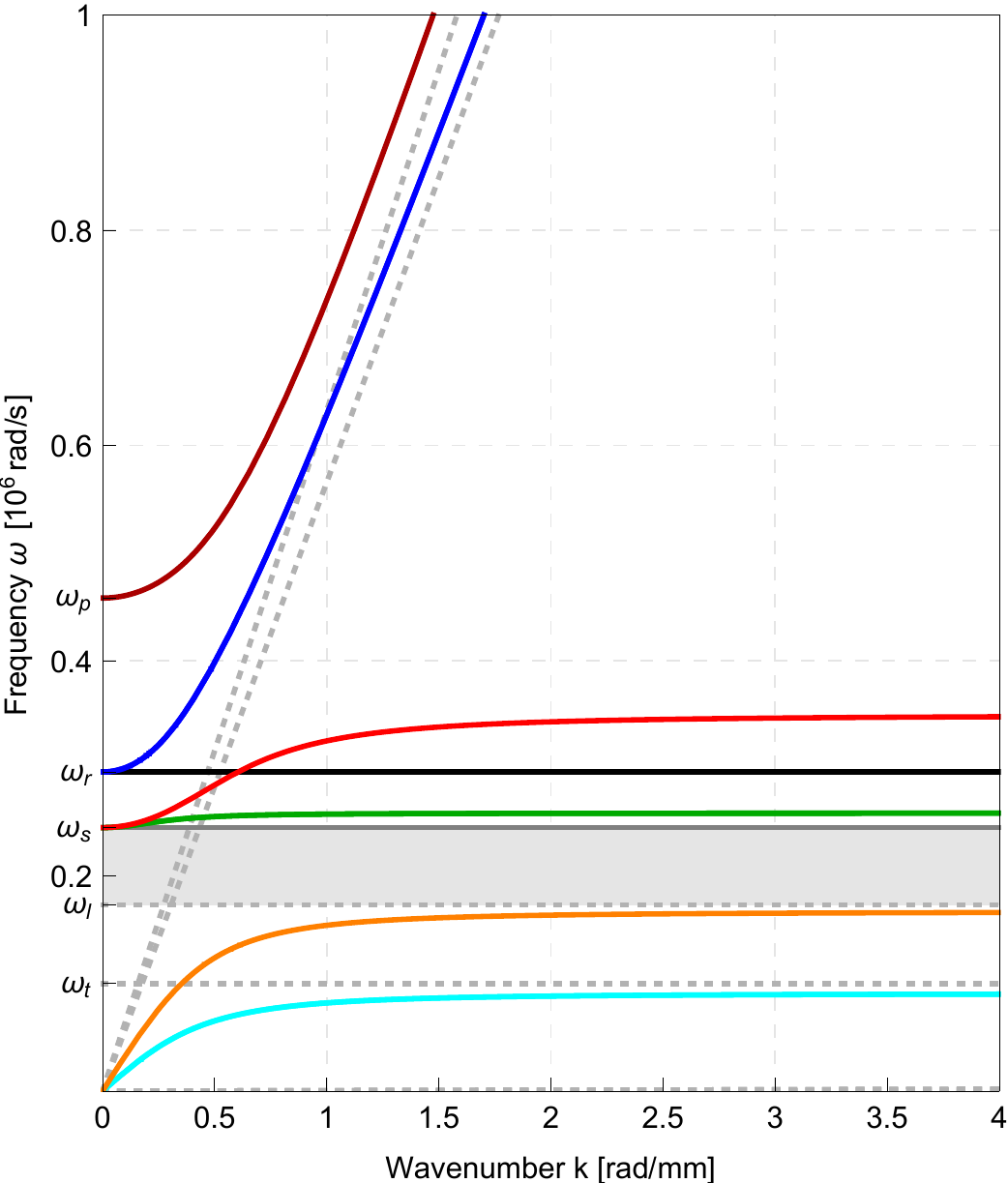} & \includegraphics[scale=0.5]{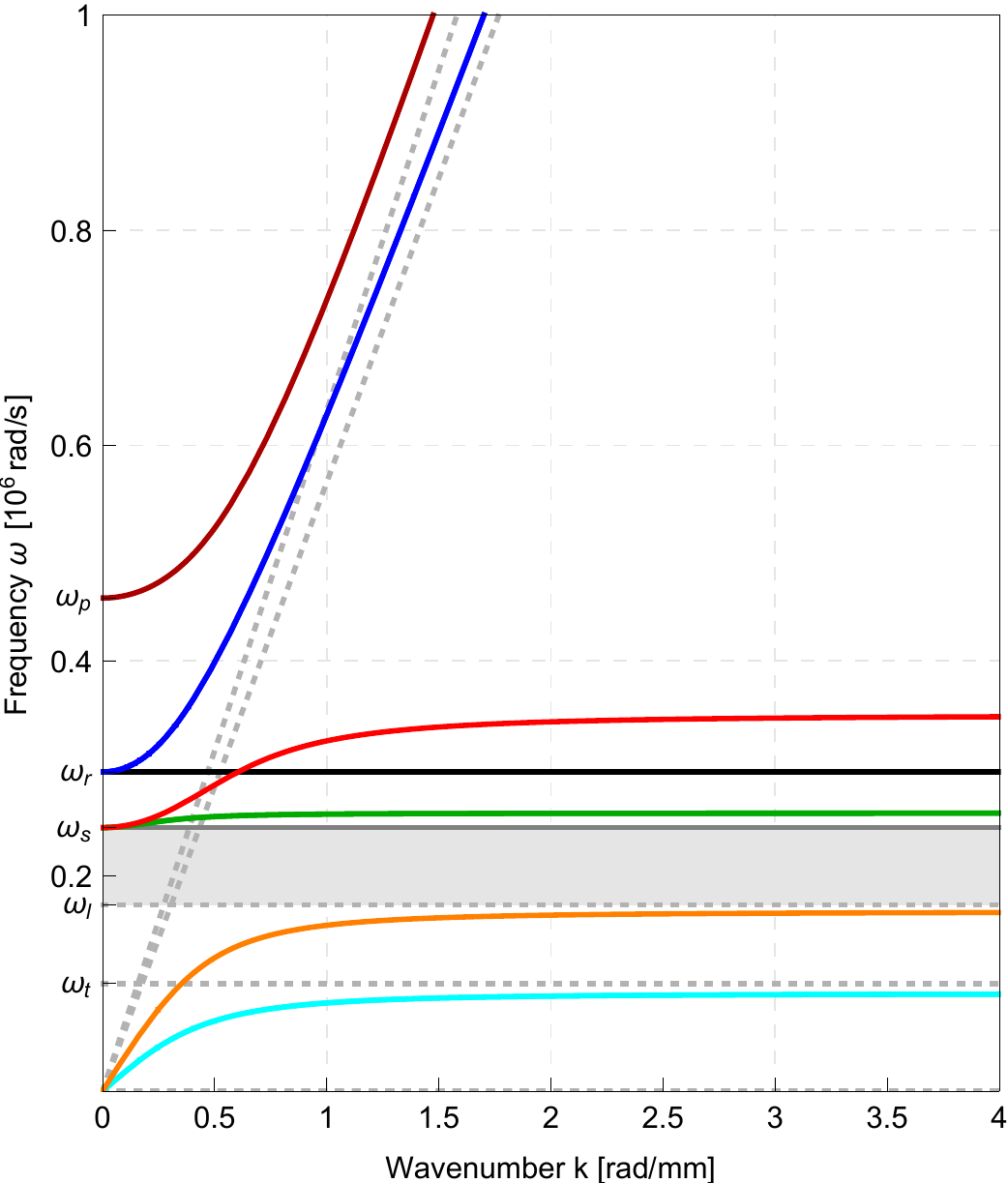}\tabularnewline
$L_{c}=3\cdot10^{-4}$ & $L_{c}=3\cdot10^{-6}$ & $L_{c}=3\cdot10^{-9}$\tabularnewline
\end{tabular}\caption{Effect of the parameter $L_{c}$ on the dispersion curves.}
\end{figure}

The band gap is always present. Nevertheless two curves become horizontal
and 4 horizontal asymptotes instead of 2 are found letting $L_{c}\fr0$.
When $L_{c}=0$ two band-gaps can be created increasing the value
of $\omega_{r}$.

\subsection{Other interesting cases}

\subsubsection{Case ${\displaystyle \protect\me\protect\fr+\infty}$ and $L_{c}$
decreasing}

\begin{figure}[H]
\centering{}%
\begin{tabular}{ccc}
\includegraphics[scale=0.5]{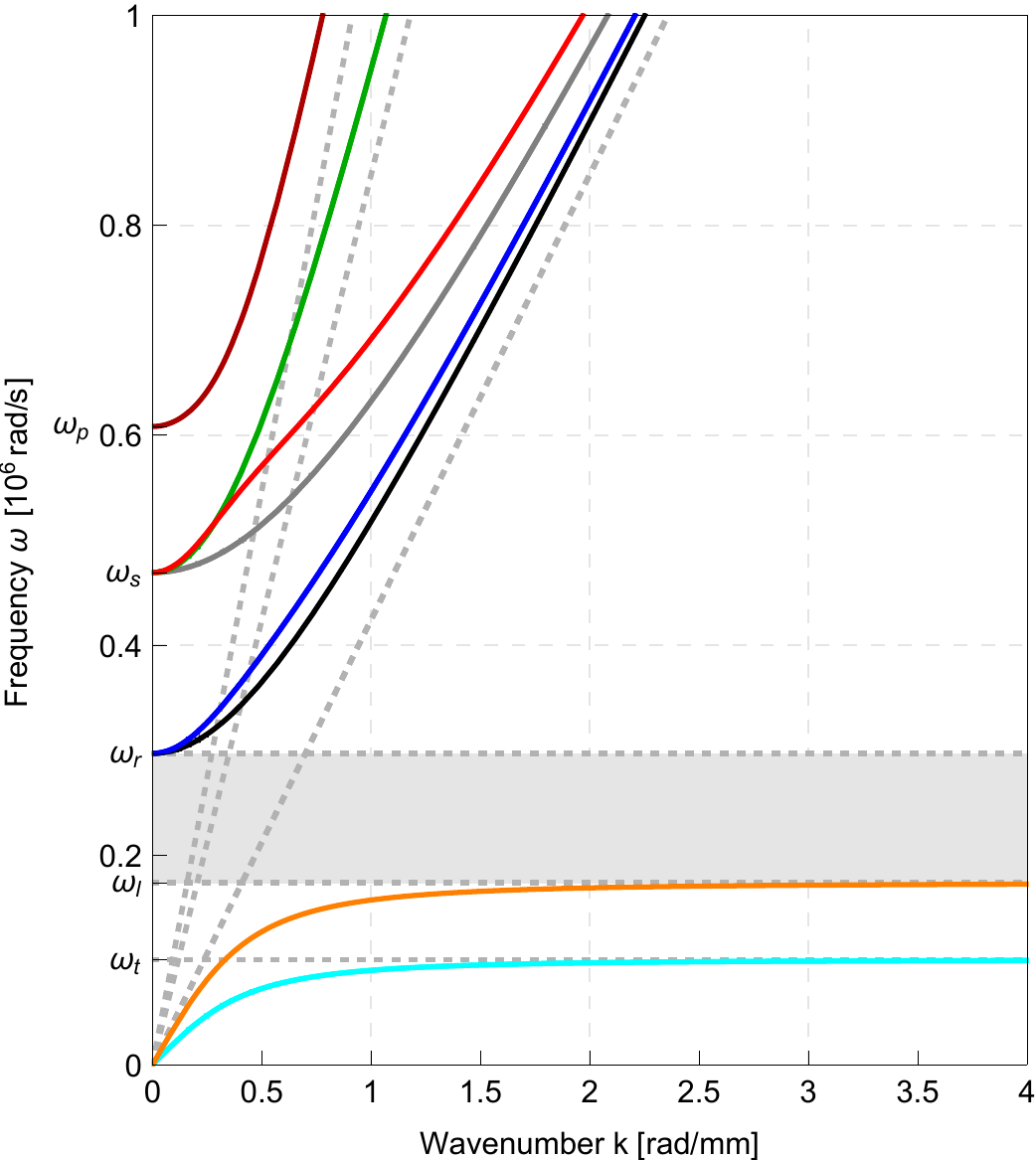} & \includegraphics[scale=0.5]{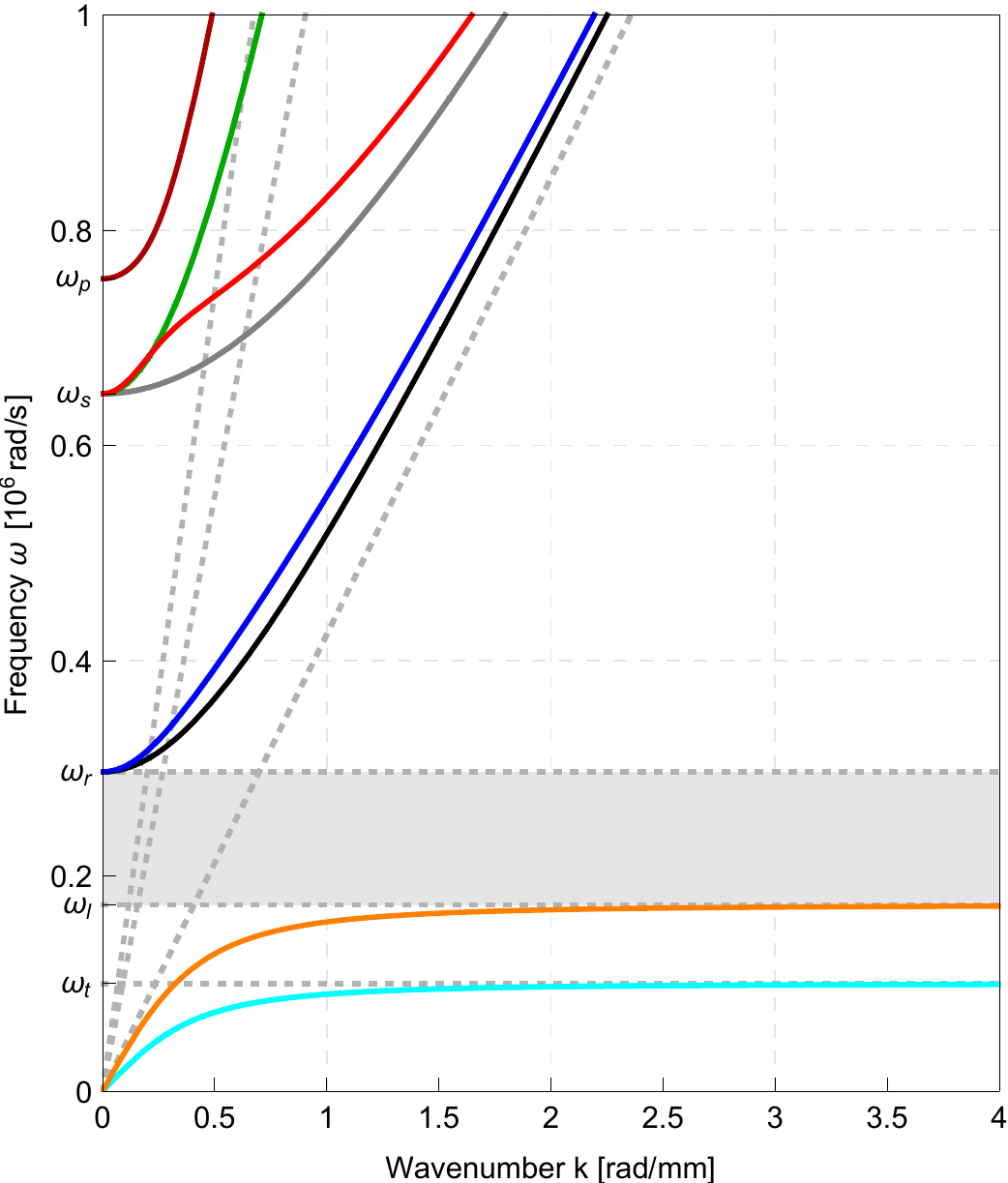} & \includegraphics[scale=0.5]{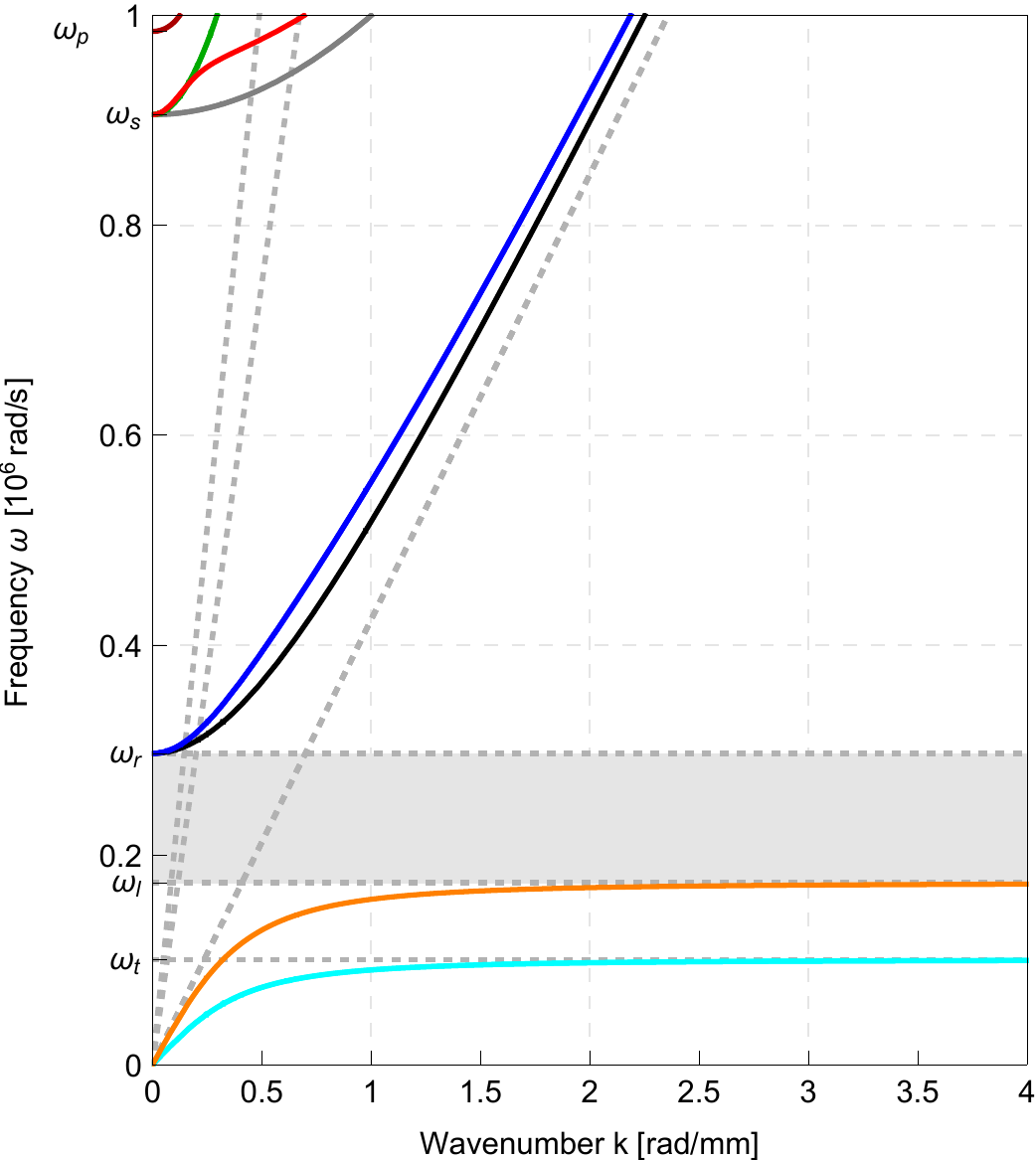}\tabularnewline
$\me=5\cdot10^{8},\;L_{c}=\frac{3\sqrt{5}}{5}\cdot10^{-3}$ & $\me=10^{9},\;L_{c}=\frac{3\sqrt{10}}{10}\cdot10^{-3}$ & $\me=2\cdot10^{9},\;L_{c}=\frac{3\sqrt{20}}{20}\cdot10^{-3}$\tabularnewline
\end{tabular}\caption{Effect of the parameter $\protect\me$ on the dispersion curves.}
\end{figure}
The effect of letting $\me\fr+\infty$ preserves the presence of the
band-gap because it does not influence the acoustic branches and the
cut-off $\omega_{r}$ sending instead the other two cut-offs to infinity.

\subsection*{}

\subsubsection{Case ${\displaystyle \protect\me,\protect\mh,\protect\mc\protect\fr+\infty}$}

\begin{figure}[H]
\centering{}%
\begin{tabular}{ccc}
\includegraphics[scale=0.5]{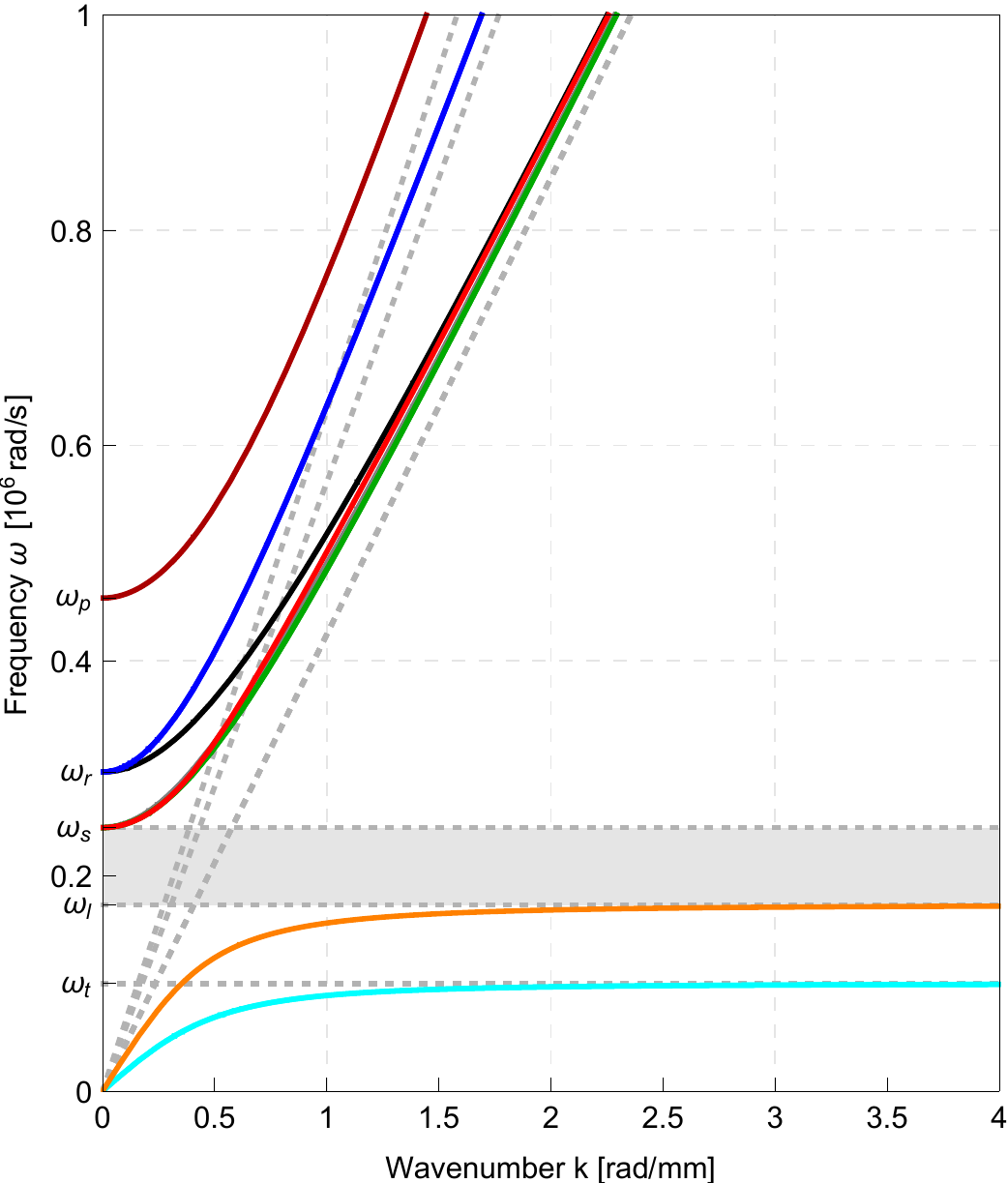} & \includegraphics[scale=0.5]{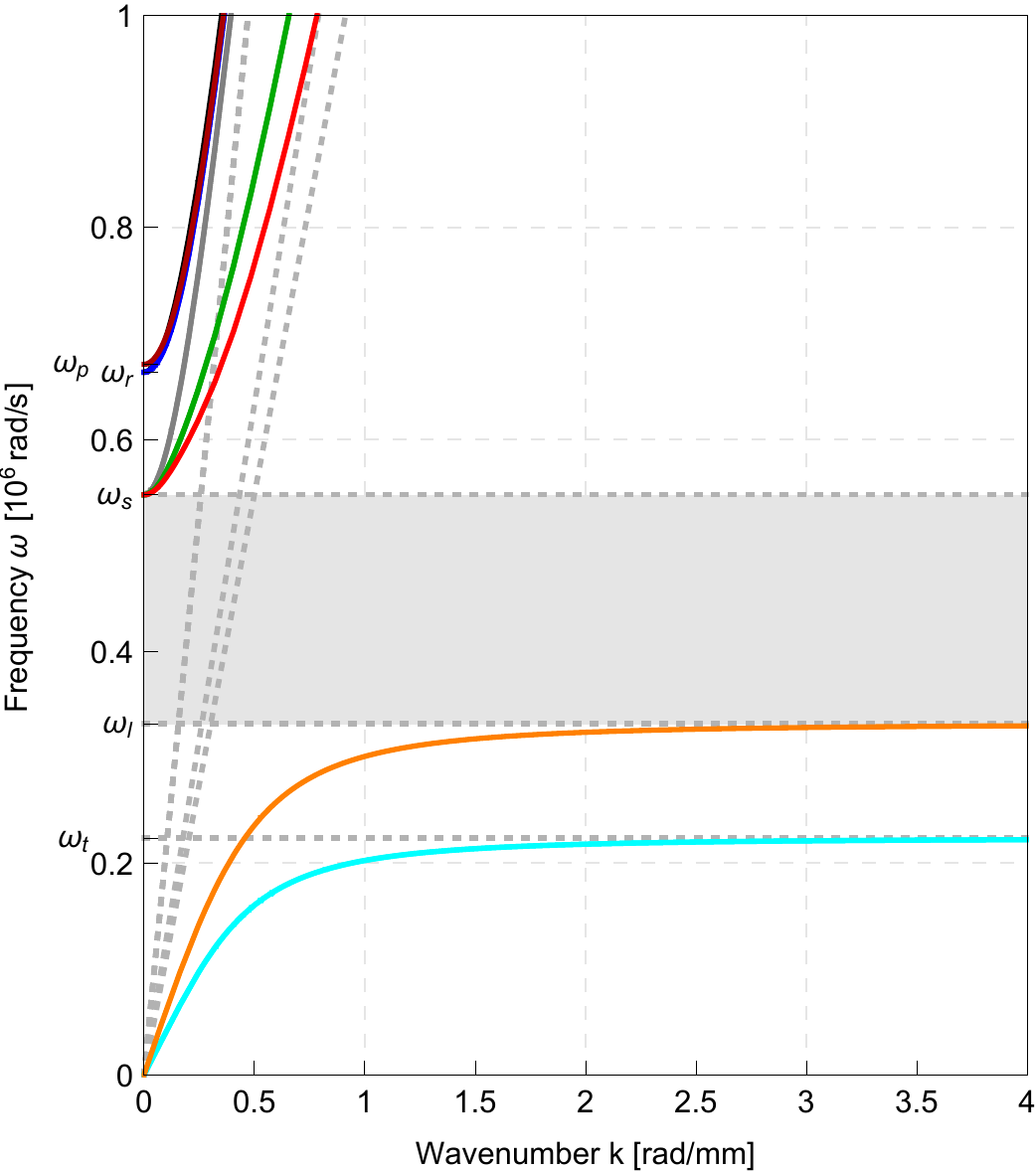} & \includegraphics[scale=0.5]{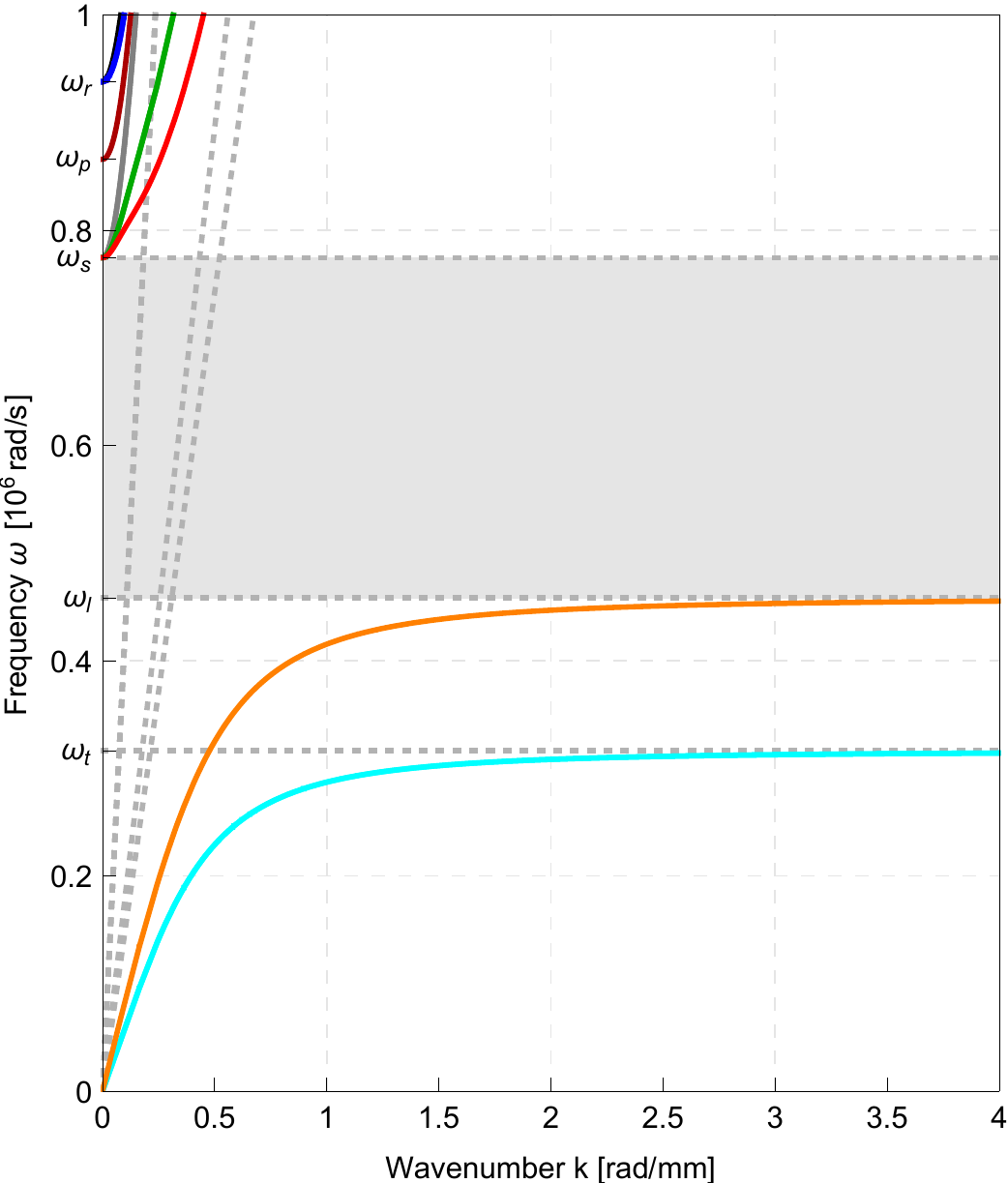}\tabularnewline
$\me=2\cdot10^{8},\mh=10^{8},\mc=4.4\cdot10^{8}$ & $\me=10^{9},\mh=5\cdot10^{8},\mc=22\cdot10^{8}$ & $\me=2\cdot10^{9},\mh=10^{9},\mc=4.4\cdot10^{9}$\tabularnewline
\end{tabular}\caption{Effect of the parameters $\protect\me$, $\protect\mh$ and $\protect\mc$
on the dispersion curves.}
\end{figure}

\subsubsection{Case ${\displaystyle \protect\mh\protect\fr+\infty}$ ``Cosserat
limit''}

\begin{figure}[H]
\centering{}%
\begin{tabular}{ccc}
\includegraphics[scale=0.5]{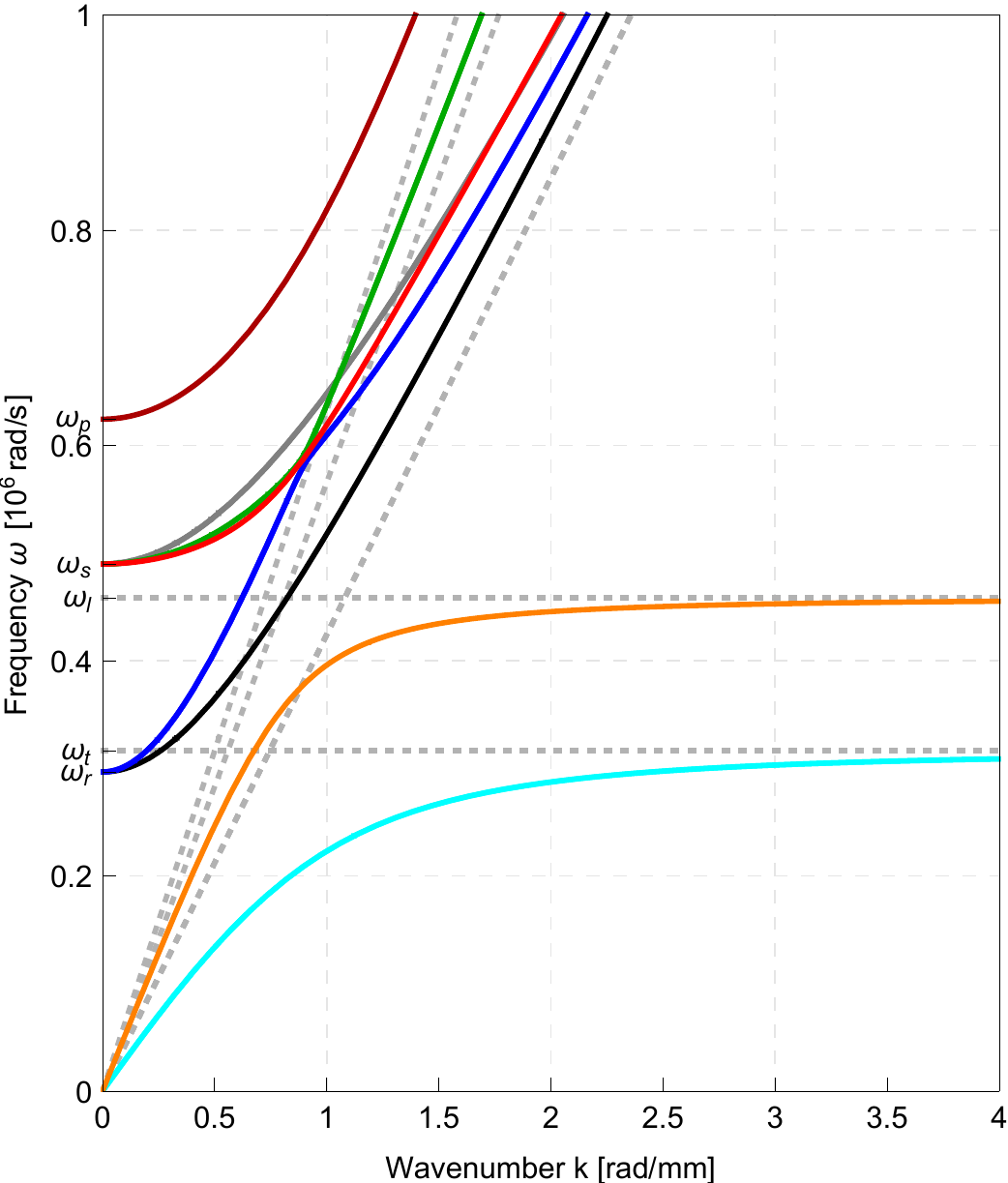} & \includegraphics[scale=0.5]{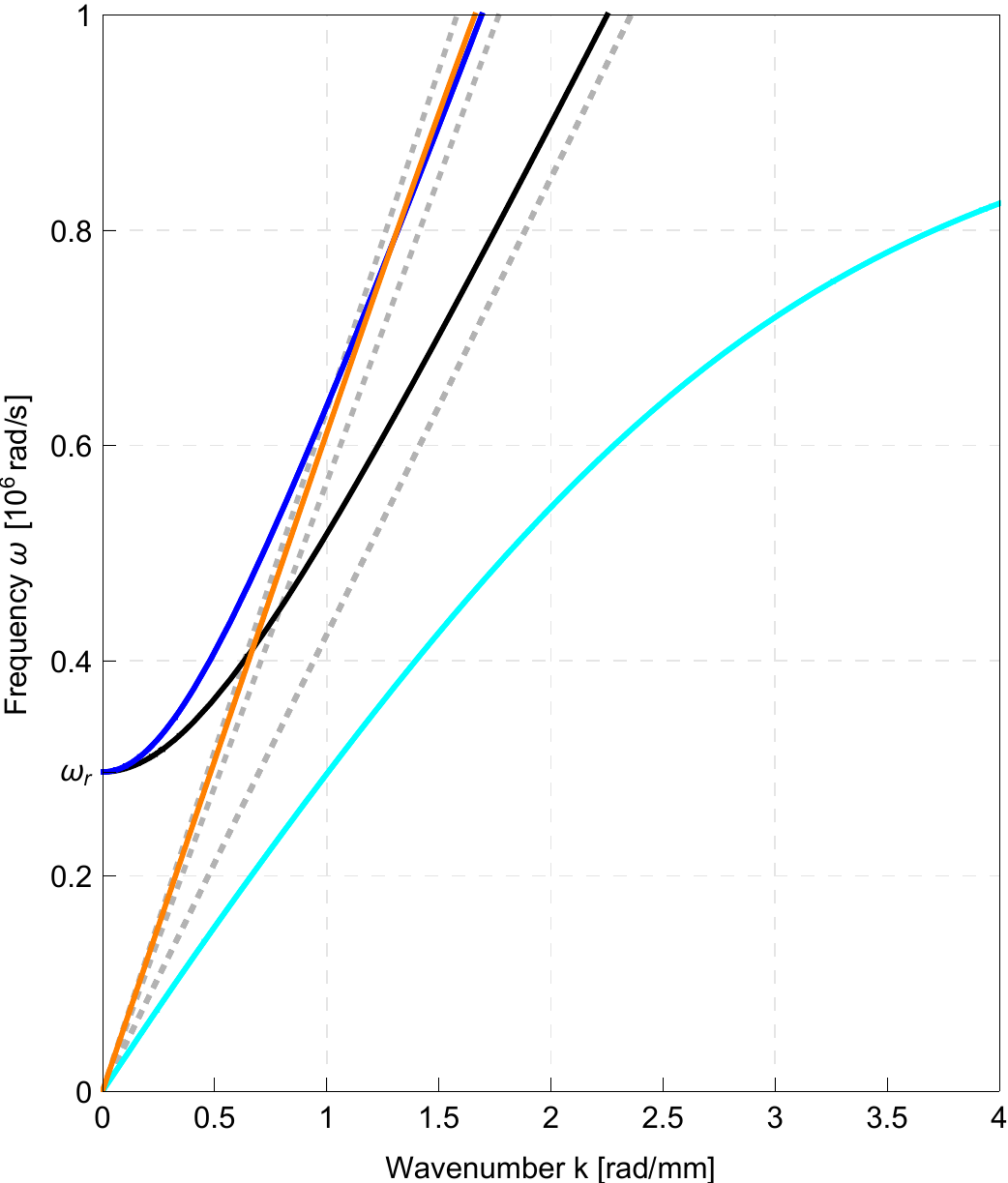} & \includegraphics[scale=0.5]{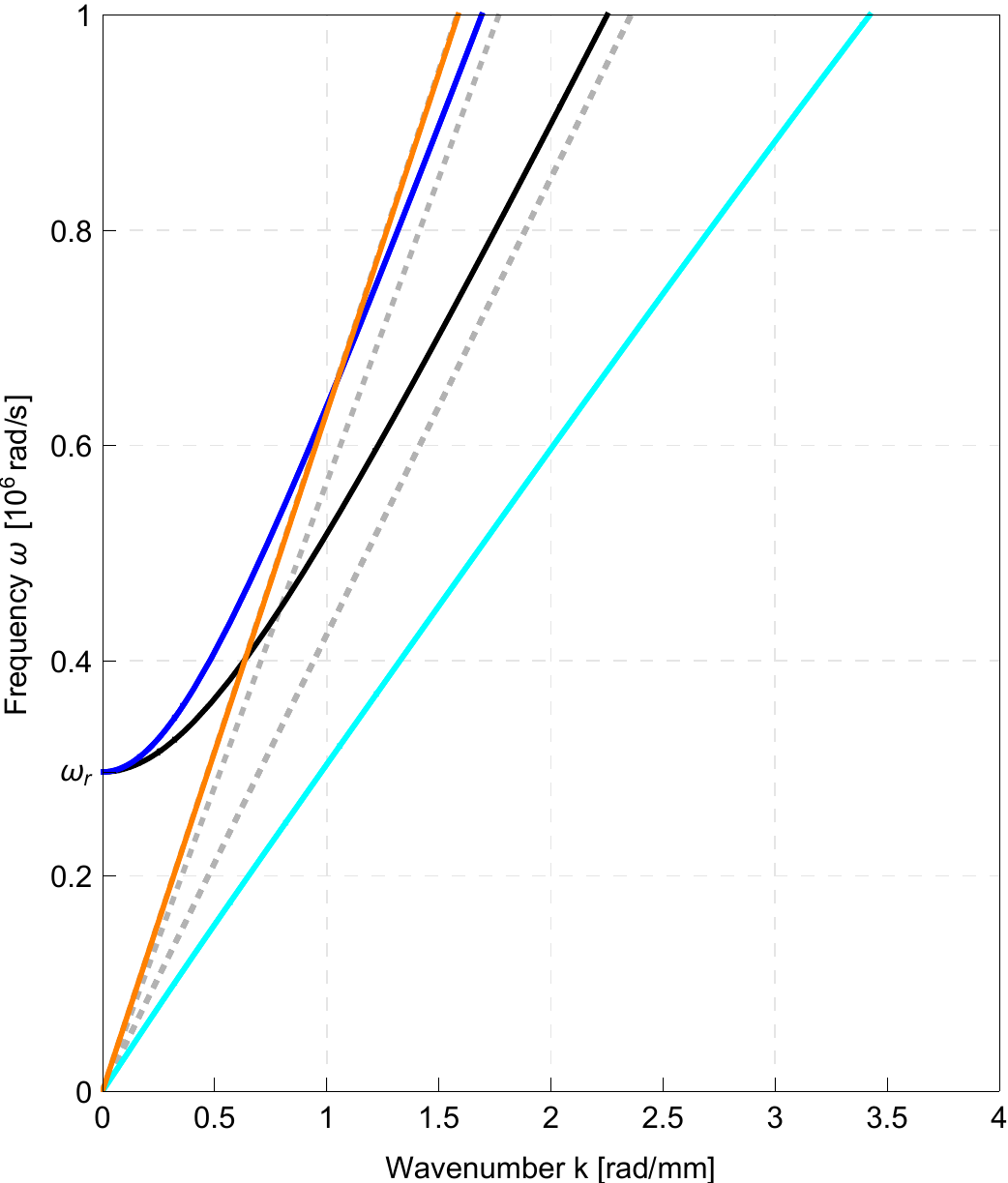}\tabularnewline
$\mh=10^{9}$ & $\mh=10^{10}$ & $\mh=2\cdot10^{11}$\tabularnewline
\end{tabular}\caption{Effect of the parameter $\protect\mh$ on the dispersion curves.\label{fig:coss}}
\end{figure}
The effect of letting $\mh\fr+\infty$ is equivalent to let $\sym\,\P\fr0$.
This means that only the skew-symmetric part $\skew\,\P$ of the micro-distortion
tensor remains both in the elastic and kinetic energies that thus
tend to the Cosserat energies given in eqs. (\ref{eq:coss1}) and
(\ref{eq:coss2}).

It is clearly seen from Fig. \ref{fig:coss} that an increasing value
of $\mh$ directly acts on the acoustic dispersion curves which become
straight lines. Moreover, the optic curves originating from $\omega_{s}$
and $\omega_{p}$ disappear from the dispersion diagrams since the
two cut-offs tend to infinity. If we compare this limit case with
the Cosserat model that we have given in (\ref{fig:Cauchy - Cosserat})
we can remark a perfect concordance.

\subsubsection{Case ${\displaystyle \protect\mh,\mu_{c}\protect\fr+\infty\textrm{ and }}{\displaystyle \eta\protect\fr0}$
``indeterminate couple stress theory''}

Letting $\mh,\mc\fr+\infty$ is equivalent to set $\sym\,\P\fr0,$
and $\P\fr\skew\P$ and $\skew\,\P\fr\skew\,\nabla u$. The corresponding
strain energy density gives rise to the so called indeterminate couple
stress model \cite{neff2014unifying,munch2010nonlinear}: since there
are no degrees of freedom related to $\P$ that remain active, the
kinematics reduces to the simple displacement field. For this reason
the term related to the micro-distortion $\P$ in the kinetic energy
must be neglected by setting $\eta\fr0$.

\begin{figure}[H]
\centering{}%
\begin{tabular}{ccc}
\includegraphics[scale=0.5]{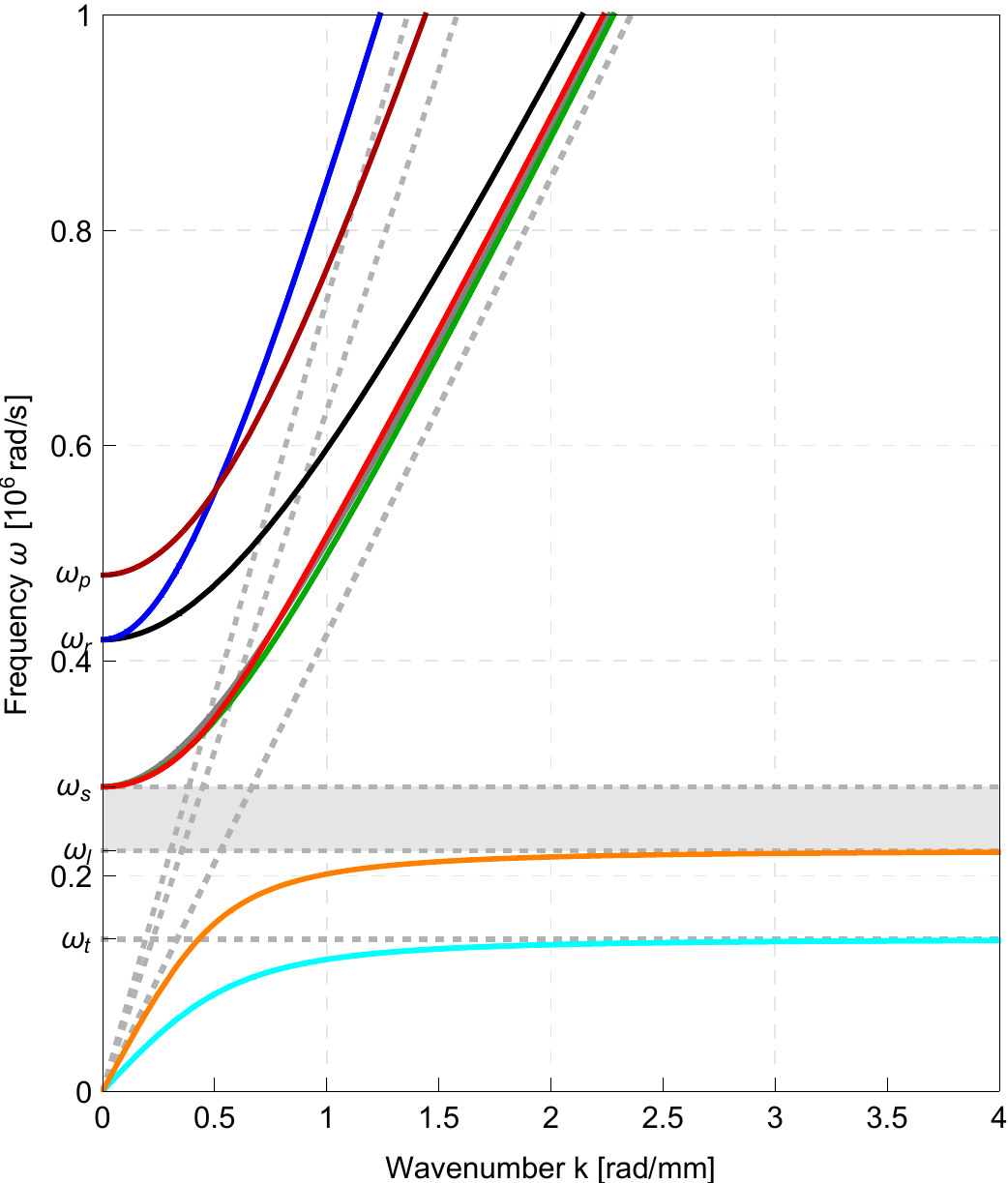} & \includegraphics[scale=0.5]{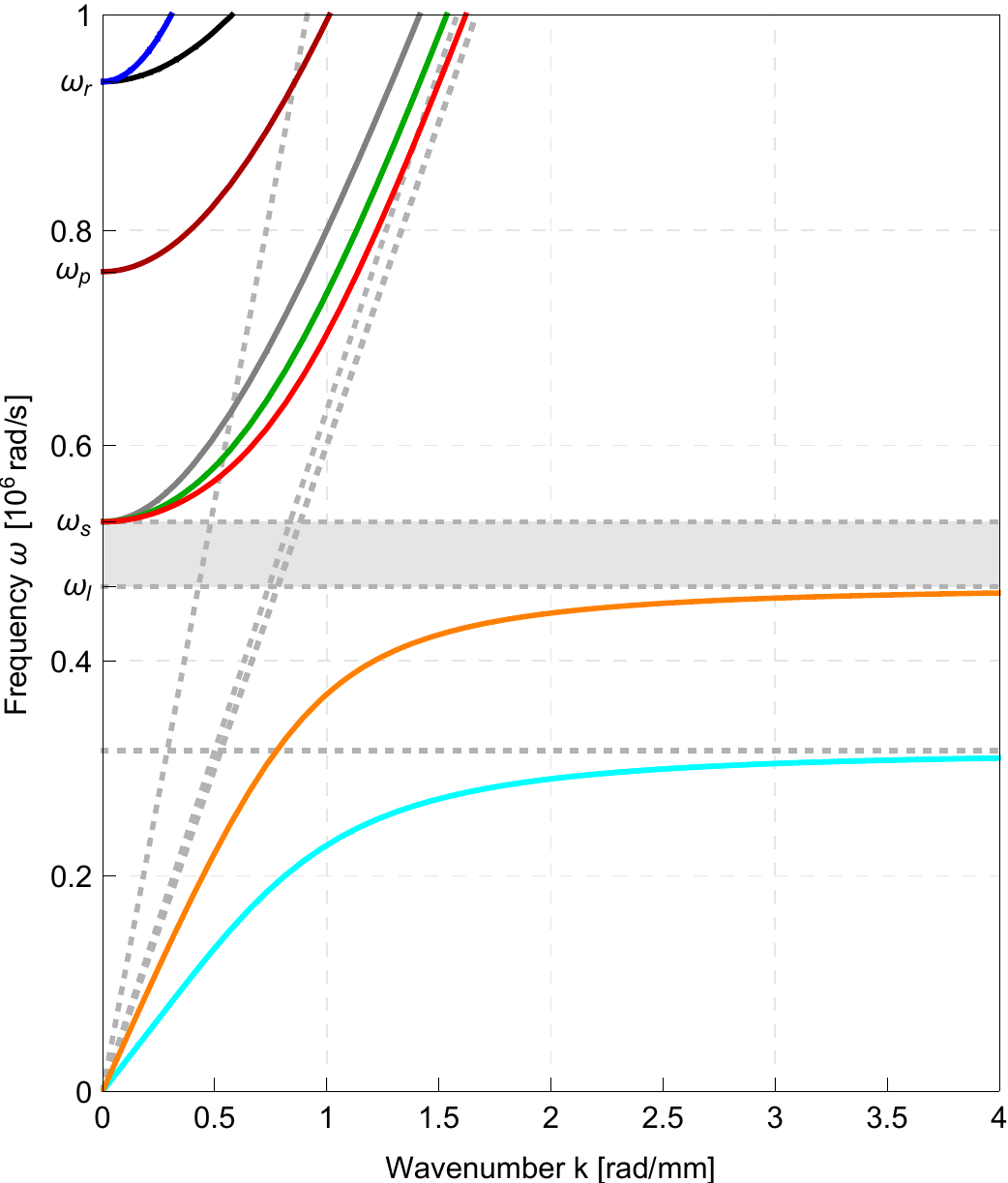} & \includegraphics[scale=0.5]{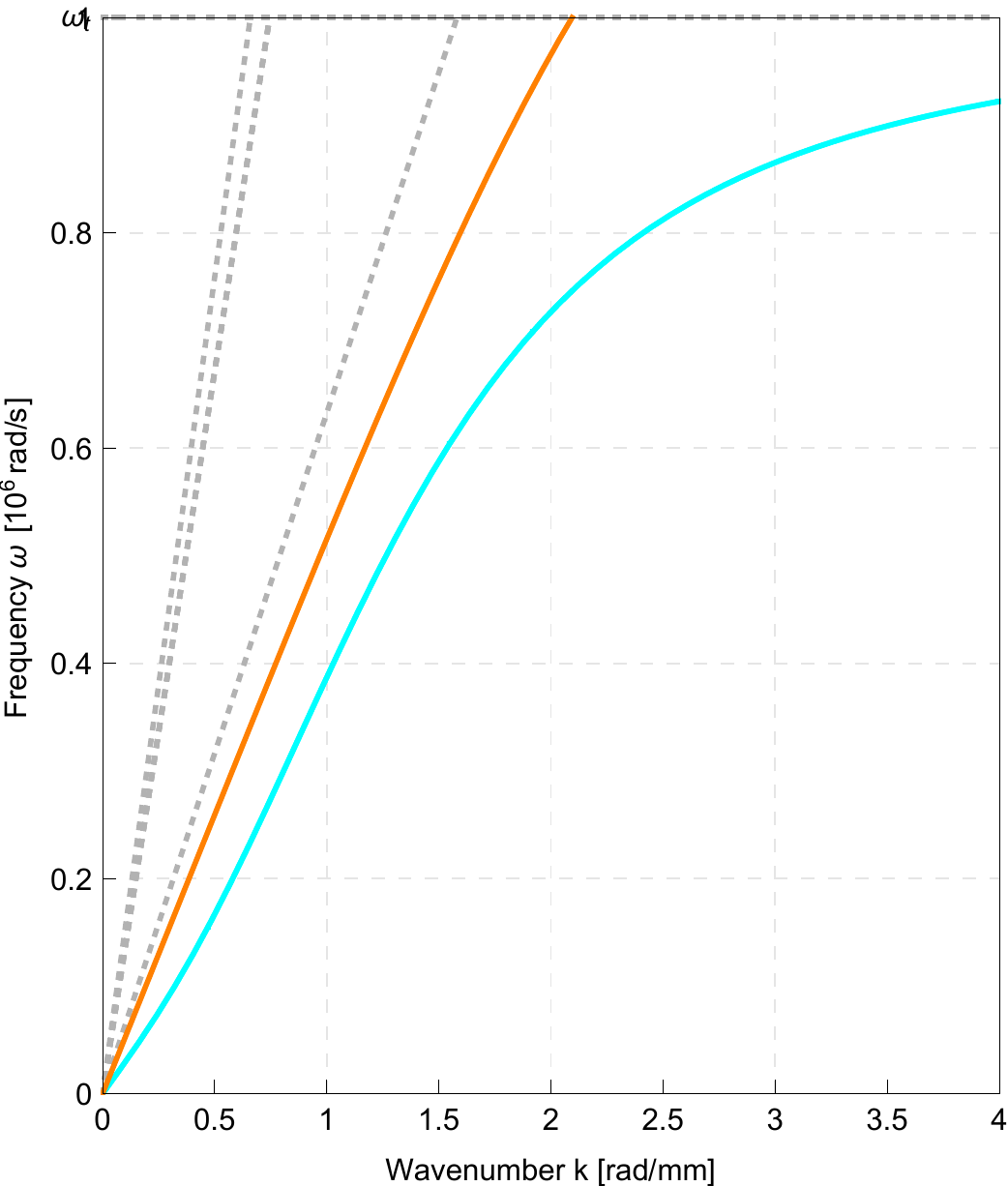}\tabularnewline
$\mc=8.8\cdot10^{8},\mh=2\cdot10^{8}$ & $\mc=2.2\cdot10^{9},\mh=5\cdot10^{8}$ & $\mc=4.4\cdot10^{9},\mh=10^{9}$\tabularnewline
$\eta_{1}=\eta_{2}=\eta_{3}=10^{-2}$ & $\eta_{1}=\eta_{2}=\eta_{3}=0.5\cdot10^{-2}$ & $\eta_{1}=\eta_{2}=\eta_{3}=10^{-3}$\tabularnewline
\includegraphics[scale=0.5]{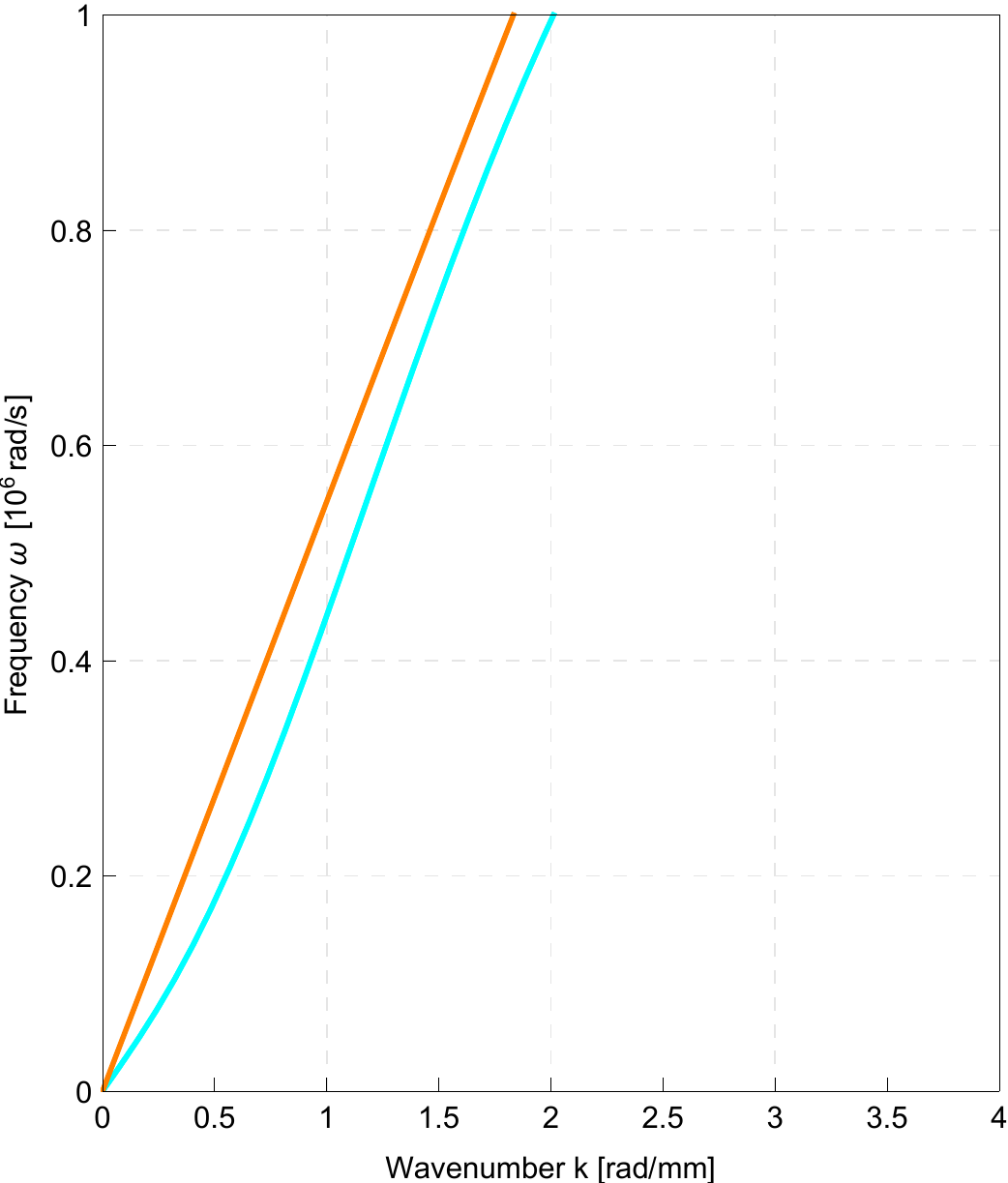} & \includegraphics[scale=0.5]{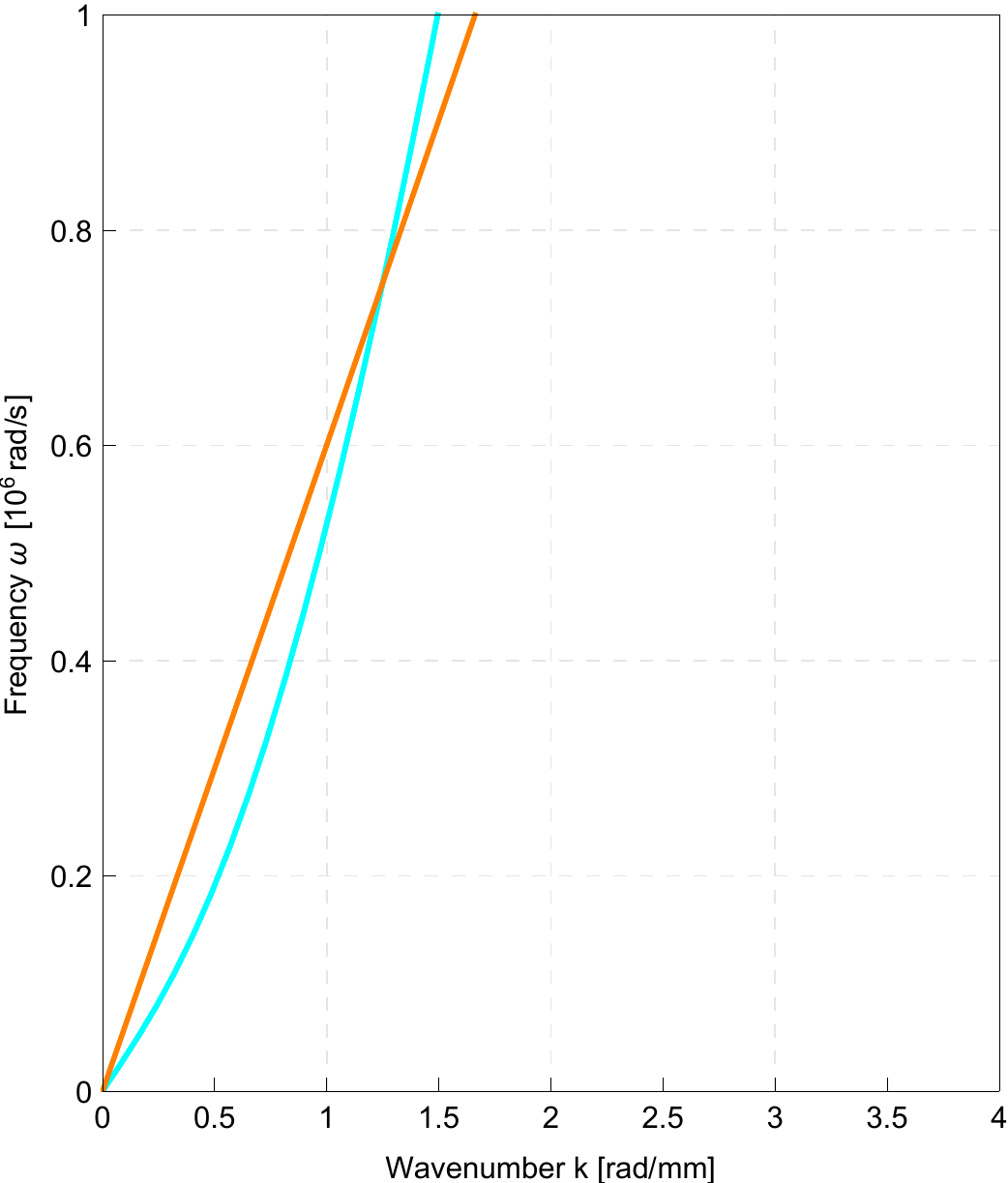} & \includegraphics[scale=0.5]{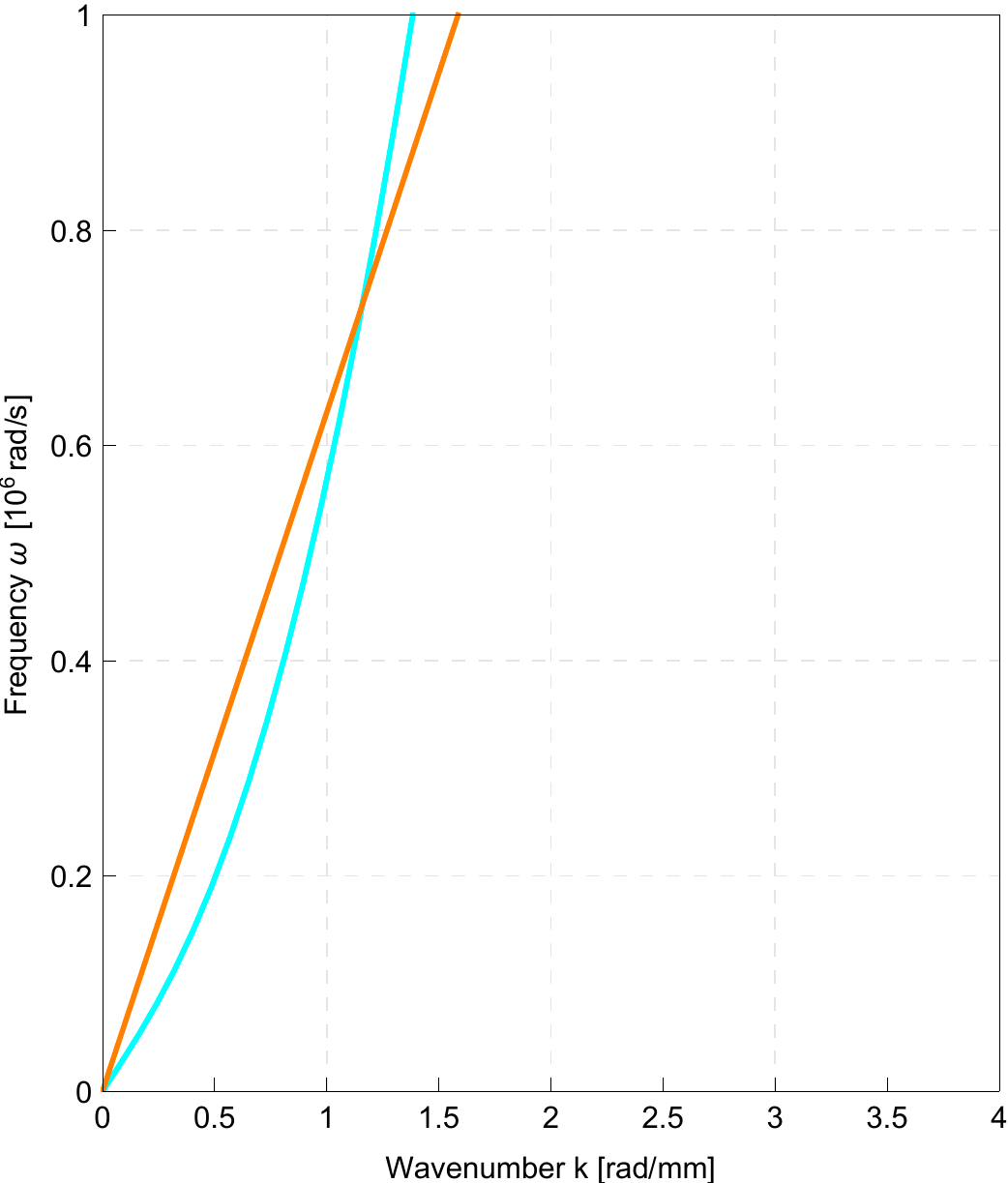}\tabularnewline
$\mc=6.6\cdot10^{9},\mh=1.5\cdot10^{9}$ & $\mc=2.2\cdot10^{10},\mh=5\cdot10^{9}$ & $\mc=4.4\cdot10^{11},\mh=10^{11}$\tabularnewline
$\eta_{1}=\eta_{2}=\eta_{3}=0.5\cdot10^{-3}$ & $\eta_{1}=\eta_{2}=\eta_{3}=10^{-4}$ & $\eta_{1}=\eta_{2}=\eta_{3}=10^{-5}$\tabularnewline
\end{tabular}\caption{Effect of the parameters $\protect\mc,\protect\mh,\eta_{1},\eta_{2},\eta_{3}$
on the dispersion curves.}
\end{figure}
As expected, and as it is known for second gradient theories, only
two acoustic curves are found as it was the case for classical elasticity.
The only extra feature with respect to the classical elasticity is
that higher gradient models may account for some dispersive effects.
In order to have a direct comparison with the indeterminate couple
stress model, we have directly implement the indeterminate couple
stress model considering the deformation energy (the relative strong
equations can be found in the appendix)
\begin{align*}
W_{\textrm{ind}}\left(\sym\,\nabla u,\nabla\skew\,\nabla u\right) & =\me\left\Vert \sym\,\nabla u\right\Vert ^{2}+\frac{\le}{2}\left(\textrm{tr}\,\nabla u\right)^{2}+\me\frac{L_{c}^{2}}{2}\left\Vert \nabla\left(\skew\,\nabla u\right)\right\Vert ^{2}.
\end{align*}
With the same choice of the material parameters, we obtain the following
dispersion curves:
\begin{figure}[H]
\centering{}%
\begin{tabular}{c}
\includegraphics[scale=0.5]{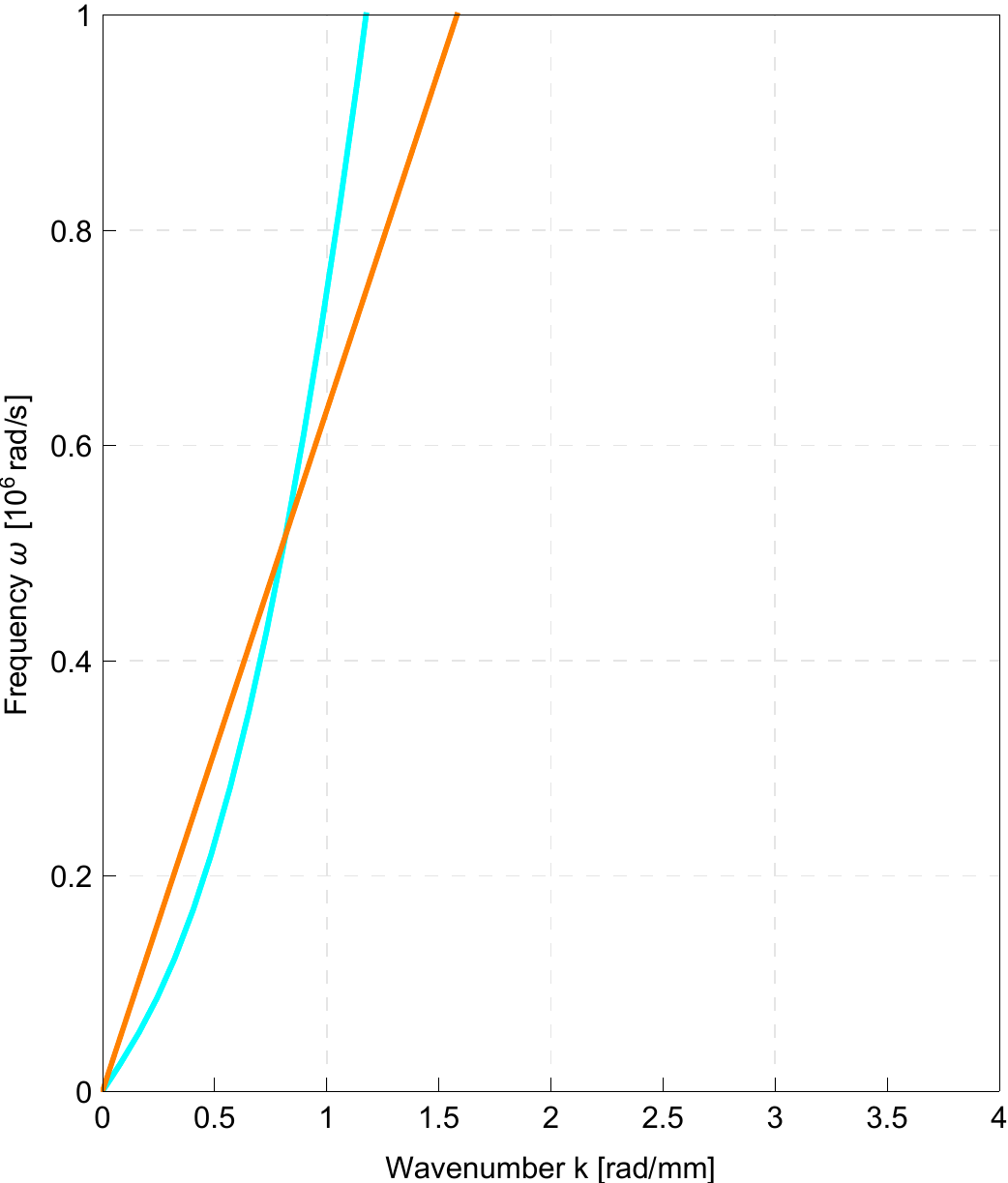}\tabularnewline
\end{tabular}\caption{Dispersion curves indeterminate couple stress model.}
\end{figure}

\section*{Conclusion}

In this paper we present for the first time the ``weighted'' relaxed
micromorphic model, by introducing the Cartan-Lie decomposition of
the second order tensors $\P_{,t}$ and $\curl\,\P$ in the kinetic
and elastic energies, respectively. It is found that the split of
the tensor $\P_{,t}$ in the micro-inertia provides a unique feature
to the model, namely the possibility of separately controlling the
cut-offs of the optic curves in the dispersion diagram. This is an
essential feature in view of the calibration of the relaxed micromorphic
parameters on real band-gap metamaterials. The split of the second
order tensor $\curl\,\P$ presents some effects on the dispersion
curves which are less clearly related to possible physical situations.
In general, we can say that the term $\dev\,\sym\,\curl\,\P$ governs,
to a large extent, the non-locality in the considered metamaterials
since it is able to give rise to a horizontal optic curve. On the
other hand, the term $\skew\,\curl\,\P$ is able to grant the onset
of (extra) horizontal asymptotes for some optic curves. No specific
effect can be attributed to the term $\textrm{tr}\left(\curl\,\P\right)$.

Another important result of the present paper is that of showing the
fundamental role of micro-inertia terms when dealing with enriched
continua. It is shown that both the cases $\eta\fr0$ and $\eta\fr\infty$
give rise to a Cauchy-type situation in which only 2 straight acoustic
branches can be recognized. Nevertheless, the physical meaning associated
to such two cases is completely different. Indeed, setting $\eta=0$
in a model with enriched kinematics can be considered to be a mistake
since one gives a complex and rich behavior to the elastic energy
(through a particular dependence on $\P$), but one does not allow
the exploitation of such constitutive behavior due to the absence
of the associated micro-inertia. It is hence not astonishing that,
no matter how complex is the constitutive choice for the elastic energy
(Mindlin, relaxed micromorphic, Cauchy, second gradient, etc.), the
resulting dispersion curves are basically those of classical elasticity:
two straight lines starting from the origin. The problem is simply
that we do not give to the model the possibility to express its dynamical
behavior since there is no micro-inertia that is able to trigger micro-vibrations.

On the other hand, the case $\eta\fr\infty$ is phenomenologically
different: we introduce the micro-inertia in the model, but it is
so ``high'' that the microstructure is ``frozen'' and this results
in Cauchy-like materials. We leave to a forthcoming paper the task
of studying in greater detail the importance of the role of micro-inertia
in enriched continuum mechanics. Finally, some parametric studies
on all of the introduced constitutive parameters are performed, thus
giving a complete panorama of all possible dispersion patterns which
are attainable in the relaxed micromorphic framework.

\section{Acknowledgments}

Angela Madeo thanks INSA-Lyon for the funding of the BQR 2016 \textquotedbl{}Caractérisation
mécanique inverse des métamatériaux: modélisation, identification
expérimentale des paramètres et évolutions possibles\textquotedbl{},
as well as the CNRS-INSIS for the funding of the PEPS project.

\section{Appendix }

\subsection{Variation of the kinetic energy}

\bigskip{}
In order to derive the Euler-Lagrange equations, we need to assume
a stronger regularity for the kinematical fields: 

\[
\left(u,P\right)\in\mathscr{C}^{2}\left(\overline{\Omega}\times I,\R^{3}\right)\times\mathscr{C}^{2}\left(\overline{\Omega}\times I,\R^{3\times3}\right).
\]
Having that 
\[
\L\in\mathscr{C}^{2}\left(\R^{3}\times\R^{3\times3}\times\R^{3\times3}\times\R^{3\times3}\times\R^{3\times3}\right),
\]
the action functional $\mathscr{A}_{\L}$ is Fréchet differentiable
(and so Gâteaux differentiable) on the affine subspace 
\[
\mathcal{Q}^{2}:=\left\{ \left(\u,\P\right)\in\mathscr{C}^{2}\left(\overline{\Omega}\times I,\R^{3}\right)\times\mathscr{C}^{2}\left(\overline{\Omega}\times I,\R^{3\times3}\right):\left(\u,\P\right)\textrm{ verifies conditions }\left(\mathsf{B}_{1}\right)\textrm{ and }\left(\mathsf{B}_{2}\right)\right\} .
\]
 Its differential\footnote{$\mathcal{Q}_{0}:=\left\{ \left(\delta\u,\delta\P\right)\in\mathscr{C}_{c}^{\infty}\left(\Omega\times I,\R^{3}\right)\times\mathscr{C}_{c}^{\infty}\left(\Omega\times I,\R^{3\times3}\right)\right\} $
is the vector space of admissible variations (test functions).} at a point $\left(u,\P\right)$ 
\[
\delta\mathscr{A}_{\L\left(u,\P\right)}:\mathcal{Q}_{0}\fr\R,
\]
evaluated at the admissible variation $\left(\delta\u,\delta\P\right)$,
is given by the variation of the part associated to the kinetic energy
\[
\delta\int_{I}\int_{\Omega}J\left(\u_{,t},\P_{,t}\right)\,dm\,dt
\]
and of that associated to the potential energy

\[
\delta\int_{I}\int_{\Omega}W\left(\nabla\u,\P,\curl\,\p\right)\,dm\,dt.
\]
We compute here the part of the action functional associated to the
kinetic energy:

\begin{gather*}
\delta\int_{I}\int_{\Omega}J\left(\u_{,t},\P_{,t}\right)\,dm\,dt=\int_{I}\int_{\Omega}\left[\left\langle D_{\u_{,t}}J\left(\u_{,t},\P_{,t}\right),\delta\u_{,t}\right\rangle +\left\langle D_{\P_{,t}}J\left(\u_{,t},\P_{,t}\right),\delta\P_{,t}\right\rangle \right]dm\,dt=\\
=\int_{I}\int_{\Omega}\frac{1}{2}\left[\left\langle D_{\u_{,t}}\left(\rho\left\langle \u_{,t},\u_{,t}\right\rangle \right),\delta\u_{,t}\right\rangle +\left\langle D_{\P_{,t}}\left(\eta_{\,1}\left\Vert \textrm{dev sym}\,\P_{,t}\right\Vert ^{2}+\eta_{\,2}\left\Vert \textrm{skew}\,\P_{,t}\right\Vert ^{2}+\frac{1}{3}\,\eta_{\,3}\left(\textrm{tr}\,\P_{,t}\right)^{2}\right),\delta\P_{,t}\right\rangle \right]dm\,dt=\\
=\int_{I}\int_{\Omega}\left[\rho\left\langle \u_{,t},\delta\u_{,t}\right\rangle +\frac{1}{2}\left\langle D_{\P_{,t}}\left(\eta_{\,1}\left\Vert \textrm{dev sym}\,\P_{,t}\right\Vert ^{2}+\eta_{\,2}\left\Vert \textrm{skew}\,\P_{,t}\right\Vert ^{2}+\frac{1}{3}\,\eta_{\,3}\left(\textrm{tr}\,\P_{,t}\right)^{2}\right),\delta\P_{,t}\right\rangle \right]dm\,dt=
\end{gather*}
\begin{gather*}
=\underbrace{\int_{I}\int_{\Omega}\rho\left\langle \u_{,t},\delta\u_{,t}\right\rangle dm\,dt}_{I}+\underbrace{\int_{I}\int_{\Omega}\eta_{\,1}\left\langle \textrm{dev sym}\,\P_{,t},\textrm{dev sym}\,\delta\P_{,t}\right\rangle dm\,dt}_{II}\\
+\underbrace{\int_{I}\int_{\Omega}\eta_{\,2}\left\langle \textrm{skew}\,\P_{,t},\textrm{skew}\,\delta\P_{,t}\right\rangle dm\,dt}_{III}+\underbrace{\int_{I}\int_{\Omega}\frac{1}{3}\,\eta_{\,3}\textrm{tr}\,\P_{,t}\,\textrm{tr}\,\delta\P_{,t}dm\,dt}_{IV}.
\end{gather*}
In order to find the Euler-Lagrange equations, we have to integrate
by parts, with respect to the time derivative, the four parts $I,II,III,IV$
:
\begin{align*}
{\displaystyle I} & =\rho\int_{\Omega}\left(\left.\left\langle \u_{,t},\delta\u\right\rangle \right|_{a}^{b}-\int_{I}\left\langle \u_{,tt},\delta\u\right\rangle dt\right)dm,\\
{\displaystyle II} & =\eta_{1}\int_{\Omega}\Bigg(\left.\left\langle \textrm{dev sym}\,\P_{,t},\textrm{dev sym}\,\delta\P\right\rangle \right|_{a}^{b}-\int_{I}\left\langle \textrm{dev sym}\,\P_{,tt},\textrm{dev sym}\,\delta\P\right\rangle dt\Bigg)dm,\\
III & {\displaystyle \,=\eta_{2}\int_{\Omega}\left(\left.\left\langle \textrm{skew}\,\P_{,t},\textrm{skew}\,\delta\P\right\rangle \right|_{a}^{b}-\int_{I}\left\langle \textrm{skew}\,\P_{,tt},\textrm{skew}\,\delta\P\right\rangle dt\right)dm,}\\
IV & {\displaystyle \,=\frac{\eta_{3}}{3}\int_{\Omega}\left(\left.\textrm{tr}\,\P_{,t}\,\textrm{tr}\,\delta\P\right|_{a}^{b}-\int_{I}\textrm{tr}\,\P_{,tt}\,\textrm{tr}\,\delta\P\,dt\right)dm.}
\end{align*}
Considering that
\begin{align}
\left\langle \textrm{dev sym}\,\P_{,tt},\textrm{dev sym}\,\delta\P\right\rangle  & =\left\langle \textrm{dev sym}\,\P_{,tt},\delta\P\right\rangle ,\nonumber \\
\left\langle \textrm{skew}\,\P_{,tt},\textrm{skew}\,\delta\P\right\rangle  & =\left\langle \textrm{skew}\,\P_{,tt},\delta\P\right\rangle ,\nonumber \\
\textrm{tr}\,\P_{,tt}\,\textrm{tr}\,\delta\P & =\left\langle \textrm{tr}\left(\P_{,tt}\right)\id,\frac{1}{3}\,\textrm{tr}\left(\delta\P\right)\id\right\rangle =\left\langle \textrm{tr}\left(\P_{,tt}\right)\id,\delta\P\right\rangle ,\label{eq:trace}
\end{align}
we find that

\begin{align*}
I & {\displaystyle \,=\rho\int_{\Omega}\left(\left.\left\langle \u_{,t},\delta\u\right\rangle \right|_{a}^{b}-\int_{I}\left\langle \u_{,tt},\delta\u\right\rangle dt\right)dm,}\\
II & {\displaystyle \,=\eta_{1}\int_{\Omega}\left(\left.\left\langle \textrm{dev sym}\,\P_{,t},\delta\P\right\rangle \right|_{a}^{b}-\int_{I}\left\langle \textrm{dev sym}\,\P_{,tt},\delta\P\right\rangle dt\right)dm,}\\
III & \,=\eta_{2}\int_{\Omega}\left(\left.\left\langle \textrm{skew}\,\P_{,t},\delta\P\right\rangle \right|_{a}^{b}-\int_{I}\left\langle \textrm{skew}\,\P_{,tt},\delta\P\right\rangle dt\right)dm,\\
IV & {\displaystyle \,=\frac{\eta_{3}}{3}\int_{\Omega}\left(\left.\left\langle \textrm{tr}\left(\P_{,t}\right)\id,\delta\P\right\rangle \right|_{a}^{b}-\int_{I}\left\langle \textrm{tr}\left(\P_{,tt}\right)\id,\delta\P\right\rangle dt\right)dm.}
\end{align*}
So considering only the bulk part , we have
\[
\delta\int_{I}\int_{\Omega}J\left(\u_{,t},\P_{,t}\right)\,dm\,dt=-\int_{\Omega}\int_{I}\left(\left\langle \u_{,tt},\delta\u\right\rangle +\left\langle \eta_{\,1}\,\textrm{dev sym}\,\P_{,tt}+\eta_{\,2}\,\textrm{skew}\P_{,tt}+\frac{1}{3}\,\eta_{\,3}\,\textrm{tr}\left(\P_{,tt}\right)\id,\delta\P\right\rangle \right)dt\,dm.
\]

\subsection{Variation of the part $\mathsf{B}$ of the potential energy}

\bigskip{}
Remembering (\ref{eq:id div}), the first variation of the action
functional is computed thanks to the following identities:
\begin{align}
\left\langle \ds\,\curl\,\p,\delta\,\ds\,\curl\,\p\right\rangle  & =\left\langle \ds\,\curl\,\p,\curl\,\delta\P\right\rangle _{\R^{3\times3}}=\sum_{i=1}^{3}\left\langle \left(\ds\,\curl\,\p\right)_{i},\left(\curl\,\delta\P\right)_{i}\right\rangle _{\R^{3}}\nonumber \\
 & =\sum_{i=1}^{3}\left\langle \left(\ds\,\curl\,\p\right)_{i},\textrm{curl}\left(\delta\P\right)_{i}\right\rangle _{\R^{3}}\nonumber \\
 & =-\sum_{i=1}^{3}\left(\textrm{div}\left(\left(\ds\,\curl\,\p\right)_{i}\times\left(\delta\P\right)_{i}\right)+\left\langle \textrm{curl}\left(\ds\,\curl\,\p\right)_{i},\left(\delta\P\right)_{i}\right\rangle _{\R^{3}}\right)\nonumber \\
 & =-\sum_{i=1}^{3}\textrm{div}\left(\left(\ds\,\curl\,\p\right)_{i}\times\left(\delta\P\right)_{i}\right)+\left\langle \textrm{Curl}\,\ds\,\curl\,\p,\delta\P\right\rangle _{\R^{3\times3}},\nonumber \\
\left\langle \skew\,\curl\,\p,\delta\,\skew\,\curl\,\p\right\rangle  & =\left\langle \skew\,\curl\,\p,\curl\,\delta\P\right\rangle =\sum_{i=1}^{3}\left\langle \left(\skew\,\curl\,\p\right)_{i},\left(\curl\,\delta\P\right)_{i}\right\rangle _{\R^{3}}\nonumber \\
 & =-\sum_{i=1}^{3}\textrm{div}\left(\left(\skew\,\curl\,\p\right)_{i}\times\left(\delta\P\right)_{i}\right)+\left\langle \textrm{Curl}\,\skew\,\curl\,\p,\delta\P\right\rangle _{\R^{3\times3}}\label{eq:curl-1-1}\\
\frac{1}{3}\tr\left(\curl\,\p\right)\delta\tr\left(\curl\,\p\right) & =\left\langle \frac{1}{3}\,\tr\left(\curl\,\p\right)\id,\curl\,\delta\P\right\rangle =\sum_{i=1}^{3}\left\langle \left(\frac{1}{3}\,\tr\left(\curl\,\p\right)\id\right)_{i},\left(\curl\,\delta\P\right)_{i}\right\rangle _{\R^{3}}\nonumber \\
 & =-\sum_{i=1}^{3}\textrm{div}\left(\left(\frac{1}{3}\,\tr\left(\curl\,\p\right)\id\right)_{i}\times\left(\delta\P\right)_{i}\right)+\left\langle \textrm{Curl}\,\frac{1}{3}\,\tr\left(\curl\,\p\right)\id,\delta\P\right\rangle _{\R^{3\times3}}.\nonumber 
\end{align}

\subsection{Derivation of PDEs in the new variables}

The following identities will be useful in the following.

\begin{align}
\left(\curl\,\ds\,\curl\,\P\right)_{ij} & =\varepsilon_{jmn}\left(\ds\,\curl\,\P\right)_{in,m}=\varepsilon_{jmn}\left(\frac{1}{2}\,\varepsilon_{npq}\P_{iq,p}+\frac{1}{2}\,\varepsilon_{ipq}\P_{nq,p}-\frac{1}{3}\delta_{in}\left(\curl\,\P\right)_{kk}\right)_{,m}\nonumber \\
 & \quad=\frac{1}{2}\left(\varepsilon_{jmn}\varepsilon_{npq}\P_{iq,pm}+\varepsilon_{jmn}\varepsilon_{ipq}\P_{nq,pm}\right)-\frac{1}{3}\,\varepsilon_{jmi}\varepsilon_{kpq}\P_{kq,pm},\nonumber \\
\left(\curl\,\skew\,\curl\,\P\right)_{ij} & =\varepsilon_{jmn}\left(\skew\,\curl\,\P\right)_{in,m}=\varepsilon_{jmn}\frac{1}{2}\left(\varepsilon_{npq}\P_{iq,p}-\varepsilon_{ipq}\P_{nq,p}\right)_{,m}\\
 & \quad=\frac{1}{2}\left(\varepsilon_{jmn}\varepsilon_{npq}\P_{iq,pm}-\varepsilon_{jmn}\varepsilon_{ipq}\P_{nq,pm}\right),\nonumber \\
\left(\curl\left(\frac{1}{3}\tr\left(\curl\,P\right)\id\right)\right)_{ij} & =\varepsilon_{jmn}\left(\frac{1}{3}\tr\left(\curl\,P\right)\id\right)_{in,m}=\frac{1}{3}\varepsilon_{jmn}\left(\left(\curl\,P\right)_{kk}\delta_{in}\right)_{,m}\nonumber \\
 & \quad=\frac{1}{3}\varepsilon_{jmn}\delta_{in}\left(\varepsilon_{kpq}\P_{kq,pm}\right)=\frac{1}{3}\varepsilon_{jmi}\varepsilon_{kpq}\P_{kq,pm}.\nonumber 
\end{align}
We set 
\begin{equation}
\Delta:=\left(\alpha_{1}\,\curl\,\ds\,\curl\,\P+\alpha_{2}\,\curl\,\skew\,\curl\,\P+\alpha_{3}\,\curl\left(\frac{1}{3}\tr\left(\curl\,P\right)\id\right)\right).\label{eq:delta}
\end{equation}

\subsection*{Equations in $u_{,tt}$}

\[
\rho\,\u_{,tt}=\textrm{Div}\left[2\,\me\,\sym\left(\grad\u-\P\right)+\le\,\tr\left(\grad\u-\P\right)\id+2\,\mc\,\skew\left(\grad\u-\P\right)\right],
\]
This set of equations can be rewrite as follows:
\begin{align}
\rho\,\u_{,tt} & =\textrm{Div}\left[2\,\me\,\sym\left(\grad\u-\P\right)+\le\,\tr\left(\grad\u-\P\right)\id+2\,\mc\,\skew\left(\grad\u-\P\right)\right]\nonumber \\
 & =\textrm{Div}\left[2\,\me\,\dev\,\sym\left(\grad\u-\P\right)+\left(\frac{2}{3}\me+\le\right)\,\tr\left(\grad\u-\P\right)\id+2\,\mc\,\skew\left(\grad\u-\P\right)\right]\nonumber \\
 & =\textrm{Div}\left[2\,\me\,\dev\,\sym\,\grad\u+\left(\frac{2}{3}\me+\le\right)\,\tr\left(\grad\u\right)\id+2\,\mc\,\skew\left(\grad\u\right)\right]\\
 & \quad-\textrm{Div}\left[2\,\me\,\dev\,\sym\,P+\left(\frac{2}{3}\me+\le\right)\,\tr\left(P\right)\id+2\,\mc\,\skew\left(P\right)\right].\nonumber 
\end{align}
Remembering the definitions of $u_{\left(1,k\right)},P_{\left(1k\right)},P^{D},P^{S},u_{\left[1,k\right]}-P_{\left[1k\right]}$
given in (\ref{eq:symmetric components}),(\ref{eq:skew symmetric components}),(\ref{eq:ps}),
the scalar components of this vectorial system are therefore
\begin{align}
\rho\,\u_{1,tt} & =\textrm{div}\left(2\,\me\begin{pmatrix}\frac{1}{3}\left(2\,u_{1,1}-u_{2,2}-u_{3,3}\right)-P^{D}\\
u_{\left(1,2\right)}-P_{\left(12\right)}\\
u_{\left(1,3\right)}-P_{\left(13\right)}
\end{pmatrix}+\left(\frac{2}{3}\,\me+\le\right)\begin{pmatrix}\sum_{\alpha}u_{\alpha,\alpha}-3\,P^{S}\\
0\\
0
\end{pmatrix}+2\,\mc\begin{pmatrix}0\\
u_{\left[1,2\right]}-P_{\left[12\right]}\\
u_{\left[1,3\right]}-P_{\left[13\right]}
\end{pmatrix}\right),\nonumber \\
\nonumber \\
\rho\,\u_{2,tt} & =\textrm{div}\left(2\,\me\begin{pmatrix}u_{\left(1,2\right)}-P_{\left(12\right)}\\
\frac{1}{3}\left(2\,u_{2,2}-u_{1,1}-u_{3,3}\right)-P_{2}^{D}\\
u_{\left(2,3\right)}-P_{\left(23\right)}
\end{pmatrix}+\left(\frac{2}{3}\,\me+\le\right)\begin{pmatrix}0\\
\sum_{\alpha}u_{\alpha,\alpha}-3\,P^{S}\\
0
\end{pmatrix}+2\,\mc\begin{pmatrix}-u_{\left[1,2\right]}+P_{\left[12\right]}\\
0\\
u_{\left[1,3\right]}-P_{\left[13\right]}
\end{pmatrix}\right),\label{eq:divergences}\\
\rho\,\u_{3,tt} & =\textrm{div}\left(2\,\me\begin{pmatrix}u_{\left(1,3\right)}-P_{\left(13\right)}\\
u_{\left(2,3\right)}-P_{\left(23\right)}\\
\frac{1}{3}\left(2\,u_{1,1}-u_{2,2}-u_{3,3}\right)-P_{3}^{D}
\end{pmatrix}+\left(\frac{2}{3}\,\me+\le\right)\begin{pmatrix}0\\
0\\
\sum_{\alpha}u_{\alpha,\alpha}-3\,P^{S}
\end{pmatrix}+2\,\mc\begin{pmatrix}-u_{\left[1,3\right]}+P_{\left[13\right]}\\
-u_{\left[2,3\right]}+P_{\left[23\right]}\\
0
\end{pmatrix}\right)\,.\nonumber 
\end{align}
Thanks to the hypothesis of dependence only on $x_{1}$, we have that
$u_{2,2},u_{3,3},u_{3,2},u_{2,3}$ are zero and $u_{\left(1,2\right)}=\frac{1}{2}\,u_{2,1},u_{\left(1,3\right)}=\frac{1}{2}\,u_{3,1},$$\,u_{\left[1,2\right]}=-\frac{1}{2}\,u_{2,1},u_{\left[1,3\right]}=-\frac{1}{2}\,u_{3,1}$.
Thus the first equation of (\ref{eq:divergences}) becomes
\begin{align*}
\rho\,\u_{1,tt} & =\left(2\,\me\left(\frac{2}{3}\,u_{1,1}-P^{D}\right)+\left(\frac{2}{3}\,\me+\le\right)\left(u_{1,1}-3\,P^{S}\right)\right)_{,1}\\
 & \quad+\underbrace{\left(2\,\me\left(\frac{1}{2}\,u_{2,1}-P_{\left(12\right)}\right)+2\,\mc\left(-\frac{1}{2}\,u_{2,1}-P_{\left[12\right]}\right)\right)_{,2}}_{=0\quad\textrm{(plane wave)}}+\underbrace{\left(2\,\me\left(\frac{1}{2}\,u_{3,1}-P_{\left(13\right)}\right)+2\,\mc\left(-\frac{1}{2}\,u_{3,1}-P_{\left[13\right]}\right)\right)_{,3}}_{=0\quad\textrm{(plane wave)}}\\
 & =\left(2\,\me+\le\right)u_{1,11}-2\,\me\,P_{,1}^{D}-3\left(\frac{2}{3}\me+\le\right)P_{,1}^{S}.
\end{align*}
Dividing by $\rho$ and remembering the definition of $c_{p}$ given
in (\ref{eq:asintoti obliqui}) we can rewrite this equation as:
\begin{equation}
\u_{1,tt}=c_{p}^{2}\,u_{1,11}-\frac{2\,\me}{\rho}\,P_{,1}^{D}-\frac{2\,\me+3\,\le}{\rho}\,P_{,1}^{S}.
\end{equation}
Repeating the same calculation for the other two equations, and remembering
the definition of $c_{s}$ and $\omega_{r}$ given in (\ref{eq:asintoti obliqui})
and (\ref{eq:omega}), we find
\begin{align}
\u_{\xi,tt} & =c_{s}^{2}\,u_{\xi,11}-\frac{2\,\me}{\rho}\,P_{\left(1\xi\right),1}+\omega_{r}^{2}\,\frac{\eta_{2}}{\rho}\,P_{\left[1\xi\right],1},\qquad\xi\in\left\{ 2,3\right\} .
\end{align}

\subsection*{Equations in $\protect\dev\,\protect\sym\,\protect\P_{,tt}$ }

The PDEs system

\begin{align}
\eta_{\,1}\,\textrm{dev sym}\,\P_{,tt} & =2\,\me\,\ds\left(\grad\u-\P\right)-2\,\mh\,\ds\,\P\label{eq:dev sym}\\
 & \quad-\me\,L_{c}^{2}\,\ds\left(\alpha_{1}\,\curl\,\ds\,\curl\,\P+\alpha_{2}\,\curl\,\skew\,\curl\,\P+\alpha_{3}\,\curl\left(\frac{1}{3}\tr\left(\curl\,P\right)\id\right)\right),\nonumber 
\end{align}
has only five independent equations. Setting $\mathbf{Eq}_{1}$ for
the system (\ref{eq:dev sym}) of PDEs, the five independent equations
that we will take are $\left(\mathbf{Eq}_{1}\right)_{11},$$\left(\mathbf{Eq}_{1}\right)_{12},$$\left(\mathbf{Eq}_{1}\right)_{13},$$\left(\mathbf{Eq}_{1}\right)_{23}$
and $\left(\mathbf{Eq}_{1}\right)_{22}-\left(\mathbf{Eq}_{1}\right)_{33}.$
In order to find the desired PDEs, we need the following calculations:
we have
\begin{align}
\left(\sym\,\Delta\right)_{11}=\Delta_{11} & =\alpha_{1}\left[\frac{1}{2}\left(\varepsilon_{1mn}\varepsilon_{npq}\P_{1q,pm}+\varepsilon_{1mn}\varepsilon_{1pq}\P_{nq,pm}\right)-\frac{1}{3}\underbrace{\varepsilon_{1m1}}_{0}\varepsilon_{kpq}\P_{kq,pm}\right]\\
 & \quad+\alpha_{2}\frac{1}{2}\left(\varepsilon_{1mn}\varepsilon_{npq}\P_{1q,pm}-\varepsilon_{1mn}\varepsilon_{1pq}\P_{nq,pm}\right)+\alpha_{3}\frac{1}{3}\underbrace{\varepsilon_{1m1}}_{0}\varepsilon_{kpq}\P_{kq,pm}\,,\nonumber 
\end{align}
where $\Delta$ is defined in (\ref{eq:delta}), and remembering that
$\P_{hk,pq}=0$ for every $p,q\neq1$ we find 
\begin{align*}
\left(\sym\,\Delta\right)_{11} & =\alpha_{1}\,\frac{1}{2}\left(\varepsilon_{11n}\varepsilon_{n1q}\P_{1q,11}+\varepsilon_{11n}\varepsilon_{11q}\P_{nq,11}\right)+\alpha_{2}\,\frac{1}{2}\left(\varepsilon_{11n}\varepsilon_{n1q}\P_{1q,11}-\varepsilon_{11n}\varepsilon_{11q}\P_{nq,11}\right)\equiv0.
\end{align*}
For the trace we find
\begin{align}
\textrm{tr}\,\Delta & =\frac{1}{2}\left[\alpha_{1}\left(\varepsilon_{kmn}\varepsilon_{npq}\P_{kq,pm}+\varepsilon_{kmn}\varepsilon_{kpq}\P_{nq,pm}\right)+\alpha_{2}\left(\varepsilon_{kmn}\varepsilon_{npq}\P_{kq,pm}-\varepsilon_{kmn}\varepsilon_{kpq}\P_{nq,pm}\right)\right]\nonumber \\
 & =\frac{\alpha_{1}+\alpha_{2}}{2}\,\varepsilon_{kmn}\varepsilon_{npq}\P_{kq,pm}+\frac{\alpha_{1}-\alpha_{2}}{2}\,\varepsilon_{kmn}\varepsilon_{kpq}\P_{nq,pm}\nonumber \\
 & =\frac{\alpha_{1}+\alpha_{2}}{2}\,\varepsilon_{k1n}\varepsilon_{n1q}\P_{kq,11}+\frac{\alpha_{1}-\alpha_{2}}{2}\,\varepsilon_{k1n}\varepsilon_{k1q}\P_{nq,11}\nonumber \\
 & =\frac{\alpha_{1}+\alpha_{2}}{2}\,\left(\varepsilon_{213}\varepsilon_{312}\P_{22,11}+\varepsilon_{312}\varepsilon_{213}\P_{33,11}\right)+\frac{\alpha_{1}-\alpha_{2}}{2}\,\left(\varepsilon_{213}\varepsilon_{213}\P_{33,11}+\varepsilon_{312}\varepsilon_{312}\P_{22,11}\right)\nonumber \\
 & =\frac{\alpha_{1}+\alpha_{2}}{2}\,\left(-\P_{22,11}-\P_{33,11}\right)+\frac{\alpha_{1}-\alpha_{2}}{2}\,\left(\P_{33,11}+\P_{22,11}\right)\nonumber \\
 & =-\alpha_{2}\left(P_{22,11}+P_{33,11}\right)=-\alpha_{2}\left(2\,\P_{,11}^{S}-P_{,11}^{D}\right).\label{eq: trace delta}
\end{align}

\subsubsection*{Equation 1}

We have

\begin{align*}
\left(\dev\,\sym\,\Delta\right)_{11} & =\left(\sym\,\Delta\right)_{11}-\frac{1}{3}\,\textrm{tr}\,\Delta=\frac{\alpha_{2}}{3}\left(2\,\P_{,11}^{S}-P_{,11}^{D}\right).
\end{align*}
Thus the first equation is
\begin{align*}
\eta_{\,1}\P_{,tt}^{D} & =2\,\me\left(\frac{2}{3}\,u_{1,1}-\P^{D}\right)-2\,\mh P^{D}-\me L_{c}^{2}\,\frac{\alpha_{2}}{3}\,\left(2\,\P_{,11}^{S}-P_{,11}^{D}\right),
\end{align*}
and remembering the definitions of $\omega_{s}$ and $c_{m}$ given
in (\ref{eq:omega}) and (\ref{eq:asintoti obliqui}) we find
\begin{align}
\P_{,tt}^{D} & =\frac{4}{3}\,\frac{\me}{\eta_{\,1}}\,u_{1,1}+\frac{1}{3}\,\frac{\alpha_{2}}{\eta_{\,1}}\,\me L_{c}^{2}\,P_{,11}^{D}-\frac{2}{3}\,\frac{\alpha_{2}}{\eta_{\,1}}\,\me L_{c}^{2}\,\P_{,11}^{S}-\omega_{s}^{2}\,\P^{D}.
\end{align}

\subsubsection*{Equations 2,3}

\begin{align}
\eta_{\,1}\P_{\left(1\xi\right),tt} & =2\,\me\left(u_{\left(1,\xi\right)}-\P_{\left(1\xi\right)}\right)-2\,\mh\P_{\left(1\xi\right)}-\me L_{c}^{2}\,\Delta_{\left(1\xi\right)}\\
 & =\me\,u_{\xi,1}-2\left(\me+\mh\right)\P_{\left(1\xi\right)}-\me L_{c}^{2}\,\Delta_{\left(1\xi\right)}.\nonumber 
\end{align}
We have to calculate the term $\Delta_{\left(1\xi\right)}$:
\begin{align*}
\left(\curl\,\ds\,\curl\,\P\right)_{\left(1\xi\right)} & =\frac{1}{4}\left(\varepsilon_{\xi mn}\varepsilon_{npq}\P_{1q,pm}+\varepsilon_{\xi mn}\varepsilon_{1pq}\P_{nq,pm}+\varepsilon_{1mn}\varepsilon_{npq}\P_{\xi q,pm}+\varepsilon_{1mn}\varepsilon_{\xi pq}\P_{nq,pm}\right)\\
 & =\frac{1}{4}\left(\varepsilon_{\xi1n}\varepsilon_{n1q}\P_{1q,11}+\varepsilon_{\xi1n}\varepsilon_{11q}\P_{nq,11}+\varepsilon_{11n}\varepsilon_{n1q}\P_{\xi q,11}+\varepsilon_{11n}\varepsilon_{\xi1q}\P_{nq,11}\right)\\
 & =\frac{1}{4}\left(\varepsilon_{\xi1n}\varepsilon_{n1q}\P_{1q,11}\right)=-\frac{1}{4}\left(\P_{\left(1\xi\right),11}+\P_{\left[1\xi\right],11}\right),\\
\left(\curl\,\skew\,\curl\,\P\right)_{\left(1\xi\right)} & =-\frac{1}{4}\left(\P_{\left(1\xi\right),11}+\P_{\left[1\xi\right],11}\right),\\
\left(\curl\left(\frac{1}{3}\tr\left(\curl\,P\right)\id\right)\right)_{\left(1\xi\right)} & =\frac{1}{2}\left(\frac{1}{3}\varepsilon_{\xi m1}\varepsilon_{kpq}\P_{kq,pm}+\frac{1}{3}\varepsilon_{1m\xi}\varepsilon_{kpq}\P_{kq,pm}\right)=\frac{1}{2}\left(\frac{1}{3}\varepsilon_{\xi11}\varepsilon_{k1q}\P_{kq,11}+\frac{1}{3}\varepsilon_{11\xi}\varepsilon_{k1q}\P_{kq,11}\right)=0.
\end{align*}
So we have 
\begin{align}
\eta_{\,1}\P_{\left(1\xi\right),tt} & =2\,\me\left(u_{\left(1,\xi\right)}-\P_{\left(1\xi\right)}\right)-2\,\mh\P_{\left(1\xi\right)}+\me L_{c}^{2}\,\frac{\alpha_{1}+\alpha_{2}}{4}\left(\P_{\left(1\xi\right),11}+\P_{\left[1\xi\right],11}\right)\\
\P_{\left(1\xi\right),tt} & =\frac{\me}{\eta_{\,1}}\,u_{\left(1,\xi\right)}+\me L_{c}^{2}\,\frac{\alpha_{1}+\alpha_{2}}{4}\,\P_{\left(1\xi\right),11}+\me L_{c}^{2}\,\frac{\alpha_{1}+\alpha_{2}}{4}\,\P_{\left[1\xi\right],11}-\omega_{s}^{2}\P_{\left(1\xi\right)}.\nonumber 
\end{align}

\subsubsection*{Equation 4}

\begin{align}
\eta_{\,1}\P_{\left(23\right),tt} & =2\,\me\left(u_{\left(2,3\right)}-\P_{\left(23\right)}\right)-2\,\mh\P_{\left(23\right)}-\me L_{c}^{2}\,\Delta_{\left(23\right)}\\
 & =-2\left(\me+\mh\right)\P_{\left(23\right)}-\me L_{c}^{2}\,\Delta_{\left(23\right)}.\nonumber 
\end{align}
We have to calculate the term $\Delta_{\left(23\right)}$:
\begin{align*}
\left(\curl\,\ds\,\curl\,\P\right)_{\left(23\right)} & =\frac{1}{4}\left(\varepsilon_{3mn}\varepsilon_{npq}\P_{2q,pm}+\varepsilon_{3mn}\varepsilon_{2pq}\P_{nq,pm}+\varepsilon_{2mn}\varepsilon_{npq}\P_{3q,pm}+\varepsilon_{2mn}\varepsilon_{3pq}\P_{nq,pm}\right)\\
 & =\frac{1}{4}\left(\varepsilon_{31n}\varepsilon_{n1q}\P_{2q,11}+\varepsilon_{31n}\varepsilon_{21q}\P_{nq,11}+\varepsilon_{21n}\varepsilon_{n1q}\P_{3q,11}+\varepsilon_{21n}\varepsilon_{31q}\P_{nq,11}\right)\\
 & =\frac{1}{4}\left(\varepsilon_{312}\varepsilon_{213}\P_{23,11}+\varepsilon_{312}\varepsilon_{213}\P_{23,11}+\varepsilon_{213}\varepsilon_{312}\P_{32,11}+\varepsilon_{213}\varepsilon_{312}\P_{32,11}\right)\\
 & =-\frac{1}{2}\left(\P_{23,11}+\P_{32,11}\right)=-\P_{\left(23\right),11},\\
\left(\curl\,\skew\,\curl\,\P\right)_{\left(23\right)} & =\frac{1}{4}\left(\varepsilon_{312}\varepsilon_{213}\P_{23,11}-\varepsilon_{312}\varepsilon_{213}\P_{23,11}+\varepsilon_{213}\varepsilon_{312}\P_{32,11}-\varepsilon_{213}\varepsilon_{312}\P_{32,11}\right)=0,\\
\left(\curl\left(\frac{1}{3}\tr\left(\curl\,P\right)\id\right)\right)_{\left(23\right)} & =\frac{1}{2}\left(\frac{1}{3}\varepsilon_{3m2}\varepsilon_{kpq}\P_{kq,pm}+\frac{1}{3}\varepsilon_{2m3}\varepsilon_{kpq}\P_{kq,pm}\right)=\frac{1}{2}\left(\frac{1}{3}\varepsilon_{312}\varepsilon_{k1q}\P_{kq,11}+\frac{1}{3}\varepsilon_{213}\varepsilon_{k1q}\P_{kq,11}\right)=0.
\end{align*}
Thus we have 
\begin{align}
\eta_{\,1}\P_{\left(23\right),tt} & =-2\left(\me+\mh\right)\P_{\left(23\right)}+\me L_{c}^{2}\,\alpha_{1}\P_{\left(23\right),11}\nonumber \\
\P_{\left(23\right),tt} & =-\omega_{s}^{2}\,\P_{\left(23\right)}+\left(c_{\textrm{m}}^{\textrm{d}}\right)^{2}\,\P_{\left(23\right),11}\,.
\end{align}

\subsubsection*{Equation 5}

We have to determine $\left(\mathbf{Eq}_{1}\right)_{22}-\left(\mathbf{Eq}_{1}\right)_{33}.$
So

\begin{align*}
\eta_{\,1}\P_{2,tt}^{D} & =2\,\me\left(u_{2,2}-\frac{1}{3}\,u_{k,k}-\P_{2}^{D}\right)-2\,\mh\P_{2}^{D}-\me L_{c}^{2}\left(\Delta_{22}-\frac{1}{3}\,\textrm{tr}\left(\Delta\right)\right)\\
 & =\frac{2}{3}\,\me u_{1,1}-\left(2\,\me+2\,\mh\right)\P_{2}^{D}-\me L_{c}^{2}\left(\Delta_{22}-\frac{1}{3}\,\textrm{tr}\left(\Delta\right)\right),
\end{align*}
and 
\begin{align*}
\eta_{\,1}\P_{3,tt}^{D} & =2\,\me\left(u_{3,3}-\frac{1}{3}\,u_{k,k}-\P_{3}^{D}\right)-2\,\mh\P_{3}^{D}-\me L_{c}^{2}\left(\Delta_{33}-\frac{1}{3}\,\textrm{tr}\left(\Delta\right)\right)\\
 & =\frac{2}{3}\,\me u_{1,1}-\left(2\,\me+2\,\mh\right)\P_{3}^{D}-\me L_{c}^{2}\left(\Delta_{33}-\frac{1}{3}\,\textrm{tr}\left(\Delta\right)\right).
\end{align*}
Thus, for $\left(\mathbf{Eq}_{1}\right)_{22}-\left(\mathbf{Eq}_{1}\right)_{33}$
we find 
\begin{align}
\eta_{\,1}\P_{,tt}^{V} & =-\left(2\,\me+2\,\mh\right)\P^{V}-\me L_{c}^{2}\left(\Delta_{22}-\Delta_{33}\right)
\end{align}
and having that 
\begin{align}
\left(\curl\,\ds\,\curl\,\P\right)_{22} & =\frac{1}{2}\left(\varepsilon_{2mn}\varepsilon_{npq}\P_{2q,pm}+\varepsilon_{2mn}\varepsilon_{2pq}\P_{nq,pm}\right)-\frac{1}{3}\,\varepsilon_{2m2}\varepsilon_{kpq}\P_{kq,pm}\nonumber \\
 & =\frac{1}{2}\left(\varepsilon_{21n}\varepsilon_{n1q}\P_{2q,11}+\varepsilon_{21n}\varepsilon_{21q}\P_{nq,11}\right)\nonumber \\
 & =\frac{1}{2}\left(\varepsilon_{213}\varepsilon_{312}\P_{22,11}+\varepsilon_{213}\varepsilon_{213}\P_{33,11}\right)\nonumber \\
 & =\frac{1}{2}\left(-\P_{22,11}+\P_{33,11}\right)=-\frac{1}{2}\P_{,11}^{V}\\
 & =-\left(\curl\,\ds\,\curl\,\P\right)_{33},\nonumber \\
\left(\curl\,\skew\,\curl\,\P\right)_{22} & =\frac{1}{2}\left(\varepsilon_{2mn}\varepsilon_{npq}\P_{2q,pm}-\varepsilon_{2mn}\varepsilon_{2pq}\P_{nq,pm}\right)\nonumber \\
 & =\frac{1}{2}\left(\varepsilon_{213}\varepsilon_{312}\P_{22,11}-\varepsilon_{213}\varepsilon_{213}\P_{33,11}\right)\nonumber \\
 & =-\frac{1}{2}\left(\P_{22,11}+\P_{33,11}\right)=\left(\curl\,\skew\,\curl\,\P\right)_{33},\nonumber \\
\left(\curl\left(\frac{1}{3}\tr\left(\curl\,P\right)\id\right)\right)_{22} & =0=\left(\curl\left(\frac{1}{3}\tr\left(\curl\,P\right)\id\right)\right)_{33}.\nonumber 
\end{align}
Thus we have
\begin{align}
\eta_{\,1}\P_{,tt}^{V} & =-\left(2\,\me+2\,\mh\right)\P^{V}+\me L_{c}^{2}\,\alpha_{1}\,\P_{,11}^{V}\nonumber \\
\P_{,tt}^{V} & =-\omega_{s}^{2}\P^{V}+\left(c_{\textrm{m}}^{d}\right)^{2}\P_{,11}^{V}.
\end{align}

\subsection*{Equations in $\protect\skew\,\protect\P_{,tt}$}

The PDEs system

\begin{align*}
\eta_{\,2}\,\textrm{skew}\,\P_{,tt} & =2\,\mc\,\skew\left(\grad\u-\P\right)\\
 & \quad-\me\,L_{c}^{2}\,\skew\left(\alpha_{1}\,\curl\,\ds\,\curl\,\P+\alpha_{2}\,\curl\,\skew\,\curl\,\P+\alpha_{3}\,\curl\left(\frac{1}{3}\tr\left(\curl\,P\right)\id\right)\right),
\end{align*}
has only three independent equations.

\subsubsection*{Equations 1,2}

\begin{align}
\eta_{\,2}\P_{\left[1\xi\right],tt} & =2\,\mc\,\left(\u_{\left[1,\xi\right]}-\P_{\left[1\xi\right]}\right)-\me\,L_{c}^{2}\,\Delta_{\left[1\xi\right]}\\
 & =2\,\mc\,\left(-\u_{\xi,1}-\P_{\left[1\xi\right]}\right)-\me\,L_{c}^{2}\,\Delta_{\left[1\xi\right]},\qquad\xi\in\left\{ 2,3\right\} \nonumber 
\end{align}
we have to calculate the term $\Delta_{\left[1\xi\right]}$:
\begin{align*}
\left(\curl\,\ds\,\curl\,\P\right)_{\left[1\xi\right]} & =\frac{1}{4}\left(\varepsilon_{\xi mn}\varepsilon_{npq}\P_{1q,pm}+\varepsilon_{\xi mn}\varepsilon_{1pq}\P_{nq,pm}-\varepsilon_{1mn}\varepsilon_{npq}\P_{\xi q,pm}-\varepsilon_{1mn}\varepsilon_{\xi pq}\P_{nq,pm}\right)\\
 & =\frac{1}{4}\left(\varepsilon_{\xi1n}\varepsilon_{n1q}\P_{1q,11}+\varepsilon_{\xi1n}\varepsilon_{11q}\P_{nq,11}-\varepsilon_{11n}\varepsilon_{n1q}\P_{\xi q,11}-\varepsilon_{11n}\varepsilon_{\xi1q}\P_{nq,11}\right)\\
 & =\frac{1}{4}\left(\varepsilon_{\xi1n}\varepsilon_{n1q}\P_{1q,11}\right)=-\frac{1}{4}\left(\P_{\left(1\xi\right),11}+\P_{\left[1\xi\right],11}\right),\\
\left(\curl\,\skew\,\curl\,\P\right)_{\left[1\xi\right]} & =-\frac{1}{4}\left(\P_{\left(1\xi\right),11}+\P_{\left[1\xi\right],11}\right),\\
\left(\curl\left(\frac{1}{3}\tr\left(\curl\,P\right)\id\right)\right)_{\left[1\xi\right]} & =\frac{1}{2}\left(\frac{1}{3}\varepsilon_{\xi m1}\varepsilon_{kpq}\P_{kq,pm}-\frac{1}{3}\varepsilon_{1m\xi}\varepsilon_{kpq}\P_{kq,pm}\right)=\frac{1}{2}\left(\frac{1}{3}\varepsilon_{\xi11}\varepsilon_{k1q}\P_{kq,11}-\frac{1}{3}\varepsilon_{11\xi}\varepsilon_{k1q}\P_{kq,11}\right)=0.
\end{align*}
So we have

\begin{align*}
\eta_{\,2}\P_{\left[1\xi\right],tt} & =2\,\mc\,\left(\u_{\left[1,\xi\right]}-\P_{\left[1\xi\right]}\right)-\me\,L_{c}^{2}\,\Delta_{\left[1\xi\right]}\\
\P_{\left[1\xi\right],tt} & =-\frac{1}{2}\,\omega_{r}^{2}\,\u_{\xi,1}-\omega_{r}^{2}\,\P_{\left[1\xi\right]}+\me\,L_{c}^{2}\,\frac{\alpha_{1}+\alpha_{2}}{4\,\eta_{2}}\left(\P_{\left(1\xi\right),11}+\P_{\left[1\xi\right],11}\right).
\end{align*}

\subsubsection*{Equation 3}

\begin{align}
\eta_{\,2}\P_{\left[23\right],tt} & =2\,\mc\left(u_{\left[2,3\right]}-\P_{\left[23\right]}\right)-\me L_{c}^{2}\,\Delta_{\left[23\right]}=-2\,\mc\,\P_{\left[23\right]}-\me L_{c}^{2}\,\Delta_{\left[23\right]}.
\end{align}
We have to calculate the term $\Delta_{\left[23\right]}$:
\begin{align*}
\left(\curl\,\ds\,\curl\,\P\right)_{\left[23\right]} & =\frac{1}{4}\left(\varepsilon_{3mn}\varepsilon_{npq}\P_{2q,pm}+\varepsilon_{3mn}\varepsilon_{2pq}\P_{nq,pm}-\varepsilon_{2mn}\varepsilon_{npq}\P_{3q,pm}-\varepsilon_{2mn}\varepsilon_{3pq}\P_{nq,pm}\right)\\
 & \quad-\frac{1}{3}\,\varepsilon_{3m2}\varepsilon_{kpq}\P_{kq,pm}\\
 & =\frac{1}{4}\left(\varepsilon_{31n}\varepsilon_{n1q}\P_{2q,11}+\varepsilon_{31n}\varepsilon_{21q}\P_{nq,11}-\varepsilon_{21n}\varepsilon_{n1q}\P_{3q,11}-\varepsilon_{21n}\varepsilon_{31q}\P_{nq,11}\right)\\
 & \quad-\frac{1}{3}\,\varepsilon_{312}\varepsilon_{k1q}\P_{kq,11}\\
 & =\frac{1}{4}\left(\varepsilon_{312}\varepsilon_{213}\P_{23,11}+\varepsilon_{312}\varepsilon_{213}\P_{23,11}-\varepsilon_{213}\varepsilon_{312}\P_{32,11}-\varepsilon_{213}\varepsilon_{312}\P_{32,11}\right)\\
 & \quad-\frac{1}{3}\left(\varepsilon_{213}\P_{23,11}+\varepsilon_{312}\P_{32,11}\right)\\
 & =\frac{1}{2}\left(-\P_{23,11}+\P_{32,11}\right)-\frac{1}{3}\left(-\P_{23,11}+\P_{32,11}\right)=-\frac{1}{3}\P_{\left[23\right],11},\\
\left(\curl\,\skew\,\curl\,\P\right)_{\left[23\right]} & =\frac{1}{4}\left(\varepsilon_{312}\varepsilon_{213}\P_{23,11}-\varepsilon_{312}\varepsilon_{213}\P_{23,11}+\varepsilon_{213}\varepsilon_{312}\P_{32,11}-\varepsilon_{213}\varepsilon_{312}\P_{32,11}\right)=0,\\
\left(\curl\left(\frac{1}{3}\tr\left(\curl\,P\right)\id\right)\right)_{\left[23\right]} & =\frac{1}{2}\left(\frac{1}{3}\varepsilon_{3m2}\varepsilon_{kpq}\P_{kq,pm}-\frac{1}{3}\varepsilon_{2m3}\varepsilon_{kpq}\P_{kq,pm}\right)=\frac{1}{2}\left(\frac{1}{3}\varepsilon_{312}\varepsilon_{k1q}\P_{kq,11}-\frac{1}{3}\varepsilon_{213}\varepsilon_{k1q}\P_{kq,11}\right)\\
 & =\frac{1}{6}\left(\varepsilon_{213}\P_{23,11}+\varepsilon_{312}\P_{32,11}+\varepsilon_{213}\P_{23,11}+\varepsilon_{312}\P_{32,11}\right)=-\frac{2}{3}\P_{\left[23\right],11}.
\end{align*}
So we have 
\begin{align*}
\eta_{\,2}\P_{\left[23\right],tt} & =-2\,\mc\,\P_{\left[23\right]}+\me L_{c}^{2}\,\left(\frac{\alpha_{1}+2\,\alpha_{3}}{3}\right)\P_{\left[23\right],11}\\
\P_{\left[23\right],tt} & =-\omega_{r}^{2}\,\P_{\left[23\right]}+\left(c_{\mathrm{m}}^{\mathrm{vd}}\right)^{2}\P_{\left[23\right],11}.
\end{align*}

\subsection*{Equations in the spherical part of $\protect\P_{,tt}$}

\begin{align*}
\frac{1}{3}\,\eta_{\,3}\,\textrm{tr}\left(\P_{,tt}\right) & {\displaystyle \,\,=\left(\frac{2}{3}\,\me+\le\right)\,\tr\left(\grad\u-\P\right)-\left(\frac{2}{3}\,\mh+\lh\right)\,\tr\left(\P\right)}\\
 & \quad-\me\,L_{c}^{2}\,\frac{1}{3}\tr\left(\alpha_{1}\,\curl\,\ds\,\curl\,\P+\alpha_{2}\,\curl\,\skew\,\curl\,\P+\frac{\alpha_{3}}{3}\,\curl\left(\tr\left(\curl\,P\right)\id\right)\right).
\end{align*}
Considering the expression (\ref{eq: trace delta}) for the $\textrm{tr}\left(\Delta\right)$,
we have 
\begin{align*}
\P_{,tt}^{S} & =\frac{2\,\me+3\,\le}{3\,\eta_{3}}\,u_{1,1}-\omega_{p}^{2}\,P^{S}+\frac{\me\,L_{c}^{2}\,\alpha_{2}}{\eta_{3}}\left(\frac{2}{3}\,\P_{,11}^{S}-\frac{1}{3}P_{,11}^{D}\right).
\end{align*}

\subsection{Determination of slopes of the acoustic branches}

We want to evaluate the first and second derivative in 0 of the expressions
\[
{\textstyle \det}\,\mathsf{E}_{\a}\left(k,\widehat{\omega}{}_{\textrm{aco};\alpha}\left(k\right)\right)=\sum_{p,q=1}^{3}\psi_{pq}^{\left(\alpha\right)}\left(\boldsymbol{m}\right)k^{2p}\widehat{\omega}_{\textrm{aco};\alpha}^{2q}\left(k\right)+\sum_{p=1}^{3}\varphi_{p}^{\left(\a\right)}\left(\boldsymbol{m}\right)k^{2p}+\sum_{q=1}^{3}\zeta_{q}^{\left(\a\right)}\left(\boldsymbol{m}\right)\widehat{\omega}_{\textrm{aco};\alpha}^{2q}\left(k\right)+\sigma^{\left(\a\right)}\left(\boldsymbol{m}\right),
\]
with $\a\in\left\{ 1,2\right\} $. In order to work with a more readable
notation, in what follows we suppress the dependence by $\boldsymbol{m}$
and $k$ of the relative functions. The derivative of $\widehat{\omega}{}_{\textrm{aco};\alpha}$
with respect to $k$ is denoted by $\widehat{\omega}'_{\textrm{aco};\alpha}$
. We have
\begin{align*}
\frac{d}{dk}\,{\textstyle \det}\,\mathsf{E}_{\a}\left(k,\widehat{\omega}{}_{\textrm{aco};\alpha}\right) & =\sum_{p,q=1}^{3}\psi_{pq}^{\left(\a\right)}\left(2p\,k^{2p-1}\widehat{\omega}_{\textrm{aco};\alpha}^{2q}+2q\,k^{2p}\widehat{\omega}_{\textrm{aco};\alpha}^{2q-1}\,\widehat{\omega}'_{\textrm{aco};\a}\right)\\
 & \quad+\sum_{p=1}^{3}2p\,\varphi_{p}^{\left(\a\right)}k^{2p-1}+\sum_{q=1}^{3}2q\,\zeta_{q}^{\left(\a\right)}\widehat{\omega}_{\textrm{aco};\alpha}^{2q-1}\,\widehat{\omega}'_{\textrm{aco};\a},
\end{align*}
and so, remembering that $\widehat{\omega}_{\textrm{aco};\alpha}\left(0\right)=0$,
this condition does not give any information on the value of $\widehat{\omega}'_{\textrm{aco};\alpha}\left(0\right)$.
For this reason, we continue with the second derivative. We compute
separately the derivative of the four terms of the first derivative
of ${\textstyle \det}\,\mathsf{E}_{\a}\left(k,\widehat{\omega}{}_{\textrm{aco};\alpha}\left(k\right)\right)$
obtaining:
\begin{align*}
\frac{d}{dk}\,\sum_{p,q=1}^{3}2p\,\psi_{pq}^{\left(\a\right)}k^{2p-1}\widehat{\omega}_{\textrm{aco};\alpha}^{2q} & =\sum_{p,q=1}^{3}2p\,\psi_{pq}^{\left(\a\right)}\left(\left(2p-1\right)k^{2p-2}\widehat{\omega}_{\textrm{aco};\alpha}^{2q}+2q\,k^{2p-1}\widehat{\omega}_{\textrm{aco};\alpha}^{2q-1}\widehat{\omega}'_{\textrm{aco};\a}\right),\\
\frac{d}{dk}\,\sum_{p,q=1}^{3}2q\,\psi_{pq}^{\left(\a\right)}k^{2p}\widehat{\omega}_{\textrm{aco};\alpha}^{2q-1}\,\widehat{\omega}'_{\textrm{aco};\a} & =\sum_{p,q=1}^{3}2q\,\psi_{pq}^{\left(\alpha\right)}\left(2p\,k^{2p-1}\widehat{\omega}_{\textrm{aco};\alpha}^{2q-1}\,\widehat{\omega}'_{\textrm{aco};\a}+k^{2p}\left(\left(2q-1\right)\widehat{\omega}_{\textrm{aco};\alpha}^{2q-2}\,\left(\widehat{\omega}'_{\textrm{aco};\a}\right)^{2}+\widehat{\omega}_{\textrm{aco};\alpha}^{2q-1}\,\widehat{\omega}''_{\textrm{aco};\a}\right)\right),\\
\frac{d}{dk}\,\sum_{p=1}^{3}2p\,\varphi_{p}^{\left(\a\right)}k^{2p-1} & =\sum_{p=2}^{3}2p\left(2p-1\right)\varphi_{p}^{\left(\a\right)}k^{2p-2}+2\varphi_{1}^{\left(\a\right)},\\
\frac{d}{dk}\,\sum_{q=1}^{3}2q\,\zeta_{q}^{\left(\a\right)}\widehat{\omega}_{\textrm{aco};\alpha}^{2q-1}\,\widehat{\omega}'_{\textrm{aco};\a} & =\sum_{p=2}^{3}2p\,\zeta_{q}^{\left(\a\right)}\left(\left(2q-1\right)\widehat{\omega}_{\textrm{aco};\alpha}^{2q-2}\,\left(\widehat{\omega}'_{\textrm{aco};\a}\right)^{2}+\widehat{\omega}_{\textrm{aco};\alpha}^{2q-1}\,\widehat{\omega}''_{\textrm{aco};\a}\right)+2\,\zeta_{1}^{\left(\a\right)}\left(\left(\widehat{\omega}'_{\textrm{aco};\a}\right)^{2}+\widehat{\omega}_{\textrm{aco};\alpha}\,\widehat{\omega}''_{\textrm{aco};\a}\right).
\end{align*}
Thus
\begin{equation}
0=\frac{d^{2}}{dk^{2}}\,\left.{\textstyle \det}\,\mathsf{E}_{\a}\left(k,\widehat{\omega}{}_{\textrm{aco};\alpha}\left(k\right)\right)\right|_{k=0}=2\,\zeta_{1}^{\left(\a\right)}\left(\widehat{\omega}'_{\textrm{aco};\a}\left(0\right)\right)^{2}+2\,\varphi_{1}^{\left(\a\right)}.\label{eq: tang in 0}
\end{equation}
Solving the equations (\ref{eq: tang in 0}) we find
\begin{gather*}
2\,\zeta_{2}^{\left(\a\right)}\left(\widehat{\omega}'_{\textrm{aco};\a}\left(0\right)\right)^{2}+2\,\varphi_{1}^{\left(\a\right)}=0\qquad\Longleftrightarrow\qquad\left(\widehat{\omega}'_{\textrm{aco};\a}\left(0\right)\right)^{2}=\frac{-\,\varphi_{1}^{\left(\a\right)}}{\zeta_{1}^{\left(\a\right)}}
\end{gather*}
and so, considering only the positive roots, 
\begin{gather*}
\widehat{\omega}'_{\textrm{aco};1}\left(0\right)=\sqrt{\frac{-\,\varphi_{1}^{\left(1\right)}}{\zeta_{1}^{\left(1\right)}}},\qquad\textrm{and}\qquad\widehat{\omega}'_{\textrm{aco};2}\left(0\right)=\sqrt{\frac{-\,\varphi_{1}^{\left(2\right)}}{\zeta_{1}^{\left(2\right)}}}.
\end{gather*}

\subsection{Derivation of strong equations for the Cosserat model and indeterminate
couple stress model}

In this appendix we give the strong field equations for the Cosserat
model and the indeterminate couple stress model.

\subsubsection*{Cosserat model}

The potential weighted energy for the Cosserat model is the following
\begin{align}
W_{\textrm{cos}} & =\me\left\Vert \sym\,\nabla u\right\Vert ^{2}+\mc\left\Vert \skew\left(\nabla u-\P\right)\right\Vert ^{2}+\frac{\le}{2}\left(\textrm{tr}\,\nabla u\right)^{2}\label{eq:coss1-1}\\
 & \quad+\mu_{e}\,\frac{L_{c}^{2}}{2}\left(\alpha_{1}\left\Vert \ds\,\curl\,\skew\,\P\right\Vert ^{2}+\alpha_{2}\left\Vert \skew\,\curl\,\skew\,\P\right\Vert ^{2}+\frac{1}{3}\,\alpha_{3}\left(\tr\,\curl\,\skew\,\P\right)^{2}\right).\nonumber 
\end{align}
The first variation is computed exactly in the same way as in (\ref{eq:PDE system}),
giving the following system of PDEs:
\begin{align}
\rho\,u_{,tt} & =\textrm{Div}\left[2\,\me\,\sym\,\nabla u+\le\,\textrm{tr}\left(\nabla u\right)\id+2\,\mc\,\skew\left(\nabla u-P\right)\right],\label{eq:Cosserat strong}\\
\eta_{2}\,\skew\,\P_{,tt} & =-\me\,L_{c}^{2}\,\skew\,\curl\left[\alpha_{1}\dev\,\sym\,\curl\,\skew\,\P+\alpha_{2}\,\skew\,\curl\,\skew\,\P+\frac{\alpha_{3}}{3}\,\textrm{tr}\left(\curl\,\skew\,\P\right)\id\right]\nonumber \\
 & \quad+2\,\mc\,\skew\left(\nabla u-P\right).\nonumber 
\end{align}
The system of PDEs (\ref{eq:Cosserat strong}) is not really suitable
for the numerical implementation. Working only with the skew symmetric
part of the micro-distortion tensor $\P$, it is more convenient to
see $\skew\,P$ like a vector thanks to the identification of the
Lie algebra $\so$ with $\R^{3}$ by means of the $\textrm{axl}$-operator.
In this way we can work only with the six independent equations of
the system (\ref{eq:Cosserat strong}). We remember that for
\[
A=\begin{pmatrix}0 & -a_{3} & a_{2}\\
a_{3} & 0 & -a_{1}\\
-a_{2} & a_{1} & 0
\end{pmatrix}\in\so
\]
the $\textrm{axl}:\so\fr\R^{3}$ operator is defined as follows:
\[
\textrm{axl}\,A:=\left(a_{1},a_{2},a_{3}\right),\qquad\left(\textrm{axl}\,A\right)_{k}=-\frac{1}{2}\varepsilon_{ijk}A_{ij}.
\]
Thanks to the identities \cite{neff2008curl}
\begin{align*}
-\curl\,A & =\left(\nabla\axl\,A\right)^{T}-\textrm{tr}\left[\left(\nabla\axl\,A\right)^{T}\right]\id,\\
\nabla\axl\,A & =-\left(\curl\,A\right)^{T}+\frac{1}{2}\,\textrm{tr}\left[\left(\curl\,A\right)^{T}\right]\id,
\end{align*}
verified for every $A\in\so$, the system of PDEs (\ref{eq:Cosserat strong})
can be rewritten in a completely equivalent form as
\begin{align*}
\rho\,u_{,tt} & =\textrm{Div}\left[2\,\me\,\sym\,\nabla u+\le\,\textrm{tr}\left(\nabla u\right)\id+2\,\mc\,\skew\left(\nabla u-P\right)\right],\\
\eta_{2}\,\left(\axl\,\skew\,\P\right)_{,tt} & =\me\,L_{c}^{2}\,\textrm{Div}\left[\frac{\alpha_{1}}{2}\,\dev\,\sym\,\nabla\left(\axl\,\skew\,\P\right)+\frac{\alpha_{2}}{2}\,\skew\,\nabla\left(\axl\,\skew\,\P\right)+\frac{2\,\alpha_{3}}{3}\,\textrm{tr}\left(\nabla\left(\axl\,\skew\,\P\right)\right)\id\right]\\
 & \quad+2\,\mc\,\axl\,\skew\left(\nabla u-P\right).
\end{align*}

\subsubsection*{Indeterminate couple stress model}

The potential weighted energy for the indeterminate couple stress
model is the following \cite{ghiba2016variant} 
\begin{align}
W_{\textrm{ind}\,1}\left(\sym\,\nabla u,\curl\,\sym\,\grad u\right) & =\me\left\Vert \sym\,\nabla u\right\Vert ^{2}+\frac{\le}{2}\left(\textrm{tr}\,\nabla u\right)^{2}\label{eq:energy ind 1}\\
 & \quad+\me\,\frac{L_{c}^{2}}{2}\left(\alpha_{1}\left\Vert \dev\,\sym\,\curl\,\sym\,\nabla u\right\Vert ^{2}+\alpha_{2}\left\Vert \skew\,\curl\,\sym\,\nabla u\right\Vert ^{2}\right).\nonumber 
\end{align}
The first variation of this energy density gives the following system
of PDEs:
\begin{align*}
\rho\,u_{,tt} & =\textrm{Div}\left[2\,\me\,\sym\,\nabla u+\le\,\textrm{tr}\left(\nabla u\right)\id\right]\\
 & \quad+\textrm{Div}\left[\me\,L_{c}^{2}\,\sym\,\curl\,\left(2\,\alpha_{1}\,\dev\,\sym\,\curl\,\sym\nabla u+2\,\alpha_{2}\,\skew\,\curl\,\sym\nabla u\right)\right].
\end{align*}
The problem can be completely reformulated in terms of the gradient
of the skew symmetric part of $\grad u$. Indeed, thanks to the equivalence
\[
\grad\left(\axl\,\skew\,\grad u\right)=\left(\curl\,\sym\,\nabla u\right)^{T},
\]
the following energy density
\begin{align}
W_{\textrm{ind}\,2}\left(\sym\,\nabla u,\grad\left(\axl\,\skew\,\grad u\right)\right) & =\me\left\Vert \sym\,\nabla u\right\Vert ^{2}+\frac{\le}{2}\left(\textrm{tr}\,\nabla u\right)^{2}\label{eq:energy ind 2}\\
 & \quad+\me\,\frac{L_{c}^{2}}{2}\left(\alpha_{1}\left\Vert \dev\,\sym\,\grad\left(\axl\,\skew\,\grad u\right)\right\Vert ^{2}+\alpha_{2}\left\Vert \skew\,\grad\left(\axl\,\skew\,\grad u\right)\right\Vert ^{2}\right).\nonumber 
\end{align}
 is completely equivalent to (\ref{eq:energy ind 1}). The associated
system of Euler- Lagrange equations is

\begin{align*}
\rho\,u_{,tt} & =\textrm{Div}\left[2\,\me\,\sym\,\nabla u+\le\,\textrm{tr}\left(\nabla u\right)\id\right]\\
 & \quad-\textrm{Div}\left[\me\,L_{c}^{2}\,\textrm{anti}\left\{ \textrm{Div}\,\left(\alpha_{1}\,\dev\,\sym\,\grad\left(\axl\,\skew\,\grad u\right)+\alpha_{2}\,\skew\,\grad\left(\axl\,\skew\,\grad u\right)\right)\right\} \right],
\end{align*}
where $\textrm{anti}$ is the inverse operator of $\axl$ defined
as follows
\[
\textrm{anti}:\R^{3}\fr\so,\qquad\left(\textrm{anti}\left(u\right)\right)_{ij}=-\varepsilon_{ijk}u_{k}.
\]

{\footnotesize{}\bibliographystyle{plain}
\bibliography{biblio}
}{\footnotesize \par}
\end{document}